\documentclass[acmlarge]{acmart}%\settopmatter{printfolios=false,printccs=false,printacmref=false}
\setcopyright{none}
\usepackage{booktabs}   %% For formal tables:
\usepackage{subcaption} %% For complex figures with subfigures/subcaptions
\usepackage{amsmath}

\usepackage{amssymb}
\usepackage{tikz-cd}
\usepackage{ifthen}
\usepackage{xspace}
\usepackage{xparse}
\usepackage{xstring}
\usepackage{mathpartir}
\usepackage{cleveref}
\usepackage{enumitem}
\usepackage{stmaryrd}
\usepackage{multicol}
\usepackage{tikz-cd}
\usepackage{mathtools}
\usepackage{stackengine}
\usepackage{scalerel}
\usepackage{wasysym}
\usepackage{fontawesome}
\usepackage{xcolor}
\usepackage{tcolorbox}
\usepackage{accents}
\usepackage{ulem}
\usepackage{multirow}
\usepackage{array,longtable}
\usepackage{turnstile}

\usetikzlibrary{matrix}
\usetikzlibrary{calc}
\usetikzlibrary{positioning}
\usetikzlibrary{decorations.pathreplacing}


\DeclareFontFamily{OT1}{pzc}{}
\DeclareFontShape{OT1}{pzc}{m}{it}{<-> s * [1.10] pzcmi7t}{}
\DeclareMathAlphabet{\mathpzc}{OT1}{pzc}{m}{it}

\crefname{figure}{fig.}{fig.}
\Crefname{figure}{Fig.}{Fig.}
\Crefname{section}{Sec.}{Sec.}

\definecolor{light-gray}{gray}{0.85}

\allowdisplaybreaks

\DeclareUnicodeCharacter{21D2}{$\Rightarrow$}
\DeclareUnicodeCharacter{2113}{$\mathpzc{l}$}
\DeclareUnicodeCharacter{2248}{$\approx$}
\DeclareUnicodeCharacter{2200}{$\forall$}
\DeclareUnicodeCharacter{2218}{$\circ$}
\DeclareUnicodeCharacter{2294}{$\sqcup$}
\DeclareUnicodeCharacter{393}{$\Gamma$}
\DeclareUnicodeCharacter{3A8}{$\Psi$}
\DeclareUnicodeCharacter{3B1}{$\alpha$}
\DeclareUnicodeCharacter{3B2}{$\beta$}
\DeclareUnicodeCharacter{3B7}{$\eta$}
\DeclareUnicodeCharacter{3BB}{$\lambda$}
\DeclareUnicodeCharacter{3BC}{$\mu$}
\DeclareUnicodeCharacter{2080}{$_0$}
\DeclareUnicodeCharacter{2081}{$_1$}
\DeclareUnicodeCharacter{2082}{$_2$}
\DeclareUnicodeCharacter{2261}{$\equiv$}
\DeclareUnicodeCharacter{2E1}{$^l$}
\DeclareUnicodeCharacter{2B3}{$^r$}
\DeclareUnicodeCharacter{2032}{$'$}
\DeclareUnicodeCharacter{21D4}{$\Leftrightarrow$}
\DeclareUnicodeCharacter{224A}{$\approxeq$}
\DeclareUnicodeCharacter{27E8}{$\langle$}
\DeclareUnicodeCharacter{27E9}{$\rangle$}
\DeclareUnicodeCharacter{27F6}{$\longrightarrow$}
\DeclareUnicodeCharacter{220E}{$\blacksquare$}
\DeclareUnicodeCharacter{27FA}{$\Leftrightarrow$}
\DeclareUnicodeCharacter{22A4}{$\top$}
\DeclareUnicodeCharacter{22A5}{$\bot$}
\DeclareUnicodeCharacter{03A3}{$\Sigma$}
\DeclareUnicodeCharacter{03C0}{$\pi$}
\DeclareUnicodeCharacter{25A1}{$\square$}
\DeclareUnicodeCharacter{2237}{$::$}
\DeclareUnicodeCharacter{2308}{$\lceil$}
\DeclareUnicodeCharacter{2309}{$\rceil$}
\DeclareUnicodeCharacter{22A2}{$\vdash$}

\newcolumntype{L}{>{$}l<{$}}
\newcolumntype{R}{>{$}r<{$}}
\newcommand{\varlevel}{\textsc{v}\xspace}
\newcommand{\codelevel}{\textsc{c}\xspace}
\newcommand{\proglevel}{\textsc{d}\xspace}
\newcommand{\metalevel}{\textsc{m}\xspace}

\newcommand{\mltt}{\textsc{MLTT}\xspace}

\newcommand{\delamlang}{\textsc{DeLaM}\xspace}

\newtheorem{law}{Law}[section]

\theoremstyle{remark}
\newtheorem*{remark}{Remark}

\makeatletter
\def\namedlabel#1#2{\begingroup
   \def\@currentlabel{#2}%
   \label{#1}\endgroup
}
\makeatother

\makeatletter
\newcommand{\customlabel}[2]{%
   \protected@write \@auxout {}{\string \newlabel {#1}{{#2}{\thepage}{#2}{#1}{}} }%
   \hypertarget{#1}{}
}
\makeatother

\makeatletter
\DeclareFontFamily{U}{MnSymbolA}{}
\DeclareFontShape{U}{MnSymbolA}{m}{n}{
    <-6>  MnSymbolA5
   <6-7>  MnSymbolA6
   <7-8>  MnSymbolA7
   <8-9>  MnSymbolA8
   <9-10> MnSymbolA9
  <10-12> MnSymbolA10
  <12->   MnSymbolA12}{}
\DeclareFontShape{U}{MnSymbolA}{b}{n}{
    <-6>  MnSymbolA-Bold5
   <6-7>  MnSymbolA-Bold6
   <7-8>  MnSymbolA-Bold7
   <8-9>  MnSymbolA-Bold8
   <9-10> MnSymbolA-Bold9
  <10-12> MnSymbolA-Bold10
  <12->   MnSymbolA-Bold12}{}
\DeclareSymbolFont{MnSyA}{U}{MnSymbolA}{m}{n}
\SetSymbolFont{MnSyA}{bold}{U}{MnSymbolA}{b}{n}

\DeclareRobustCommand{\overleftharpoon}{\mathpalette{\overarrow@\leftharpoonfill@}}
\DeclareRobustCommand{\overrightharpoon}{\mathpalette{\overarrow@\rightharpoonfill@}}
\def\leftharpoonfill@{\arrowfill@\leftharpoondown\mn@relbar\mn@relbar}
\def\rightharpoonfill@{\arrowfill@\mn@relbar\mn@relbar\rightharpoonup}

\DeclareMathSymbol{\leftharpoondown}{\mathrel}{MnSyA}{'112}
\DeclareMathSymbol{\rightharpoonup}{\mathrel}{MnSyA}{'100}
\DeclareMathSymbol{\mn@relbar}{\mathrel}{MnSyA}{'320}
\makeatother

\DeclareFontFamily{U}{mathx}{}
\DeclareFontShape{U}{mathx}{m}{n}{<-> mathx10}{}
\DeclareSymbolFont{mathx}{U}{mathx}{m}{n}
\DeclareMathAccent{\widecheck}{0}{mathx}{"71}

\newcommand{\vect}[1]{\overrightarrow{#1}}

\newcommand{\vGamma}{\vect \Gamma}

\newcommand{\Trm}{\textsf{Trm}}
\newcommand{\To}{\Longrightarrow}
\newcommand{\STo}{\Rightarrow}
\newcommand{\D}{\ensuremath{\mathcal{D}}}
\newcommand{\E}{\ensuremath{\mathcal{E}}}
\newcommand{\F}{\ensuremath{\mathcal{F}}}

\newcommand{\Ac}{\ensuremath{\mathcal{A}}}
\newcommand{\Bc}{\ensuremath{\mathcal{B}}}
\newcommand{\Cc}{\ensuremath{\mathcal{C}}}

\newcommand{\Nat}{\ensuremath{\texttt{Nat}}}

\newcommand{\tbox}{\ensuremath{\texttt{box}}\xspace}

\newcommand{\tletbox}{\ensuremath{\texttt{letbox}}}

\newcommand{\ze}{\textsf{zero}}
\newcommand{\tsucc}{\textsf{succ}}
\newcommand{\su}[1]{\tsucc~#1}
\newcommand{\tapp}{\textsf{app}}

\newcommand{\Nf}{\textsf{Nf}}
\newcommand{\Ne}{\textsf{Ne}}
\newcommand{\wk}{\textsf{wk}}

\newcommand{\Ctx}{\textsf{Ctx}}
\newcommand{\Typ}{\textsf{Typ}}
\newcommand{\Exp}{\textsf{Exp}}

\newcommand{\N}{\mathbb{N}}

\newcommand{\boxit}[1]{\tbox\ #1}

\newcommand{\letbox}[3]{\tletbox\ {#1} \shortleftarrow #2\ \texttt{in}\ #3}

\newcommand{\telimn}{\textsf{elim}_{\Nat}}

\newcommand\ELIMN[5]{\ensuremath{\telimn^{#1}~(#2)~#3~(#4)~#5}}

\newcommand{\countf}[1]{\ensuremath{\textsf{count}(#1)}}

\newcommand{\flatten}[2]{\ensuremath{\textsf{flatten}_{#1}(#2)}}
\newcommand{\mergef}[2]{\ensuremath{\textsf{merge}(#1, #2)}}
\newcommand{\ttypeof}{\ensuremath{\Uparrow}}
\newcommand{\typeof}[1]{\ensuremath{\Uparrow(#1)}}
\newcommand{\compt}[1]{\ensuremath{#1~\textsf{computable}}}

\newcommand{\sep}{\;|\;}
\newcommand{\func}{\longrightarrow}

\newcommand\id{\ensuremath{\textsf{id}}}

\newcommand{\Se}{\texttt{Ty}}

\newcommand{\Ty}[1]{\ensuremath{\Se_{#1}}}
\newcommand{\tEl}{\texttt{El}}
\newcommand{\Elt}[2]{\ensuremath{\tEl^{#1}~#2}}
\newcommand{\Level}{\ensuremath{\textsf{Level}}}

\newcommand\El{\textbf{El}}

\newcommand{\PI}[5]{\ensuremath{\Pi^{#1,#2}(#3 : #4). #5}}
\newcommand{\LAM}[5]{\ensuremath{\lambda^{#1,#2}(#3 : #4). #5}}
\newcommand{\APP}[7]{\ensuremath{(#1 : \PI{#2}{#3}{#4}{#5}{#6})~#7}}
\newcommand{\UPI}[3]{\ensuremath{\vect{#1} \STo^{#2} #3}}
\newcommand{\ULAM}[3]{\ensuremath{\Lambda^{#1}~{\vect{#2}}. #3}}
\newcommand{\UAPP}[2]{\ensuremath{#1~\$~\vect{#2}}}
\newcommand{\LETBOX}[6]{\ensuremath{\tletbox^{#1}_{#2}~#3 \leftarrow (#4 : #5)~\textsf{in}~#6}}
\newcommand{\ELIMC}[6]{\ensuremath{\textsf{elim}^{#1,#2}~{#3}~#4~(#5 : #6)}}
\newcommand{\LETBTYP}[7]{\LETBOX{#1}{#4}{#5}{#7}{\CTyp[#3]{#2}}{#6}}
\newcommand{\LETBTRM}[8]{\LETBOX{#1}{#5}{#6}{#8}{\CTrm[#3]{#4}{#2}}{#7}}
\newcommand{\ELIMTYPn}[7]{\ELIMC{#1}{#2}{#3}{#4}{#7}{\CTyp[#6]{#5}}}
\newcommand{\ELIMTRMn}[8]{\ELIMC{#1}{#2}{#3}{#4}{#8}{\CTrm[#6]{#7}{#5}}}
\newcommand{\ELIMTYP}[7]{\ELIMTYPn{#1}{#2}{\vect{#3}}{\vect{#4}}{#5}{#6}{#7}}
\newcommand{\ELIMTRM}[8]{\ELIMTRMn{#1}{#2}{\vect{#3}}{\vect{#4}}{#5}{#6}{#7}{#8}}
\newcommand{\redd}{\ensuremath{\rightsquigarrow}}
\newcommand{\reds}{\ensuremath{\rightsquigarrow^\ast}}

\newcommand\tand{\text{ and }}
\newcommand\byIH{\tag{by IH}}

\newcommand{\at}[1]{\scalebox{0.75}{\color{gray}$\,@~#1$}}

\newcommand{\vertrule}[1][1ex]{\rule{.45pt}{#1}}

\newcommand{\tttstileA}{\ensuremath{\mathrel{\raisebox{0.01pt}{\vertrule[1.6ex]}\hspace{.15em}{\vertrule[1.6ex]}{\hspace{.1em}\vDash}}}}
\newcommand{\sttstileA}{\ensuremath{\mathrel{\raisebox{0.01pt}{\vertrule[1.6ex]}{\hspace{.1em}\vDash}}}}

\DeclareDocumentCommand{\genjudge}{ m m m } {
  \ensuremath{#1 #2 #3}
}

\DeclareDocumentCommand{\singlejudge}{ o m m } {
  \IfNoValueTF {#1}
  {\genjudge{\Gamma}{#2}{#3}}
  {\genjudge{#1}{#2}{#3}}
}

\DeclareDocumentCommand{\dualjudge}{o o m m m } {
  \IfNoValueTF {#1}
  {\IfNoValueTF {#2}
    {\genjudge{\Psi #3 \Gamma}{#4}{#5}}
    {\genjudge{\Psi #3 #2}{#4}{#5}}}
  {\IfNoValueTF {#2}
    {\genjudge{#1 #3 \Gamma}{#4}{#5}}
    {\genjudge{#1 #3 #2}{#4}{#5}}}
}

\DeclareDocumentCommand{\triplejudge}{o o o m m } {
  \IfNoValueTF {#1}
  {\IfNoValueTF {#2}
    {\IfNoValueTF {#3}
      {\genjudge{L \sep \Psi; \Gamma}{#4}{#5}}
      {\genjudge{#3 \sep \Psi; \Gamma}{#4}{#5}}}
    {\IfNoValueTF {#3}
      {\genjudge{L \sep \Psi; #2}{#4}{#5}}
      {\genjudge{#3 \sep \Psi; #2}{#4}{#5}}}}
  {\IfNoValueTF {#2}
    {\IfNoValueTF {#3}
      {\genjudge{L \sep #1; \Gamma}{#4}{#5}}
      {\genjudge{#3 \sep #1; \Gamma}{#4}{#5}}}
    {\IfNoValueTF {#3}
      {\genjudge{L \sep #1; #2}{#4}{#5}}
      {\genjudge{#3 \sep #1; #2}{#4}{#5}}}}  
}

\DeclareDocumentCommand{\dcjudge}{o o m m } {
  \dualjudge[#1][#2]{;}{#3}{#4}
}

\DeclareDocumentCommand{\dSemjudgebf}{ o o m } {
  \dcjudge[#1][#2]{\sttstileA}{#3}
}

\DeclareDocumentCommand{\judge}{ o m } {
  \singlejudge[#1]{\vdash}{#2}
}

\DeclareDocumentCommand{\gequiv}{ o m m } {
  \IfNoValueTF {#1}
  {\judge[L]{#2 \approx #3}}
  {\judge[#1]{#2 \approx #3}}
}

\DeclareDocumentCommand{\cont}{ o m } {
  \square (\judge[#1]{#2})
}

\DeclareDocumentCommand{\rjudge}{ o m } {
  \IfNoValueTF {#1}
  {\singlejudge[\vGamma]{\vdash_r}{#2}}
  {\singlejudge[#1]{\vdash_r}{#2}}
}

\DeclareDocumentCommand{\semjudge}{ o m } {
  \singlejudge[#1]{\vDash}{#2}
}

\DeclareDocumentCommand{\Semjudge}{ o m } {
  \singlejudge[#1]{\Vdash}{#2}
}

\DeclareDocumentCommand{\mjudge}{ o m } {
  \IfNoValueTF {#1}
  {\singlejudge[\vGamma]{\vdash}{#2}}
  {\singlejudge[#1]{\vdash}{#2}}
}

\DeclareDocumentCommand{\djudge}{ o o m } {
  \dualjudge[#1][#2]{;}{\vdash}{#3}
}

\DeclareDocumentCommand{\ljudge}{ o m m } {
  \IfNoValueTF {#1}
  {\genjudge{\Psi}{\vdash_{#2}}{#3}}
  {\genjudge{#1}{\vdash_{#2}}{#3}}
}

\DeclareDocumentCommand{\ljudgel}{ o m m } {
  \IfNoValueTF {#1}
  {\genjudge{\Gamma}{\vdash_{#2}}{#3}}
  {\genjudge{#1}{\vdash_{#2}}{#3}}
}

\DeclareDocumentCommand{\DTyp}{o m m} {
  \ensuremath{(\ljudgel[#1]{#2}{\at{#3}})}
}

\DeclareDocumentCommand{\DTrm}{o m m m} {
  \ensuremath{(\ljudgel[#1]{#2}{#3\at{#4}})}
}

\DeclareDocumentCommand{\CTyp}{o m} {
  \ensuremath{\square \DTyp[#1]{\codelevel}{#2}}
}

\DeclareDocumentCommand{\CTrm}{o m m} {
  \ensuremath{\square \DTrm[#1]{\codelevel}{#2}{#3}}
}

\DeclareDocumentCommand{\TPI}{m o m m m}{
  \ensuremath{(#1 : \DTyp[#2]{\proglevel}{#3}) \STo^{#4} #5}
}
\newcommand{\TLAM}[4]{\ensuremath{\Lambda^{#1,#2}_p~{#3}. #4}}
\newcommand{\TAPP}[2]{\ensuremath{#1~\$_p~#2}}
\newcommand{\CPI}[3]{\ensuremath{(#1 : \Ctx) \STo^{#2} #3}}
\newcommand{\CLAM}[3]{\ensuremath{\Lambda^{#1}~{#2}. #3}}
\newcommand{\CAPP}[2]{\ensuremath{#1~\$~#2}}

\DeclareDocumentCommand{\lpgenjudge}{ o o m m } {
  \IfNoValueTF {#1}
  {\IfNoValueTF {#2}
    {\genjudge{L \sep \Psi}{#3}{#4}}
    {\genjudge{#2 \sep \Psi}{#3}{#4}}}
  {\IfNoValueTF {#2}
    {\genjudge{L \sep #1}{#3}{#4}}
    {\genjudge{#2 \sep #1}{#3}{#4}}}
}

\DeclareDocumentCommand{\pjudge}{ o o m } {
  \lpgenjudge[#1][#2]{\vdash}{#3}
}

\DeclareDocumentCommand{\ptyping}{o o m m } {
  \pjudge[#1][#2]{#3 : #4}
}

\DeclareDocumentCommand{\ptyequiv}{o o m m m } {
  \pjudge[#1][#2]{#3 \approx #4 : #5}
}

\DeclareDocumentCommand{\lpjudge}{ o o m m } {
  \pjudge[#1][#2]{_{#3}#4}
}

\DeclareDocumentCommand{\lpequiv}{ o o m m m } {
  \lpjudge[#1][#2]{#3}{#4 \approx #5}
}

\DeclareDocumentCommand{\lpcobj}{ o o m m m } {
  \lpjudge[#1][#2]{#3}{#4 : #5}
}

\DeclareDocumentCommand{\tjudge}{ o o o m } {
  \triplejudge[#1][#2][#3]{\vdash}{#4}
}

\DeclareDocumentCommand{\ltjudge}{o o o m m} {
  \tjudge[#1][#2][#3]{_{#4} #5}
}

\DeclareDocumentCommand{\lttypwf}{o o o m m m} {
  \ltjudge[#1][#2][#3]{#4}{#5 \at{#6}}
}

\DeclareDocumentCommand{\lttyping}{o o o m m m m} {
  \ltjudge[#1][#2][#3]{#4}{#5 : #6 \at{#7}}
}

\DeclareDocumentCommand{\lttypingd}{o o o m m m m} {
  \ltjudge[#1][#2][#3]{#4}{#5 : \Ty{#6}\;(\at{#7})}
}

\DeclareDocumentCommand{\ltsubst}{o o o m m m} {
  \ltjudge[#1][#2][#3]{#4}{#5 : #6}
}

\DeclareDocumentCommand{\ltsubstv}{o o o m m m} {
  \ltjudge[#1][#2][#3]{#4}{^{\#} #5 : #6}
}

\DeclareDocumentCommand{\ltsubeq}{o o o m m m m} {
  \ltjudge[#1][#2][#3]{#4}{#5 \approx #6 : #7}
}

\DeclareDocumentCommand{\lttypeq}{o o o m m m m} {
  \ltjudge[#1][#2][#3]{#4}{#5 \approx #6 \at{#7}}
}

\DeclareDocumentCommand{\lttyequiv}{o o o m m m m m} {
  \ltjudge[#1][#2][#3]{#4}{#5 \approx #6 : #7 \at{#8}}
}

\DeclareDocumentCommand{\lttyequivv}{o o o m m m m m} {
  \ltjudge[#1][#2][#3]{#4}{^{\#} #5 \approx #6 : #7 \at{#8}}
}

\DeclareDocumentCommand{\lttyequivd}{o o o m m m m m} {
  \ltjudge[#1][#2][#3]{#4}{#5 \approx #6 : \Ty{#7}\;(\at{#8})}
}

\DeclareDocumentCommand{\lsemjudge}{ o m m } {
  \IfNoValueTF {#1}
  {\genjudge{\Psi}{\sttstileA_{#2}}{#3}}
  {\genjudge{#1}{\sttstileA_{#2}}{#3}}
}

\DeclareDocumentCommand{\lsemvjudge}{ o m m } {
  \IfNoValueTF {#1}
  {\genjudge{\Psi}{\Vdash_{#2}}{#3}}
  {\genjudge{#1}{\Vdash_{#2}}{#3}}
}

\DeclareDocumentCommand{\lmjudge}{ o o m m } {
  \djudge[#1][#2]{_{#3} #4}
}

\DeclareDocumentCommand{\dpjudge}{ o o m } {
  \lpgenjudge[#1][#2]{\sttstileA}{#3}
}

\DeclareDocumentCommand{\dtjudge}{ o o o m } {
  \triplejudge[#1][#2][#3]{\tttstileA}{#4}
}

\DeclareDocumentCommand{\ldpjudge}{ o o m m m } {
  \dpjudge[#1][#2]{_{#3}^{#4} #5}
}

\DeclareDocumentCommand{\ldjudge}{ o o o m m m } {
  \dtjudge[#1][#2][#3]{_{#4}^{#5} #6}
}

\DeclareDocumentCommand{\dpSjudge}{ o o m } {
  \lpgenjudge[#1][#2]{\sttstileA}{#3}
}

\DeclareDocumentCommand{\dSjudge}{ o o o m } {
  \triplejudge[#1][#2][#3]{\sttstileA}{#4}
}

\DeclareDocumentCommand{\dsjudge}{ o o o m } {
  \triplejudge[#1][#2][#3]{\Vdash}{#4}
}

\DeclareDocumentCommand{\ldpSjudge}{ o o m m m } {
  \dpSjudge[#1][#2]{_{#3}^{#4} #5}
}

\DeclareDocumentCommand{\ldSjudge}{ o o o m m m } {
  \dSjudge[#1][#2][#3]{_{#4}^{#5} #6}
}

\DeclareDocumentCommand{\dsemjudge}{ o o m } {
  \dualjudge[#1][#2]{;}{\sttstileA}{#3}
}

\DeclareDocumentCommand{\dSemjudge}{ o o m } {
  \dualjudge[#1][#2]{;}{\Vdash}{#3}
}

\DeclareDocumentCommand{\wmjudge}{ o m } {
  \IfNoValueTF {#1}
  {\vGamma \vdash_{\! w} #2}
  {#1 \vdash_{\! w} #2}
}

\DeclareDocumentCommand{\msemjudge}{ o m } {
  \IfNoValueTF {#1}
  {\vGamma \vDash #2}
  {#1 \vDash #2}
}

\DeclareDocumentCommand{\mSemjudge}{ o m } {
  \IfNoValueTF {#1}
  {\vGamma \Vdash #2}
  {#1 \Vdash #2}
}

\DeclareDocumentCommand{\typing}{ o m m } {
  \judge[#1]{#2 : #3}
}

\DeclareDocumentCommand{\rtyping}{ o m m } {
  \rjudge[#1]{#2 : #3}
}

\DeclareDocumentCommand{\semtyp}{ o m m } {
  \semjudge[#1]{#2 : #3}
}

\DeclareDocumentCommand{\semvtyp}{ o m m } {
  \IfNoValueTF {#1}
  {\semtyp[\Psi]{ #2}{#3}}
  {\semtyp[#1]{ #2}{#3}}
}

\DeclareDocumentCommand{\Semtyp}{ o m m } {
  \Semjudge[#1]{#2 : #3}
}

\DeclareDocumentCommand{\tyequiv}{ o m m m } {
  \judge[#1]{#2 \approx #3 : #4}
}

\DeclareDocumentCommand{\semtyeq}{ o m m m } {
  \semjudge[#1]{#2 \approx #3 : #4}
}

\DeclareDocumentCommand{\semvtyeq}{ o m m m } {
  \IfNoValueTF {#1}
  {\semjudge[\Psi]{ #2 \approx #3 : #4}}
  {\semjudge[#1]{ #2 \approx #3 : #4}}
}

\DeclareDocumentCommand{\lSemtypPrime}{ o o m m m } {
  \dSemjudgebf[#1][#2]{_{#3} #4 : #5}
}

\DeclareDocumentCommand{\mtyping}{ o m m } {
  \mjudge[#1]{#2 : #3}
}

\DeclareDocumentCommand{\lmtyping}{ o m m m } {
  \mjudge[#1]{_{#2} #3 : #4}
}

\DeclareDocumentCommand{\dtyping}{ o o m m } {
  \djudge[#1][#2]{#3 : #4}
}

\DeclareDocumentCommand{\ltyping}{ o o m m m } {
  \djudge[#1][#2]{_{#3} #4 : #5}
}

\DeclareDocumentCommand{\wmtyping}{ o m m } {
  \wmjudge[#1]{#2 : #3}
}

\DeclareDocumentCommand{\mSemtyp}{ o m m } {
  \mSemjudge[#1]{#2 : #3}
}

\DeclareDocumentCommand{\lsemtyp}{ o o m m m } {
  \dsemjudge[#1][#2]{_{#3} #4 : #5}
}

\DeclareDocumentCommand{\lsemvtyp}{ o o m m m } {
  \dSemjudge[#1][#2]{_{#3} #4 : #5}
}

\DeclareDocumentCommand{\lSemtyp}{ o o m m m } {
  \dSemjudge[#1][#2]{_{#3} #4 : #5}
}

\DeclareDocumentCommand{\mrarr}{ o m m } {
  \mjudge[#1]{#2 > #3}
}

\DeclareDocumentCommand{\wmrarr}{ o m m } {
  \wmjudge[#1]{#2 > #3}
}

\DeclareDocumentCommand{\mtyequiv}{ o m m m } {
  \mjudge[#1]{#2 \approx #3 : #4}
}

\DeclareDocumentCommand{\dtyequiv}{ o o m m m } {
  \djudge[#1][#2]{#3 \approx #4 : #5}
}

\DeclareDocumentCommand{\ltyequiv}{ o o m m m m } {
  \djudge[#1][#2]{_{#3} #4 \approx #5 : #6}
}

\DeclareDocumentCommand{\ltygneeq}{ o o m m m } {
  \djudge[#1][#2]{_1 #3 \sim #4 : #5}
}

\DeclareDocumentCommand{\ltygteq}{ o o m m m } {
  \djudge[#1][#2]{_1 #3 \simeq #4 : #5}
}

\DeclareDocumentCommand{\ltyred}{ o o m m m } {
  \djudge[#1][#2]{_1 #3 \rightsquigarrow #4 : #5}
}

\DeclareDocumentCommand{\ltyreds}{ o o m m m } {
  \djudge[#1][#2]{_1 #3 \rightsquigarrow^\ast #4 : #5}
}

\DeclareDocumentCommand{\dtconv}{ o o m m m } {
  \djudge[#1][#2]{_1 #3 \ \hat{\Longleftrightarrow}\ #4 : #5}
}

\DeclareDocumentCommand{\dtconvnf}{ o o m m m } {
  \djudge[#1][#2]{_1 #3 \Longleftrightarrow #4 : #5}
}

\DeclareDocumentCommand{\dtconvne}{ o o m m m } {
  \djudge[#1][#2]{_1 #3 \longleftrightarrow #4 : #5}
}

\DeclareDocumentCommand{\ttypred}{ o o o m m m m } {
  \ltjudge[#1][#2][#3]{#4}{ #5 \rightsquigarrow #6 \at{#7}}
}

\DeclareDocumentCommand{\ttypreds}{ o o o m m m m } {
  \ltjudge[#1][#2][#3]{#4}{ #5 \rightsquigarrow^\ast #6 \at{#7}}
}

\DeclareDocumentCommand{\ttrmred}{ o o o m m m m m } {
  \ltjudge[#1][#2][#3]{#4}{ #5 \rightsquigarrow #6 : #7 \at{#8}}
}

\DeclareDocumentCommand{\ttrmreds}{ o o o m m m m m } {
  \ltjudge[#1][#2][#3]{#4}{ #5 \rightsquigarrow^\ast #6 : #7 \at{#8}}
}

\DeclareDocumentCommand{\tconvtyp}{ o o o m m m m } {
  \ltjudge[#1][#2][#3]{#4}{ #5 \ \hat{\Longleftrightarrow}\ #6 \at{#7}}
}

\DeclareDocumentCommand{\tconvtypnf}{ o o o m m m m } {
  \ltjudge[#1][#2][#3]{#4}{ #5 \Longleftrightarrow #6 \at{#7}}
}

\DeclareDocumentCommand{\tconvtypne}{ o o o m m m m } {
  \ltjudge[#1][#2][#3]{#4}{ #5 \longleftrightarrow #6 \at{#7}}
}

\DeclareDocumentCommand{\tconvctx}{o o m m m}{
  \lpjudge[#1][#2]{#3}{#4 \ \hat{\Longleftrightarrow}\ #5}
}

\DeclareDocumentCommand{\tconvtrm}{ o o o m m m m m } {
  \ltjudge[#1][#2][#3]{#4}{ #5 \ \hat{\Longleftrightarrow}\ #6 : #7 \at{#8}}
}

\DeclareDocumentCommand{\tconvtrmnf}{ o o o m m m m m } {
  \ltjudge[#1][#2][#3]{#4}{ #5 \Longleftrightarrow #6 : #7 \at{#8}}
}

\DeclareDocumentCommand{\tconvtrmne}{ o o o m m m m m } {
  \ltjudge[#1][#2][#3]{#4}{ #5 \ \hat{\longleftrightarrow}\ #6 : #7 \at{#8}}
}

\DeclareDocumentCommand{\tconvtrmnee}{ o o o m m m m m } {
  \ltjudge[#1][#2][#3]{#4}{ #5 \longleftrightarrow #6 : #7 \at{#8}}
}

\DeclareDocumentCommand{\tconvsub}{ o o o m m m m } {
  \ltjudge[#1][#2][#3]{#4}{ #5 \ \hat{\Longleftrightarrow}\ #6 : #7}
}

\DeclareDocumentCommand{\lttypgneeq}{ o o o m m m m } {
  \ltjudge[#1][#2][#3]{#4}{#5 \sim #6 \at{#7}}
}

\DeclareDocumentCommand{\lttypgeq}{ o o o m m m m } {
  \ltjudge[#1][#2][#3]{#4}{#5 \simeq #6 \at{#7}}
}

\DeclareDocumentCommand{\lttrmgneeq}{ o o o m m m m m } {
  \ltjudge[#1][#2][#3]{#4}{#5 \sim #6 : #7 \at{#8}}
}

\DeclareDocumentCommand{\lttrmgeq}{ o o o m m m m m } {
  \ltjudge[#1][#2][#3]{#4}{#5 \simeq #6 : #7 \at{#8}}
}

\DeclareDocumentCommand{\ltctxgeq}{o o m m m}{
  \lpjudge[#1][#2]{#3}{#4 \simeq #5}
}

\DeclareDocumentCommand{\ltsubgeq}{ o o o m m m m } {
  \ltjudge[#1][#2][#3]{#4}{ #5 \simeq #6 : #7}
}

\DeclareDocumentCommand{\lsemtyp}{ o o m m m } {
  \dsemjudge[#1][#2]{_{#3} #4 : #5}
}

\DeclareDocumentCommand{\dsemtyeq}{ o o m m m } {
  \dsemjudge[#1][#2]{#3 \approx #4 : #5}
}

\DeclareDocumentCommand{\ldtypwf}{ o o o m m m m } {
  \ldjudge[#1][#2][#3]{#4}{#5}{#6 \at{#7}}
}

\DeclareDocumentCommand{\ldtypeq}{ o o o m m m m m } {
  \ldjudge[#1][#2][#3]{#4}{#5}{#6 \approx #7 \at{#8}}
}

\DeclareDocumentCommand{\ldtyeqnat}{ o o o m m m } {
  \ldjudge[#1][#2][#3]{#4}{}{#5 \approx #6 : \Nat}
}

\DeclareDocumentCommand{\ldtynfeqnat}{ o o o m m m } {
  \ldjudge[#1][#2][#3]{#4}{}{#5 \simeq #6 : \Nat}
}

\DeclareDocumentCommand{\ldtynfeqm}{ o o o m m m } {
  \ldjudge[#1][#2][#3]{}{}{#4 \simeq #5 : #6}
}

\DeclareDocumentCommand{\ldtyping}{ o o o m m m m m } {
  \ldjudge[#1][#2][#3]{#4}{#5}{#6 : #7 \at{#8}}
}

\DeclareDocumentCommand{\ldtyequivt}{ o o o m m m m m m } {
  \ldjudge[#1][#2][#3]{#4}{#5}{#6 \approx #7 : #8 \at{#9}}
}

\DeclareDocumentCommand{\ldtyequiv}{ o o o m m m m m } {
  \ldjudge[#1][#2][#3]{#4}{#5}{#6 \approx #7 : \tEl(#8)}
}

\DeclareDocumentCommand{\ldctxwf}{ o o m m m } {
  \ldpjudge[#1][#2]{#3}{#4}{#5}
}

\DeclareDocumentCommand{\ldctxeq}{ o o m m m m } {
  \ldpjudge[#1][#2]{#3}{#4}{#5 \approx #6}
}

\DeclareDocumentCommand{\ldsubeq}{ o o o m m m m m } {
  \ldjudge[#1][#2][#3]{#4}{#5}{#6 \approx #7 : #8}
}

\DeclareDocumentCommand{\ldsubst}{ o o o m m m m } {
  \ldjudge[#1][#2][#3]{#4}{#5}{#6 : #7}
}

\DeclareDocumentCommand{\spjudge}{ o o m } {
  \lpgenjudge[#1][#2]{\Vdash}{#3}
}

\DeclareDocumentCommand{\stjudge}{ o o o m } {
  \triplejudge[#1][#2][#3]{\Vdash}{#4}
}

\DeclareDocumentCommand{\lspjudge}{ o o m m } {
  \spjudge[#1][#2]{_{#3} #4}
}

\DeclareDocumentCommand{\lsjudge}{ o o o m m } {
  \stjudge[#1][#2][#3]{_{#4} #5}
}

\DeclareDocumentCommand{\lsgctx}{ o m } {
  \IfNoValueTF {#1}
  {\genjudge{L}{\Vdash}{#2}}
  {\genjudge{#1}{\Vdash}{#2}}
}

\DeclareDocumentCommand{\lstypwf}{ o o o m m m } {
  \lsjudge[#1][#2][#3]{#4}{#5 \at{#6}}
}

\DeclareDocumentCommand{\lstypeq}{ o o o m m m m } {
  \lsjudge[#1][#2][#3]{#4}{#5 \approx #6 \at{#7}}
}

\DeclareDocumentCommand{\lstyping}{ o o o m m m m } {
  \lsjudge[#1][#2][#3]{#4}{#5 : #6 \at{#7}}
}

\DeclareDocumentCommand{\lstyequiv}{ o o o m m m m m } {
  \lsjudge[#1][#2][#3]{#4}{#5 \approx #6 : #7 \at{#8}}
}

\DeclareDocumentCommand{\lsgsubeq}{ o o m m m } {
  \lspjudge[#1][#2]{}{#3 \approx #4 : #5}
}

\DeclareDocumentCommand{\lsgsubst}{ o o m m } {
  \lspjudge[#1][#2]{}{#3 : #4}
}

\DeclareDocumentCommand{\lsctxwf}{ o o m m } {
  \lspjudge[#1][#2]{#3}{#4}
}

\DeclareDocumentCommand{\lsctxeq}{ o o m m m } {
  \lspjudge[#1][#2]{#3}{#4 \approx #5}
}

\DeclareDocumentCommand{\lssubeq}{ o o o m m m m } {
  \lsjudge[#1][#2][#3]{#4}{#5 \approx #6 : #7}
}

\DeclareDocumentCommand{\lssubst}{ o o o m m m } {
  \lsjudge[#1][#2][#3]{#4}{#5 : #6}
}

\DeclareDocumentCommand{\ldStypeq} { o o m m m m m m } {
  \ldSjudge[#1][#3][#2]{#4}{#5}{#6 \approx #7 \at{#8}}
}

\DeclareDocumentCommand{\ldStyequiv} { o o m m m m m m } {
  \ldSjudge[#1][#3][#2]{#4}{#5}{#6 \approx #7 : \tEl(#8)}
}

\DeclareDocumentCommand{\ldStypwfge} { o o o m m m m } {
  \ldjudge[#1][#2][#3]{\ge #4}{#5}{#6 \at{#7}}
}

\DeclareDocumentCommand{\ldStypingge} { o o o m m m m m } {
  \ldjudge[#1][#2][#3]{\ge #4}{#5}{#6 : #7 \at{#8}}
}

\DeclareDocumentCommand{\ldStypeqge} { o o o m m m m m } {
  \ldjudge[#1][#2][#3]{\ge #4}{#5}{#6 \approx #7 \at{#8}}
}

\DeclareDocumentCommand{\ldStyequivge} { o o o m m m m m m } {
  \ldjudge[#1][#2][#3]{\ge #4}{#5}{#6 \approx #7 : #8 \at{#9}}
}

\DeclareDocumentCommand{\ldSctxeqge} { o o m m m m } {
  \ldpSjudge[#1][#2]{\ge #3}{#4}{#5 \approx #6}
}

\DeclareDocumentCommand{\ldSctxwfge} { o o m m m } {
  \ldpSjudge[#1][#2]{\ge #3}{#4}{#5}
}

\DeclareDocumentCommand{\ldSsubeqge} { o o o m m m m m } {
  \ldjudge[#1][#2][#3]{\ge #4}{#5}{#6 \approx #7 : #8}
}

\DeclareDocumentCommand{\ldSsubstge} { o o o m m m m } {
  \ldjudge[#1][#2][#3]{\ge #4}{#5}{#6 : #7}
}

\DeclareDocumentCommand{\ldSctxeq} { o o m m m m } {
  \ldpSjudge[#1][#2]{#3}{#4}{#5 \approx #6}
}

\DeclareDocumentCommand{\ltSjudge}{ o o o m m m } {
  \triplejudge[#1][#2][#3]{\vDash_{#4}^{#5} #6}
}

\DeclareDocumentCommand{\ltStypwf} { o o m m m m m } {
  \ltSjudge[#1][#3][#2]{#4}{#5}{#6 \at{#7}}
}

\DeclareDocumentCommand{\ltStypeq} { o o m m m m m m } {
  \ltSjudge[#1][#3][#2]{#4}{#5}{#6 \approx #7 \at{#8}}
}

\DeclareDocumentCommand{\ltSctxwf} { o m m } {
  \lpgenjudge[#2][#1]{\vDash_{p}^{p}}{#3}
}

\DeclareDocumentCommand{\ltSctxeq} { o m m m } {
  \lpgenjudge[#2][#1]{\vDash_{p}^{p}}{#3 \approx #4}
}

\DeclareDocumentCommand{\ldSsubeq} { o o m m m m m m } {
  \ldSjudge[#1][#3][#2]{#4}{#5}{#6 \approx #7 : #8}
}

\DeclareDocumentCommand{\ldSgctx}{ o m } {
  \IfNoValueTF {#1}
  {\genjudge{L}{\sttstileA}{#2}}
  {\genjudge{#1}{\sttstileA}{#2}}
}

\DeclareDocumentCommand{\ldSgctxeq}{ o m m } {
  \IfNoValueTF {#1}
  {\genjudge{L}{\vDash}{#2 \approx #3}}
  {\genjudge{#1}{\vDash}{#2 \approx #3}}
}

\DeclareDocumentCommand{\ldSgsubst}{ o o m m } {
  \lpgenjudge[#1][#2]{\sttstileA}{#3 : #4}
}

\DeclareDocumentCommand{\ldSgsubeq}{ o o m m m } {
  \lpgenjudge[#1][#2]{\sttstileA}{#3 \approx #4 : #5}
}

\DeclareDocumentCommand{\ldssubst} { o o o m m m  } {
  \dsjudge[#1][#2][#3]{_{#4} #5 : #6}
}

\DeclareDocumentCommand{\ldstypwf} { o o o m m m  } {
  \dsjudge[#1][#2][#3]{_{#4} #5 \at{#6}}
}

\DeclareDocumentCommand{\ldstyping} { o o o m m m m } {
  \dsjudge[#1][#2][#3]{_{#4} #5 : #6 \at{#7}}
}

\DeclareDocumentCommand{\ltsemjudge} { o o o m m m } {
  \triplejudge[#1][#2][#3]{\sttstile{}{}_{#4}^{#5}}{#6}
}

\DeclareDocumentCommand{\ltsemtypwf} { o o o m m m m } {
  \ltsemjudge[#1][#2][#3]{#4}{#5}{#6 \at{#7}}
}

\DeclareDocumentCommand{\ltsemtyping} { o o o m m m m m } {
  \ltsemjudge[#1][#2][#3]{#4}{#5}{#6 : #7 \at{#8}}
}

\DeclareDocumentCommand{\ltsemsubst} { o o o m m m m } {
  \ltsemjudge[#1][#2][#3]{#4}{#5}{#6 : #7}
}

\DeclareDocumentCommand{\lsemtyeqp}{ o o m m m m m } {
  \dsemjudge[#1][#2]{_{#3}^{#4} {#5} \approx #6 : #7}
}

\DeclareDocumentCommand{\lsemtyeq}{ o o m m m m } {
  \dsemjudge[#1][#2]{_1^{#3} {#4} \approx #5 : #6}
}

\DeclareDocumentCommand{\lsemtypingc}{ o o m m } {
  \dsemjudge[#1][#2]{_0^0 {#3} : #4}
}

\DeclareDocumentCommand{\lsemtypingp}{ o o m m } {
  \dualjudge[#1][#2]{;}{\vDash_1^0}{#3 : #4}
}

\DeclareDocumentCommand{\lsemvtyeq}{ o o m m m m } {
  \dSemjudge[#1][#2]{_{#3} #4 \approx #5 : #6}
}

\DeclareDocumentCommand{\msemtyeq}{ o m m m } {
  \msemjudge[#1]{#2 \approx #3 : #4}
}
\DeclareDocumentCommand{\msemtyp}{ o m m } {
  \msemjudge[#1]{#2 : #3}
}

\DeclareDocumentCommand{\mgluty}{ o m o m } {
  \IfNoValueTF {#3}
  {\mjudge[#1]{#2 \; \circledR \; #4}}
  {\mjudge[#1]{#2 \; \circledR_{#3} \; #4}}
}

\DeclareDocumentCommand{\mglutm}{ o m m m o m } {
   \IfNoValueTF {#5}
   {\mjudge[#1]{#2 : #3 \; \circledR \; #4 \in \El(#6)}}
   {\mjudge[#1]{#2 : #3 \; \circledR_{#5} \; #4 \in \El_{#5}(#6)}}
 }

 \DeclareDocumentCommand{\mglunat}{ o m m } {
   \mjudge[#1]{#2 : \Nat \; \circledR \; #3 \in Nat}
 }

\DeclareDocumentCommand{\mglutms}{ o m m m } {
  \mjudge[#1]{#2 : #3 \; \circledR \; #4}
}

\DeclareDocumentCommand{\mglutyu}{ o m o m } {
  \IfNoValueTF {#3}
  {\mjudge[#1]{#2 \; \overline{\circledR} \; #4}}
  {\mjudge[#1]{#2 \; \overline{\circledR}_{#3} \; #4}}
}

\DeclareDocumentCommand{\mglutmu}{ o m m m o m } {
  \IfNoValueTF {#5}
  {\mjudge[#1]{#2 : #3 \; \overline{\circledR} \; #4 \in \El(#6)}}
  {\mjudge[#1]{#2 : #3 \; \overline{\circledR}_{#5} \; #4 \in \El_{#5}(#6)}}
}

\DeclareDocumentCommand{\mglutmd}{ o m m m o m } {
  \IfNoValueTF {#5}
  {\mjudge[#1]{#2 : #3 \; \underline{\circledR} \; #4 \in \El(#6)}}
  {\mjudge[#1]{#2 : #3 \; \underline{\circledR}_{#5} \; #4 \in \El_{#5}(#6)}}
}

\newcommand{\labeledit}[1]{\label{#1}}

\newcommand{\mhighlight}[1]{\colorbox{light-gray}{\ensuremath{#1}}}

\newcommand{\iscore}[1]{\ensuremath{#1\ \texttt{core}}}
\newcommand{\istype}[1]{\ensuremath{#1\ \texttt{type}}}

\newcommand{\JH}[1]{ }

\AtBeginDocument{%
  \providecommand\BibTeX{{%
    \normalfont B\kern-0.5em{\scshape i\kern-0.25em b}\kern-0.8em\TeX}}}

\begin{document}

%% Title information
\title{\delamlang: A Dependent Layered Modal Type Theory for Meta-programming}         %% [Short Title] is optional;
                                        %% when present, will be used in
                                        %% header instead of Full Title.
% \titlenote{with title note}             %% \titlenote is optional;
%                                         %% can be repeated if necessary;
%                                         %% contents suppressed with 'anonymous'
% \subtitle{Subtitle}                     %% \subtitle is optional
% \subtitlenote{with subtitle note}       %% \subtitlenote is optional;
%                                         %% can be repeated if necessary;
%                                         %% contents suppressed with 'anonymous'

%% Author information
%% Contents and number of authors suppressed with 'anonymous'.
%% Each author should be introduced by \author, followed by
%% \authornote (optional), \orcid (optional), \affiliation, and
%% \email.
%% An author may have multiple affiliations and/or emails; repeat the
%% appropriate command.
%% Many elements are not rendered, but should be provided for metadata
%% extraction tools.

%% Author with single affiliation.
\author{Jason Z. S. Hu}
\email{zhong.s.hu@mail.mcgill.ca}
\author{Brigitte Pientka}
\email{bpientka@cs.mcgill.ca}
\affiliation{%
  \department{School of Computer Science}
  \institution{McGill University}
  \streetaddress{McConnell Engineering Bldg. , 3480 University St.}
  \city{Montr\'eal}
  \state{Qu\'ebec}
  \country{Canada}
  \postcode{H3A 0E9}
}

%% Abstract
%% Note: \begin{abstract}...\end{abstract} environment must come
%% before \maketitle command
\begin{abstract}
  We scale layered modal type theory to dependent types, introducing \delamlang,
  dependent layered modal type theory. %
  This type theory is novel in that we have one uniform type theory in which we can
  not only compose and execute code, but also intensionally analyze the code of types
  and terms. %
  The latter in particular allows us to write tactics as meta-programs and use regular
  libraries when writing tactics. %
  \delamlang provides a sound foundation for proof assistants to support type-safe
  tactic mechanism.
\end{abstract}

%% 2012 ACM Computing Classification System (CSS) concepts
%% Generate at 'http://dl.acm.org/ccs/ccs.cfm'.
% <concept_id>10011007.10011006.10011008</concept_id>
% <concept_desc>Software and its engineering~General programming languages</concept_desc>
% <concept_id>10003456.10003457.10003521.10003525</concept_id>
% <concept_desc>Social and professional topics~History of programming languages</concept_desc>
% <concept_significance>300</concept_significance>

% \ccsdesc[500]{Software and its engineering~General programming languages}
% \ccsdesc[300]{Social and professional topics~History of programming languages}
%% End of generated code

%% Keywords
%% comma separated list
% \keywords{modal logic, dependent types, normalization by
%   evaluation} %% \keywords are mandatory in final camera-ready submission

%% \maketitle
%% Note: \maketitle command must come after title commands, author
%% commands, abstract environment, Computing Classification System
%% environment and commands, and keywords command.
\maketitle

\section{Introduction}

\citet{hu2024layered} develop a layered modal type theory which supports pattern
matching on code. %
A critical idea lying in this system is \emph{the layering principle}. %
The layering principle begins with a core language, e.g. simply typed
$\lambda$-calculus (STLC). %
Then we extend the core language with a layer of the $\square$ modality which supports
meta-programming and intensional analysis. %
Though only two layers are demonstrated by \citet{hu2024layered}, in principle, this
extension can be iterated indefinite number of times, forming an arbitrary $n$-layered
modal type theory. %

In another dimension, instead of adding more and more layers, we could also increase
the expressive power of the core language. %
One interesting candidate for a core language is Martin-L\"of type theory (MLTT). %
MLTT is the foundation for many type-theory-based proof assistants, including Coq,
Agda and Lean. %
Treating MLTT as the core language and applying the layering principle to it could
yield a dependently typed system that allows to meta-program and intensionally analyze
code of itself without forgoing the consistency of the overall system. %
This feature gives a solid foundation for proof assistants to support truly type-safe
meta-programming. %
Due to the layering principle, libraries written for bare MLTT can also be used during
meta-programming. %
For example, we can use the same data structures for natural numbers and lists for
both programs and meta-programs. %
Meanwhile in reality, e.g. in Coq, we have at least four unexchangeable notions
natural numbers: the natural numbers defined inductively in Gallina, Ltac's natural
numbers, failure levels and hint database's search levels. %
Therefore, we can foresee that the layering principle also has the additional benefit
in engineering.

In this technical report, we first extend the previous layered modal type theory with
context variables. %
They are necessary to enable recursion on code. %
We justify the decidability of conversion checking following
\citet{abel_decidability_2017}'s reducibility proof. %
We then scale the setup all the way to MLTT, introducing \delamlang,
\textbf{De}pendent \textbf{La}yered \textbf{M}odal type theory. %
We then scale the reducibility proof to \delamlang and therefore justify its
decidability of conversion checking. %
A corollary is the consistency of \delamlang, hence showing that this type theory can
be used as a foundation for proof assistants. 

\section{Supporting Context Variables}\labeledit{sec:cv}

\citet{hu2024layered} present a layered modal type theory which supports pattern
matching on code and establish a normalization proof via a presheaf model. %
However, this form of intentional analysis is not in the most desired form: we cannot
perform recursion on the structure of code. %
This limitations comes from two aspects:
\begin{itemize}
\item STLC lacks a generic notion of types. %
  When we do recursion on, for example, a $\lambda$ expression, the type of the
  $\lambda$ and the type of the body necessarily differ. %
  Therefore, we are not able to formulate a recursion principle without type variables
  as in System F, or a variable of type \texttt{Set} in dependent type theory. %
  This problem naturally goes away if we employ a stronger core type system, so it is
  not our primary concern.
  
\item However, System F or dependent types does not give us a way to capture local
  contexts using a variable. %
  Consider again the $\lambda$ case, even if we have type variables, the recursion on
  the body is in an extended local context, so for the recursion to work, we must be
  able to capture \emph{context variables}, which capture local contexts as global
  variables. 
\end{itemize}

In this section, we focus on context variables. %
We develop the syntactic theory for our $2$-layered contextual model type theory with
context variables, and show its consistency via a reducibility predicate argument. %

\subsection{Well-formedness of Contexts and Types}\labeledit{sec:cv:ctx-types}

With context variables, the type theory becomes ``slightly'' dependently typed, in
that both global and local contexts, and types can depend on context variables, so
their well-formedness requires dedicated judgments. %
The syntax of contexts and types is:
\begin{alignat*}{2}
  i & && \tag{Layer, $i \in[0, 1]$} \\
  x, y & && \tag{Local variables} \\
  u & && \tag{Global variables} \\
  g & && \tag{Context variables} \\
  S, T &:=&&\ \Nat \sep \cont T \sep S \func T \sep (g : \Ctx) \STo T
             \tag{Types, \Typ} \\
  B &:=&&\ u : (\judge T) \sep g : \Ctx
          \tag{Global bindings} \\
  \Phi, \Psi &:= &&\ \cdot \sep \Phi, B
  \tag{Global contexts} \\
  \Gamma, \Delta &:= &&\ \cdot \sep g \sep \Gamma, x : T
                        \tag{Local contexts}
\end{alignat*}
Their well-formedness judgments are:
\begin{mathpar}
  \inferrule
  { }
  {\vdash \cdot}

  \inferrule
  {\vdash \Psi}
  {\vdash \Psi, g : \Ctx}
  
  \inferrule
  {\ljudge 0 \Gamma \\ \ljudge 0 T}
  {\vdash \Psi, u : (\judge T)}

  \inferrule
  {\vdash \Psi}
  {\ljudge i \cdot}

  \inferrule
  {\vdash \Psi \\ g : \Ctx \in \Psi}
  {\ljudge i{g}}

  \inferrule
  {\ljudge i \Gamma \\ \ljudge i T}
  {\ljudge[\Psi] i{\Gamma, x : T}}

  \inferrule
  {\vdash \Psi}
  {\ljudge i \Nat}

  \inferrule
  {\ljudge i S \\ \ljudge i T}
  {\ljudge i {S \func T}}

  \inferrule
  {\ljudge 0 \Delta \\ \ljudge 0 T}
  {\ljudge 1 {\cont[\Delta] T}}

  \inferrule
  {\ljudge[\Psi, g : \Ctx] 1 T}
  {\ljudge 1 {(g : \Ctx) \STo T}}
\end{mathpar}
$\vdash \Psi$ states the well-formedness of a global context $\Psi$. %
There are two kinds of bindings in a global context: either $x : (\judge T)$ as in 
previous modal type theory, or $g : \Ctx$ which is a context variable representing a local
context. %
Eventually, $g$ will be substituted by a concrete and well-formed local context. %
Due to the introduction of context variables, the system clearly becomes
dependently typed. %
$\ljudge i \Gamma$ states the well-formedness of a local context $\Gamma$ at
layer $i$. %
Since we are dealing with a $2$-layered system now, we know $i \in [0, 1]$. %
The base case of a local context can either be an empty local context, or be
a well-scoped context variable. %
The layer $i$ is propagated to the well-formedness judgment of types $\ljudge i T$,
which states that $T$ is well-formed in $\Psi$ at layer $i$. %
Notice that the well-formedness of $T$ does not depend on any local context. %
This judgment essentially combines $\iscore T$ and $\istype T$
predicates in \citet{hu2024layered}, with context variables taken into consideration. %
Context variables is introduced by a special meta-function type $(g : \Ctx)
\STo T$, which pushes a context variable $g$ to the global context. %
This type can be seen as a ``meta-function space'' in which we define macros or
meta-programs, so intuitively, this type and its terms can only live at layer $1$. %
We prove the presupposition lemma of these three judgments:
\begin{lemma}[Presupposition]$ $
  \begin{itemize}
  \item If $\ljudge i \Gamma$,  then $\vdash \Psi$.
  \item If $\ljudge i T$, then $\vdash \Psi$.
  \end{itemize}
\end{lemma}
\begin{lemma}[Lifting] $ $
  \begin{itemize}
  \item If $\ljudge 0 \Gamma$,  then $\ljudge 1 \Gamma$.
  \item If $\ljudge 0 T$,  then $\ljudge 1 T$.
  \end{itemize}
\end{lemma}

\subsection{Weakenings}

Similar to previous layered modal type theories, we need two notions of weakenings: a global one and
a local one. %
In this section, due to context variables, we change the definition of weakenings
based on their counterparts in \citet{hu2024layered}:
\begin{align*}
  \gamma &:= \id \sep q(\gamma) \sep p(\gamma)
  \tag{Global weakenings}
  \\
  \tau &:= \id \sep q(\tau) \sep p(\tau)
  \tag{Local weakenings}
\end{align*}
We will use a global weakening $p(\id)$ in the typing rule of meta-functions in
$\Gamma[p(\id)]$ to account for the insertion of $g: \Ctx$ to the global context
$\Psi$. %
In the following, we examine the properties of weakenings. %
First, we define the composition of weakenings. %
We only define the one for global weakenings and the one for local weakenings is
completely identical:
\begin{align*}
  \id \circ \gamma' &:= \gamma' \\
  \gamma \circ \id &:= \gamma \\
  p(\gamma) \circ q(\gamma') &:= p(\gamma \circ \gamma') \\
  q(\gamma) \circ q(\gamma') &:= q(\gamma \circ \gamma') \\
  \gamma \circ p(\gamma') &:= p(\gamma \circ \gamma')
\end{align*}
\begin{lemma}[Associativity] $ $
  \begin{itemize}
  \item $(\gamma \circ \gamma') \circ \gamma'' = \gamma \circ (\gamma' \circ \gamma'')$
  \item $(\tau \circ \tau') \circ \tau'' = \tau \circ (\tau' \circ \tau'')$
  \end{itemize}
\end{lemma}
Then we apply global weakenings to types and contexts:
\begin{align*}
  \Nat[\gamma] &:= \Nat \\
  S \func T[\gamma] &:= (S[\gamma]) \func (T[\gamma]) \\
  \cont T [\gamma] &:= \cont[\Gamma[\gamma]]{T[\gamma]} \\
  (g: \Ctx) \STo T[\gamma] &:= (g: \Ctx) \STo (T[q(\gamma)]) \\[0.25em]
  \cdot[\gamma] &:= \cdot \\
  g[\gamma] &:= g
                 \tag{properly weakened depending on the name representation} \\
  \Gamma, x : T[\gamma] &:= (\Gamma[\gamma]), x : (T[\gamma])
\end{align*}
This definition of global weakenings admits the following lemma:
\begin{lemma}[Algebra of global weakenings] $ $
  \begin{itemize}
  \item $T[\id] = T$
  \item $\Gamma[\id] = \Gamma$
  \item $T[\gamma][\gamma'] = T[\gamma \circ \gamma']$
  \item $\Gamma[\gamma][\gamma'] = \Gamma[\gamma \circ \gamma']$
  \end{itemize}
\end{lemma}
The well-formedness of weakenings is given by the following rules:
\begin{mathpar}
  \inferrule
  {\vdash \Psi}
  {\id: \Psi \To_g \Psi}

  \inferrule
  {\gamma : \Psi \To_g \Phi \\ \judge[\Psi] B}
  {p(\gamma) : \Psi, B \To_g \Phi}

  \inferrule
  {\gamma : \Psi \To_g \Phi \\ \judge[\Phi] B \\ \judge[\Psi]{B[\gamma]}}
  {q(\gamma) : \Psi, B[\gamma] \To_g \Phi, B}
  \\
  
  \inferrule
  {\ljudge[\Psi] i \Gamma}
  {\id: \Psi;\Gamma \To_i \Gamma}

  \inferrule
  {\tau : \Psi; \Gamma \To_i \Delta \\ \ljudge i T}
  {p(\tau) : \Psi; \Gamma, x : T \To_i \Delta}

  \inferrule
  {\tau : \Psi; \Gamma \To_i \Delta \\ \ljudge i T}
  {q(\tau) : \Psi; \Gamma, x : T \To_i \Delta, x : T}
\end{mathpar}
where the well-formedness of global bindings $\judge[\Psi] B$ is given as follows:
\begin{mathpar}
  \inferrule
  {\vdash \Psi}
  {\judge[\Psi]{g: \Ctx}}

  \inferrule
  {\ljudge 0 \Gamma \\ \ljudge 0 T}
  {\judge[\Psi]{u: (\judge T)}}
\end{mathpar}
The identity case for local weakenings is slightly more complex because we must take
the context variables into consideration. %

Then we can prove the following global weakening lemma:
\begin{lemma}[Global weakenings]\labeledit{lem:cvar:glob-wk-ty} $ $
  \begin{itemize}
  \item If $\ljudge[\Phi] i \Gamma$ and $\gamma : \Psi \To_g \Phi$, then $\ljudge i{ \Gamma[\gamma]}$.
  \item If $\ljudge[\Phi] i T$ and $\gamma : \Psi \To_g \Phi$, then $\ljudge i{ T[\gamma]}$.
  \end{itemize}
\end{lemma}
\begin{proof}
  Mutual induction on $\ljudge[\Phi] i \Gamma$ and $\ljudge[\Phi] i T$.
\end{proof}

The action of global weakenings on a type at layer $0$ is no-op:
\begin{lemma}
  If $\ljudge 0 T$, then $T[\gamma] = T$.
\end{lemma}
Global weakenings do not really affect local weakenings:
\begin{lemma}
  If $\gamma : \Psi \To_g \Phi$ and $\tau : \Phi; \Gamma \To_i \Delta$, then $\tau :
  \Psi; \Gamma[\gamma] \To_i \Delta[\gamma]$. 
\end{lemma}
\begin{proof}
  Induction on $\tau : \Phi; \Gamma \To_i \Delta$.
\end{proof}

The actions of local weakenings only affect terms so they will be looked into in the
next section, after we consider the syntax of the type theory.

\subsection{Syntax and Typing}

In this section, we define the syntax of the type theory with context variables. %
To isolate concerns, we use $\tletbox$ for elimination, instead of pattern matching. %
Nevertheless, pattern matching on code should work with proper adjustments to our
development:
\begin{alignat*}{2}
  m & && \tag{Natural numbers, $\N$} \\
  \delta &:=&& \cdot_{g?}^{m} \sep \wk_g^m \sep \delta, t/x \tag{Local substitutions} \\
  s, t &:=&&\ x \sep u^\delta \tag{Terms, $\Exp$} \\
  & && \sep \ze \sep \su t
  \tag{natural numbers}\\
  & && \sep \boxit t \sep \letbox u s t
  \tag{box}\\
  & &&\sep \lambda x. t \sep s\ t \tag{functions}  \\
  & &&\sep \Lambda g. t \sep t ~\$~ \Gamma \tag{meta-functions} 
\end{alignat*}
Similar to before, we have natural numbers as our base type. %
To construct and eliminate meta-functions, we have $\Lambda g. t$ and $t ~\$~ \Gamma$
respectively. %
Since we have contextual types, each global variable must be associated with a local
substitution. %
In a local substitution, due to how a local context is structured, there are also two
base cases: it can either be an empty local substitution $\cdot_{g?}^{m}$, or a weakening $\wk^m_g$ of a
context variable $g$. %
The number $m$ associated with both cases are the number of local weakening $p$'s. %
Effectively, $m$ equals to the length of the codomain local context in the typing
judgment, as we will specify later in this section. %
This number can be fetched from a local substitution by the following function:
\begin{align*}
  \widehat{\cdot_{g?}^m} &:= m \\
  \widehat{\wk_{g}^m} &:= m \\
  \widehat{\delta, t/x} &:= \widehat{\delta}
\end{align*}
In addition, $\cdot$ is optionally associated with a context variable $g$. %
If there is such a $g$, it is the base case of the codomain local context. %
These information are necessary in order to define the local and global substitution
operations. %
We rely on the following function to return the context variable inside of a local
substitution, if it exists:
\begin{align*}
  \widecheck{\cdot_{g?}^m} &:= g? \\
  \widecheck{\wk_g^m} &:= g \\
  \widecheck{\delta, t/x} &:= \widecheck{\delta}
\end{align*}
Next, we define the action of global weakenings on terms and local substitutions:
\begin{align*}
  x[\gamma] &:= x \\
  u^\delta[\gamma] &:= u^{\delta[\gamma]}
                     \tag{with $u$ properly weakened depending on name representation} \\
  \ze[\gamma] &:= \ze \\
  \su t[\gamma] &:= \su {(t[\gamma])} \\
  \boxit t[\gamma] &:= \boxit {(t[\gamma])} \\
  \letbox u s t[\gamma] &:= \letbox u {s[\gamma]} {(t[q(\gamma)])} \\
  \lambda x. t[\gamma] &:= \lambda x. {(t[\gamma])} \\
  t~s[\gamma] &:= (t[\gamma])~(s[\gamma]) \\
  \Lambda g. t[\gamma] &:= \Lambda g. {(t[q(\gamma)])} \\
  t~\$~\Gamma[\gamma] &:= (t[\gamma])~\$~(\Gamma[\gamma]) \\[0.25em]
  \wk_g^m[\gamma] &:= \wk_g^m
                     \tag{with $g$ properly weakened} \\
  \cdot_{g?}^m[\gamma] &:= \cdot_{g?}^m
                         \tag{with $g$ properly weakened if exists} \\
  \delta, t/x[\gamma] &:= (\delta[\gamma]), (t[\gamma])/x
\end{align*}
The application of global weakenings satisfy the following lemma:
\begin{lemma}[Algebra of global weakenings]$ $
  \begin{itemize}
  \item $t[\id] = t$
  \item $\delta[\id] = \delta$
  \item $t[\gamma][\gamma'] = t[\gamma \circ \gamma']$
  \item $\delta[\gamma][\gamma'] = \delta[\gamma \circ \gamma']$
  \end{itemize}
\end{lemma}

Applying local weakenings on terms and local substitutions is defined as follows:
\begin{align*}
  x[\tau] &:= x
  \tag{properly weakened} \\
  u^\delta[\tau] &:= u^{\delta[\tau]} \\
  \ze[\tau] &:= \ze \\
  \su t[\tau] &:= \su {(t[\tau])} \\
  \boxit t[\tau] &:= \boxit t \\
  \letbox u s t[\tau] &:= \letbox u {s[\tau]} {(t[\tau])} \\
  \lambda x. t[\tau] &:= \lambda x. {(t[q(\tau)])} \\
  t~s[\tau] &:= (t[\tau])~(s[\tau]) \\
  \Lambda g. t[\tau] &:= \Lambda g. {(t[\tau])} \\
  t~\$~\Gamma[\tau] &:= (t[\tau])~\$~\Gamma \\[0.25em]
  \wk_g^m[\tau] &:= \wk_g^{m + m'}
  \tag{where $m'$ is the number of $p$ constructor in $\tau$} \\
  \cdot_{g?}^m[\tau] &:= \cdot_{g?}^{m + m'}
                       \tag{where $m'$ is the number of $p$ constructor in $\tau$} \\
  \delta, t/x[\tau] &:= (\delta[\tau]), (t[\tau])/x
\end{align*}
The application of local weakenings satisfy the following lemma:
\begin{lemma}[Algebra of local weakenings]$ $
  \begin{itemize}
  \item $t[\gamma][\tau] = t[\tau][\gamma]$
  \item $\delta[\gamma][\tau] = \delta[\tau][\gamma]$
  \item $t[\tau][\tau'] = t[\tau \circ \tau']$
  \item $\delta[\tau][\tau'] = \delta[\tau \circ \tau']$
  \item $t[\delta][\tau] = t[\delta[\tau]]$
  \item $(\delta \circ \delta')[\tau] = \delta \circ (\delta'[\tau])$
  \end{itemize}
\end{lemma}

Then the weakenings of dual-contexts are defined as just tuples of global and local
weakenings:
\begin{align*}
  &\inferrule
  {\gamma : \Psi \To_g \Phi \\ \tau : \Psi;\Gamma \To_i \Delta[\gamma]}
  {\gamma; \tau : \Psi; \Gamma \To_i \Phi; \Delta} \\
  &t[\gamma; \tau] := t[\gamma][\tau] \\
  &\delta[\gamma; \tau] := \delta[\gamma][\tau]
\end{align*}
To disambiguate, when we apply a local weakening literal, we usually pair it with an
identity global weakening. %
For example, $t [ \id; p(\id)]$ \emph{locally} weakens $t$ with a local weakening $p(\id)$. %
When we write $t[p(\id)]$, we mean that $t$ is \emph{globally} weakened by $p(\id)$. %
The correctness lemma for these operations can only be proved after defining the typing
rules. %
In order to define the typing rules, we must first give the definition of global
substitutions of types and local contexts:
\begin{alignat*}{2}
  \sigma &:=&& \cdot \sep \sigma, t/u \sep \sigma, \Gamma/g
               \tag{Global substitutions}
\end{alignat*}
Global substitutions can obviously be globally weakened by iteratively applying a
global weakening to the terms and local contexts within. %
The global substitution operation of types and local contexts is defined as follows:
\begin{align*}
  \Nat [\sigma] &:= \Nat \\
  S \func T [\sigma] &:= (S[\sigma]) \func (T[\sigma]) \\
  \cont T [\sigma] &:= \cont[\Gamma[\sigma]]{T[\sigma]} \\
  (g: \Ctx) \STo T [\sigma] &:= (g: \Ctx) \STo (T [\sigma[p(\id)], g/g])
  \\[0.25em]
  \cdot [\sigma] &:= \cdot \\
  g [\sigma] &:= \sigma(g)
                  \tag{lookup $g$ in $\sigma$; undefined if $g$ is not bound or result is not a
                  local context} \\
  \Gamma, x : T[\sigma] &:= (\Gamma[\sigma]) , x : (T[\sigma])
\end{align*}
For consistency of notations, we often write $q(\sigma)$ for $\sigma[p(\id)], g/g$ or
$\sigma[p(\id)], u^\id/u$. %
This notation relates similar operations of weakenings and substitutions. %
The global substitutions on types satisfy the following lemma:
\begin{lemma}[Algebra of Global Substitutions]$ $
  \begin{itemize}
  \item $T[\gamma][\gamma'] = T[\gamma \circ \gamma]$
  \item $\sigma[\gamma][\gamma'] = \sigma[\gamma \circ \gamma]$
  \item $T[\sigma][\gamma] = T[\sigma[\gamma]]$
  \item $\Gamma[\sigma][\gamma] = \Gamma[\sigma[\gamma]]$
  \item $t[\sigma][\gamma] = t[\sigma[\gamma]]$
  \item $\delta[\sigma][\gamma] = \delta[\sigma[\gamma]]$
  \end{itemize}
\end{lemma}

Next, we give the application operation of local substitutions on terms and
composition of local substitutions:
\begin{align*}
  x[\delta] &:= \delta(x)
              \tag{lookup of $x$ in $\delta$} \\
  u^{\delta'}[\delta] &:= u^{\delta' \circ \delta} \\
  \ze[\delta] &:= \ze \\
  \su t[\delta] &:= \su{(t[\delta])} \\
  \lambda x. t[\delta] &:= \lambda x. (t[\delta[\id; p(\id)], x/x]) \\
  t~s[\delta] &:= (t[\delta])~(s[\delta]) \\
  \boxit t [\delta] &:= \boxit t \\
  \letbox u s t[\delta] &:= \letbox u {s[\delta]} {(t[\delta[p(\id)]])} \\
  \Lambda g. t [\delta] &:= \Lambda g. (t[\delta[(p(\id))]]) \\
  t~\$~\Gamma [\delta] &:= (t[\delta]) ~\$~\Gamma \\[.25em]
  \wk_g^m \circ \delta &:= \wk_g^{\widehat\delta} \\
  \cdot^m \circ \delta &:= \cdot_{\widecheck{\delta}}^{\widehat\delta} \\
  \cdot_{g}^m \circ \delta &:= \cdot_{g}^{\widehat\delta} \\
  (\delta', t/x) \circ \delta &:= (\delta' \circ \delta), t[\delta]/x
\end{align*}
Similarly, we might also write $q(\delta)$ for $\delta[\id; p(\id)], x/x$. %
Notice that in the definition of composition, we make use of the $\widehat\delta$
function to fetch the number of weakenings. %
This number is used in the application operation of global substitutions given below
in the application of global substitutions to specify the number of local $p$
weakenings when a context variable is substituted by a concrete context. %
In the composition of $\cdot^m$, we use $\widecheck{\delta}$ to query whether
$\delta$'s codomain context starts from a context variable. %
If it does, then we use that context variable in the result of the composition. %
Local substitutions and global weakenings interact in the following way:
\begin{lemma}[Algebra of Local Substitutions] $ $
  \begin{itemize}
  \item $t[\delta][\gamma] = (t[\gamma])[\delta[\gamma]]$
  \item $(\delta \circ \delta')[\gamma] = (\delta[\gamma]) \circ (\delta'[\gamma])$
  \end{itemize}
\end{lemma}

Now, we define the application of global substitutions to terms and local substitutions:
\begin{align*}
  x[\sigma] &:= x \\
  u^{\delta}[\sigma] &:= \sigma(u)[\delta[\sigma]]
                        \tag{lookup of $u$ in $\sigma$} \\
  \ze[\sigma] &:= \ze \\
  \su t[\sigma] &:= \su{(t[\sigma])} \\
  \lambda x. t[\sigma] &:= \lambda x. (t[\sigma]) \\
  t~s[\sigma] &:= (t[\sigma])~(s[\sigma]) \\
  \boxit t [\sigma] &:= \boxit {(t[\sigma])} \\
  \letbox u s t[\sigma] &:= \letbox u {s[\sigma]} {(t[\sigma[p(\id)], u^\id/u])} \\
  \Lambda g. t [\sigma] &:= \Lambda g. (t[\sigma[(p(\id))], g/g]) \\
  t~\$~\Gamma [\sigma] &:= (t[\sigma]) ~\$~(\Gamma[\sigma]) \\[.25em]
  \wk_g^m [\sigma] &:= \id_{\sigma(g)}[\id; p^m(\id)]
                     \tag{defined only when $\sigma(g)$ is a local context}\\
  \cdot^m [\sigma] &:= \cdot^m \\
  \cdot_{g}^m [\sigma] &:= \cdot^{|\Gamma| + m}
  \tag{if $\sigma(g) = \Gamma$ and $\Gamma$ ends with a $\cdot$} \\
  \cdot_{g}^m [\sigma] &:= \cdot_{g'}^{|\Gamma| + m}
  \tag{if $\sigma(g) = \Gamma$ and $\Gamma$ ends with a $g'$} \\
  (\delta, t/x) [\sigma] &:= (\delta[\sigma]), t[\sigma]/x
\end{align*}
In the definition of global substitution operation, we make use of the local identity
substitution, which is defined through local weakening substitutions as below:
\begin{align*}
  \wk^m_{\cdot} &:=  \cdot^m \\
  \wk^m_{g} &:=  \wk_g^m \\
  \wk^m_{\Gamma, x : T} &:= \wk^{1 + m}_{\Gamma}, x/x
\end{align*}
The local identity substitution is defined as a special case of the local weakening
substitutions by setting $m$ to be $0$:
\begin{align*}
  \id_{\Gamma} := \wk^0_{\Gamma}
\end{align*}
Though we make heavy use of symbol overloading, but hopefully the exact meanings of
symbols should be disambiguated by the surrounding textual contexts. %
Usually, the symbols are designed such that their behaviors remain the same for their
ambiguous readings (e.g. various uses of $\id$ and $\wk$).  %

At last, we need to define the global identity substitution before giving the typing
rules. %
The global identity substitution is defined in the same principle; it is a special
case of the global weakening substitutions:
\begin{align*}
  \wk^m_{\cdot} &:=  \cdot \\
  \wk^m_{\Psi, g : \Ctx} &:=  \wk_{\Psi}^{1+m}, g/g \\
  \wk^m_{\Psi, u : (\judge T)} &:= \wk^{1 + m}_{\Psi}, u^{\id_\Gamma[p^{1+m}(\id)]}/u
\end{align*}
Notice that in the cons case for $u : (\judge T)$, the local identity substitution
$\id_{\Gamma}$ must be weakened by $p^{1+m}(\id)$, because its global typing
environment is weakened by $\wk^{1 + m}_{\Psi}$, which takes the same effect. %
This weakening is necessary to make the typing to go through. %
Then identity is just a special case:
\begin{align*}
  \id_{\Psi} := \wk^0_{\Psi}
\end{align*}
We sometimes omit the subscript for different $\id$, as we know their effects on
types, local contexts, terms and local substitutions are just \emph{identity}. 

The composition of global substitutions is defined intuitively:
\begin{align*}
  \cdot \circ \sigma' &:= \cdot \\
  (\sigma, t/u) \circ \sigma' &:= (\sigma \circ \sigma'), t[\sigma']/u \\
  (\sigma, \Gamma/g) \circ \sigma' &:= (\sigma \circ \sigma'), \Gamma[\sigma']/g
\end{align*}
Essentially, composition just iteratively applies $\sigma'$ to all terms and contexts
in the first global substitution. %
Notice that composition can only be defined, after the applications of global
substitutions to both terms and local contexts are defined. %
This definition does lead to some complication when we tries to prove the global
substitution lemma of terms and local substitutions in the next section. %

At last, the following gives the typing rules of terms and local substitutions.
\begin{mathpar}
  \inferrule
  {\ljudge i \Gamma \\ x : T \in \Gamma}
  {\ltyping{i}{x}{T}}

  \inferrule
  {\ltyping i \delta \Delta \\ u : (\judge[\Delta] T) \in \Psi}
  {\ltyping{i}{u^\delta}{T}}

  \inferrule
  {\ljudge i \Gamma}
  {\ltyping{i}{\ze}{\Nat}}

  \inferrule
  {\ltyping{i}{t}{\Nat}}
  {\ltyping{i}{\su t}{\Nat}}
  
  \inferrule
  {\ltyping[\Psi][\Gamma, x : S]{i}{t}{T}}
  {\ltyping{i}{\lambda x. t}{S \func T}}

  \inferrule
  {\ltyping{i}{t}{S \func T} \\ \ltyping{i}{s}{S}}
  {\ltyping{i}{t\ s}{T}}

  \inferrule
  {\ljudge 1 \Gamma \\ \ltyping[\Psi][\Delta] 0 t T}
  {\ltyping{1}{\boxit t}{\cont[\Delta] T}}
  
  \inferrule
  {\ltyping 1 {s}{\cont[\Delta] T} \\ \mhighlight{\ljudge[\Psi] 0 \Delta} \\
    \mhighlight{\ljudge[\Psi] 0 T} \\
    \ljudge[\Psi] 1{T'} \\
    \ltyping[\Psi, u : (\judge[\Delta] T)][\Gamma[p(\id)]] 1 {t}{T'[p(\id)]}}
  {\ltyping 1 {\letbox u s t} T'}

  \inferrule
  {\ljudge[\Psi] 1 \Gamma \\ \ltyping[\Psi, g: \Ctx][\Gamma[p(\id)]] 1 t T}
  {\ltyping{1}{\Lambda g. t}{(g : \Ctx) \STo T}}

  \inferrule
  {\ltyping{1}{t}{(g : \Ctx) \STo T} \\ \ljudge[\Psi] 0 \Delta}
  {\ltyping{1}{t~\$~\Delta}{T[\id_{\Psi}, \Delta/g]}}
  
  \inferrule
  {\ljudge i \Gamma \\ \text{$\Gamma$ ends with $\cdot$} \\ |\Gamma| = m}
  {\ltyping i {\cdot^m}{\cdot}}

  \inferrule
  {\ljudge i \Gamma \\ g : \Ctx \in \Psi \\ \text{$\Gamma$ ends with $g$} \\ |\Gamma| = m}
  {\ltyping i {\cdot_g^m}{\cdot}}

  \inferrule
  {\ljudge i \Gamma \\ g : \Ctx \in \Psi \\ \text{$\Gamma$ ends with $g$} \\ |\Gamma| = m}
  {\ltyping i {\wk_g^m}{g}}

  \inferrule
  {\ltyping i {\delta}{\Delta} \\ \ltyping i {t}{T}}
  {\ltyping i {\delta, t/x}{\Delta, x : T}}
\end{mathpar}
In the typing rules, there are premises highlighted by \mhighlight{\text{shades}}. %
These shaded premises are necessary to establish the theorem of \emph{presupposition}
or \emph{syntactic validity}. %
After establishing presupposition, these premises can be derived from other premises
and thus technically can be omitted afterwards. %
Moreover, in the rule for global variables, the lookup of a global context $\Psi$ must
consider the effect of global weakenings as follows:
\begin{mathpar}
  \inferrule
  { }
  {u : B[p(\id)] \in \Psi, u : B}

  \inferrule
  {u : B \in \Psi}
  {u : B[p(\id)] \in \Psi , u' : B'}
\end{mathpar}
The typing rules also support lifting:
\begin{lemma}[Lifting] $ $
  \begin{itemize}
  \item If $\ltyping{0}{t}{T}$, then $\ltyping{1}{t}{T}$.
  \item If $\ltyping{0}{\delta}{\Delta}$, then $\ltyping{1}{\delta}{\Delta}$.
  \end{itemize}
\end{lemma}
This lemma ensures that terms at layer $0$ are \emph{included} in layer $1$. 

Typing rules for global substitutions are defined as follows:
\begin{mathpar}
  \inferrule*
  {\vdash \Psi}
  {\typing[\Psi]{\cdot}{\cdot}}

  \inferrule*
  {\typing[\Psi]{\sigma}{\Phi} \\ \ljudge[\Psi] 0 \Gamma \\ \ljudge[\Psi] 0 T \\ \ltyping[\Psi][\Gamma[\sigma]] 0 {t}{T[\sigma]}}
  {\typing[\Psi]{\sigma, t/u}{\Phi, u : (\judge T)}}

  \inferrule*
  {\typing[\Psi]{\sigma}{\Phi} \\ \ljudge 0 \Gamma}
  {\typing[\Psi]{\sigma, \Gamma/g}{\Phi, g: \Ctx}}
\end{mathpar}

In the next section, we establish a set of syntactic properties as a basic sanity
check of the definitions, which are also useful in later section proving normalization. 

\subsection{Syntactic Properties of 2-layered Modal Type Theory with Context Variables}

In this section, we list syntactic properties that eventually leads to the
substitution lemma of terms and local substitutions for global substitutions. %
This lemma is our ``benchmark'' to ensure that the rules for the system make sense. %
During the process, we must establish other necessary syntactic properties. %
We elaborate the proofs an important and selected few. %
Other proofs in this section have been mechanized in Agda.

\begin{lemma}\labeledit{lem:cvar:gsubst-skip} $ $
  \begin{itemize}
  \item If $n$ is the length of $\sigma'$, then $t [q^n(p(id))][\sigma,t/u,\sigma'] =
    t [q^n(p(id))][\sigma,\Gamma/g,\sigma'] = t[\sigma,\sigma']$.
  \item If $n$ is the length of $\sigma'$, then $\delta [q^n(p(id))][\sigma,t/u,\sigma'] =
    \delta [q^n(p(id))][\sigma,\Gamma/g,\sigma'] = \delta[\sigma,\sigma']$.
  \end{itemize}
\end{lemma}
This lemma allows to skip a binding in the middle of a global substitution according
to a global weakening.

A similar lemma holds for local substitutions:
\begin{lemma}\labeledit{lem:cvar:lsubst-skip} $ $
  \begin{itemize}
  \item If $n$ is the length of $\delta'$, then $t [q^n(p(id))][\delta,t/x,\delta'] =
    t[\delta,\delta']$.
  \item If $n$ is the length of $\delta'$, then $\delta [q^n(p(id))] \circ (\delta,t/x,\delta') =
    \delta \circ (\delta,\delta')$.
  \end{itemize}
\end{lemma}

\begin{lemma}[Composition of Global Substitutions]$ $
  \begin{itemize}
  \item $T[\sigma][\sigma'] = T[\sigma \circ \sigma']$
  \item $\Gamma[\sigma][\sigma'] = \Gamma[\sigma \circ \sigma']$
  \end{itemize}
\end{lemma}

\begin{lemma}[Composition and Associativity of Local Substitutions] $ $
  \begin{itemize}
  \item $t[\delta][\delta'] = t[\delta \circ \delta']$
  \item $(\delta \circ \delta') \circ \delta'' = \delta \circ (\delta' \circ \delta'')$
  \end{itemize}
\end{lemma}

\begin{lemma}[Typing of Local Weakening Substitutions]
  If $\ljudge i {\Delta, \Gamma}$, then $\ltyping[\Psi][\Delta, \Gamma] i
  {\wk^{|\Gamma|}_{\Delta}}{\Delta}$. 
\end{lemma}
The corollary is the well-typedness of local identity substitution:
\begin{corollary}
  If $\ljudge i \Gamma$, then $\ltyping i {\id_{\Gamma}}{\Gamma}$. 
\end{corollary}

The next few questions require well-formedness or typing judgments to work:
\begin{lemma}\labeledit{lem:cvar:wk-p-n-ty} $ $
  \begin{itemize}
  \item If $\ljudge i T$, then $T[\wk^n_{\Psi}] = T[p^n(\id)]$. 
  \item If $\ljudge i \Gamma$, then $\Gamma[\wk^n_{\Psi}] = \Gamma[p^n(\id)]$. 
  \end{itemize}
\end{lemma}
This lemma proves that the global weakening substitutions behaves exactly like global
weakenings.

\begin{lemma}[Naturality]$ $
  \begin{itemize}
  \item If $\ljudge[\Psi,g:\Ctx] i T$ and $\gamma : \Phi \To_g \Psi$, then
    $T[\id_{\Psi}, \Gamma/g][\gamma] = T[q(\gamma)][\id_{\Phi},\Gamma[\gamma]/g]$.
  \item If $\ljudge[\Psi,g:\Ctx] i \Delta$ and $\gamma : \Phi \To_g \Psi$, then
    $\Delta[\id_{\Psi}, \Gamma/g][\gamma] = \Delta[q(\gamma)][\id_{\Phi},\Gamma[\gamma]/g]$.
  \end{itemize}
\end{lemma}
The naturality lemma finds correspondence in the characterization of a presheaf
category of a type theory in general, which instructs how $q$ weakenings can be used
to swap a global weakening and a global substitution. %

\begin{lemma}[Local Identity]\labeledit{lem:cvar:loc-id}$ $
  \begin{itemize}
  \item If $\ltyping i t T$, then $t[\id_{\Gamma}] = t$. 
  \item If $\ltyping i \delta \Delta$, then $\delta \circ \id_{\Gamma} = \delta$. 
  \end{itemize}
\end{lemma}
This lemma shows that the local identity substitution has no effect on terms and that
the right identity property of local substitutions. %

Next, we establish the global weakening lemma for typing rules:
\begin{lemma}[Global weakenings]\labeledit{lem:cvar:glob-wk-tm} $ $
  \begin{itemize}
  \item If $\ltyping i t T$ and $\gamma : \Psi' \To_g \Psi$, then
    $\ltyping[\Psi'][\Gamma[\gamma]] i {t[\gamma]}{ T[\gamma]}$.
  \item If $\ltyping i \delta \Delta$ and $\gamma : \Psi' \To_g \Psi$, then
    $\ltyping[\Psi'][\Gamma[\gamma]] i {\delta[\gamma]}{\Delta[\gamma]}$.
  \end{itemize}
\end{lemma}
\begin{proof}
  Mutual induction on $\ltyping[\Phi] i t T$ and $\ltyping[\Phi] i \delta \Delta$. %
  We only consider a few interesting cases:
  \begin{itemize}[label=Case]
  \item
    \begin{mathpar}
      \inferrule
      {\ltyping[\Phi] 1 {s}{\cont[\Delta] T} \\ \ljudge[\Phi] 1 {T'} \\ \ltyping[\Phi, u : (\judge[\Delta] T)][\Gamma[p(\id)]] 1 {t}{T'[p(\id)]}}
      {\ltyping[\Phi] 1 {\letbox u s t} T'}
    \end{mathpar}
    \begin{align*}
      & \ltyping[\Psi][\Gamma[\gamma]] 1 {s[\gamma]}{\cont[\Delta] T[\gamma]}
        \byIH \\
      & \ltyping[\Psi, u : (\judge[\Delta[\gamma]]
        T[\gamma])][\Gamma[p(\id)][q(\gamma)]] 1 {t[q(\gamma)]}{T'[p(\id)][q(\gamma)]}
        \byIH \\
      & \ltyping[\Psi, u : (\judge[\Delta[\gamma]]
        T[\gamma])][\Gamma[p(\gamma)]] 1 {t[q(\gamma)]}{T'[p(\gamma)]}
        \tag{by computation} \\
      & \ltyping[\Psi, u : (\judge[\Delta[\gamma]]
        T[\gamma])][\Gamma[\gamma][p(\id)]] 1 {t[q(\gamma)]}{T'[\gamma][p(\id)]} \\
      & \ltyping[\Psi][\Gamma[\gamma]] 1 {\letbox u s t[\gamma]}{T'[\gamma]}
        \tag{by constructor}
    \end{align*}
    
  \item
    \begin{mathpar}
      \inferrule
      {\ltyping[\Phi, g: \Ctx][\Gamma[p(\id)]] 1 t T}
      {\ltyping[\Phi]{1}{\Lambda g. t}{(g : \Ctx) \STo T}}
    \end{mathpar}
    \begin{align*}
      & q(\gamma) : \Psi', g : \Ctx \To_g \Psi, g : \Ctx
        \tag{by typing rules} \\
      & \ltyping[\Psi', g: \Ctx][\Gamma[p(\id)][q(\gamma)]]
        1{t[q(\gamma)]}{T[q(\gamma)]}
        \byIH \\
      & \ltyping[\Psi', g: \Ctx][\Gamma[p(\id) \circ q(\gamma)]]
        1{t[q(\gamma)]}{T[q(\gamma)]}
        \tag{by algebraic law} \\
      & \ltyping[\Psi', g: \Ctx][\Gamma[p(\gamma)]]
        1{t[q(\gamma)]}{T[q(\gamma)]} \\
      & \ltyping[\Psi', g: \Ctx][\Gamma[\gamma][p(\id)]]
        1{t[q(\gamma)]}{T[q(\gamma)]} \\
      & \ltyping[\Psi'][\Gamma[\gamma]]
        1{\Lambda g. (t[q(\gamma)])}{(g : \Ctx) \STo (T[q(\gamma)])}
        \tag{by typing rule}
    \end{align*}
    
  \item
    \begin{mathpar}
      \inferrule
      {\ltyping[\Phi]{1}{t}{(g : \Ctx) \STo T} \\ \ljudge[\Phi] 0 \Delta}
      {\ltyping[\Phi]{1}{t~\$~\Delta}{T[\id_{\Phi}, \Delta/g]}}
    \end{mathpar}
    \begin{align*}
      & \ltyping[\Psi][\Gamma[\gamma]]{1}{t[\gamma]}{(g : \Ctx) \STo T[\gamma]}
        \byIH \\
      & \ljudge[\Psi] 0 {\Delta[\gamma]}
        \tag{by \Cref{lem:cvar:glob-wk-ty}} \\
      &
        \ltyping[\Psi][\Gamma[\gamma]]{1}{(t[\gamma])~\$~(\Delta[\gamma])}{T[q(\gamma)][\id_{\Psi}, \Delta[\gamma]/g]}
        \tag{by constructor} \\
      & \ltyping[\Psi][\Gamma[\gamma]]{1}{(t[\gamma])~\$~(\Delta[\gamma])}{T[\id_{\Phi},
        \Delta/g][\gamma]}
        \tag{by naturality}
    \end{align*}
    
  \end{itemize}
\end{proof}

\begin{lemma}[Global Weakening]
  If $\typing[\Psi] \sigma \Phi$ and $\gamma : \Psi' \To_g \Psi$, then $\typing[\Psi']{\sigma[\gamma]}\Phi$.
\end{lemma}

\begin{lemma}[Global Weakening Substitutions]
  If $\vdash \Psi,\Phi$, then $\typing[\Psi,\Phi]{\wk^{|\Phi|}_{\Psi}}{\Psi}$. 
\end{lemma}
\begin{proof}
  This lemma is actually requires a bit preliminaries to establish and so this is the
  earliest point where this lemma can be proven. %
  From $\vdash \Psi,\Phi$, we know $\vdash \Psi$, which we do induction on. %
  We consider only one case:
  \begin{mathpar}
    \inferrule
    {\ljudge 0 \Gamma \\ \ljudge 0 T}
    {\vdash \Psi, u : (\judge T)}
  \end{mathpar}
  \begin{align*}
    & \wk^{|\Phi|}_{\Psi, u : (\judge T)} = \wk^{1 + |\Phi|}_{\Psi},
      u^{\id_\Gamma[p^{1+|\Phi|}(\id)]}/u
      \tag{by definition} \\
    & \typing[\Psi, u : (\judge T),\Phi]{\wk^{1 + |\Phi|}_{\Psi}}{\Psi}
      \byIH 
  \end{align*}
  At last, we must prove $\ltyping[\Psi, u : (\judge
  T),\Phi][\Gamma[\wk^{1 + |\Phi|}_{\Psi}]] 0
  {u^{\id_\Gamma[p^{1+|\Phi|}(\id)]}}{T[\wk^{1 + |\Phi|}_{\Psi}]}$. %
  But we know that this goal is the same as the following due to
  \Cref{lem:cvar:wk-p-n-ty}:
  \begin{align*}
    \ltyping[\Psi, u : (\judge
    T),\Phi][\Gamma[p^{1 + |\Phi|}(\id)]] 0
    {u^{\id_\Gamma[p^{1+|\Phi|}(\id)]}}{T[p^{1 + |\Phi|}(\id)]}
  \end{align*}
  It remains to prove that the substitution is well-typed:
  \begin{align*}
    \ltyping[\Psi, u : (\judge
    T),\Phi][\Gamma[p^{1 + |\Phi|}(\id)]] 0
    {\id_\Gamma[p^{1+|\Phi|}(\id)]}{\Gamma[p^{1 + |\Phi|}(\id)]}
  \end{align*}
  This goal is immediate due to \Cref{lem:cvar:loc-id,lem:cvar:glob-wk-tm}.
\end{proof}

\begin{corollary}
  If $\vdash \Psi$, then $\typing[\Psi]{\id_{\Psi}}{\Psi}$. 
\end{corollary}

Finally, we can establish the presupposition of terms and local substitutions:
\begin{lemma}[Presupposition]$ $
  \begin{itemize}
  \item If $\ltyping i t T$, then $\ljudge i \Gamma$ and $\ljudge i T$. 
  \item If $\ltyping i \delta \Delta$, then $\ljudge i \Gamma$ and $\ljudge i \Delta$. 
  \end{itemize}
\end{lemma}
\begin{proof}
  We do a mutual induction.
\end{proof}

Next, we need a similar lemma to \Cref{lem:cvar:wk-p-n-ty} but for terms and local
substitutions:
\begin{lemma}\labeledit{lem:cvar:wk-p-n-tm} $ $
  \begin{itemize}
  \item If $\ltyping i t T$, then $t[\wk^n_{\Psi}] = t[p^n(\id)]$. 
  \item If $\ltyping i \delta \Delta$, then $\delta[\wk^n_{\Psi}] = \delta[p^n(\id)]$. 
  \end{itemize}
\end{lemma}
\begin{proof}
  The proof of this lemma requires an intrigued generalization in order to handle
  extensions of global contexts due to $\tletbox$ and $\Lambda$. %
  Details of the generalization are technical and too elaborate to put in this
  technical report, and thus we choose to leave them in the Agda mechanization for
  readers' reference. 
\end{proof}

Next, we should verify the identity rules of composition of global substitutions. %
Notice that by applying \Cref{lem:cvar:wk-p-n-ty,lem:cvar:wk-p-n-tm}, we obtain the
right identity immediately:
\begin{lemma}
  If $\typing[\Psi] \sigma \Phi$, then $\sigma \circ \id_{\Psi} = \sigma$. 
\end{lemma}

The left identity, on the other hand, requires certain generalization which
must incorporate global weakenings. %
We again leave the details in the Agda mechanization:
\begin{lemma}
  If $\typing[\Psi] \sigma \Phi$, then $\id_{\Phi} \circ \sigma = \sigma$. 
\end{lemma}

Another useful equation is that global substitutions and local weakenings commute:
\begin{lemma}[Commutativity of Global Substitutions and Local Weakenings] $ $
  \begin{itemize}
  \item If $\ltyping i t T$ and $\tau : \Psi ; \Delta \To_i \Gamma$, then
    $t[\tau][\sigma] = t[\sigma][\tau]$.
  \item If $\ltyping i \delta {\Gamma'}$ and $\tau : \Psi ; \Delta \To_i \Gamma$, then
    $\delta[\tau][\sigma] = \delta[\sigma][\tau]$.
  \end{itemize}
\end{lemma}

Next, we move on to the global substitution lemma for terms and local substitutions. %
Prior to that, we must first show the local weakening lemma and local substitution lemma:
\begin{lemma}[Local Weakenings] $ $
  \begin{itemize}
  \item If $\ltyping i t T$ and $\tau : \Psi; \Delta \To_i \Gamma$, then
    $\ltyping[\Psi][\Delta] i{t[\tau]} T$. 
  \item If $\ltyping i \delta {\Delta'}$ and $\tau : \Psi; \Delta \To_i \Gamma$, then
    $\ltyping[\Psi][\Delta] i{\delta[\tau]}{\Delta'}$. 
  \end{itemize}
\end{lemma}

\begin{lemma}[Local Substitutions] $ $
  \begin{itemize}
  \item If $\ltyping i t T$ and $\ltyping[\Psi][\Delta] i \delta \Gamma$, then
    $\ltyping[\Psi][\Delta] i{t[\delta]} T$. 
  \item If $\ltyping i \delta \Delta$ and $\ltyping[\Psi][\Gamma'] i {\delta'} \Gamma$, then
    $\ltyping[\Psi][\Gamma'] i{\delta \circ \delta'} {\Delta}$. 
  \end{itemize}
\end{lemma}
\begin{proof}
  We do a mutual induction. %
  Notice that in this lemma, it is somewhat more cumbersome to establish the proof for
  local substitutions. %
  When $\delta = \cdot^m$, then we must reason about the properties of
  $\widecheck{\delta'}$. %
  If $\widecheck{\delta'} = g$ for some $g$, then we must show that both $\Gamma$ and
  $\Gamma'$ start with this $g$. %
  The details are given in the Agda mechanization. 
\end{proof}

Finally, we give the global substitution lemma:
\begin{lemma}[Global Substitutions] $ $
  \begin{itemize}
  \item If $\ltyping i t T$ and $\typing[\Psi']\sigma{\Psi}$, then
    $\ltyping[\Psi'][\Gamma[\sigma]] i{t[\sigma]}{T[\sigma]}$. 
  \item If $\ltyping i \delta \Delta$ and $\typing[\Psi']\sigma{\Psi}$, then
    $\ltyping[\Psi'][\Gamma[\sigma]] i{\delta[\sigma]}{\Delta[\sigma]}$. 
  \end{itemize}
\end{lemma}
\begin{proof}
  We do a mutual induction. %
  We consider a few interesting cases:
  \begin{itemize}[label=Case]
  \item
    \begin{mathpar}
      \inferrule
      {\ltyping 1 {s}{\cont[\Delta] T} \\
        \ljudge[\Psi] 1{T'} \\
        \ltyping[\Psi, u : (\judge[\Delta] T)][\Gamma[p(\id)]] 1 {t}{T'[p(\id)]}}
      {\ltyping 1 {\letbox u s t} T'}
    \end{mathpar}
    \begin{align*}
      & \ltyping[\Psi'][\Gamma[\sigma]] 1 {s[\sigma]}{\cont[\Delta] T [\sigma] = \cont[\Delta[\sigma]]{T[\sigma]} }
        \byIH \\
      & \typing[\Psi', u : (\judge[\Delta[\sigma]]{T[\sigma]})]{\sigma[p(\id)],
        u^{\id_{\Delta[\sigma[p(\id)]]}}/u}{\Psi, u : (\judge[\Delta] T)}
        \tag{by typing rule} \\
      & \ltyping[\Psi', u : (\judge[\Delta[\sigma]]{T[\sigma]})][\Gamma[p(\id)][\sigma[p(\id)],
        u^{\id_{\Delta[\sigma[p(\id)]]}}/u]] 1 {t[\sigma[p(\id)],
        u^{\id_{\Delta[\sigma[p(\id)]]}}/u]}{T'[p(\id)][\sigma[p(\id)],
        u^{\id_{\Delta[\sigma[p(\id)]]}}/u]}
        \byIH \\
      & \ltyping[\Psi', u : (\judge[\Delta[\sigma]]{T[\sigma]})][\Gamma[\sigma[p(\id)]]] 1 {t[\sigma[p(\id)],
        u^{\id_{\Delta[\sigma[p(\id)]]}}/u]}{T'[\sigma[p(\id)]]}
        \tag{by \Cref{lem:cvar:gsubst-skip}} \\
      & \ltyping[\Psi', u : (\judge[\Delta[\sigma]]{T[\sigma]})][\Gamma[\sigma][p(\id)]] 1 {t[\sigma[p(\id)],
        u^{\id_{\Delta[\sigma][p(\id)]}}/u]}{T'[\sigma][p(\id)]}
        \tag{by algebraic laws} \\
      & \ltyping[\Psi'][\Gamma[\sigma]] 1 {\letbox u s t}{T'[\sigma]}
        \tag{by typing rule}
    \end{align*}
    
  \item
    \begin{mathpar}
      \inferrule
      {\ljudge[\Psi] 1 \Gamma \\ \ltyping[\Psi, g: \Ctx][\Gamma[p(\id)]] 1 t T}
      {\ltyping{1}{\Lambda g. t}{(g : \Ctx) \STo T}}
    \end{mathpar}
    Basically we do the same as above, but instead of providing a global substitution
    of terms, we provide a global substitution of a local context.
    
  \item
    \begin{mathpar}
      \inferrule
      {\ltyping{1}{t}{(g : \Ctx) \STo T} \\ \ljudge[\Psi] 0 \Delta}
      {\ltyping{1}{t~\$~\Delta}{T[\id_{\Psi}, \Delta/g]}}
    \end{mathpar}
    \begin{align*}
      &\ltyping[\Psi'][\Gamma[\sigma]]{1}{t[\sigma]}{(g : \Ctx) \STo T[\sigma] = (g :
      \Ctx) \STo (T[\sigma[p(\id)], g/g])}
      \byIH \\
      & \ljudge[\Psi'] 0 {\Delta[\sigma]}
        \byIH \\
      & \ltyping[\Psi'][\Gamma[\sigma]]{1}{t~\$~\Delta[\sigma]}{T[\sigma[p(\id)],
        g/g][\id_{\Psi'}, \Delta[\sigma]/g]}
        \tag{by typing rule}
    \end{align*}
    Finally we must show the following equation:
    \begin{align*}
      T[\sigma[p(\id)],
      g/g][\id_{\Psi'}, \Delta[\sigma]/g] =
      T[\id_{\Psi}, \Delta/g][\sigma]
    \end{align*}
    We reason as follows:
    \begin{align*}
      T[\sigma[p(\id)], g/g][\id_{\Psi'}, \Delta[\sigma]/g]
      &= T[(\sigma[p(\id)], g/g) \circ (\id_{\Psi'}, \Delta[\sigma]/g)] \\
      &= T[(\sigma[p(\id)] \circ (\id_{\Psi'}, \Delta[\sigma]/g)), \Delta[\sigma] ] \\
      &= T[(\sigma \circ \id_{\Psi'}), \Delta[\sigma] ]
      \tag{by \Cref{lem:cvar:gsubst-skip}}\\
      &= T[\sigma, \Delta[\sigma] ]
        \tag{by right identity}
    \end{align*}
    On the right hand side,
    \begin{align*}
      T[\id_{\Psi}, \Delta/g][\sigma]
      &= T[(\id_{\Psi} \circ \sigma), \Delta[\sigma]/g] \\
      &= T[\sigma, \Delta[\sigma]/g]
        \tag{by left identity}
    \end{align*}
    Thus both sides agree. 
  \item
    \begin{mathpar}
      \inferrule
      {\ljudge i \Gamma \\ g : \Ctx \in \Psi \\ \text{$\Gamma$ ends with $g$} \\ |\Gamma| = m}
      {\ltyping i {\cdot_g^m}{\cdot}}
    \end{mathpar}
    In this case, we must look up $g$ in $\sigma$, and branch depending on
    $\sigma(g)$.
    \begin{itemize}[label=Subcase]
    \item
      If $\sigma(g) = g', \Delta$ meaning that $\sigma(g)$ ends with a context
      variable $g'$, then we must construct a typing judgment for $\cdot_{g'}^{m'}$
      for some $m'$. %
      In this case, say if $\Gamma = g, \Gamma'$, then $\Gamma[\sigma] = g', \Delta,
      (\Gamma'[\sigma])$. %
      Therefore, $m' = |\Delta| + |\Gamma'[\sigma]| = |\Delta| + m$. 
    \item If $\sigma(g) = \cdot, \Delta$, then we proceed similarly except that we
      must construct a typing judgment for $\cdot^{m'}$ instead. 
    \end{itemize}
    
  \end{itemize}
\end{proof}

\begin{lemma}[Distributivity of Global Substitutions] $ $
  \begin{itemize}
  \item If $\ltyping i t T$, $\ltyping[\Psi][\Delta] i \delta \Gamma$ and
    $\typing[\Phi] \sigma \Psi$, then
    $t[\delta][\sigma] = (t[\sigma][\delta[\sigma]])$.
  \item If $\ltyping i \delta \Delta$, $\ltyping[\Psi][\Gamma'] i {\delta'} \Gamma$
    and $\typing[\Phi] \sigma \Psi$, then
    $(\delta \circ \delta') [\sigma] = (\delta[\sigma]) \circ (\delta'[\sigma])$.
  \end{itemize}
\end{lemma}
\begin{proof}
  We do a mutual induction on $\ltyping i t T$ and $\ltyping i \delta \Delta$. %
  The difficulty is coming from the base cases of local substitutions. %
  We must incorporate the shape of $\delta'$ depending on the cases of $\delta$. %
  We refer the readers to the Agda mechanization for the detailed proof and how we do
  case analysis on the base cases. 

\end{proof}

\subsection{Equivalence Rules}

In this section, we describe the equivalence rules. %
They follow closely to the equivalence rules by \citet[Sec. 4]{hu2024layered}. %
We only show the rules for the newly added constructs.
\begin{mathpar}
  \inferrule
  {\ljudge[\Psi] 1 \Gamma \\ \ltyequiv[\Psi, g: \Ctx][\Gamma[p(\id)]] 1 t {t'} T}
  {\ltyequiv{1}{\Lambda g. t}{\Lambda g. t'}{(g : \Ctx) \STo T}}

  \inferrule
  {\ltyequiv{1}{t}{t'}{(g : \Ctx) \STo T} \\ \ljudge[\Psi] 0 \Delta}
  {\ltyequiv{1}{t~\$~\Delta}{t'~\$~\Delta}{T[\id_{\Psi}, \Delta/g]}}

  \inferrule
  {\ljudge[\Psi] 1 \Gamma \\ \ltyping[\Psi, g: \Ctx][\Gamma[p(\id)]] 1 t T \\ \ljudge[\Psi] 0 \Delta}
  {\ltyequiv{1}{(\Lambda g. t)~\$~\Delta}{t[\id_{\Psi}, \Delta/g]}{T[\id_{\Psi}, \Delta/g]}}

  \inferrule
  {\ltyping{1}{t}{(g : \Ctx) \STo T}}
  {\ltyequiv{1}{t}{\Lambda g. (t[p(\id)])~\$~g}{(g : \Ctx) \STo T}}
\end{mathpar}
In the rules above, we specify the congruence for meta-abstraction $\Lambda$ and the
meta-application. %
Moreover, they also have $\beta$ and $\eta$ rules in the expected ways.

We also have a equivalence judgment for local substitutions
\[
  \ltyequiv 1 \delta{\delta'} \Delta
\]
The only rules are the congruence rules determined by all possible constructors. %
Both equivalence judgments for terms and local substitutions must be mutually
defined. %

Note that there is no need to define the equivalence for global substitutions. %
Effectively, the equivalence for global substitutions is defined as the equality. %
This is because all terms stored in a global substitution are at layer $0$, and thus
they do not have meaningful dynamics. %
Therefore, equivalent global substitutions must also be equal. %

We first establish the presupposition lemma for equivalence:
\begin{lemma}[Presupposition] $ $
  \begin{itemize}
  \item If $\ltyequiv 1 t {t'} T$, then $\ltyping 1 t T$ and $\ltyping 1 {t'} T$.
  \item If $\ltyequiv 1 \delta {\delta'} \Delta$, then $\ltyping 1 \delta \Delta$ and $\ltyping 1 {\delta'} \Delta$.
  \end{itemize}
\end{lemma}
\begin{proof}
  We perform a mutual induction. %
  For the $\beta$ and $\eta$ rules, we apply substitution lemmas proved in the
  previous section. 
\end{proof}

\begin{lemma}[Local Weakenings] $ $
  \begin{itemize}
  \item If $\ltyequiv[\Psi][\Delta] 1 t{t'} T$ and $\tau : \Psi; \Gamma \To_1 \Delta$,
    then $\ltyequiv 1 {t[\tau]}{t'[\tau]} T$.
  \item If $\ltyequiv[\Psi][\Delta] 1 \delta{\delta'} \Delta'$ and
    $\tau : \Psi; \Gamma \To_1 \Delta$, then
    $\ltyequiv 1 {\delta[\tau]}{\delta'[\tau]} \Delta'$.
  \end{itemize}
\end{lemma}
\begin{proof}
  We follow the local weakening property for typing above. 
\end{proof}

A counterpart is w.r.t to global weakenings:
\begin{lemma}[Global Weakenings] $ $
  \begin{itemize}
  \item If $\ltyequiv[\Phi] 1 t{t'} T$ and $\gamma : \Psi \To_g \Phi$,
    then $\ltyequiv[\Psi][\Gamma[\gamma]] 1 {t[\gamma]}{t'[\gamma]}{T[\gamma]}$.
  \item If $\ltyequiv[\Phi] 1 \delta{\delta'} \Delta$ and $\gamma : \Psi \To_g \Phi$,
    then $\ltyequiv[\Psi][\Gamma[\gamma]] 1 {\delta[\gamma]}{\delta'[\gamma]}{\Delta[\gamma]}$.
  \end{itemize}
\end{lemma}
\begin{proof}
  Similarly, we follow the global weakening property for typing above. 
\end{proof}

\begin{lemma}[Congruence of Local Substitutions] $ $
  \begin{itemize}
  \item If $\ltyping 1 t T$ and $\ltyequiv[\Psi][\Delta] 1 \delta{\delta'} \Gamma$, then
    $\ltyequiv[\Psi][\Delta] 1{t[\delta]}{t[\delta']} T$. 
  \item If $\ltyping 1 {\delta} {\Delta'}$ and
    $\ltyequiv[\Psi][\Delta] 1 {\delta'}{\delta''} \Gamma$, then
    $\ltyequiv[\Psi][\Delta] 1{\delta \circ \delta'}{\delta \circ \delta''}
    {\Delta'}$.
  \end{itemize}
\end{lemma}
\begin{proof}
  We do a mutual induction on $\ltyping 1 t T$ and $\ltyping 1 {\delta} {\Delta'}$. %
  We only consider a few interesting cases.
  \begin{itemize}[label=Case]
  \item
    \begin{mathpar}
      \inferrule
      {\ltyping 1 {s}{\cont[\Delta'] T} \\
        \ljudge[\Psi] 1{T'} \\
        \ltyping[\Psi, u : (\judge[\Delta'] T)][\Gamma[p(\id)]] 1 {t}{T'[p(\id)]}}
      {\ltyping 1 {\letbox u s t} T'}
    \end{mathpar}
    In this case, we must use the global weakening lemma above to derive
    \[
      \ltyequiv[\Psi][\Delta[p(\id)]] 1 {\delta[p(\id)]}{\delta'[p(\id)]}{\Gamma[p(\id)]}
    \]
    and then we apply IH.
    
  \item
    For all base cases of local substitutions, we realize that given
    $\ltyequiv[\Psi][\Delta] 1 \delta{\delta'} \Gamma$,
    we have
    \begin{itemize}
    \item $\widecheck{\delta} = \widecheck{\delta'}$, and
    \item $\widehat\delta = \widehat{\delta'}$.
    \end{itemize}
    because they characterize $\Delta$, so the exact $\delta$ and $\delta'$ are
    irrelevant. 
    
    The target goal follows immediate. 
    
  \end{itemize}
\end{proof}

\begin{lemma}[Local Substitutions] $ $
  \begin{itemize}
  \item If $\ltyequiv 1 t {t'} T$ and $\ltyequiv[\Psi][\Delta] 1 \delta{\delta'} \Gamma$, then
    $\ltyequiv[\Psi][\Delta] 1{t[\delta]}{t'[\delta']} T$. 
  \item If $\ltyequiv 1 {\delta''} {\delta'''} {\Delta'}$ and
    $\ltyequiv[\Psi][\Delta] 1 \delta{\delta'} \Gamma$, then
    $\ltyequiv[\Psi][\Delta] 1{\delta'' \circ \delta}{\delta''' \circ \delta'}
    {\Delta'}$.
  \end{itemize}
\end{lemma}
\begin{proof}
  We proceed by a mutual induction on $\ltyequiv 1 t {t'} T$ and
  $\ltyequiv 1 {\delta''} {\delta'''} {\Delta'}$. %
  We only look into the $\beta$ and $\eta$ rule for meta-functions, because we cannot
  apply IH:
  \begin{itemize}[label=Case]
  \item
    \begin{mathpar}
      \inferrule
      {\ljudge[\Psi] 1 \Gamma \\ \ltyping[\Psi, g: \Ctx][\Gamma[p(\id)]] 1 t T \\ \ljudge[\Psi] 0 \Delta'}
      {\ltyequiv{1}{(\Lambda g. t)~\$~\Delta'}{t[\id_{\Psi}, \Delta'/g]}{T[\id_{\Psi}, \Delta'/g]}}
    \end{mathpar}
    \begin{align*}
      & \ltyequiv[\Psi, g: \Ctx][\Delta[p(\id)]] 1
        {t[\delta[p(\id)]]}{t[\delta'[p(\id)]]} T
        \tag{by local substitution lemma} \\
      & \ltyequiv[\Psi][\Delta]{1}{(\Lambda g. (t[\delta[p(\id)]]))~\$~\Delta'}{t[\delta'[p(\id)]][\id_{\Psi}, \Delta'/g]}{T[\id_{\Psi}, \Delta'/g]}
    \end{align*}
    Now we have to align up the right hand side. %
    The target right hand side is
    \[
      t[\id_{\Psi}, \Delta'/g][\delta']
    \]
    We reason as follows
    \begin{align*}
      t[\delta'[p(\id)]][\id_{\Psi}, \Delta'/g]
      &= t[\id_{\Psi}, \Delta'/g][\delta'[p(\id)][\id_{\Psi}, \Delta'/g]]
        \tag{by distributivity of global substitutions} \\
      &= t[\id_{\Psi}, \Delta'/g][\delta']
        \tag{by \Cref{lem:cvar:gsubst-skip}}
    \end{align*}
    and we have the target goal. 
  \item
    \begin{mathpar}
      \inferrule
      {\ltyping{1}{t}{(g : \Ctx) \STo T}}
      {\ltyequiv{1}{t}{\Lambda g. (t[p(\id)])~\$~g}{(g : \Ctx) \STo T}}
    \end{mathpar}
    We apply the local weakening lemma:
    \begin{align*}
      & \ltyequiv[\Psi][\Delta] 1 {t[\delta]}{t[\delta']} T 
    \end{align*}
    On the right hand side, we have
    \begin{align*}
      \Lambda g. (t[p(\id)])~\$~g[\delta']
      & = \Lambda g. (t[p(\id)][\delta'[p(\id)]])~\$~g \\
      & = \Lambda g. (t[\delta'][p(\id)])~\$~g
        \tag{by algebraic rule}
    \end{align*}
    
  \end{itemize}
\end{proof}

Next, we consider the global substitution lemma:
\begin{lemma}[Global Substitutions] $ $
  \begin{itemize}
  \item If $\ltyequiv 1 t {t'} T$ and $\typing[\Phi]\sigma\Psi$, then
    $\ltyequiv[\Phi][\Gamma[\sigma]] 1{t[\sigma]}{t'[\sigma]}{T[\sigma]}$. 
  \item If $\ltyequiv 1{\delta} {\delta'}{\Delta}$ and $\typing[\Phi]\sigma\Psi$, then
    $\ltyequiv[\Phi][\Gamma[\sigma]] 1{\delta[\sigma]}{\delta'[\sigma]}{\Delta[\sigma]}$. 
  \end{itemize}
\end{lemma}
\begin{proof}
  As we have noted above, there is no need for an equivalence judgment between two
  global substitutions. %
  Therefore we directly apply the same $\sigma$ on both sides. %
  Let us consider the $\beta$ and $\eta$ rules for meta-functions:
  \begin{itemize}[label=Case]
  \item
    \begin{mathpar}
      \inferrule
      {\ljudge[\Psi] 1 \Gamma \\ \ltyping[\Psi, g: \Ctx][\Gamma[p(\id)]] 1 t T \\ \ljudge[\Psi] 0 \Delta'}
      {\ltyequiv{1}{(\Lambda g. t)~\$~\Delta'}{t[\id_{\Psi}, \Delta'/g]}{T[\id_{\Psi}, \Delta'/g]}}
    \end{mathpar}
    \begin{align*}
      & \typing[\Phi, g : \Ctx]{\sigma[p(\id)], g/g}{\Psi, g : \Ctx} \\
      & \ltyping[\Phi, g: \Ctx][\Gamma[p(\id)][\sigma[p(\id)], g/g]] 1
        {t[\sigma[p(\id)], g/g]}{T[\sigma[p(\id)], g/g]}
        \tag{by global substitution lemma} \\
      & \ltyping[\Phi, g: \Ctx][\Gamma[\sigma][p(\id)]] 1
        {t[\sigma[p(\id)], g/g]}{T[\sigma[p(\id)], g/g]} \\
      &  \ljudge[\Phi] 0 {\Delta'[\sigma]}\\
      & \ltyequiv[\Phi][\Gamma[\sigma]]{1}{(\Lambda g. (t[\sigma[p(\id)], g/g]))~\$~(\Delta'[\sigma])}{t[\sigma[p(\id)], g/g][\id_{\Psi}, \Delta'[\sigma]/g]}{T[\sigma[p(\id)], g/g][\id_{\Psi}, \Delta'[\sigma]/g]}
    \end{align*}
    We now reason a few equations:
    \begin{align*}
      (\Lambda g. (t[\sigma[p(\id)], g/g]))~\$~(\Delta'[\sigma])
      = (\Lambda g. t)~\$~\Delta' [\sigma]
    \end{align*}
    and
    \begin{align*}
      t[\sigma[p(\id)], g/g][\id_{\Psi}, \Delta'[\sigma]/g]
      & = t[(\sigma[p(\id)], g/g) \circ (\id_{\Psi}, \Delta'[\sigma]/g)]  \\
      & = t[\sigma[p(\id)] \circ (\id_{\Psi}, \Delta'[\sigma]/g), \Delta'[\sigma]/g]  \\
      & = t[\sigma \circ \id_{\Psi}, \Delta'[\sigma]/g]  \\
      & = t[\id_{\Phi} \circ \sigma, \Delta'[\sigma]/g]  \\
      & = t[\id_{\Phi}, \Delta'/g][\sigma] 
    \end{align*}
    Similarly, we have
    \begin{align*}
      T[\sigma[p(\id)], g/g][\id_{\Psi}, \Delta'[\sigma]/g]
      = T[\id_{\Phi}, \Delta'/g][\sigma]
    \end{align*}
    This concludes the goal. 
  \item
    \begin{mathpar}
      \inferrule
      {\ltyping{1}{t}{(g : \Ctx) \STo T}}
      {\ltyequiv{1}{t}{\Lambda g. (t[p(\id)])~\$~g}{(g : \Ctx) \STo T}}
    \end{mathpar}
    \begin{align*}
      & \ltyping[\Phi][\Gamma[\sigma]]{1}{t[\sigma]}{(g : \Ctx) \STo T[\sigma]} \\
      & \ltyping[\Phi][\Gamma[\sigma]]{1}{t[\sigma]}{(g : \Ctx) \STo (T[\sigma[p(\id)],g/g])} \\
      & \ltyequiv[\Phi][\Gamma[\sigma]]{1}{t[\sigma]}{\Lambda g. t[\sigma][p(\id)]~\$~g}{(g : \Ctx) \STo (T[\sigma[p(\id)],g/g])} 
    \end{align*}
    We reason the right hand side
    \begin{align*}
      \Lambda g. (t[p(\id)])~\$~g[\sigma]
      &= \Lambda g. (t[p(\id)][\sigma[p(\id)], g/g])~\$~g \\
      &= \Lambda g. (t[\sigma[p(\id)]])~\$~g \\
      &= \Lambda g. (t[\sigma][p(\id)])~\$~g
    \end{align*}
    and hence establish the goal.
  \end{itemize}
\end{proof}

\subsection{Weak Head Reduction}

To compute the decide whether two terms are equivalence based on the typing rules
above, we take the approach of reducibility candidates. %
This approach requires a reduction strategy, and then we use a type directed
convertibility checking to determine whether two normal forms after reduction are
equivalent. %
In this approach, it is sufficient to have a term reduced to weak head normal forms,
specified as below:
\begin{alignat*}{2}
  w &:= &&\ v \sep \ze \sep \su t \sep \boxit t \sep \lambda x. t \sep \Lambda g. t
  \tag{Weak head normal form ($\Nf$)} \\
  v &:= &&\ x \sep u^\delta \sep v\ t \sep \letbox u v t \sep v~\$~\Gamma \tag{Neutral form
    ($\Ne$)}
\end{alignat*}
Basically the one-step weak head reduction simply takes the $\beta$ rules from the
equivalence relations and adds head reduction:

Head reductions:
\begin{mathpar}
  \inferrule*
  {\ltyred{t}{t'}{S \func T} \\ \ltyping 1{s}{S}}
  {\ltyred{t\ s}{t'\ s}{T}}

  \inferrule*
  {\ltyred {s}{s'}{\square T} \\ \ltyping[\Psi, u : T][\Gamma[p(\id)]] 1 {t}{T'[p(\id)]}}
  {\ltyred {\letbox u s t}{\letbox u{s'}{t}}{T'}}

  \inferrule*
  {\ltyred{t}{t'}{(g : \Ctx) \STo T} \\ \ljudge[\Psi] 0 \Delta}
  {\ltyred{t~\$~\Delta}{t'~\$~\Delta}{T[\id_{\Psi}, \Delta/g]}}
\end{mathpar}

$\beta$ reductions:
\begin{mathpar}
  \inferrule*
  {\ltyping[\Psi][\Gamma, x : S] 1 t T \\ \ltyping 1 s S}
  {\ltyred{(\lambda x. t)\ s}{t[\id_{\Gamma}, s/x]}{T}}

  \inferrule*
  {\ljudge[\Psi] 1 \Gamma \\ \ltyping[\Psi, g: \Ctx][\Gamma[p(\id)]] 1 t T \\ \ljudge[\Psi] 0 \Delta}
  {\ltyred{(\Lambda g. t)~\$~\Delta}{t[\id_{\Psi}, \Delta/g]}{T[\id_{\Psi}, \Delta/g]}}

  \inferrule*
  {\ltyping[\Psi][\cdot] 0 s T \\ \ljudge 1 \Gamma \\ \ljudge 1 T' \\ \ltyping[\Psi, u : T][\Gamma[p(\id)]] 1 {t}{T'[p(\id)]}}
  {\ltyred{\letbox u {\boxit s} t}{t[\id_{\Psi}, s/u]}{T'}}
\end{mathpar}

Notice that weak head reduction only occurs at layer $1$, so we do not need to rewrite
down the layer explicitly. %
This is because $\beta$ reduction only occurs at layer $1$ and all terms at layer $0$
are identified by their exact syntactic structure. %
We define the reflexive transitive closure $\ltyreds t{t'} T$ of weak head reduction
in the usual way. %
In an implementation, we would repeatedly do one-step weak head reduction, until it
reaches a normal form. %
To compare two terms of the same type, we would first compute the weak head normal
forms of both sides, and then based on their type, we perform a type-directed
convertibility check. %
This check also performs $\eta$ expansion when necessary, so the resulting algorithm
is complete w.r.t. the equivalence relation. %

Notice that the weak head reduction relation is a subrelation of the equivalence
above. %
Therefore, several theorems are simply carried over. %
We omit the proofs and only state the theorems below:
\begin{lemma}[Local Weakenings]
  If $\ltyred[\Psi][\Delta] t{t'} T$ and $\tau : \Psi; \Gamma \To_1 \Delta$,
    then $\ltyred {t[\tau]}{t'[\tau]} T$.
\end{lemma}

\begin{lemma}[Global Weakenings]
  If $\ltyred[\Phi] t{t'} T$ and $\gamma : \Psi \To_g \Phi$,
  then $\ltyred[\Psi][\Gamma[\gamma]] {t[\gamma]}{t'[\gamma]}{T[\gamma]}$.
\end{lemma}

\begin{theorem}[Preservation]
  If $\ltyred t {t'} T$, then $\ltyping 1 t T$ and $\ltyping 1 {t'} T$. 
\end{theorem}

\begin{lemma}[Local Substitutions]
  If $\ltyred t {t'} T$ and $\ltyping[\Psi][\Delta] 1 \delta \Gamma$, then
  $\ltyred[\Psi][\Delta] {t[\delta]}{t'[\delta]} T$. 
\end{lemma}
\begin{lemma}[Globale Substitutions]
  If $\ltyred t {t'} T$ and $\typing[\Phi]\sigma\Psi$, then
    $\ltyred[\Phi][\Gamma[\sigma]]{t[\sigma]}{t'[\sigma]}{T[\sigma]}$. 
\end{lemma}
\begin{lemma}[Uniqueness]
  If $\ltyred t {t'} T$ and $\ltyred t {t''} T$, then $t' = t''$. 
\end{lemma}
\begin{proof}
  Induction on $\ltyred t {t'} T$ and analyze $\ltyred t {t''} T$. 
\end{proof}

These conclusions can be generalized to its reflexive transitive closure. 

In the next section, we will first look into the logical relations which establish
normalization property, and then we show that this checking strategy indeed is
complete.

\section{Normalization and Convertibility}\labeledit{sec:cv:logrel}

In the previous section, we have developed the syntactic theory of context
variables and the weak head reduction relation. %
In this section, our goal is to show the termination of the weak head reduction,
i.e. normalization, develop a type-directed convertibility checking algorithm, and
show that this algorithm is sound and complete w.r.t. the equivalence of terms. %
In the first step to normalization, we establish the reducibility predicate. %
This predicate is defined inductively and then recursively in order to describe the
semantic well-formedness of types and terms. %
Then from the reducibility predicate, we prove the semantic soundness of the model and
prove the normalization property as a corollary. %
Next, we give the convertibility algorithm as a judgment. %
This judgment will be shown equivalent to the equivalence of terms, which wraps up
our discussion on simple types.

\subsection{Generic Equivalence}

We follow \citet{abel_decidability_2017} to define a modular generic equivalence to ease
subsequent proof constructions. %
Since we do not have type-level computation for this almost simply typed theory, we
only have to be concerned about two kinds of equivalence over terms:
$\ltygteq t {t'} T$ describes a generic equivalence between two terms, and
$\ltygneeq v {v'} T$ describes a generic equivalence between two neutral terms. %
Furthermore, $\ltygteq t {t'} T$ is generalized to $\ltygteq{\delta}{\delta'}\Delta$
by recursion on $\Delta$, which denotes a generic equivalence between local
substitutions and the base cases are just congruence. %
The definition is
\begin{mathpar}
  \inferrule
  {\ljudge 1 \Gamma \\ \text{$\Gamma$ ends with $\cdot$} \\ |\Gamma| = m}
  {\ltygteq{\cdot^m}{\cdot^m}{\cdot}}

  \inferrule
  {\ljudge 1 \Gamma \\ g : \Ctx \in \Psi \\ \text{$\Gamma$ ends with $g$} \\ |\Gamma| = m}
  {\ltygteq{\cdot_g^m}{\cdot_g^m}{\cdot}}

  \inferrule
  {\ljudge 1 \Gamma \\ g : \Ctx \in \Psi \\ \text{$\Gamma$ ends with $g$} \\ |\Gamma| = m}
  {\ltygteq{\wk_g^m}{\wk_g^m}{g}}

  \inferrule
  {\ltygteq{\delta}{\delta'}{\Delta} \\ \ltygteq{t}{t'}T}
  {\ltygteq{\delta, t/x}{\delta',t'/x}{\Delta, x: T}}
\end{mathpar}

The generic equivalence will be instantiated twice, once for each layer. %
Following the layering principle, the laws for layer $0$ are subsumed by those for
layer $1$. %
The reason why we also need an instantiation for layer $0$ is that code from layer $0$
can be lifted by $\tletbox$ and become a program, so that its computation is
recovered. %
Therefore, a logical relation is needed to capture its computational behavior. %

These two relations must satisfy the following laws. %
Law statements for $\square$ and meta-functions only apply when the generic
equivalence is indexed by layer $1$. 

\begin{law}[Subsumption]$ $
  \begin{itemize}
  \item If $\ltygneeq v {v'} T$, then $\ltygteq v {v'} T$.
  \item If $\ltygteq t {t'} T$, then $\ltyequiv 1 t {t'} T$.
  \end{itemize}
\end{law}

From the subsumption of generic equivalence of terms, we have the subsumption of
generic equivalence of local substitutions as a lemma:
\begin{lemma}[Subsumption]
  If $\ltygteq \delta {\delta'} \Delta$, then $\ltyequiv 1 \delta {\delta'} \Delta$.  
\end{lemma}

\begin{law}[PER]
  Both relations are PERs. 
\end{law}

\begin{law}[Monotonicity]
  Given $\gamma; \tau : \Psi; \Gamma \To \Phi; \Delta$, if
  $\ltygteq[\Phi][\Delta] t {t'} T$, or $\ltygneeq[\Phi][\Delta] v {v'} T$, then
  $\ltygteq{t[\gamma;\tau]}{t'[\gamma;\tau]}{T[\gamma]}$, or
  $\ltygneeq{v[\gamma;\tau]}{v'[\gamma;\tau]}{T[\gamma]}$, respectively.
\end{law}

As a lemma, we have monotonicity generalized to local substitutions:
\begin{lemma}[Monotonicity]
  Given $\gamma; \tau : \Psi; \Gamma \To \Phi; \Delta$, if
  $\ltygteq[\Phi][\Delta] \delta {\delta'}{\Delta'}$, then
  $\ltygteq{\delta[\gamma;\tau]}{\delta'[\gamma;\tau]}{\Delta'[\gamma]}$.
\end{lemma}

\begin{law}[Weak head closure]
  If $\ltyreds t w T$, $\ltyreds{t'}{w'}T$ and $\ltygteq w {w'} T$, then $\ltygteq t {t'} T$.
\end{law}

\begin{law}[Congruence] $ $
  \begin{itemize}
  \item If $\ljudge[\Psi] 1 \Gamma$, then $\ltygteq \ze \ze \Nat$. 
  \item If $\ltygteq {t}{t'} \Nat$, then $\ltygteq{\su t}{\su t'} \Nat$.
  \item If $\ltyping 1 t {S \func T}$, $\ltyping 1 {t'}{S \func T}$ and
    $\ltygteq[\Psi][\Gamma, x : S]{t[\id;p(\id)]~x}{t'[\id;p(\id)]~x}{T}$, then
    $\ltygteq{t}{t'}{S \func T}$.
    
  \item If $\ljudge[\Psi] 1 \Gamma$ and $\ltyping[\Psi][\Delta] 0 t {T}$,
    then $\ltygteq{\boxit t}{\boxit t}{T}$.
    
  \item If $\ltyping 1 t {(g : \Ctx) \STo T}$, $\ltyping 1 {t'}{(g : \Ctx) \STo T}$ and
    $\ltygteq[\Psi, g : \Ctx]{t[p(\id)]~\$~g}{t'[p(\id)]~\$~g}{T}$, then
    $\ltygteq{t}{t'}{(g : \Ctx) \STo T}$.
  \end{itemize}
\end{law}

\begin{law}[Congruence of neutrals]$ $
  \begin{itemize}
  \item If $\ltyping 1 x T$, then $\ltygneeq x x T$. 
  \item If $u : (\judge[\Delta] T) \in \Psi$ and $\ltygteq{\delta}{\delta'}\Delta$, then
    $\ltygneeq {u^{\delta}}{u^{\delta'}} T$.
  \item If $\ltygneeq v {v'}{S\func T}$ and $\ltygteq{t}{t'}S$, then
    $\ltygneeq{v~t}{v~t'}T$. 
  \item If $\ljudge[\Psi] 1 {T'}$, $\ltygneeq{v}{v'}{\cont[\Delta]T}$ and
    $\ltygteq[\Psi, u : (\judge[\Delta] T)]{t}{t'}{T'[p(\id)]}$, then
    $\ltygneeq{\letbox u v t}{\letbox u{v'}{t'}}{T'}$.
  \item If $\ltygneeq{v}{v'}{(g : \Ctx) \STo T}$ and $\ljudge[\Psi] 0 \Delta$,
    then $\ltygneeq{v~\$~\Delta}{v'~\$~\Delta}{T[\id_\Psi, \Delta/g]}$
  \end{itemize}
\end{law}

From the congruence of local variables, we have that local weakening substitutions,
specifically, local identity substitutions, are reflexive in the generic equivalence:
\begin{lemma}[Reflexivity of Local Weakening Substitutions]
  If $\ljudge[\Psi] 1 \Delta, \Gamma$, then
  $\ltygteq[\Psi][\Delta,
  \Gamma]{\wk^{|\Gamma|}_{\Delta}}{\wk^{|\Gamma|}_{\Delta}}{\Delta}$.
\end{lemma}
\begin{lemma}[Reflexivity of Local Identity Substitutions]
  If $\ljudge[\Psi] 1 \Gamma$, then
  $\ltygteq[\Psi][\Gamma]{\id_{\Gamma}}{\id_{\Gamma}}{\Gamma}$.
\end{lemma}
This further implies
\begin{lemma}[Congruence of Global Variables]
  If $\vdash \Psi$ and $u : (\judge T) \in \Psi$, then
  $\ltygneeq {u^{\id_{\Gamma}}}{u^{\id_{\Gamma}}} T$.
\end{lemma}

We give the first instantiation of both relations as follows:
\begin{itemize}
\item $\ltygteq t {t'} T := \ltyequiv 1 t {t'} T$, and
\item $\ltygneeq v {v'} T := \ltyequiv 1 v {v'} T$.
\end{itemize}
The laws are instantiated to appropriate rules. %
Later, we will instantiate the relations to algorithmic equivalence, showing that the
algorithmic rules are complete, following \citet{abel_decidability_2017}. %
Before that, let us first define the reducibility predicates parameterized by the
generic equivalence relations for terms. 

\subsection{Reducibility Predicates}

Following~\citet{abel_decidability_2017}, we first give the semantic well-formedness
of types. %
The predicates do not need to be defined recursive-inductively, because unlike
dependent types, there is no type-level computation here in our system. %
\begin{mathpar}
  \inferrule
  {\vdash \Psi}
  {\lsemjudge i \Nat}

  \inferrule
  {\forall \gamma : \Phi \To_g \Psi ~.~ \lsemjudge[\Phi] i{S[\gamma]} \\
    \forall \gamma : \Phi \To_g \Psi ~.~ \lsemjudge[\Phi] i{T[\gamma]}}
  {\lsemjudge i {S \func T}}

  \inferrule
  {\ljudge 0 \Delta \\ \forall \gamma : \Phi \To_g \Psi ~.~ \lsemjudge[\Phi] 0 {T[\gamma]}}
  {\lsemjudge 1 {\cont[\Delta] T}}
  \qquad
  \inferrule
  {\ljudge[\Psi, g : \Ctx] 1 T \\
    \forall \gamma : \Phi \To_g \Psi \tand \ljudge[\Phi] 0 \Gamma ~.~ \lsemjudge[\Phi]
    1{T[q(\gamma)][\id_\Phi, \Gamma/g]}}
  {\lsemjudge 1 {(g : \Ctx) \STo T}}
\end{mathpar}
Compared to the syntactic well-formedness judgment in~\Cref{sec:cv:ctx-types}, the
semantic counterpart has extra universal quantifications over global weakenings. %
These universal quantifications are necessary when we give semantics to terms. %
Moreover, the semantic well-formedness of types is monotonic w.r.t. global weakenings:
\begin{lemma}[Monotonicity]
  If $\gamma : \Phi \To_g \Psi$ and $\lsemjudge i T$, then $\lsemjudge[\Phi] i
  {T[\gamma]}$. 
\end{lemma}
\begin{proof}
  Induction. %
\end{proof}

The lifting lemma also has a semantic counterpart:
\begin{lemma}[Lifting]
  If $\lsemjudge 0 T$, then $\lsemjudge 1 T$.
\end{lemma}
\begin{proof}
  Induction.
\end{proof}

\begin{lemma}[Escape]
  If $\lsemjudge i T$, then $\ljudge i T$. 
\end{lemma}
\begin{proof}
  We do induction on $\lsemjudge i T$. %
  In the meta-function case, we instantiate the global weakening to
  $p(\id) : \Psi, g : \Ctx \To_g \Psi$ and $\Gamma$ to $g$. %
  Apply IH again to obtain the goal. 
\end{proof}

Now we move on to defining the reducibility for terms $\lsemtyeq i t {t'} T$ and
local substitutions $\lsemtyeq i \delta {\delta'} \Delta$. %
Both relations are defined first by recursion on the layer $i$. %
Then, the reducibility predicate for terms $\lsemtyeq i t {t'} T$ is defined by recursion on
$\lsemjudge i T$. %
The reducibility predicate for local substitutions $\lsemtyeq i \delta {\delta'}
\Delta$ is defined inductively. %
The predicates are defined in this way because the layer-1 predicate for terms refers
to the predicates at layer $0$, as in our presheaf models by \citet{hu2024layered}. %
The definition that we present here does not focus too much on the
layered nature to save space and be more modular. %
However, we imagine that if this model must be put into a proof assistant, then
attention must be paid to ensure the well-foundness. %
The predicate for terms is a Kripke model, as it is indexed by both global and local contexts. %
We first define the semantic natural numbers:
\begin{mathpar}
  \inferrule
  {% \ltyping i t \Nat \\ \ltyping i{t'}\Nat \\ 
    \ltyreds t w \Nat \\ \ltyreds{t'}{w'} \Nat \\ \ltygteq{w}{w'}\Nat \\ \dsemtyeq {^{\Nf} w}{w'} \Nat}
  {\lsemtyeq i {t}{t'}\Nat}

  \inferrule
  { }
  {\dsemtyeq{^{\Nf} \ze}\ze \Nat}

  \inferrule
  {\dsemtyeq t{t'} \Nat}
  {\dsemtyeq {^{\Nf} {\su t}}{\su t'} \Nat}

  \inferrule
  {\ltygneeq v{v'} \Nat}
  {\dsemtyeq{^{\Nf} {v}}{v'} \Nat}
\end{mathpar}
Then we let $\lsemtyeq i {t}{t'}\Nat := \dsemtyeq {t}{t'}\Nat$.

Next, we define the case for function. %
$\lsemtyeq i t {t'}{S \func T}$ holds iff
\begin{itemize}
% \item $\ltyping i t {S \func T}$ and $\ltyping i{t'}{S \func T}$, and
\item $\ltyreds t w {S \func T}$, and
\item $\ltyreds{t'}{w'}{S \func T}$, and
\item $\ltygteq{w}{w'}{S \func T}$, and
\item for any $\gamma; \tau: \Phi; \Delta \To \Psi; \Gamma$, and given
  $\lsemtyeq[\Phi][\Delta] i s{s'}{S[\gamma]}$, then we have
  $\lsemtyeq[\Phi][\Delta] i{w[\gamma;\tau]~s}{w'[\gamma;\tau]~s'}{T[\gamma]}$.
\end{itemize}
It means that $t$ reduces to some weak head normal form, and the result of applying
this weak head normal form to a reducible term remains reducible. 

Next, we define the reducibility for $\cont[\Delta] T$,
\begin{mathpar}
  \inferrule
  {\ltyreds t w{\cont[\Delta] T} \\ \ltyreds{t'}{w'}{\cont[\Delta] T} \\\\ \ltygteq{w}{w'}{\cont[\Delta] T} \\ \lsemtyeq {1}{^{\Nf} w}{w'}{\cont[\Delta] T}}
  {\lsemtyeq 1{t}{t'}{\cont[\Delta] T}}
  \\
  \inferrule
  {\ltyping[\Psi][\Delta] 0 t T \\ \forall~\lsemtyeq[\Psi][\Delta'] 0{\delta}{\delta'}{\Delta} ~.~ \lsemtyeq[\Psi][\Delta'] 0{t[\delta]}{t[\delta']}{T}}
  {\lsemtyeq {1}{^{\Nf} {\boxit t}}{\boxit t}{\cont[\Delta] T}}

  \inferrule
  {\ltygneeq v{v'}{\cont[\Delta] T}}
  {\lsemtyeq {1}{^{\Nf} {v}}{v'}{\cont[\Delta] T}}
\end{mathpar}
Similar to our presheaf models, we must refer to $\lsemtyeq[\Psi][\Delta]
0{\delta}{\delta'}{\Gamma}$ when giving semantics for terms of type $\cont[\Delta] T$,
so the predicates at layer $0$ must be finished before defining the predicate at layer
$1$. %
To support pattern matching on code, instead of this universal quantification, we must
use an inductively defined layer-$0$ model instead as done in
\citet[Sec. 5.4.1]{hu2024layered}. 

Next, we define the case for meta-functions. %
$\lsemtyeq 1 t {t'}{(g : \Ctx) \STo T}$ holds iff
\begin{itemize}
\item $\ltyreds t w {(g : \Ctx) \STo T}$, and
\item $\ltyreds{t'}{w'}{(g : \Ctx) \STo T}$, and
\item $\ltygteq{w}{w'}{(g : \Ctx) \STo T}$, and
\item for any $\gamma; \tau: \Phi; \Delta \To \Psi; \Gamma$, and given
  $\ljudge[\Phi] 0{\Delta'}$, then we have
  $\lsemtyeq[\Phi][\Delta]
  1{w[\gamma;\tau]~\$~\Delta'}{w'[\gamma;\tau]~\$~\Delta'}{T[q(\gamma)][\id_\Phi, \Delta'/g]}$.
\end{itemize}

We generalize the reducibility for terms to local contexts and local substitutions by
doing an inductive-recursive definition:
\begin{mathpar}
  \inferrule
  {\vdash \Psi}
  {\lsemjudge i \cdot}
\end{mathpar}
Then $\lsemtyeq i \delta{\delta'}\cdot$ holds by considering the cases for $\Gamma$,
\begin{itemize}
\item if $\Gamma$ ends with $\cdot$, then $\delta = \delta' = \cdot^{|\Gamma|}$;
\item if $\Gamma$ ends with $g$, then $g : \Ctx \in \Psi$ and $\delta = \delta' = \cdot^{|\Gamma|}_g$.
\end{itemize}
\begin{mathpar}
  \inferrule
  {\vdash \Psi \\ g : \Ctx \in \Psi}
  {\lsemjudge i g}
\end{mathpar}
Then $\lsemtyeq i \delta{\delta'}g$ holds iff $\Gamma$ ends with $g$ and $\delta = \delta' =
\wk_g^{|\Gamma|}$. 

\begin{mathpar}
  \inferrule
  {\lsemjudge i \Delta \\ \lsemjudge i T}
  {\lsemjudge i {\Delta, x : T}}
\end{mathpar}
Then $\lsemtyeq i \delta{\delta'}{\Delta, x : T}$ holds iff
\begin{itemize}
\item $\lsemtyeq i{\delta}{\delta'}{\Delta}$,
\item $\lsemtyeq i {t}{t'}{T}$.
\end{itemize}
Note that one should consider the rules above really give two predicates, one for
layer $0$ and one for layer $1$. %
In this technical report, however, we do not really type out the replication. %

At this point, we have finished defining reducibility predicates for all types. %
We further let $\lsemtyp i t T$ be $\lsemtyeq i t t T$. %
This predicate basically means that $t$ can be reduced to some weak head normal form. 

By the definition of the predicates, we have the following lemmas:
\begin{lemma}\labeledit{lem:cvar:semred}
  If $\lsemtyeq i {t}{t'} T$, then $\ltyreds {s}{t}{T}$, and $\ltyreds{s'}{t'}{T}$, and
  $\lsemtyeq i{s}{s'}{T}$.
\end{lemma}
\begin{proof}
  By induction on $\lsemjudge i T$. 
\end{proof}
\begin{lemma}
  If $\lsemtyeq i {t}{t'} T$, then $\ltyreds t w {T}$, and $\ltyreds{t'}{w'}{T}$, and
  $\ltygteq{w}{w'}{T}$.
\end{lemma}
\begin{proof}
  By induction on $\lsemjudge i T$ and transitivity of multi-step weak head
  reduction. 
\end{proof}
\begin{corollary}\labeledit{lem:cvar:semeq2geneq}
  If $\lsemtyeq i {t}{t'} T$, then $\ltygteq{t}{t'}{T}$.
\end{corollary}
\begin{lemma}\labeledit{lem:cvar:semeq2geneq-lsubst}
  If $\lsemtyeq i {\delta}{\delta'} \Delta$, then
  $\ltygteq{\delta}{\delta'}{\Delta}$. 
\end{lemma}
\begin{proof}
  Generalization of \Cref{lem:cvar:semeq2geneq}. 
\end{proof}

\begin{lemma}[Reducibility of Neutrals]\labeledit{lem:cvar:red-ne}
  If $\lsemjudge i T$, $\ltyping i v T$ and $\ltygneeq v {v'} T$, then $\lsemtyeq i v {v'} T$. 
\end{lemma}
\begin{proof}
  Induction by $\lsemjudge i T$ and apply IHs. 
\end{proof}
\begin{lemma}[Reducibility of Weakenings]
  If $\lsemjudge i \Gamma,\Delta$, then $\lsemtyeq[\Psi][\Gamma,\Delta]i {\wk^{|\Delta|}_{\Gamma}}{\wk^{|\Delta|}_{\Gamma}}{\Gamma}$.
\end{lemma}
\begin{proof}
  A direct consequence of \Cref{lem:cvar:red-ne}.
\end{proof}
\begin{corollary}[Reducibility of Identity]
  If $\lsemjudge i \Gamma$, then $\lsemtyeq i {\id_{\Gamma}}{\id_{\Gamma}}{\Gamma}$.  
\end{corollary}

\begin{lemma}
  If $\lsemjudge 0 T$, then $\lsemtyeq 0 {t}{t'} T$ iff $\lsemtyeq 1 {t}{t'} T$.
\end{lemma}
\begin{proof}
  We only consider the cases for $\Nat$ and functions and apply IHs directly. 
\end{proof}
This is derived from the fact that the predicate for $\Nat$ is invariant at
different layers. %

We verify several important properties for the reducibility predicates:
\begin{lemma}[Escape]
  If $\lsemtyeq i {t}{t'} T$, then $\ltyequiv 1 {t}{t'}T$.
\end{lemma}
\begin{proof}
  We do induction on $\lsemjudge i T$. %
  Notice that generic equivalence eventually implies syntactic equivalence by the
  subsumption law. 
\end{proof}
\begin{lemma}[Escape]
  If $\lsemtyeq i {\delta}{\delta'} \Delta$, then $\ltyequiv 1 {\delta}{\delta'}\Delta$.
\end{lemma}

\begin{lemma}[Monotonicity]
  If $\lsemjudge i T$ and $\lsemtyeq i {t}{t'} T$, given $\gamma; \tau : \Phi;\Delta
  \To \Psi; \Gamma$, then $\lsemtyeq[\Phi][\Delta] i
  {t[\gamma;\tau]}{t'[\gamma;\tau]}{T[\gamma]}$.
\end{lemma}
\begin{proof}
  We do induction on $\lsemjudge i T$.
  \begin{itemize}[label=Case]
  \item $T = \Nat$, then we have the goal by the monotonicity of multi-step weak head
    reduction and the generic equivalence. %
    We further do induction on $\dsemtyeq{^{\Nf} {w}}{w'} \Nat$.
    
  \item $T = S \func T'$, then we assume another $\gamma';\tau' : \Phi'; \Delta' \To
    \Phi; \Delta$ and $\lsemtyeq[\Phi'][\Delta'] i {s}{s'}{S[\gamma \circ
      \gamma']}$, we must show $\lsemtyeq[\Phi'][\Delta'] i {t[\gamma; \tau][\gamma'; \tau']~s}{t'[\gamma; \tau][\gamma'; \tau']~s'}{S[\gamma \circ
      \gamma']}$. %
    But this is immediate due to the composition of weakenings.
    
  \item $T = \cont[\Delta]{T'}$, then it is also immediate after a case analysis of
    $\dsemtyeq{^{\Nf} {w}}{w'}{\cont[\Delta]{T'}}$. %
    We apply the monotonicity of multi-step weak head
    reduction and the generic equivalence appropriately.
    
  \item $T = (g : \Ctx) \STo T'$, then this case is very similar to the function
    case. %
    We assume another $\gamma';\tau' : \Phi'; \Delta' \To
    \Phi; \Delta$ and $\ljudge[\Phi'] 0 \Gamma'$. %
    Our goal is to show
    \[
      \lsemtyeq[\Phi'][\Delta']
      1{w[\gamma;\tau][\gamma';\tau']~\$~\Gamma'}{w'[\gamma;\tau][\gamma';\tau']~\$~\Gamma'}{T[q(\gamma
        \circ \gamma')][\id_{\Phi'}, \Gamma'/g]}
    \]
    This again has been given by the composition of weakenings and notice that
    $q(\gamma \circ \gamma') = q(\gamma) \circ q(\gamma')$. 
  \end{itemize}
\end{proof}
\begin{lemma}[Monotonicity]
  If $\lsemtyeq i {\delta}{\delta'}{\Delta'}$, given $\gamma; \tau : \Phi;\Delta
  \To \Psi; \Gamma$, then $\lsemtyeq[\Phi][\Delta] i
  {\delta[\gamma;\tau]}{\delta'[\gamma;\tau]}{\Delta'[\gamma]}$.
\end{lemma}

\begin{lemma}[PER]
  The $\lsemtyeq i {t}{t'} T$ relation satisfies symmetry and transitivity.
\end{lemma}
\begin{proof}
  Transitivity relies on the uniqueness of weak head reduction. 
\end{proof}
\begin{lemma}[PER]
  The $\lsemtyeq i {\delta}{\delta'} \Delta$ relation satisfies symmetry and transitivity.
\end{lemma}

The following lemma is the semantic counterpart for the layering principle. %
It ensures that terms inhabiting types in STLC have the same behaviors at both layers.
\begin{lemma}[Layering Restriction] $ $
  \begin{itemize}
  \item If $\lsemjudge 0 T$, then $\lsemtyeq 0 t{t'} T$ is equivalent to $\lsemtyeq 1 t{t'} T$.
  \item If $\lsemjudge 0 \Delta$, then $\lsemtyeq 0 \delta{\delta'} \Delta$ is
    equivalent to $\lsemtyeq 1 \delta{\delta'} \Delta$.
  \end{itemize}
\end{lemma}
This lemma is particularly useful to help treat global variables the same at both layers.

Our goal is then to show the following theorem:
\begin{theorem}[Completeness]\labeledit{thm:cvar:comp} $ $
  \begin{itemize}
  \item If $\ljudge[\Psi] i T$, then $\lsemjudge[\Psi] i T$.
  \item If $\ltyequiv i {t}{t'} T$, then $\lsemtyeq i {t}{t'} T$.
  \item If $\ltyequiv i {\delta}{\delta'} \Delta$, then $\lsemtyeq i {\delta}{\delta'} \Delta$.
  \end{itemize}
\end{theorem}
If the $\ltygteq{t}{t'}{T}$ is the algorithmic convertibility checking algorithm, then
we show that syntactic equivalence implies algorithmic equivalence. %
In other words, algorithmic convertibility is complete w.r.t. syntactic equivalence. %
However, to arrive at that solution, we must first show that the completeness theorem
above holds w.r.t. the generic equivalence relations. %
Following \citet{abel_decidability_2017}, however, we still have one step missing to
conclude this goal. %
We must define a set of validity judgments to handle the meta-function case in the
semantic well-formedness of types.

\subsection{Validity Judgments}

According to \citet{abel_decidability_2017}, validity judgments are introduced to
characterize the reducible terms that are also closed under substitutions. %
In the same spirit, we also need the validity judgments to handle the case for
meta-functions:
\begin{mathpar}
  \inferrule
  {\ljudge[\Psi, g : \Ctx] 1 T \\
    \forall \gamma : \Phi \To_g \Psi \tand \ljudge[\Phi] 0 \Gamma ~.~ \lsemjudge[\Phi]
    1{T[q(\gamma)][\id_\Phi, \Gamma/g]}}
  {\lsemjudge 1 {(g : \Ctx) \STo T}}
\end{mathpar}
When we attempt to prove the following statement from the completeness theorem:
\begin{center}
  If $\ljudge[\Psi] i T$, then $\lsemjudge[\Psi] 1 T$.
\end{center}
This case breaks down, because the IH only provides
\[
  \lsemjudge[\Psi, g : \Ctx] 1{T}
\]
and by monotonicity, we derive 
\[
  \lsemjudge[\Phi, g : \Ctx] 1{T[q(\gamma)]}
\]
but then we are stuck. %
We have no way to prove that the semantic well-formedness of types is closed under
substitutions. %
Following~\citet{abel_decidability_2017}, we define validity judgments which further
restrict reducibility predicates to subsets that are closed under substitutions. %
Since types are only affected by the global contexts, the validity judgments are
defined by induction-recursion on global contexts:
\begin{mathpar}
  \inferrule
  { }
  {\vDash^v \cdot}
\end{mathpar}

The validity equivalence for global substitutions
$\semvtyp[\Phi]{\sigma}{\cdot}$ is defined as $\vdash \Phi$ and $\sigma =\cdot$. 

\begin{mathpar}
  \inferrule
  {\vDash^v \Psi}
  {\vDash^v \Psi, g : \Ctx}
\end{mathpar}
$\semvtyp[\Phi]{\sigma}{\Psi, g : \Ctx}$ is defined as
\begin{itemize}
\item $\sigma = \sigma_1, \Gamma/g$,
\item $\ljudge[\Phi] 0 \Gamma$, and
\item $\semvtyp[\Phi]{\sigma_1}{\Psi}$.
\end{itemize}

\begin{mathpar}
  \inferrule
  {\vDash^v \Psi \\ \lsemvjudge 0 \Gamma \\ \lsemvjudge 0 T}
  {\vDash^v \Psi, u : (\judge T)}
\end{mathpar}
$\semvtyp[\Phi]{\sigma}{\Psi, u : (\judge T)}$ is defined as
\begin{itemize}
\item $\sigma = \sigma_1, t/u$,
\item $\semvtyp[\Phi]{\sigma_1}{\Psi}$, and
\item for all $\lsemtyeq[\Phi][\Delta]0{\delta}{\delta'}{\Gamma[\sigma_1]}$, we have
  $\lsemtyeq[\Phi][\Delta] 0 {t[\delta]}{t[\delta']}{T[\sigma_1]}$.
\end{itemize}

The validity of types $\lsemvjudge i T$ is defined as: $\lsemjudge i T$ and given
$\semvtyp[\Phi]{\sigma}{\Psi}$, we have $\lsemjudge[\Phi] i {T[\sigma]}$. %
The validity of local contexts $\lsemvjudge i \Gamma$ is defined by applying the validity
of types pointwise. %
The judgments are simplified because there is no need to have an equivalence judgment
between global substitutions. %
In a global substitution, there are only two kinds of components: code of STLC and
local contexts. %
The former is determined to be identified by syntax. %
Local contexts are also identified by syntax because we are dealing with simple
types. %

With these definitions ready, we put a universal quantification on top of the
reducibility predicates, which specifies the reducible terms that are closed under
valid global substitutions:
\begin{itemize}
\item $\lsemvtyeq i{t}{t'}{T}$ iff for any $\semvtyp[\Phi]{\sigma}{\Psi}$ and
  $\lsemtyeq[\Phi][\Delta]i{\delta}{\delta'}{\Gamma[\sigma]}$, we have
  $\lsemtyeq[\Phi][\Delta] i{t[\sigma][\delta]}{t'[\sigma][\delta]}{T[\sigma]}$.
\item $\lsemvtyeq i{\delta}{\delta'}{\Delta}$ iff for any
  $\semvtyp[\Phi]{\sigma}{\Psi}$ and
  $\lsemtyeq[\Phi][\Delta']i{\delta''}{\delta'''}{\Gamma[\sigma]}$, we have
  $\lsemtyeq[\Phi][\Delta']
  i{\delta[\sigma] \circ \delta''}{\delta'[\sigma] \circ \delta''}{\Delta[\sigma]}$.
\end{itemize}

Now we work out several lemmas:
\begin{lemma}[Escape]
  If $\semvtyp[\Phi]{\sigma}{\Psi}$, then $\typing[\Phi]\sigma\Psi$. 
\end{lemma}
\begin{proof}
  We do induction on $\vDash^v \Psi$. %
  In the case for extension of code, we apply the escape lemma for semantically
  well-formed types and reducible terms.
\end{proof}

\begin{lemma}[Monotonicity]
  If $\semvtyp[\Phi]{\sigma}{\Psi}$ and $\gamma : \Phi' \To_g \Phi$, then
  $\semvtyp[\Phi']{\sigma[\gamma]}{\Psi}$.
\end{lemma}
\begin{proof}
  We do induction on $\vDash^v \Psi$. %
  Use monotonicity of reducible terms and the algebra of global substitutions. %
\end{proof}

\begin{lemma}[Validity of Global Weakening Substitutions]
  If $\vdash \Psi, \Phi$ and $\vDash^v \Psi$, then $\semvtyp[\Psi,\Phi]{\wk^{|\Phi|}_{\Psi}}{\Psi}$. 
\end{lemma}
\begin{proof}
  We do induction on $\vDash^v \Psi$. %
  The most interesting case is the extension case for code. %
  If $\Psi = \Psi', u : (\judge T)$, then our goal is given $\lsemtyeq[\Psi', u : (\judge T), \Phi][\Delta]0{\delta}{\delta'}{\Gamma[\wk^{1 + |\Phi|}_{\Psi'}]}$, to prove
  \[
    \lsemtyeq[\Psi', u : (\judge T), \Phi][\Delta] 0 {u^{\id_{\Gamma}[p^{1 + |\Phi|}(\id)]}[\delta]}{u^{\id_{\Gamma}[p^{1 + |\Phi|}(\id)]}[\delta']}{T[\wk^{1 + |\Phi|}_{\Psi'}]}
  \]
  In the last section, we have established that $[\wk^{1 + |\Phi|}_{\Psi'}]$ has the
  same effect as $[p^{1 + |\Phi|}(\id)]$, so this goal becomes
  \[
    \lsemtyeq[\Psi', u : (\judge T), \Phi][\Delta] 0 {u^{\id_{\Gamma}[p^{1 + |\Phi|}(\id)]}[\delta]}{u^{\id_{\Gamma}[p^{1 + |\Phi|}(\id)]}[\delta']}{T[p^{1 + |\Phi|}(\id)]}
  \]
  By computation, the goal further becomes:
  \[
    \lsemtyeq[\Psi', u : (\judge T), \Phi][\Delta] 0 {u^{\delta}}{u^{\delta'}}{T[p^{1 + |\Phi|}(\id)]}
  \]
  But this is immediate due to \Cref{lem:cvar:semeq2geneq-lsubst}, the reducibility of
  neutrals and congruence of the generic equivalence
  \[
    \ltygneeq[\Psi', u : (\judge T), \Phi][\Delta] {u^{\delta}}{u^{\delta'}}{T[p^{1 + |\Phi|}(\id)]}
  \]
\end{proof}

In particular, it proves that the identity is valid:
\begin{corollary}
  If $\vdash \Psi$, then $\semvtyp{\id}{\Psi}$. 
\end{corollary}

\begin{theorem}[Fundamental] $ $
  \begin{itemize}
  \item If $\vdash \Psi$, then $\vDash^v \Psi$.
  \item If $\ljudge i T$, then $\lsemvjudge i T$. 
  \item If $\ljudge i \Gamma$, then $\lsemvjudge i \Gamma$. 
  \end{itemize}
\end{theorem}
\begin{proof}
  We do induction. The cases for global contexts are simple.
  \begin{itemize}[label=Case]
  \item
    \begin{mathpar}
      \inferrule
      {\vdash \Psi}
      {\ljudge i \Nat}
    \end{mathpar}
    Assuming $\semvtyp[\Phi]{\sigma}{\Psi}$, by escape, we have
    $\typing[\Phi]{\sigma}{\Psi}$, and then $\vdash \Phi$ by presupposition. %
    We then conclude the goal. 
    
  \item
    \begin{mathpar}
      \inferrule
      {\ljudge i S \\ \ljudge i T}
      {\ljudge i {S \func T}}
    \end{mathpar}
    We assume $\semvtyp[\Phi]{\sigma}{\Psi}$. %
    We now should prove $\lsemjudge[\Phi]i{(S \func T)[\sigma]}$. %
    This can be concluded from $\lsemjudge[\Phi]i{S[\sigma]}$ and
    $\lsemjudge[\Phi]i{T[\sigma]}$. %
    They are immediate from $\lsemvjudge i S$ and $\lsemvjudge i T$ by IH. %
    
  \item
    \begin{mathpar}
      \inferrule
      {\ljudge 0 \Delta \\ \ljudge 0 T}
      {\ljudge 1 {\cont[\Delta] T}}
    \end{mathpar}
    We assume $\semvtyp[\Phi]{\sigma}{\Psi}$. %
    We now should prove $\lsemjudge[\Phi]i{(\cont[\Delta] T)[\sigma]}$. %
    We further assume $\gamma : \Phi' \To_g \Phi$. %
    The goal can be concluded from $\lsemjudge[\Phi']1{T[\sigma][\gamma]}$. %
    Since $T[\sigma][\gamma] = T[\sigma[\gamma]]$ and $\lsemvjudge 0 T$ from IH,
    we only need $\semvtyp[\Phi']{\sigma[\gamma]}{\Psi}$, which we get from
    monotonicity. %
    
  \item
    \begin{mathpar}
      \inferrule
      {\ljudge[\Psi, g : \Ctx] 1 T}
      {\ljudge 1 {(g : \Ctx) \STo T}}
    \end{mathpar}
    We assume $\semvtyp[\Phi]{\sigma}{\Psi}$. %
    We now should prove $\lsemjudge[\Phi]i{((g : \Ctx) \STo T)[\sigma]}$. %
    We further assume $\gamma : \Phi' \To_g \Phi$ and $\ljudge[\Phi'] 0 \Gamma$. %
    The goal can be concluded from $\lsemjudge[\Phi']1{T[q(\sigma)][q(\gamma)][\id,\Gamma/g]}$. %
    We compute:
    \begin{align*}
      T[q(\sigma)][q(\gamma)][\id,\Gamma/g]
      & = T[q(\sigma)[q(\gamma)]][\id,\Gamma/g] \\
      & = T[q(\sigma[\gamma])][\id,\Gamma/g] \\
      & = T[\sigma[\gamma][p(\id)], g/g][\id,\Gamma/g] \\
      & = T[(\sigma[\gamma][p(\id)] \circ (\id,\Gamma/g)), g[\id,\Gamma/g]/g] \\
      & = T[\sigma[\gamma], \Gamma/g]
    \end{align*}
    Therefore the goal becomes to prove $\lsemjudge[\Phi']1{T[\sigma[\gamma],
      \Gamma/g]}$. %
    By IH, we have $\lsemvjudge[\Psi, g : \Ctx] 1 T$. %
    We simply need $\semvtyp[\Phi']{\sigma[\gamma], \Gamma/g}{\Psi, g : \Ctx}$. %
    This is immediate from monotonicity and definition. 
  \end{itemize}
\end{proof}

\begin{theorem}[Fundamental] $ $
  \begin{itemize}
  \item If $\ltyping i t T$, then $\lsemvtyp i{t}{T}$.
  \item If $\ltyping i \delta \Delta$, then $\lsemvtyp i{\delta}{\Delta}$.
  \item If $\ltyequiv i t {t'} T$, then $\lsemvtyeq i{t}{t'}{T}$.
  \item If $\ltyequiv i \delta{\delta'}\Delta$, then $\lsemvtyeq i{\delta}{\delta'}{\Delta}$. 
  \end{itemize}
\end{theorem}
\begin{proof}
  We again do mutual induction. %
  We focus on modal cases. %
  \begin{itemize}[label=Case]
  \item
    \begin{mathpar}
      \inferrule
      {\ltyequiv i \delta{\delta'} \Delta \\ u : (\judge[\Delta] T) \in \Psi}
      {\ltyequiv{i}{u^\delta}{u^{\delta'}}{T}}
    \end{mathpar}
    \begin{align*}
      H_0: & \lsemvtyeq i \delta{\delta'} \Delta
             \byIH \\
      H_1: & \semvtyp[\Phi]{\sigma}{\Psi}
             \tag*{(by assumption)} \\
           & \lsemtyeq[\Phi][\Delta']i{\delta''}{\delta'''}{\Gamma[\sigma]}
             \tag*{(by assumption)} \\
      H_2: & \lsemtyeq[\Phi][\Delta'] i{\delta[\sigma] \circ
             \delta''}{\delta'[\sigma] \circ \delta'''}{\Delta[\sigma]}
             \tag*{(by $H_0$)} \\
      H_3: & \lsemtyeq[\Phi][\Delta'] 0{\delta[\sigma] \circ
             \delta''}{\delta'[\sigma] \circ \delta'''}{\Delta[\sigma]}
      \tag{by layering restriction and $\Delta$ is well-formed at layer $0$}\\
           & \lsemtyeq[\Phi][\Delta']i{\sigma(u)[\delta[\sigma] \circ
             \delta'']}{\sigma(u)[\delta'[\sigma] \circ
             \delta''']}{T[\sigma[p^{1 + u}(\id)]]}
             \tag*{(by $H_1$ and $H_3$)}
    \end{align*}

  \item
    \begin{mathpar}
      \inferrule
      {\ljudge 1 \Gamma \\ \ltyequiv[\Psi][\Delta] 0{t}{t'}T}
      {\ltyequiv{1}{\boxit t}{{\boxit{t'}}}{\cont[\Delta] T}}
    \end{mathpar}
    \begin{align*}
      H_0: & \lsemvtyeq[\Psi][\Delta]0{t}{t'}{T}
             \byIH \\
      H_1: & \semvtyp[\Phi]{\sigma}{\Psi}
             \tag*{(by assumption)} \\
      H_2: & \lsemtyeq[\Phi][\Delta']1{\delta}{\delta'}{\Gamma[\sigma]}
             \tag*{(by assumption)} \\
           & \ljudge[\Phi] 1 {\Delta'}
             \tag*{(by $H_2$, escape and presupposition)} \\
      H_3: & \lsemtyeq[\Phi][\Delta'']0{\delta''}{\delta'''}{\Delta[\sigma]}
             \tag*{(by assumption)} \\
           & t = t' \tag*{(by static code)} \\
           & \lsemtyeq[\Phi][\Delta'']0{t[\sigma][\delta'']}{t'[\sigma][\delta''']}{T[\sigma]}
             \tag*{(by $H_0$, $H_1$ and $H_3$)} \\
           & \lsemtyeq[\Phi][\Delta[\sigma]]1{^\Nf \boxit t[\sigma]}{\boxit{t'}[\sigma]}{T[\sigma]}
             \tag*{(by definition)} \\
           & \ltyequiv[\Phi][\Delta']{1}{\boxit t [\sigma][\delta]}{{\boxit{t'}[\sigma][\delta']}}{\cont[\Delta] T[\sigma]}
    \end{align*}
    In the last step, notice that local substitutions do not propagate into $\tbox$. 
  \item
    \begin{mathpar}
      \inferrule
      {\ltyequiv 1 {s}{s'}{\cont[\Delta] T} \\ \ljudge[\Psi] 0 \Delta \\
        \ljudge[\Psi] 0 T \\
        \ljudge[\Psi] 1{T'} \\
        \ltyequiv[\Psi, u : (\judge[\Delta] T)][\Gamma[p(\id)]] 1 {t}{t'}{T'[p(\id)]}}
      {\ltyequiv 1 {\letbox u s t}{\letbox u {s'}{t'}} T'}
    \end{mathpar}
    \begin{align*}
      H_0: & \lsemvtyeq 1 {s}{s'}{\cont[\Delta] T}
             \byIH \\
      H_1: & \lsemvtyeq[\Psi, u : (\judge[\Delta] T)][\Gamma[p(\id)]] 1
             {t}{t'}{T'[p(\id)]}
             \byIH \\
      H_2: & \semvtyp[\Phi]{\sigma}{\Psi}
             \tag*{(by assumption)} \\
      H_3: & \lsemtyeq[\Phi][\Delta']1{\delta}{\delta'}{\Gamma[\sigma]}
             \tag*{(by assumption)} \\
           & \lsemtyeq[\Phi][\Delta']0{s[\sigma][\delta]}{s'[\sigma][\delta']}{\cont[\Delta]T[\sigma]}
             \tag*{(by $H_0$, $H_2$ and $H_3$)}
    \end{align*}
    Now, we consider $H_3$, where we know for some $w$ and $w'$
    \begin{align*}
      & \ltyreds[\Phi][\Delta']{s[\sigma][\delta]} w{\cont[\Delta[\sigma]]{(T[\sigma])}} \\
      & \ltyreds[\Phi][\Delta']{s'[\sigma][\delta]}{w'}{\cont[\Delta[\sigma]]{(T[\sigma])}} \\
      H_4: & \lsemtyeq[\Phi][\Delta']{1}{^{\Nf} w}{w'}{\cont[\Delta[\sigma]]{(T[\sigma])}}
    \end{align*}
    We then case analyze $H_4$:
    \begin{itemize}[label=Subcase]
    \item
      \begin{mathpar}
        \inferrule
        {\ltyping[\Phi][\Delta[\sigma]] 0{t''}{T[\sigma]} \\ \forall~\lsemtyeq[\Phi'][\Delta''] 0{\delta''}{\delta'''}{\Delta[\sigma]} ~.~ \lsemtyeq[\Phi'][\Delta''] 0{t''[\delta'']}{t''[\delta''']}{T}}
        {\lsemtyeq[\Phi][\Delta']{1}{^{\Nf} {\boxit{t''}}}{\boxit{t''}}{\cont[\Delta] T[\sigma]}}
      \end{mathpar}
      Then we have
      \begin{align*}
        H_5: & \semvtyp[\Phi]{\sigma, t''/u}{\Psi, u: (\judge[\Delta] T)}
               \tag*{(by definition)} \\
             & \lsemtyeq[\Phi][\Delta']1{t[\sigma,t''/u][\delta]}{t'[\sigma,t''/u][\delta']}{T'[p(\id)][\sigma,t''/u]}
               \tag*{(by $H_1$ and $H_5$)} \\
             & T'[p(\id)][\sigma,t''/u] = T'[\sigma]
               \tag*{(by computation)} \\
             & \letbox u s t[\sigma][\delta] \\
               = & \letbox u {s[\sigma][\delta]}{(t[q(\sigma)][\delta[p(\id)]])} \\
        \rightsquigarrow^\ast & \letbox u
                                {\boxit{t''}}{(t[q(\sigma)][\delta[p(\id)]])} \\
        \rightsquigarrow & t[q(\sigma)][\delta[p(\id)]][\id,t''/u] \\
        = & t[q(\sigma) \circ (\id,t''/u)][\delta[p(\id)][\id,t''/u]] \\
        = & t[\sigma,t''/u][\delta]
            \tag*{(by computation)} \\
             & \letbox u{s'}{t'}[\sigma][\delta'] \rightsquigarrow^\ast t'[\sigma,t''/u][\delta']
               \tag*{(similarly)} \\
           & \lsemtyeq[\Phi][\Delta'] 1 {\letbox u s t[\sigma][\delta]}{\letbox u {s'}{t'}[\sigma][\delta]}{T'[\sigma]}
      \end{align*}
      
    \item
      \begin{mathpar}
        \inferrule
        {\ltygneeq[\Phi][\Delta'] v{v'}{\cont[\Delta] T[\sigma]}}
        {\lsemtyeq[\Phi][\Delta'] {1}{^{\Nf} {v}}{v'}{\cont[\Delta] T[\sigma]}}
      \end{mathpar}
      Then we have
      \begin{align*}
        H_5: & \semvtyp[\Phi,u : (\judge[\Delta[\sigma]]{(T[\sigma])})]{\sigma[p(\id)],
               u^\id/u}{\Psi, u: (\judge[\Delta] T)}
               \tag*{(by monotonicity and \Cref{lem:cvar:red-ne})} \\
        H_6: & \lsemtyeq[\Phi, u : (\judge[\Delta[\sigma]]{(T[\sigma])})][\Delta'[p(\id)]]1{\delta[p(\id)]}{\delta'[p(\id)]}{\Gamma[\sigma][p(\id)]}
               \tag*{(by monotonicity)} \\
             & \Gamma[p(\id)][\sigma[p(\id)],u^\id/u] = \Gamma[\sigma[p(\id)]] =
               \Gamma[\sigma][p(\id)]
               \tag*{(by computation)} \\
             & T'[p(\id)][\sigma[p(\id)],u^\id/u] = T'[\sigma[p(\id)]] =
               T'[\sigma][p(\id)]
               \tag*{(by computation)} \\
             & \lsemtyeq[\Phi, u : (\judge[\Delta[\sigma]]{(T[\sigma])})][\Delta'[p(\id)]]1{t[\sigma[p(\id)],u^\id/u][\delta[p(\id)]]}{t'[\sigma[p(\id)],u^\id/u][\delta'[p(\id)]]}{T'[\sigma][p(\id)]}
               \tag*{(by $H_1$, $H_5$ and $H_6$)} \\
             & \ltygteq[\Phi, u :
               (\judge[\Delta[\sigma]]{(T[\sigma])})][\Delta'[p(\id)]]{t[\sigma[p(\id)],u^\id/u][\delta[p(\id)]]}{t'[\sigma[p(\id)],u^\id/u][\delta'[p(\id)]]}{T'[\sigma][p(\id)]}
               \tag*{(by \Cref{lem:cvar:semeq2geneq})} \\
             & \letbox u s t[\sigma][\delta] \\
        = & \letbox u {s[\sigma][\delta]}{(t[q(\sigma)][\delta[p(\id)]])} \\
        \rightsquigarrow^\ast & \letbox u
                                {v}{(t[q(\sigma)][\delta[p(\id)]])} 
                                \tag*{(by computation)} \\        
             & \letbox u{s'}{t'}[\sigma][\delta'] 
               \rightsquigarrow^\ast \letbox u
               {v'}{(t'[q(\sigma)][\delta'[p(\id)]])} 
               \tag*{(similarly)} \\
             & \ltygneeq[\Phi][\Delta']{\letbox u
               {v}{(t[q(\sigma)][\delta[p(\id)]])} }
               {\letbox u {v'}{(t'[q(\sigma)][\delta'[p(\id)]])} }{T[\sigma]}
               \tag*{(neutral terms)} \\
             & \lsemtyeq[\Phi][\Delta']1{\letbox u
               {v}{(t[q(\sigma)][\delta[p(\id)]])} }
               {\letbox u {v'}{(t'[q(\sigma)][\delta'[p(\id)]])} }{T[\sigma]}
               \tag*{(by \Cref{lem:cvar:red-ne})} \\
             & \lsemtyeq[\Phi][\Delta'] 1 {\letbox u s t[\sigma][\delta]}{\letbox u {s'}{t'}[\sigma][\delta]}{T'[\sigma]}
      \end{align*}
    \end{itemize}

  \item
    \begin{mathpar}
      \inferrule
      {\ljudge[\Psi] 1 \Gamma \\ \ltyequiv[\Psi, g: \Ctx][\Gamma[p(\id)]] 1 t{t'} T}
      {\ltyequiv{1}{\Lambda g. t}{\Lambda g. t'}{(g : \Ctx) \STo T}}
    \end{mathpar}
    \begin{align*}
      H_0: & \lsemvtyeq[\Psi, g: \Ctx][\Gamma[p(\id)]] 1 t{t'} T
             \byIH \\
      H_1: & \semvtyp[\Phi]{\sigma}{\Psi}
             \tag*{(by assumption)} \\
      H_2: & \lsemtyeq[\Phi][\Delta']1{\delta}{\delta'}{\Gamma[\sigma]}
             \tag*{(by assumption)} \\
      H_3: & \gamma;\tau : \Phi'; \Delta'' \To \Phi;\Delta'
             \tag*{(by assumption)} \\
      H_4: & \ljudge[\Phi']0{\Gamma'}
             \tag*{(by assumption)} \\
      H_5: & \semvtyp[\Phi']{\sigma[\gamma], \Gamma'/g}{\Psi, g: \Ctx}
             \tag*{(by monotonicity)} \\
      H_6: & \lsemtyeq[\Phi'][\Delta'']1{\delta[\gamma;\tau]}{\delta'[\gamma;\tau]}{\Gamma[\sigma][\gamma]}
             \tag*{(by monotonicity)} \\
           & \lsemtyeq[\Phi'][\Delta'']
             1{t[\sigma[\gamma],\Gamma'/g][\delta[\gamma;\tau]]}{t'[\sigma[\gamma],\Gamma'/g][\delta'[\gamma;\tau]]}{T[\sigma[\gamma], \Gamma'/g]}
             \tag*{(by $H_0$, $H_5$ and $H_6$)} \\
           & \Lambda g. t[\sigma][\delta][\gamma;\tau]~\$~\Gamma' \\
      = & \Lambda g. (t[q(\sigma[\gamma])][\delta[\gamma;\tau][p(\id)]])~\$~\Gamma' \\
      \rightsquigarrow & t[q(\sigma[\gamma])][\delta[\gamma;\tau][p(\id)]][\id,\Gamma'/g] \\
      = & t[q(\sigma[\gamma])][\id,\Gamma'/g][\delta[\gamma;\tau][p(\id)][\id,\Gamma'/g]] \\
      = & t[\sigma[\gamma],\Gamma'/g][\delta[\gamma;\tau]]
          \tag*{(by computation)} \\
           & \Lambda g. t'[\sigma][\delta'][\gamma;\tau]~\$~\Gamma' \rightsquigarrow t'[\sigma[\gamma],\Gamma'/g][\delta'[\gamma;\tau]]
             \tag*{(similarly)}\\
           & \lsemtyeq[\Phi'][\Delta'']
             1{\Lambda g. t[\sigma][\delta][\gamma;\tau]~\$~\Gamma'}{\Lambda g. t'[\sigma][\delta'][\gamma;\tau]~\$~\Gamma'}{T[\sigma[\gamma],
             \Gamma'/g]}
      \tag*{(by \Cref{lem:cvar:semred})}\\
           & \lsemtyeq[\Phi][\Delta']{1}{\Lambda g. t[\sigma][\delta]}{\Lambda g. t'[\sigma][\delta']}{(g : \Ctx) \STo T[\sigma]}
    \end{align*}
    
  \item
    \begin{mathpar}
      \inferrule
      {\ltyequiv{1}{t}{t'}{(g : \Ctx) \STo T} \\ \ljudge[\Psi] 0 \Delta}
      {\ltyequiv{1}{t~\$~\Delta}{t'~\$~\Delta}{T[\id_{\Psi}, \Delta/g]}}
    \end{mathpar}
    \begin{align*}
      H_0: & \lsemvtyeq 1 t{t'}{(g : \Ctx) \STo T}
             \byIH \\
      H_1: & \semvtyp[\Phi]{\sigma}{\Psi}
             \tag*{(by assumption)} \\
      H_2: & \lsemtyeq[\Phi][\Delta']i{\delta}{\delta'}{\Gamma[\sigma]}
             \tag*{(by assumption)} \\ 
           & \lsemtyeq[\Phi][\Delta'] 1{t[\sigma][\delta]}{t'[\sigma][\delta']}{(g :
             \Ctx) \STo T[\sigma]} \\
           & \ljudge[\Phi] 0 {\Delta[\sigma]} \\
           & \lsemtyeq[\Phi][\Delta']
             1{t[\sigma][\delta]~\$~\Delta[\sigma]}{t'[\sigma][\delta']~\$~\Delta[\sigma]}{T[\sigma][\id,\Delta[\sigma]/g]}
      \\
           & T[\sigma][\id,\Delta[\sigma]/g] = T[\id,\Delta/g][\sigma]
             \tag*{(by computation)}
    \end{align*}
    The goal is then concluded. 

  \item
    \begin{mathpar}
      \inferrule
      {\ljudge i \Gamma \\ \text{$\Gamma$ ends with $\cdot$} \\ |\Gamma| = m}
      {\ltyequiv i {\cdot^m}{\cdot^m}{\cdot}}
    \end{mathpar}
    \begin{align*}
      & \semvtyp[\Phi]{\sigma}{\Psi}
             \tag*{(by assumption)} \\
      & \lsemtyeq[\Phi][\Delta']i{\delta}{\delta'}{\Gamma[\sigma]}
             \tag*{(by assumption)}
    \end{align*}
    We must compare $\cdot^m[\sigma] \circ \delta =
    \cdot_{\widecheck{\delta}}^{\widehat\delta}$ and $\cdot^m[\sigma] \circ \delta' =
    \cdot_{\widecheck{\delta'}}^{\widehat{\delta'}}$.
    But they are immediately equal, as we can show that
    $\lsemtyeq[\Phi][\Delta']i{\delta}{\delta'}{\Gamma[\sigma]}$ implies
    $\widecheck{\delta} = \widecheck{\delta'}$ and $\widehat{\delta} = \widehat{\delta'}$.

  \item
    \begin{mathpar}
      \inferrule
      {\ljudge i \Gamma \\ g : \Ctx \in \Psi \\ \text{$\Gamma$ ends with $g$} \\ |\Gamma| = m}
      {\ltyequiv i {\cdot_g^m}{\cdot_g^m}{\cdot}}
    \end{mathpar}
    \begin{align*}
      & \semvtyp[\Phi]{\sigma}{\Psi}
        \tag*{(by assumption)} \\
      & \lsemtyeq[\Phi][\Delta']i{\delta}{\delta'}{\Gamma[\sigma]}
        \tag*{(by assumption)}
    \end{align*}
    We look up $\sigma(g)$ and consider what it ends with.
    \begin{itemize}[label=Subcase]
    \item $\sigma(g)$ ends with $\cdot$. Then we must compare
      $\cdot_g^m[\sigma]\circ \delta = \cdot_{\widecheck{\delta}}^{\widehat\delta}$ and
      $\cdot_g^m[\sigma]\circ \delta' = \cdot_{\widecheck{\delta'}}^{\widehat{\delta'}}$,
      which we know are equal.
      
    \item $\sigma(g)$ ends with some $g'$. Then we must compare
      $\cdot_g^m[\sigma]\circ \delta = \cdot_{g'}^{\widehat\delta}$ and
      $\cdot_g^m[\sigma]\circ \delta' = \cdot_{g'}^{\widehat{\delta'}}$,
      which we know are also equal.
    \end{itemize}
    
  \item 
    \begin{mathpar}
      \inferrule
      {\ljudge i \Gamma \\ g : \Ctx \in \Psi \\ \text{$\Gamma$ ends with $g$} \\ |\Gamma| = m}
      {\ltyequiv i {\wk_g^m}{\wk_g^m}{g}}
    \end{mathpar}
    \begin{align*}
      & \semvtyp[\Phi]{\sigma}{\Psi}
        \tag*{(by assumption)} \\
      H_0: & \lsemtyeq[\Phi][\Delta']i{\delta}{\delta'}{\Gamma[\sigma]}
        \tag*{(by assumption)}
    \end{align*}
    Then we have
    \[
      \wk_g^m[\sigma]\circ \delta = \id_{\sigma(g)}[id;p^m(\id)] \circ \delta =
      \wk^m_{\sigma(g)} \circ \delta
    \]
    and 
    \[
      \wk_g^m[\sigma]\circ \delta' = \id_{\sigma(g)}[id;p^m(\id)] \circ \delta' =
      \wk^m_{\sigma(g)} \circ \delta'
    \]
    We show $\lsemtyeq[\Phi][\Delta']i{\wk^m_{\sigma(g)} \circ
      \delta}{\wk^m_{\sigma(g)} \circ \delta'}{\sigma(g)}$
    by unraveling $H_0$ $m$ times. 
  \end{itemize}
\end{proof}

As a corollary of the fundamental theorems, we can prove the completeness theorem:
\begin{proof}[Proof of \Cref{thm:cvar:comp}]
  Notice that the reducibility predicates are just special cases of the validity
  judgments. 
\end{proof}

\subsection{Convertibility Checking}

In this section, we will write down the converibility checking rules and instantiate
the generic equivalence with it, proving that equivalence terms can be checked. %
We define three judgments: $\dtconv{t}{t'}{T}$ checks the convertibility of two terms
$t$ and $t'$. %
$\dtconvnf{w}{w'}{T}$ checks the convertibility of two normal forms $w$ and $w'$. %
This operation is directed by types. %
$\dtconvne{v}{v'}T$ checks the convertibility of two neutral forms $v$ and $v'$. %
This operation is structural on the neutral forms. %
We give all the rules below:
\begin{mathpar}
  \inferrule
  {\ltyreds t w T \\ \ltyreds {t'}{w'}T \\ \dtconvnf{w}{w'}T}
  {\dtconv t{t'} T}

  \inferrule
  {\ljudge[\Psi]1\Gamma}
  {\dtconvnf{\ze}{\ze}\Nat}

  \inferrule
  {\dtconv t{t'} \Nat}
  {\dtconvnf{\su{t}}{\su{t'}}\Nat}

  \inferrule
  {\dtconvne v{v'} \Nat}
  {\dtconvnf{v}{v'}\Nat}

  \inferrule
  {\dtconv[\Psi][\Gamma, x : S]{w[\id;p(\id)]~x}{w'[\id;p(\id)]~x}T}
  {\dtconvnf{w}{w'}{S \func T}}

  \inferrule
  {\ljudge[\Psi]1\Gamma \\ \ltyping[\Psi][\Delta]0 t T}
  {\dtconvnf{\boxit{t}}{\boxit{t}}{\cont[\Delta] T}}

  \inferrule
  {\dtconvne v{v'}{\cont[\Delta] T}}
  {\dtconvnf{v}{v'}{\cont[\Delta] T}}

  \inferrule
  {\dtconv[\Psi,g:\Ctx]{w[p(\id)]~\$~g}{w'[p(\id)]~\$~g}T}
  {\dtconvnf{w}{w'}{(g : \Ctx) \STo T}}

  \inferrule
  {\ljudge 1 \Gamma \\ x : T \in \Gamma}
  {\dtconvne x{x} T}

  \inferrule
  {\dtconvnf{\delta}{\delta'}\Delta \\ x : (\judge[\Delta]T) \in \Psi}
  {\dtconvne{u^\delta}{u^{\delta'}}T}

  \inferrule
  {\dtconvne{v}{v'}{S \func T} \\ \dtconv{t}{t'}S}
  {\dtconvne{v~t}{v'~t'} T}

  \inferrule
  {\dtconvne{v}{v'}{\cont[\Delta] T} \\ \ljudge 1 {T'} \\\\
    \dtconv[\Psi,u:(\judge[\Delta] T)][\Gamma[p(\id)]]{t}{t'}{T'[p(\id)]}}
  {\dtconvne{\letbox u v t}{\letbox u{v'}{t'}}{T'}}

  \inferrule
  {\dtconvne{v}{v'}{(g : \Ctx) \STo T} \\ \ljudge 0 \Delta}
  {\dtconvne{v~\$~\Delta}{v'~\$~\Delta}{T[\id,\Delta/g]}}
\end{mathpar}

We then instantiate the generic equivalence. We instantiate $\ltygteq t {t'} T$ with
$\dtconv{t}{t'}{T}$ and $\ltygneeq t {t'} T$ with $\dtconvne{t}{t'}{T}$.

Most laws are immediate. %
We discuss a few of them.
\begin{lemma}[PERs]
  All three relations above are PERs.
\end{lemma}
\begin{proof}
  When we prove transitivity of $\dtconv{t}{t'}{T}$, we use the uniqueness of
  multi-step weak head reduction.
\end{proof}

\begin{lemma}[Congruence of $\tbox$]
  If $\ljudge[\Psi] 1 \Gamma$, $\ltyping[\Psi][\Delta] 0 t {T}$ and $\dtconv[\Psi][\Delta]{t}{t}{T}$,
  then $\dtconv{\boxit t}{\boxit t}{T}$.
\end{lemma}
\begin{proof}
  Notice that we are almost there, except that we must prove
  $\dtconv[\Psi][\Delta]{t}{t}{T}$ for $t$ at layer $0$. %
  This premise is met due to our layered model, where we instantiate layer $0$ and
  layer $1$ separately, so that the fundamental theorem of layer $0$ gives
  $\dtconv[\Psi][\Delta]{t}{t}{T}$. 
\end{proof}

A successful instantiation gives us the following desired completeness theorem for
converibility checking:
\begin{theorem}[Completeness]$ $
  \begin{itemize}
  \item If $\ltyequiv 1 t {t'} T$, then $\dtconv{t}{t'}{T}$.
  \item If $\ltyequiv 1 \delta{\delta'}\Delta$, then $\dtconv{\delta}{\delta'}{\Delta}$. 
  \end{itemize}
\end{theorem}
Soundness is easy by a simple induction:
\begin{theorem}[Soundness]$ $
  \begin{itemize}
  \item If $\dtconv{t}{t'}{T}$, then $\ltyequiv 1 t {t'} T$.
  \item If $\dtconvnf{w}{w'}{T}$, then $\ltyequiv 1 w {w'} T$.
  \item If $\dtconvne{v}{v'}{T}$, then $\ltyequiv 1 v {v'} T$.
  \item If $\dtconv{\delta}{\delta'}{\Delta}$, then $\ltyequiv 1 \delta{\delta'}\Delta$.
  \end{itemize}
\end{theorem}
This concludes our discussion about context variables.

\section{Dependent Layered Modal Type Theory}\labeledit{sec:dt}

In this section, we combine the work by \citet{hu2024layered} and what we have built
up in the previous sections and scale all the way up to dependent types. %
We present \delamlang, \textbf{De}pendent \textbf{La}yered \textbf{M}odal Type
Theory. %
With dependent types, we can not only analyze the syntax of programs, but also that of
types. %
This ability, therefore, gives us the power to write tactics that could potentially
fill in proof obligations in a proof environment. %
In particular, this type theory addresses a number of problems that we often see in
proof assistants like Coq, Lean, and Agda. %
In Coq, tactics are written in a separate language, Ltac or Ltac2, where the
advantages of dependent types in Gallina, the core language, are lost. %
Stratifying the tactic language and the core language into two also cause
duplications: there are multiple notions for natural numbers, functions, etc.. %
On the other hand, in Lean and Agda, we use reflection to convert a Lean or an Agda
term into an AST and then use the core language to manipulate the AST. %
An instrumentation in the kernel is responsible for converting this AST back to a
valid term, if type-checked. %
This mechanism superficially provides a uniform way to tactics, but reflection
generally fails to guarantee the well-formedness of ASTs, making type malformedness
run-time errors and necessitating exception mechanisms exclusively for macro executions.

We believe that this type theory provides an example to address all aforementioned
problems. %
Starting this section, let us dive into dependent types.

\subsection{Highlights}

On a high level, we would continue to apply the layering principle in \delamlang to
guide us in the design of this type theory. %
In particular, we would want the layer for static code to be subsumed by layer for
computing programs. %
Moreover, with context variables, we are now able to formulate a recursive
elimination principle for code, which was not possible in simple types. %
However, these two requirements combining together causes some high-level technical
effects to the design of the type theory, which are worth mentioning before presenting the
type theory itself.

\subsubsection{Code Promotion}

Since we are going to introduce elimination principle for terms with dependent types,
we must also consider how equivalences are handled for code. %
For example, if we know a given piece of code has type $(\lambda x. x)~\Nat$, should
we regard this type the ``same'' as $\Nat$? %
Intuitively, the answer should be yes. %
After all, we only want to capture the syntax of the term, not its type. %
Effectively, for code of type $\cont T$, it should also be regarded as an inhabitant
of another type $\cont{T'}$, as long as $T \approx T'$ in $\Gamma$. 

In the context of dependent types, however, that causes some problems in the presence
of function applications. %
Consider a function $f : \Pi(x : S). T$ and an argument $t : S$. %
Then in general, the type of $f~t$ is $T[t/x]$. %
Now, let us construct this term as code. %
Even though $t$ is constructed as part of the code, the type of the overall code
$\cont{(T[t/x])}$ contains $t$ and therefore part of the dynamics of $t$ is in fact
promoted to the type. %
For a more concrete example, if $f : \Pi(x : \Ty 0). x$ and let the argument be
$(\lambda x. x)~\Nat$, then $\boxit{(f~((\lambda x. x)~\Nat))}$ has type
$\cont{(\lambda x. x)~\Nat}$, which we agree is just $\cont{\Nat}$. %
Clearly, the argument computes and is not purely static code as in simple types. %
We cannot avoid this phenomenon because of dependent types, so we must handle it with
care. %
This phenomenon is call a \emph{code promotion}. 

\subsubsection{Non-cumulativity}

Due to code promotions, we must permit non-trivial equivalences in types of code. %
This causes particular problems when we want an elimination principle for code with
intensional analyses. %
When we split code into cases in the elimination form, we must specify in each case
how do we construct the original code from its components. %
Therefore, it is the most convenient, when each term has a unique type, leading to a
conclusion of preferring non-cumulative universes. %
Whereas with cumulative universes, types live in higher universes for free. %
Cumulativity forms a pre-order of types which cannot be captured solely by equivalence
rules and makes the typing rule for the elimination principle extremely difficult to express, if
not completely impossible.

\subsubsection{Universe Polymorphism}

Though it is often omitted in other work, universe levels and universe polymorphism
are important ingredients in dependent type theory. %
They are typically considered as ``details'' and are not very much paid attention
to. %
However, in this work, we must be explicit about universe polymorphism. %
Consider some code for function application $t~s$. %
We in general do not know the type of $s$, let alone its universe level (though it
must be uniquely determined due to non-cumulativity discussed above). %
Therefore, the elimination principle for code must work for \emph{all} universe
levels, leading to a formalism of universe polymorphism. 

\subsubsection{Tarski-style Universes}

Another ingredient to consider when approaching an elimination principle for code with
dependent types and intensional analyses is the separation between types and terms. %
Consider Russell-style universes where types and terms are not distinguished
naturally. %
It would probably suffice to say that $\cont T$ can represent code for both types and
terms. %
In particular, code for some types just has type $\cont{\Ty 0}$, for example. %
Unfortunately, this thought is too naive. %
When we consider $\cont T$ as the type of code, we are considering this type with two
indices, the (local) context $\Gamma$ and the type of the code body $T$. %
But what is the type of $T$? %
Well, it is $\cont{\Ty l}$ for some $l$, which is just a special case of $\cont T$! %
A type clearly should not be indexed by a special case of itself. %
The intertwine between types and terms in Russell-style universes seems to even
prevent a proper statement of indices of types for code. %
However, Tarski-style universes, where types and terms are clearly distinguished, introduce
mutually inductive relations between types and terms, and safely bail us out of this
problem, as we will see very soon.

\subsection{Syntax}\labeledit{sec:dt:syntax}

Let us start with the syntax of \delamlang. %
Since we employ Tarski-style universes, the syntax of terms and types are separate. %
Due to non-cumulativity, certain constructs must remember universe levels. %
Due to the elimination principle of code, some constructs must include additional
sub-structures so that the elimination eventually checks out. %
Let us begin with the subset that is basically just Martin-L\"of type theory (MLTT).
\begin{alignat*}{2}
  x, y & && \tag{Local variables} \\
  \ell & && \tag{Universe variables} \\
  l & := &&\ \ell \sep 0 \sep 1 + l \sep l \sqcup l' \sep \omega \tag{Universe levels} \\
  M, S, T & := && \ \Nat \sep \PI{l}{l'}x S T \sep \Ty l \sep \UPI \ell l T \sep \Elt l t \tag{Types, $\Typ$} \\
  s, t &:=&&\ x \tag{Terms, $\Exp$} \\
  & && \sep \Nat \sep \PI{l}{l'}x s t \sep \Ty l \tag{encoding of types}  \\
  & && \sep \ze \sep \su t \sep \ELIMN l{x.M}{s}{x, y.s'}{t} \tag{natural numbers $N$} \\
       & &&\sep \LAM l{l'} x S t \sep \APP t l{l'} x S T s \tag{dependent functions}  \\
  & && \sep \ULAM l \ell t \sep \UAPP t l \tag{universe polymorphic functions} \\
  \Gamma, \Delta &:= &&\ \cdot \sep g \sep \Gamma, x : T \at l
                        \tag{Local contexts, \Ctx} \\
  L &:= && \cdot \sep L, \ell \tag{Universe contexts}
\end{alignat*}
Due to three kinds of contexts and the scale of the system, we omit the discussion of
weakenings, which we diligently kept track of in the previous sections. %
We take various weakenings for granted. %
Nevertheless, they will appear in the semantics. %
Weakenings in general are pretty obvious, as we permit arbitrary lookups for all
variables. %
Since we now must deal with universe levels, we use $\ell$ to range over variables for
universe levels. %
The syntax for universe levels follows Agda's conventions. %
Universe levels form an idempotent commutative monoid, the laws of which we will show
in the next subsection. %
Here we use $\sqcup$ to denote taking the max of two universe levels. %
% Our syntax above has not introduced the syntax to introduce a universe level
% variables.
The ability to take maximum between two levels induces a partial order:
\begin{align*}
  l \le l' := l' \approx l \sqcup l'
\end{align*}
where we use $\approx$ to express the equivalence between universe levels. %
Thus, with $\le$, universe levels form a bounded inf-lattice. %
A strict order is given by requiring the pre-order to hold for the successor of $l$:
\begin{align*}
  l < l' := l' \approx (1 + l) \sqcup l'
\end{align*}
This strict order, as to be shown later, is well-founded, based on which we will give
semantics to universes. %
Due to universe polymorphism, we must also include an $\omega$ level, which will be
used to represent the universe level of a universe-polymorphic function. %
The $\omega$ level must not appear in any program, does not participate in the bounded
inf-lattice specified above, is only used in type-checking, and therefore can be
ignored most of the time. %
The formalization of universe polymorphism here follows~\citet{bezem_type_2023}
tightly. %
We use $\UPI \ell l T$ to denote the type of a universe-polymorphic function. %
It introduce a non-empty list of universe variables at once, and lives at $\Ty
\omega$. %
The type $T$ lives at $\Ty l$, where $l$ may refer to all variables from $\vect
\ell$. %
Since $l$ cannot be $\omega$, we must have all universe variables introduced in one
go. %
The introduction form is $\ULAM l \ell t$, which also introduces a non-empty list of
universe variables first and then the function body as expected. %
The elimination form $\UAPP t l$ symmetrically eliminates a universe-polymorphic
function with the same number of universe level expressions. %

The rest of the expressions are pretty much standard from MLTT. %
We have natural numbers (\Nat), their introduction forms and a recursion principle. %
We always use $M$ to exclusively represent the motives of a recursion principle. %
For regular dependent functions $\PI{l}{l'}x S T$, we must specify the universe levels
of $S$ and $T$, following~\citet{pujet_impredicative_2023}. %
We might omit the universe levels if they are not important in the discussion. %
The function abstraction $\LAM l{l'} x S t$ is standard; we might omit $l$ and $S$ if they
are not important or can be inferred from the context. %
The function application $\APP t l{l'} x S T s$ is arguably more complex. %
We explicitly specify the type of the function to prepare for a better formulation of
the elimination principle for code later in the section. %
By requiring explicit type annotations in elimination forms, types that are usually
hidden in the core syntax become 
sub-structures in the elimination form for code. %
This verbosity has no negative impact for programmers: after all, we are discussing a core theory, and
we can let a type-inferring front-end to fill in these types for the users if they
choose so. %
Following conventions, we may simply write $s~t$ if the types are not important. %

Since we are employing Tarski-style universes, as we have specified in the syntax,
types and terms are separated. %
As terms, we use \emph{encodings} of types, i.e. the overloaded $\Nat$,
$\PI l{l'} x s t$ and $\Ty l$, which are members of some universes. %
They are \emph{decoded} into actual types through $\Elt l t$, converting the encodings
to actual corresponding types. %
This part is basically identical to~\citet{palmgren_universes_1998}.  %
For simplicity, we have omitted the type lifter, which is responsible for raising the
universe levels explicitly, similar to \texttt{Lift} in Agda's standard library. %
According~\citet{palmgren_universes_1998}, the type lifter bears additional
equivalences and therefore we omit them here for conciseness, as lifting of the levels
is an orthogonal issue here. %

\subsection{Universe Levels}\labeledit{sec:dt:ulevel}

In the syntax, we deliberately group all the universe variables into a separate
context. %
This is beneficial as both local and global contexts (to be discussed later) will need to refer
to universe variables. %
It is also helpful for the future work of extending \delamlang to more
layers, by simply inserting more contexts after $L$. %
In this section, we state the well-formedness and equivalence judgments for universe
levels. %
Note that all syntactically valid universe contexts are already well-formed as they
only contain universe variables.
\begin{mathpar}
  \inferrule
  {\ell \in L}
  {\typing[L]{\ell}{\Level}}

  \inferrule
  { }
  {\typing[L]{0}{\Level}}

  \inferrule
  {\typing[L]{l}{\Level}}
  {\typing[L]{1 + l}{\Level}}
  
  \inferrule
  {\typing[L]{l}{\Level} \\ \typing[L]{l'}{\Level}}
  {\typing[L]{l \sqcup l'}{\Level}}
\end{mathpar}
Notice that $\omega$ is not well-formed. %
Indeed, the judgment $\typing[L]l\Level$ only captures the well-formed universe levels
that can be written by a programmer. %
The level $\omega$, on the other hand, only appears during type-checking to denote the
type of universe-polymorphic functions. %

As discussed above, universe levels themselves form an idempotent commutative
monoid. %
Hence they have the following equivalence rules:
\begin{mathpar}
  \inferrule
  {\typing[L]{l}{\Level}}
  {\tyequiv[L]{l}{l}\Level}
  
  \inferrule
  {\tyequiv[L]{l}{l'}{\Level}}
  {\tyequiv[L]{l'}{l}\Level}

  \inferrule
  {\tyequiv[L]{l}{l'}\Level \\ \tyequiv[L]{l'}{l''}\Level}
  {\tyequiv[L]{l}{l''}\Level}
\end{mathpar}
First we specify the basic PER rules. %
Then we have congruence rules:
\begin{mathpar}
  \inferrule
  {\tyequiv[L]{l}{l'}{\Level}}
  {\tyequiv[L]{1 + l}{1 + {l'}}\Level}

  \inferrule
  {\tyequiv[L]{l}{l'}\Level \\ \tyequiv[L]{l''}{l'''}\Level}
  {\tyequiv[L]{l \sqcup l''}{l' \sqcup l'''}\Level}  
\end{mathpar}
Finally we have the algebraic rules:
\begin{mathpar}
  \inferrule
  {\typing[L]{l}{\Level}}
  {\tyequiv[L]{l \sqcup 0}{l}\Level}
  
  \inferrule
  {\typing[L]{l}{\Level} \\ \typing[L]{l'}{\Level} \\ \typing[L]{l''}{\Level}}
  {\tyequiv[L]{(l \sqcup l') \sqcup l''}{l \sqcup (l' \sqcup l'')}\Level}

  \inferrule
  {\typing[L]{l}{\Level} \\ \typing[L]{l'}{\Level}}
  {\tyequiv[L]{l \sqcup l'}{l' \sqcup l}\Level}

  \inferrule
  {\typing[L]{l}{\Level}}
  {\tyequiv[L]{l \sqcup l}{l}\Level}

  \inferrule
  {\typing[L]{l}{\Level} \\ \typing[L]{l'}{\Level}}
  {\tyequiv[L]{1 + (l \sqcup l')}{1 + l \sqcup 1 + {l'}}\Level}

  \inferrule
  {\ell \in L}
  {\tyequiv[L]{\ell \sqcup (1 + \ell)}{1 + \ell}{\Level}}
\end{mathpar}
The second last rule is distributivity of $\tsucc$ over $\sqcup$. %
The last rule is absorption of $\tsucc$ over $\sqcup$. %
The equivalence judgment confirms the well-formedness of both components:
\begin{lemma}[Presupposition]
  If $\tyequiv[L]{l}{l'}\Level$, then $\typing[L]l\Level$ and $\typing[L]{l'}\Level$. 
\end{lemma}
\begin{proof}
  Induction.
\end{proof}
We can prove a more general absorption rule by doing a few inductions.
\begin{lemma}[Absorption]
  If $\typing[L]l\Level$, then $\tyequiv[L]{l \sqcup (1 + l)}{1 + l}\Level$. 
\end{lemma}
\begin{proof}
  Induction. Only the following case is interesting:
  \begin{mathpar}
    \inferrule
    {\typing[L]{l}{\Level} \\ \typing[L]{l'}{\Level}}
    {\typing[L]{l \sqcup l'}{\Level}}
  \end{mathpar}
  We reason as follows:
  \begin{align*}
    l \sqcup l' \sqcup (1 +{l \sqcup l'})
    & \approx (l \sqcup (1 + l)) \sqcup (l' \sqcup (1 +{l'})) \\
    & \approx (1 + l) \sqcup (1 +{l'}) \byIH \\
    & \approx 1 +{(l \sqcup l')}
  \end{align*}
  Hence the proof is complete. 
\end{proof}
Then we generalize further:
\begin{lemma}[Absorption]
  If $\typing[L]l\Level$, then for any $n$, $\tyequiv[L]{l \sqcup (n + {l})}{n + {l}}\Level$.   
\end{lemma}
\begin{proof}
  We proceed by induction on $n$. %
  The cases for $n=0$ and $n=1$ are simple.  We consider the step case, where
  $n = 1 + n'$ and we know
  $\tyequiv[L]{l \sqcup (n' + {l})}{n' + {l}}\Level$. %
  We reason as follows:
  \begin{align*}
    l \sqcup (1 + (n' + l))
    & \approx l \sqcup (1 +{(l \sqcup (n' + {l}))})
      \byIH \\
    & \approx l \sqcup (1 + l) \sqcup (1 + n' + {l}) \\
    & \approx (1 + l) \sqcup (1 + n' + {l})
    \tag{by absorption} \\
    & \approx (1 + l) \sqcup (n' + (1 + l)) \\
    & \approx n' + {(1 + l)} \byIH \\
    & \approx n + l
  \end{align*}
  The proof is complete. 
\end{proof}

As readers might have noticed, the theory for universe levels are self-contained and
their equivalence is decidable, as per implemented by Agda's type-checker. %
For this reason, in the remainder of the discussion, we undermine the importance of
well-formedness and equivalence of universe levels, unless it is essential in the
surrounding context. 

Next, we define substitutions for universe levels:
\begin{align*}
  \phi := \cdot \sep \phi,l/\ell \tag{Substitutions for universe levels}
\end{align*}
\begin{mathpar}
  \inferrule
  { }
  {\typing[L]{\cdot}{\cdot}}

  \inferrule
  {\typing[L]\phi{L'} \\ \typing[L]l\Level}
  {\typing[L]{\phi, l/\ell}{L',\ell}}
\end{mathpar}
Applying substitutions is intuitive:
\begin{align*}
  \ell [\phi] &:= \phi(\ell) \\
  0 [\phi] &:= 0 \\
  1 + l [\phi] &:= 1 + (l[\phi]) \\
  l \sqcup l' [\phi] &:= (l [\phi]) \sqcup (l'[\phi]) \\
  \omega [\phi] &:= \omega
\end{align*}
We then have the following lemmas:
\begin{lemma}
  If $\typing[L]l\Level$ and $\typing[L']\phi L$, then $\typing[L']{l[\phi]}\Level$.
\end{lemma}
\begin{proof}
  Induction.
\end{proof}
\begin{lemma}
  If $\tyequiv[L]l{l'}\Level$ and $\typing[L']\phi L$, then $\tyequiv[L']{l[\phi]}{l'[\phi]}\Level$.
\end{lemma}
\begin{proof}
  Induction.
\end{proof}

\newcommand{\adjustf}[1]{\ensuremath{\textsf{adjust}(#1)}}

The well-foundedness of the strict order $<$ is intuitive. %
The only bottom element is $0$. %
We simply keep removing $\tsucc$ from all components of $\sqcup$, and we must
eventually stop. %
Thus, the simplest way to argue the well-foundedness of $<$ is to define a measure
based on the number of $\tsucc$'s that can be removed from an $l$. %
This number is defined over all $\typing[L]l\Level$ recursively as follows:
\begin{align*}
  \countf \ell &:= \{ \ell \mapsto 0, 0 \mapsto 0 \} \\
  \countf 0 &:= \{ 0 \mapsto 0 \} \\
  \countf{1 + l} &:= \mergef{\{ \ell \mapsto 1 + n \sep \ell \mapsto n \in \countf{l}
                   \}}{\{0 \mapsto 1 + \countf{l}(0)\}} \\
  \countf{l \sqcup l'} &:= \mergef{\countf{l}}{\countf{l'}}
\end{align*}
Here \textsf{count} returns a map that counts the number of $\tsucc$'s over all
universe variables and $0$. %
The function \textsf{merge} merges two maps and takes the maximum in a conflict. %
The following definition makes sure equivalent universe levels to have the same
representation as maps:
\begin{align*}
  \adjustf m :=
  \begin{cases}
    \{ \ell \mapsto n \sep m(\ell) \mapsto n \} & \text{if $m(0) \le \max_{\ell
                                                  \mapsto n \in m} n$} \\
    m & \text{otherwise}
  \end{cases}
\end{align*}
In the first branch, we check if there is a variable which has a higher universe level
than the constant. %
If so, we drop the constant completely. %
For example, in $2 \sqcup (2 + \ell)$, $2$ is
redundant, as we know $2 + \ell$ is at least as large as $2$. %
On the other hand, in $2 \sqcup (1 + \ell)$, it is possible for $1 + \ell$
to be smaller than $2$ when we take $\ell$ as $0$. %
Therefore, in this case, we must keep $2$.  Thus, the finiteness of
decreasing steps of the universe levels can be seen as taking some finite steps by
removing all $1 +$ from maps returned by $\adjustf{\countf l}$. %
Then we can just take away variables until we can no longer descend. %
In fact, $\adjustf{\countf l}$ should be considered a normalization algorithm for
universe levels. %
We can simply compare equality between maps computed as such decide whether two
universe levels are equivalent. %
The correctness is as follows:
\begin{lemma}
  If $\tyequiv[L]l{l'}\Level$, then $\adjustf{\countf l} = \adjustf{\countf{l'}}$.
\end{lemma}
\begin{proof}
  Induction. Take advantage of the idempotent commutative monoidal nature of maximum. 
\end{proof}
This lemma ensures that the procedure respects equivalence between universe levels. %
The other direction is seen by providing a ``decoding'' function, which converts a map
to a universe level. %
% We then simply need to agree on a serialization of the map which is a mere detail
% here. %
We give one possible function that converts a map to a $\Level$.
\begin{align*}
  \flatten L m := \begin{cases}
    \bigsqcup_{\ell \mapsto n \in m} n + \ell & \text{if $0$ is not in $m$} \\
    m(0) \sqcup (\bigsqcup_{\ell \mapsto n \in m} n + \ell)
    & \text{otherwise}
  \end{cases}
\end{align*}
where the order of $\ell$'s respects their order in $L$ and all $\sqcup$ are right
associative. %
These requirements give a syntactically unique flattening of a map. %
Then we prove
\begin{lemma}
  If $\typing[L]l\Level$, then $\tyequiv l{\flatten{L}{\adjustf{\countf l}}}\Level$. 
\end{lemma}
\begin{proof}
  We proceed by induction. %
  It is rather immediate. For the $\tsucc$ case, we use its distributivity to
  propagate it inwards. For the $\sqcup$ case, we use commutativity to rearrange levels
  within and absorption to eliminate small levels whenever necessary.
\end{proof}
% We can in fact just compare $\Level$'s by considering the results of \textsf{count}. %
% This is pretty simple so we often consider the equalities of our types and terms
% modulo the equivalence of $\Level$s. 

\subsection{Typing and Equivalence Judgments}\labeledit{sec:dt:mltt-rules}

In this section, we introduce the typing and equivalence judgments, only for the MLTT
part. %
We will consider the modal part next altogether. %
The typing and equivalence judgments are defined mutually as usual. %
All the related judgments are:
\begin{itemize}
\item $\judge[L]\Psi$ denotes the well-formedness of the global context $\Psi$ under
  $L$. 
\item $\lpjudge i \Gamma$ denotes the well-formedness of $\Gamma$ given the universe
  context $L$ and the global context $\Psi$ at layer $i$. %
  In this section, we are not very concerned about layers yet as most parts about
  meta-programming and intensional analysis come in a later section
  (\Cref{sec:dt:modal}). 
\item $\lpequiv i \Gamma\Delta$ denotes the equivalence between two local contexts
  $\Gamma$ and $\Delta$. 
\item $\lttypwf i T l$ denotes the well-formedness of the type $T$ living in the
  universe level $l$ at layer $i$ in the given contexts. 
\item $\lttyping i t T l$ denotes the well-typedness of $t$ of type $T$, which is in
  the universe level $l$ at layer $i$.
  In the special occasion of $T$ being some $\Se$, we might write $\lttypingd i T
  l{1 + l}$ to simultaneously denote two judgments at the same time to save space.
\item $\lttypeq i{T}{T'}l$ denotes the equivalence between types $T$ and $T'$ living in the
  universe level $l$ at layer $i$ in the given contexts. 
\item $\lttyequiv i t {t'} T l$ denotes the equivalence between $t$ and $t'$ of type
  $T$, which is in the universe level $l$ at layer $i$. %
  The shorthand $\lttyequivd i T {T'} l {1 + l}$ has a meaning similar to above.
\item $\ltsubst i \delta \Delta$ denotes the well-formedness of a local substitution
  $\delta$ which substitute all local variables in $\Delta$ into terms referring to
  $\Gamma$. %
  We will introduce this judgment and the next when we discuss the modal part. %
\item $\ltsubeq i\delta{\delta'}\Delta$ denotes the equivalence between local
  substitutions $\delta$ and $\delta'$. 
% \item $\lttypingv i t T l$ is a sub-judgment of $\lttyping i t T l$ which makes sure
%   $t$ is a local variable or a global variable that can only be substituted to a local
%   variable.
% \item $\ltsubstv i \delta\Delta$ is the extension of $\lttypingv i t T l$ to local substitution.
\end{itemize}

The judgments for MLTT are rather routine. %
Many are just generalization of the judgments in~\Cref{sec:cv}. 
Let us first consider the well-formedness and equivalence of local contexts:
\begin{mathpar}
  \inferrule
  {\judge[L] \Psi}
  {\lpjudge i \cdot \\ \lpequiv i \cdot\cdot}

  \inferrule
  {\judge[L] \Psi \\ g : \Ctx \in \Psi}
  {\lpjudge i{g} \\ \lpequiv i{g}{g}}

  \inferrule
  {\lpjudge i \Gamma \\\\ \lttypwf i T l \\ \typing[L]l\Level}
  {\lpjudge i{\Gamma, x : T \at l}}

  \inferrule
  {\lpequiv i \Gamma\Delta \\ \lttypeq i T{T'} l \\ \lttypeq[\Psi][\Delta] i T{T'} l
    \\ \lttypwf i T l \\ \lttypwf[\Psi][\Delta] i{T'}{l'} \\ \tyequiv[L]l{l'}\Level}
  {\lpequiv i{\Gamma, x : T \at l}{\Delta, x : T' \at{l'}}}
\end{mathpar}
The well-formedness of types are also immediate, following~\citet{pujet_impredicative_2023}. %
When we encounter $\tEl$, we resort that to the typing judgment of terms. %
Overlapping rules for well-typedness of encoding as terms are also listed:
\begin{mathpar}
  \inferrule
  {\lpjudge{\typeof i} \Gamma}
  {\lttypingd i \Nat 0 {1}}

  \inferrule
  {\typing[L]l\Level  \\ \typing[L]{l'}\Level \\\\
    \lttypwf i S l \\ \lttypwf[\Psi][\Gamma, x : S \at l] i T{l'}}
  {\lttypwf i{\PI l{l'} x S T}{l \sqcup l'}}

  \inferrule
  {\typing[L]l\Level  \\ \typing[L]{l'}\Level \\\\
    \lttyping i s{\Ty l}{1 + l} \\ \lttyping[\Psi][\Gamma, x : \Elt l s \at l] i t{\Ty{l'}}{1 + l'}}
  {\lttyping i{\PI l{l'} x s t}{\Ty{l \sqcup l'}}{1 + {(l \sqcup l')}}}

  \inferrule
  {\typing[L] l \Level \\\\ \lttyping i t{\Ty l}{1 + l}}
  {\lttypwf i{\Elt l t}{l}}

  \inferrule
  {\lpjudge{\typeof i} \Gamma \\ \typing[L]l\Level}
  {\lttypingd i{\Ty l}{1 + l}{2 + l}}

  \inferrule
  {\lpjudge \metalevel \Gamma \\ \lttypwf[\Psi][\Gamma][L,\vect\ell] 1 T l \\\\ |\vect \ell| > 0 \\ \typing[L,\vect\ell]l\Level}
  {\lttypwf \metalevel {\UPI \ell l T}{\omega}}

  \inferrule
  {\lttypwf i T{l'} \\\\ \tyequiv[L]l{l'}\Level}
  {\lttypwf i{T}{l}}
\end{mathpar}
Note that all type constructors with explicitly specified universe levels must not
refer to $\omega$. %
Indeed, $\omega$ level only appears when we validate a universe-polymorphic function
and nowhere else. %
Nor can we pass around a universe-polymorphic function. %
Moreover, universe-polymorphic functions are only available at the highest layer, which is the
layer with capability to do meta-programming and recursive intensional analysis. %
This is because that universe variables must also be visible by the bindings in global
contexts. %
% In general, universe-polymorphic functions must only be available at the highest
% layer, if \delamlang is extended to more than two layers. % 
In the well-formedness rule for $\Nat$, we use a function $\typeof i$ which alters the
layer, in which local contexts live. %
This treatment is necessary to accommodate code promotions and permit computation in the
local contexts and on the type level. %
We give the actual definition of $\typeof i$ in \Cref{sec:dt:modal}.

The equivalence between types is composed of three parts. %
The first part is the PER rules. %
\begin{mathpar}
  \inferrule
  {\lttypeq i{T}{T'}l}
  {\lttypeq i{T'}{T}l}

  \inferrule
  {\lttypeq i{T}{T'}l \\ \lttypeq i{T'}{T''}l}
  {\lttypeq i{T}{T''}l}
\end{mathpar}
Then we have the congruence rules, which simply straightforwardly propagate
equivalence downwards:
\begin{mathpar}
  \inferrule
  {\lpjudge{\typeof i} \Gamma}
  {\lttyequivd i \Nat \Nat 0 {1}}

  \inferrule
  {\tyequiv[L]{l_1}{l_3}\Level \\ \tyequiv[L]{l_2}{l_4}\Level \\\\ \lttypeq i S{S'}{l_1} \\ \lttypeq[\Psi][\Gamma, x : S \at{l_1}] iT{T'}{l_2}}
  {\lttypeq i{\PI{l_1}{l_2} x S T}{\PI{l_3}{l_4} x{S'}{T'}}{l_1 \sqcup l_2}}

  \inferrule
  {\tyequiv[L]{l_1}{l_3}\Level \\ \tyequiv[L]{l_2}{l_4}\Level \\\\
    \lttyequiv i s{s'}{\Ty{l_1}}{1 + l_1} \\
    \lttyequiv[\Psi][\Gamma, x : \Elt{l_1} s \at{l_1}] i t{t'}{\Ty{l_2}}{1 + l_2}}
  {\lttyequiv i{\PI{l_1}{l_2} x s t}{\PI{l_3}{l_4} x{s'}{t'}}{\Ty{l_1 \sqcup l_2}}{1 + {(l_1 \sqcup l_2)}}}

  \inferrule
  {\tyequiv[L]{l}{l'}\Level \\ \lttyequiv i t{t'}{\Ty l}{1 + l}}
  {\lttypeq i{\Elt l t}{\Elt{l'}{t'}}{l}}

  \inferrule
  {\lpjudge{\typeof i}\Gamma \\ \tyequiv[L]l{l'}\Level}
  {\lttyequivd i{\Ty l}{\Ty{l'}}{1 + l}{2 + l}}

  \inferrule
  {\lttypwf[\Psi][\Gamma][L,\vect\ell] \metalevel T l \\
    \lttypeq[\Psi][\Gamma][L,\vect\ell] \metalevel T{T'} l \\\\ |\vect \ell| > 0 \\ \tyequiv[L,\vect\ell]l{l'}\Level}
  {\lttypeq \metalevel {\UPI \ell l T}{\UPI \ell{l'}{T'}}{\omega}}

  \inferrule
  {\lttypeq i T{T'}{l'} \\ \tyequiv[L]l{l'}\Level}
  {\lttypeq i{T}{T'}{l}}
\end{mathpar}

Finally, we have a number of computation rules that decode terms into types:
\begin{mathpar}
  \inferrule
  {\lpjudge{\typeof i} \Gamma}
  {\lttypeq i \Nat{\Elt 0 \Nat} 0}

  \inferrule
  {\typing[L]l\Level \\ \lpjudge{\typeof i} \Gamma}
  {\lttypeq i{\Ty l}{\Elt{1 + l}{\Ty l}}{1 + l}}
  
  \inferrule
  {\typing[L]l\Level  \\ \typing[L]{l'}\Level \\
    \lttyping i s{\Ty l}{1 + l} \\ \lttyping[\Psi][\Gamma, x : \Elt l s \at l] i t{\Ty{l'}}{1 + l'}}
  {\lttypeq i{\PI l{l'} x{\Elt l s}{\Elt{l'} t}}{\Elt{l \sqcup l'}{\PI l{l'} x s t}}{l \sqcup l'}}
\end{mathpar}
We do not have an encoding for universe-polymorphic functions, so there is not a
decoding rule for them.

Then we move on to defining the typing rules for terms. %
They are pretty much straightforward:
\begin{mathpar}
  \inferrule
  {\lpjudge{\typeof i}\Gamma \\ x : T \at l \in \Gamma}
  {\lttyping i x T l}

  \inferrule
  {\lpjudge{\typeof i}\Gamma}
  {\lttyping i \ze \Nat 0}

  \inferrule
  {\lttyping i t \Nat 0}
  {\lttyping i {\su t} \Nat 0}

  \inferrule
  {\typing[L]l\Level \\ \lttypwf[\Psi][\Gamma, x : \Nat \at 0] i M l \\
    \lttyping i s {M[\ze/x]}l \\
    \lttyping[\Psi][\Gamma, x : \Nat \at 0, y : M \at l] i {s'}{M[\su x/x]}l \\
  \lttyping i t \Nat 0}
  {\lttyping i{\ELIMN l{x.M}s{x,y. s'}t}{M[t/x]}{l}}

  \inferrule
  {\typing[L]l\Level  \\ \typing[L]{l'}\Level \\ \lttypwf i S l \\ \lttyping[\Psi][\Gamma, x : S \at l] i t{T}{l'}}
  {\lttyping i{\LAM l{l'} x S t}{\PI l{l'} x S T}{l \sqcup l'}}

  \inferrule
  {\typing[L]l\Level  \\ \typing[L]{l'}\Level \\ \lttypwf i S l \\ \lttypwf[\Psi][\Gamma, x : S \at l] i{T}{l'} \\
    \lttyping i t{\PI l{l'} x S T}{l \sqcup l'} \\ \lttyping i s S l}
  {\lttyping i{\APP t l{l'} x S T s}{T[s/x]}{l'}}

  \inferrule
  {\lttyping[\Psi][\Gamma][L, \vect \ell]\metalevel{t}{T}{l} \\\\ |\vect \ell| > 0 \\ \typing[L,\vect\ell]l\Level}
  {\lttyping \metalevel {\ULAM l \ell t}{\UPI \ell l T}{\omega}}

  \inferrule
  {\lttypwf[\Psi][\Gamma][L,\vect\ell] \metalevel T l \\
    \lttyping \metalevel{t}{\UPI \ell l T}{\omega} \\\\ |\vect\ell| = |\vect l| > 0 \\ \forall l' \in \vect l ~.~ \typing[L]{l'}\Level}
  {\lttyping \metalevel{\UAPP t l}{T[\vect l/\vect \ell]}{l[\vect l/\vect \ell]}}

  \inferrule
  {\lttyping i t{T'}l \\ \lttypeq{\typeof i}T{T'}l}
  {\lttyping i t T l}

  \inferrule
  {\lttyping i tT{l'} \\ \tyequiv[L]l{l'}\Level}
  {\lttyping i t T l}
\end{mathpar}

The equivalence rules for terms are also composed of three parts. %
The PER rules are immediate:
\begin{mathpar}
  \inferrule
  {\lttyequiv i t{t'}T l}
  {\lttyequiv i{t'}tT l}

  \inferrule
  {\lttyequiv i t{t'}T l \\ \lttyequiv i{t'}{t''}T l}
  {\lttyequiv i t{t''}T l}
\end{mathpar}
The congruence rules are naturally induced by the typing rules:
\begin{mathpar}
  \inferrule
  {\lpjudge{\typeof i}\Gamma \\ x : T \at l \in \Gamma}
  {\lttyequiv i x x T l}

  \inferrule
  {\lpjudge{\typeof i}\Gamma}
  {\lttyequiv i \ze \ze \Nat 0}

  \inferrule
  {\lttyequiv i t{t'} \Nat 0}
  {\lttyequiv i {\su t}{\su{t'}} \Nat 0}

  \inferrule
  {\tyequiv[L]l{l'}\Level \\ \lttypeq[\Psi][\Gamma, x : \Nat \at 0] i M{M'} l \\
    \lttyequiv i {s_1}{s_3} {M[\ze/x]}l \\
    \lttyequiv[\Psi][\Gamma, x : \Nat \at 0, y : M \at l] i{s_2}{s_4}{M[\su x/x]}l \\
  \lttyequiv i t{t'} \Nat 0}
  {\lttyequiv i{\ELIMN l{x.M}{s_1}{x,y. s_2}t}{\ELIMN{l'}{x.M'}{s_3}{x,y. s_4}{t'}}{M[t/x]}{l}}

  \inferrule
  {\tyequiv[L]{l_1}{l_3}\Level  \\ \tyequiv[L]{l_2}{l_4}\Level \\ \lttypwf i S{l_1} \\ \lttypeq i
    S{S'}{l_1} \\
    \lttyequiv[\Psi][\Gamma, x : S \at {l_1}] i t{t'}{T}{l_2}}
  {\lttyequiv i{\LAM {l_1}{l_2} x S t}{\LAM {l_3}{l_4} x{S'}{t'}}{\PI l{l'} x S T}{l_1 \sqcup l_2}}

  \inferrule
  {\tyequiv[L]{l_1}{l_3}\Level  \\ \tyequiv[L]{l_2}{l_4}\Level \\ \lttypwf i S{l_1} \\ \lttypeq i S{S'}{l_1}
    \\
    \lttypeq[\Psi][\Gamma, x : S \at{l_1}] i{T}{T'}{l_2} \\
    \lttyequiv i t{t'}{\PI {l_1}{l_2} x S T}{l_1 \sqcup l_2} \\ \lttyequiv i s{s'} S{l_1}}
  {\lttyequiv i{\APP t {l_1}{l_2} x S T s}{\APP{t'} {l_3}{l_4} x {S'}{T'}{s'}}{T[s/x]}{l_2}}

  \inferrule
  {\lttyequiv[\Psi][\Gamma][L, \vect \ell]\metalevel{t}{t'}{T}{l} \\ |\vect \ell| > 0 \\ \tyequiv[L,\vect\ell]l{l'}\Level}
  {\lttyequiv \metalevel {\ULAM l \ell t}{\ULAM {l'} \ell{t'}}{\UPI \ell l T}{\omega}}

  \inferrule
  {\lttypwf[\Psi][\Gamma][L,\vect\ell] \metalevel T l \\
    \lttyequiv \metalevel{t}{t'}{\UPI \ell l T}{\omega} \\ |\vect\ell| = |\vect l| = |\vect
    l'| > 0 \\
    \forall 0 \le n < |\vect l| ~.~ \tyequiv[L]{\vect l(n)}{\vect l'(n)}\Level}
  {\lttyequiv \metalevel{\UAPP t l}{t'~\$~\vect l'}{T[\vect l/\vect \ell]}{l[\vect l/\vect \ell]}}

  \inferrule
  {\lttyequiv i t{t'}{T'}l \\ \lttypeq{\typeof i}T{T'}l}
  {\lttyequiv i t {t'} T l}

  \inferrule
  {\lttyequiv i t{t'}T{l'} \\ \tyequiv[L]l{l'}\Level}
  {\lttyequiv i t{t'}T l}
\end{mathpar}

\subsection{Meta-programming Modalities}\labeledit{sec:dt:modal}

In this part, we introduce the modalities for meta-programming and intensional
analysis. %
We use the $\square$ modality to represent the type of
code and we use layers to control the computational behaviors of the type theory. %
However, there are two points that we need to pay attention to:
\begin{enumerate}
\item Previously, we have discussed code promotions. %
  Code promotions imply that we must permit computation of code on the type level. %
  This further implies that we must introduce an intermediate layer between layer for code and
  that for meta-programs, which restricts the language to still MLTT but permits computation. 
\item Due to Tarski-style universes, we must introduce two kinds of contextual types,
  one for types and one for terms. %
  As seen in the syntax of types and terms, they are mutually defined, so the
  recursive principles for code of types and terms must also be mutual.
\end{enumerate}

\begin{table}[]
\begin{tabular}{|l|l|l|l|l|}
\hline
Layer            & $\varlevel$ & $\codelevel$  & $\proglevel$  & $\metalevel$                          \\ \hline
Language         & Variables only &MLTT & MLTT & MLTT extended with meta-programming \\ \hline
Computation      & No & No   & Yes  & Yes                          \\ \hline
Meta-programming & No & No   & No  & Yes                          \\ \hline
Layer of types   & $\proglevel$ & $\proglevel$  & $\proglevel$  & $\metalevel$                          \\ \hline
\end{tabular}
\caption{Features at each layer}\labeledit{tab:dt:feat}
\end{table}

Having set up the basic theme, let us begin with the syntax of the extension to MLTT:
\begin{alignat*}{2}
  i &:= &&\ \varlevel \sep \codelevel \sep \proglevel \sep \metalevel \tag{Layers} \\
  U & && \tag{Global variables as types} \\
  u & && \tag{Global variables as terms} \\
  % u_\# & && \tag{Global variables as variables} \\
  k & && \tag{Natural numbers, $\N$} \\
  \delta &:=&& \cdot_{g?}^k \sep \wk_g^k \sep \delta, t/x \tag{Local substitutions} \\
  B &:=&&\ g : \Ctx \sep U : \DTyp i l \sep u : \DTrm i T l
          \tag{Global bindings} \\
  \Phi, \Psi &:= &&\ \cdot \sep \Phi, B \tag{Global contexts} \\
  S, T & := && \cdots \sep U^\delta \sep \CPI g l T \sep \TPI U l{l'}T \sep \CTyp l \sep \CTrm T l \\
  s, t & := && \cdots \sep u^\delta \sep \CLAM l g t \sep \CAPP t \Gamma
               \sep \TLAM l{l'} U t \sep \TAPP t T
               \sep \boxit T \sep \boxit t \\
    & && \sep \LETBTYP{l'}l\Gamma{x_T. M}{U}{t'}t \sep \LETBTRM{l'}l\Gamma
         T{x_t. M}{u}{t'}t \\
    & && \sep \ELIMTYP{l_1}{l_2}Mbl\Gamma t \sep \ELIMTRM{l_1}{l_2}Mbl\Gamma T t \\
  \vect M & := &&\ (\ell,g,x_T. M_\Typ)~(\ell,g,U_T,x_t. M_\Trm)
                  \tag{Two motives for mutual recursion of code} \\
  \vect b &:= && \vect{b}_\Typ~\vect{b}_\Trm
                 \tag{Branches for mutual recursion of code} \\
  b_\Typ &:= &&\ (g. t_\Nat) \sep (\ell,\ell',g, U_S,U_T,x_S,x_T. t_\Pi) \sep (\ell,g. t_{\Se})
                \sep (\ell,g,u_t,x_t. t_\tEl)
                       \tag{Branches for code of types} \\
  b_\Trm &:= &&\ (\ell,g,U_T,u_x. t_x) \sep (g.t'_\Nat) \sep (\ell,\ell',g,u_s,u_t,x_s,x_t.t'_\Pi) \sep
                (\ell,g.t'_\Se) 
                \tag{Branches for code of terms} \\
    & && \sep (g.t_\ze) \sep (g,u_t,x_t.t_{\tsucc}) \sep
         (\ell,g,U_M,u_s,u_{s'},u_t,x_M,x_s,x_{s'},x_t.t_{\telimn}) \\
    & && \sep (\ell,\ell',g,U_S,U_T,u_t,x_S,x_t.t_\lambda) \sep (\ell,\ell',g,U_S,U_T,u_t,u_s,x_S,x_T,x_t,x_s.t_\tapp)
\end{alignat*}
Following the layering principle before, we index our judgments with an layer index
$i$. %
We include four layers and these layers are summarized in \Cref{tab:dt:feat}. %
To elaborate, we begin with layer $\varlevel$, which is the layer that contains only
variables. %
This layer is needed to describe the base case of the recursive principles for code
when a local variable is hit. %
Layer $\varlevel$ (for variables) is \emph{not} available for most rules given in
\Cref{sec:dt:mltt-rules} other than the local variable rule and its congruence. %
We will follow this convention for the rest of this report, unless we specifically
state that layer $\varlevel$ is available for particular rules. %
Layer $\codelevel$ (for code) is the layer for code of MLTT. %
This layer is akin to layer $0$ in \Cref{sec:cv,sec:cv:logrel}, where static code
resides and no computation is allowed. %
However, in order to capture code promotions, we must introduce another layer, $\proglevel$ (for
programs), to permit computation of MLTT \emph{programs} in local contexts and on the
type level. %
This layer is especially crucial in the recursive principle for terms for the argument
$T$ where code promotions are implicitly handled by definitional equivalence. %
However, no meta-programs are allowed at layer $\proglevel$; in other words, the language at
layer $\proglevel$ is virtually vanilla MLTT. %
Therefore, unlike simple types, there are two layers in \delamlang permitting
computation. %
At last, we have layer $\metalevel$ (for meta-programs) where we have the power to do
meta-programming. %
At this layer, we have access to not only universe-polymorphic functions, but also
contextual types and recursive principles for code. %
All meta-functions must reside at this layer. %
All layers are related by a strict order of $\varlevel < \codelevel < \proglevel < \metalevel$. %

The reason to introduce layer $\proglevel$ becomes obvious by considering which layer the type
of a given MLTT term should live in. %
For instance, given a judgment $\lttyping \codelevel t T l$, we know $t$ lives at layer $\codelevel$ as
code, but what about $T$? %
Since $T$ is a type and we want $T$ to compute, it cannot live at layer $\codelevel$, but also
not $\metalevel$ as it must be a well-formed pure MLTT type. %
Indeed, $T$ ought to live at layer $\proglevel$, i.e. $\lttypwf \proglevel T l$. %
What about $\lttyping \proglevel t T l$? %
In this case, $t$ lives at layer $\proglevel$. %
Since $T$ must still be a well-formed pure MLTT type and compute, we must have
$\lttypwf \proglevel T l$. %
The type of a term living at layer $\metalevel$ simply also lives at layer $\metalevel$.  %
The relation of layers of terms and types is summarized by the following function:
\begin{align*}
  \typeof \varlevel &:= \proglevel \\
  \typeof \codelevel &:= \proglevel \\
  \typeof \proglevel &:= \proglevel \\
  \typeof \metalevel &:= \metalevel
\end{align*}
The judgment $\compt i$ quantifies computable layers:
\begin{mathpar}
  \inferrule
  { }
  {\compt \proglevel}

  \inferrule
  { }
  {\compt \metalevel}
\end{mathpar}

Then we extend our system with five types, following~\Cref{sec:cv}:
\begin{itemize}
\item $U^\delta$ is a global variable for types. %
  Due to separation of types and terms, we need a way to refer to code of types on the
  type level. 
\item $\CPI g l T$ is a meta-function type for introducing a context variable $g$
  to the global context. %
  We also have this in~\Cref{sec:cv}.
\item $\TPI U l{l'}T$ is similarly a meta-function type for introducing a type at
  layer $\proglevel$ to the global context. %
  This type is introduced to provide an index for the contextual type for terms to be
  discussed in the second next item.
\item $\CTyp l$ is a contextual type for types in MLTT. %
  It represents a static code of types. %
\item Finally, $\CTrm T l$ is a contextual type for terms in MLTT. %
  This $T$ may refer to the index type at layer $\proglevel$ introduced by meta-functions above. 
\end{itemize}

Since there are two kinds of contextual types now, there four kinds of bindings in a
global context:
\begin{itemize}
\item context variables $g : \Ctx$ representing a local context;
\item global variables $U : \DTyp i l$ representing a type in MLTT (note that $i \in
  \{\codelevel, \proglevel\}$);
\item global variables $u : \DTrm i T l$ representing a term in MLTT (note that
  $i \in \{\varlevel, \codelevel\}$ and there is no way to introduce a term at layer $\proglevel$ to global
  context).
\end{itemize}

Now, let us move on to discuss the extended terms. %
\begin{itemize}
\item First, we also introduce global variables and local substitutions. %
  Their syntax is identical to one in \Cref{sec:cv}. 
\item Then we have the introduction and elimination forms for meta-functions of
  context variables, $\CLAM l g t$ and $\CAPP t \Gamma$. 
\item Similarly, we have the introduction and elimination forms for meta-functions of
  types, $\TLAM l{l'} U t$ and $\TAPP t T$. 
\item Then we have the introduction forms of two kinds of contextual types.
\item Following \citet[Sec. 4]{hu2024layered}, we have two elimination forms for each kind of
  contextual types, $\tletbox$ and the recursive principles. %
  Same as before, $\tletbox$ is responsible for code composition and evaluation. %
  Intentional analyses are done through the recursive principles. %
  In \delamlang, $\tletbox$ is a bit more complex as it requires a specified motive. %
  We alter the syntax a little bit to make $\tletbox$ more like an operation:
  $\LETBTYP{l'}l\Gamma{x_T. M}{u}{t'}t$ and $\LETBTRM{l'}l\Gamma T{x_t. M}{u}{t'}t$.

\item Finally, we extend the recursive principles for code. %
  As indicated before, code of types and terms in MLTT are mutually defined, so the
  recursive principles must also be mutual. %
  The two recursive principles $\ELIMTYP{l_1}{l_2}Mbl\Gamma t$ and
  $\ELIMTRM{l_1}{l_2}Mbl\Gamma T t$ require two motives, one for types and one for terms,
  and contain all branches for code of types and terms. %
  Their difference is what exactly eventually being eliminated, as indicated by their
  subscript. %
  The branches $\vect b$ are a list of branches $\vect b_\Typ$ and $\vect b_\Trm$,
  where $\vect b_\Typ$ and $\vect b_\Trm$ contain \emph{all} branches for types and
  terms, respectively. %

\end{itemize}
In the branches, there are four kinds of variables.
\begin{itemize}
\item There is a globally introduced context variable $g$ which represents the
  local context where the code lives.
\item There could be some universe variables that are used to tell the universe levels
  of some types.
\item There could be some global variables $u$ and $U$ represent the sub-structures of
  a given case. %
  The subscripts correspond tightly to the syntax given in \Cref{sec:dt:syntax}.
\item For each sub-structure, there is one corresponding recursive variable $x$. %
  Again, the subscripts correspond tightly to the sub-structure. 
\end{itemize}

\subsection{More Typing and Equivalence Judgments}

In this section, we specify the remainder of the rules. %
We begin with the well-formedness rule for the global contexts. %
Recall that layer $\varlevel$ does not apply for most rules below, unless the otherwise is
specifically stated.
\begin{mathpar}
  \inferrule
  { }
  {\judge[L]\cdot}

  \inferrule
  {\judge[L]\Psi}
  {\judge[L]{\Psi, g : \Ctx}}

  \inferrule
  {\judge[L]\Psi \\ \lpjudge \proglevel \Gamma \\\\ \typing[L]l\Level \\ i \in \{\codelevel, \proglevel\}}
  {\judge[L]{\Psi, U : \DTyp i l}}

  \inferrule
  {\judge[L]\Psi \\ \lttypwf \proglevel T l \\\\ \typing[L]l\Level \\ i \in \{\varlevel, \codelevel\}}
  {\judge[L]{\Psi, u : \DTrm i T l}}
\end{mathpar}
The judgments for local substitutions follow very closely \Cref{sec:cv}. %
In these rules, $i$ might take $\varlevel$. %
This permission has a particular effect on the step case, which forces all terms in a
local substitution must be variables. %
\begin{mathpar}
  \inferrule
  {\lpjudge{\typeof i}\Gamma \\\\ \text{$\Gamma$ ends with $\cdot$} \\ |\Gamma| = k}
  {\ltsubst i {\cdot^k}{\cdot} \\ \ltsubeq i {\cdot^k}{\cdot^k}{\cdot}}

  \inferrule
  {\lpjudge{\typeof i}\Gamma \\ g : \Ctx \in \Psi \\\\ \text{$\Gamma$ ends with $g$} \\ |\Gamma| = k}
  {\ltsubst i {\cdot_g^k}{\cdot} \\ \ltsubeq i {\cdot_g^k}{\cdot_g^k}{\cdot}}

  \inferrule
  {\lpjudge{\typeof i}\Gamma \\ g : \Ctx \in \Psi \\ \text{$\Gamma$ ends with $g$} \\ |\Gamma| = k}
  {\ltsubst i {\wk_g^k}{g} \\ \ltsubeq i {\wk_g^k}{\wk_g^k}{g}}

  \inferrule
  {\typing[L]l\Level \\ \ltsubst i {\delta}{\Delta} \\ \lttypwf[\Psi][\Delta]{\typeof i}T l \\ \lttyping i {t}{T[\delta]} l}
  {\ltsubst i {\delta, t/x}{\Delta, x : T \at l}}

  \inferrule
  {\typing[L]l\Level \\ \ltsubeq i {\delta}{\delta'}{\Delta} \\ \lttypwf[\Psi][\Delta]{\typeof i}T l \\
    \lttyping i {t}{T[\delta]} l \\ \lttyping i {t'}{T[\delta]} l \\ \lttyequiv i {t}{t'}{T[\delta]} l}
  {\ltsubeq i {\delta, t/x}{\delta', t'/x}{\Delta, x : T \at l}}
\end{mathpar}
In the step case for equivalence above, we ask for two redundant premises of the
well-typedness of $t$ and $t'$ to provide an early presupposition for equivalence of
local substitutions. %
We will need this early presupposition in \Cref{lem:dt:lsubst-eq}. %
It breaks the dependency loop so that reaching the full presupposition lemma becomes
viable. 

Now let us consider the extended types and their equivalence.
\begin{mathpar}
  \inferrule
  {\lpjudge{\typeof i}\Gamma \\ U : \DTyp[\Delta]{i'} l \in \Psi \\ i' \in \{\codelevel, \proglevel\} \\ i' \le i \\ \ltsubst i \delta \Delta}
  {\lttypwf i{U^\delta}{l}}

  \inferrule
  {\lpjudge \metalevel\Gamma \\\\ \lttypwf[\Psi, g : \Ctx] \metalevel T l \\ \typing[L]l\Level}
  {\lttypwf \metalevel{\CPI g l T}{l}}

  \inferrule
  {\lpjudge \metalevel\Gamma \\ \lttypwf[\Psi, U : \DTyp[\Delta] \proglevel l] \metalevel T{l'} \\\\ \typing[L]l\Level \\ \typing[L]{l'}\Level}
  {\lttypwf \metalevel{\TPI U[\Delta] l{l'} T}{l'}}

  \inferrule
  {\lpjudge \metalevel\Gamma \\ \lpjudge \proglevel \Delta \\ \typing[L]l\Level}
  {\lttypwf \metalevel{\CTyp[\Delta]l}{0}}

  \inferrule
  {\lpjudge \metalevel\Gamma \\ \lttypwf[\Psi][\Delta] \proglevel T l \\ \typing[L]l\Level}
  {\lttypwf \metalevel{\CTrm[\Delta]T l}{0}}
\end{mathpar}
Notice that contextual types always live on the universe level $0$. %
This is fine because the meta-language at layer $\metalevel$ is strictly more expressive than
that at layer $\codelevel$, i.e. \mltt. %
Therefore, the universe level $0$ is large enough to encode all intrinsically typed
syntax of \mltt without causing any size issue. %
This intuition is further justified in the logical relations, where the semantics of
contextual types need not to refer to other semantics at layer $\metalevel$, so contextual
types can safely reside on level $0$. %
Similarly, meta-functions with open types as inputs also do not need to be concerned
about the universe levels of the types. %

The additional equivalence rules are just their congruence rules:
\begin{mathpar}
  \inferrule
  {\lpjudge{\typeof i}\Gamma \\ U : \DTyp[\Delta]{i'} l \in \Psi \\ i' \in \{\codelevel, \proglevel\} \\ i' \le i \\ \ltsubeq i \delta{\delta'} \Delta}
  {\lttypeq i{U^\delta}{U^{\delta'}}{l}}

  \inferrule
  {\lpjudge \metalevel\Gamma \\ \lttypeq[\Psi, g : \Ctx] \metalevel T{T'} l \\ \tyequiv[L]l{l'}\Level}
  {\lttypeq \metalevel{\CPI g l T}{\CPI g{l'}{T'}}{l}}

  \inferrule
  {\lpjudge \metalevel\Gamma \\ \lpequiv \proglevel \Delta{\Delta'} \\ \lttypeq[\Psi, U : \DTyp[\Delta] \proglevel l] \metalevel T{T'}{l'} \\ \tyequiv[L]{l_1}{l_3}\Level \\ \tyequiv[L]{l_2}{l_4}\Level}
  {\lttypeq \metalevel{\TPI U[\Delta]{l_1}{l_2}T}{\TPI U[\Delta']{l_3}{l_4}{T'}}{l_2}}

  \inferrule
  {\lpjudge \metalevel\Gamma \\ \lpequiv \proglevel \Delta{\Delta'} \\ \tyequiv[L]l{l'}\Level}
  {\lttypeq \metalevel{\CTyp[\Delta]l}{\CTyp[\Delta']{l'}}{0}}

  \inferrule
  {\lpjudge \metalevel\Gamma \\ \lpequiv \proglevel \Delta{\Delta'} \\ \lttypeq[\Psi][\Delta] \proglevel T{T'} l \\ \tyequiv[L]l{l'}\Level}
  {\lttypeq \metalevel{\CTrm[\Delta]T l}{\CTrm[\Delta']{T'}{l'}}{0}}
\end{mathpar}

Next, we list the extended typing judgments:
\begin{mathpar}
  \inferrule
  {\lpjudge{\typeof i}\Gamma \\ u : \DTrm[\Delta]{i'}T l \in \Psi \\ i' \in \{\varlevel, \codelevel\} \\ i \in
    \{\varlevel, \codelevel, \proglevel, \metalevel\} \\ i' \le i \\ \ltsubst i \delta \Delta}
  {\lttyping i{u^\delta}{T[\delta]}{l}}

  \inferrule
  {\lpjudge \metalevel\Gamma \\\\ \lttyping[\Psi, g : \Ctx]\metalevel t{T}{l} \\ \typing[L]l\Level}
  {\lttyping \metalevel{\CLAM l g t}{\CPI g l T}{l}}

  \inferrule
  {\lttypwf[\Psi, g : \Ctx] \metalevel T l  \\\\ \lttyping \metalevel{t}{\CPI g l T}{l} \\ \lpjudge \proglevel \Delta}
  {\lttyping \metalevel{\CAPP t \Delta}{T[\Delta/g]}{l}}

  \inferrule
  {\lpjudge \metalevel\Gamma \\ \lttyping[\Psi, U : \DTyp[\Delta] \proglevel l]\metalevel{t}{T}{l'} \\ \typing[L]{l}\Level \\ \typing[L]{l'}\Level}
  {\lttyping \metalevel{\TLAM{l}{l'}U t}{\TPI U[\Delta] l{l'} T}{l'}}

  \inferrule
  {\lttypwf[\Psi, U : \DTyp[\Delta] \proglevel l] \metalevel {T'}{l'} \\ \lttyping \metalevel{t}{\TPI U[\Delta] l{l'}{T'}}{l'} \\ \lttypwf[\Psi][\Delta] \proglevel{T} l}
  {\lttyping \metalevel{\TAPP t{T}}{T'[T/U]}{l'}}

  \inferrule
  {\lpjudge \metalevel \Gamma \\ \lttypwf[\Psi][\Delta]\codelevel T l}
  {\lttyping \metalevel{\boxit T}{\CTyp[\Delta] l}{0}}

  \inferrule
  {\lpjudge \metalevel \Gamma \\ \lttyping[\Psi][\Delta]\codelevel t T l}
  {\lttyping \metalevel{\boxit t}{\CTrm[\Delta] T l}{0}}

  \inferrule
  {\typing[L]{l'}\Level \\ \typing[L]l\Level \\ \lpjudge \proglevel \Delta \\
    \lttyping \metalevel{t}{\CTyp[\Delta]l}{1 + l} \\
    \lttypwf[\Psi][\Gamma,x_T : \CTyp[\Delta]l \at{0}]\metalevel{M}{l'} \\
  \lttyping[\Psi, U : \DTyp[\Delta]\codelevel l]\metalevel{t'}{M[\boxit U/x_T]}{l'}}
  {\lttyping \metalevel{\LETBTYP{l'}l \Delta{x_T.M}{U}{t'}t}{M[t/x_T]}{l'}}

  \inferrule
  {\typing[L]{l'}\Level \\ \typing[L]l\Level \\ \lpjudge \proglevel \Delta \\
    \lttypwf[\Psi][\Delta]\proglevel T l \\
    \lttyping \metalevel{t}{\CTrm[\Delta] T l}{l} \\
    \lttypwf[\Psi][\Gamma,x_t : \CTrm[\Delta] T l \at{0}]\metalevel{M}{l'} \\
  \lttyping[\Psi, u : \DTrm[\Delta]\codelevel T l]\metalevel{t'}{M[\boxit u/x_t]}{l'}}
  {\lttyping \metalevel{\LETBTRM{l'}l \Delta T{x_t.M}{u}{t'}t}{M[t/x_t]}{l'}}
\end{mathpar}
For the typing rule of $\CAPP t \Delta$, we require $\Delta$ to be a context at layer
$\proglevel$. %
This is because a context variable represents a local context for an MLTT term. %
A local context for an MLTT term necessarily lives at layer $\proglevel$, so we can only
substitute a context living at layer $\proglevel$ with a context variable. %

Now we shall mentally prepare ourselves to write down the typing rules for the two
recursive principles. %
They are conceptually easy but simply verbose to write down. %
We will only write down the rules for this time for completeness and in later
discussions, we simply omit the premises. %
% As both recursive principles share most premises, we write them
Our goal is to provide the following conclusions:
\begin{gather*}
  \lttyping \metalevel {\ELIMTYP{l_1}{l_2}M b{l'}\Delta t}{M_\Typ[l'/l,\Delta/g,t/x_T]}{l_1} \\
  \lttyping \metalevel {\ELIMTRM{l_1}{l_2}M b{l'}\Delta T t}{M_\Trm[l'/l,\Delta/g,T/U_T,t/x_t]}{l_2}
\end{gather*}
We group the premises into different parts. %
First we give the premises related to the motives:
\begin{mathpar}
  \typing[L]{l_1}\Level

  \typing[L]{l_2}\Level

  \lttypwf[\Psi, g : \Ctx][\Gamma, x_T : \CTyp[g]\ell \at{0}][L,\ell] \metalevel{M_\Typ}{l_1}

  \lttypwf[\Psi, g : \Ctx, U_T : \DTyp[g]\proglevel\ell][\Gamma, x_t : \CTrm[g]{U_T^\id}\ell \at{0}][L,\ell] \metalevel{M_\Trm}{l_2}
\end{mathpar}
where $\vect M =  (\ell,g,x_T. M_\Typ)~(\ell,g,U_T,x_t. M_\Trm)$. %
In the premises above, we give the well-formedness of two motives for code of types
and terms, respectively. %
Let us call this group $G_M$. %
We move on to considering the branches. %
We first consider the branches for code of types. %
It is relatively easy as there are only four cases:
\begin{itemize}
\item 
  \[
    \lttyping[\Psi, g : \Ctx]\metalevel{t_\Nat}{M_\Typ[0/\ell,g/g,\boxit \Nat/x_T]}{l_1}
  \]
\item We explain this premise more carefully. Consider some code of type
  $\boxit{\PI{l}{l'}x S T}$, then we have the matching premise
  \[
    \lttyping[\Psi'][\Gamma'][L,\ell,\ell']\metalevel{t_\Pi}{M_\Typ[\ell\sqcup \ell'/\ell,g/g,\boxit{\PI{\ell}{\ell'}x{U_S^\id}{U_T^\id}}/x_T]}{l_1}
  \]
  where
  \begin{align*}
    \Psi' :=&~ \Psi \\
    ,&~ g : \Ctx \\
    ,&~U_S : \DTyp[g] \codelevel \ell
       \tag{the global variable for the input type, which captures $S$}\\
    ,&~U_T :\DTyp[g,x:U_S^\id \at{\ell}]\codelevel {\ell'}
       \tag{the global variable for the output type, which captures $T$; note that it
       lives in an extended local context}
  \end{align*}
  and
  \begin{align*}
    \Gamma' :=&~ \Gamma \\
    ,&~ x_S: M_\Typ[\ell/\ell,g/g,\boxit{U_S^\id}/x_T] \at{l_1}
       \tag{the recursive call for $S$ of type $M$ that is properly substituted}\\
    ,&~x_T:M_\Typ[\ell'/\ell,(g,x:U_S^\id \at{\ell})/g,\boxit{U_T^\id}/x_T]\at {l_1}
       \tag{the recursive call for $T$; see how the local context is extended}
  \end{align*}
\item
  Further,
  \[
    \lttyping[\Psi,g:\Ctx][\Gamma][L,\ell]\metalevel{t_\Se}{M_\Typ[1 + \ell/\ell,g/g,\boxit{\Ty \ell}/x_T]}{l_1}
  \]
\item
  \[
    \lttyping[\Psi,g:\Ctx, u_t:\DTrm[g]\codelevel{\Ty \ell}{1 + \ell}][\Gamma'][L,\ell]\metalevel{t_\tEl}{M_\Typ[\ell/\ell,g/g,\boxit{(\Elt \ell{u_t^\id})}/x_T]}{l_1}
  \]
  where
  \[
    \Gamma' := \Gamma, x_t : M_\Trm[1 + \ell/\ell,g/g,\Ty \ell/U_T,\boxit{u_t^\id}/x_t]  \at{l_2}
  \]
\end{itemize}
Let us call this group $G_\Typ$. 

Lastly, let us consider the nine cases for terms.
\begin{itemize}
\item
  \[
    \lttyping[\Psi, g : \Ctx,U_T : \DTyp[g] \proglevel \ell, u_x : \DTrm[g]\varlevel{U_T^\id}{\ell}][\Gamma][L,\ell]\metalevel{t_x}{M_\Trm[\ell/\ell,g/g,U_T^\id/U_T,\boxit u_x/x_t]}{l_2}
  \]
  In this case, the type of the variable is captured by $U_T$. %
  It has to live at layer $\proglevel$ because it is not a sub-structure of the variable,
  i.e. it is obtained externally, from the indexing arguments of the recursive principles. %
\item
  \[
    \lttyping[\Psi, g : \Ctx]\metalevel{t'_\Nat}{M_\Trm[1/\ell,g/g,\Ty 0/U_T,\boxit \Nat/x_t]}{l_2}    
  \]
\item
  \[
    \lttyping[\Psi'][\Gamma'][L,\ell,\ell']\metalevel{t'_\Pi}{M_\Trm[1 +{(\ell\sqcup \ell')}/\ell,g/g,\Ty{\ell\sqcup \ell'}/U_T,\boxit{\PI{\ell}{\ell'}x{u_s^\id}{u_t^\id}}/x_t]}{l_2}
  \]
  where
  \[
    \Psi' := \Psi, g : \Ctx,u_s : \DTrm[g] \codelevel {\Ty \ell}{1 + \ell},u_t :\DTrm[g,x:\Elt{\ell}{u_s^\id}]\codelevel{\Ty{\ell'}} {1 + \ell'}
  \]
  and
  \begin{align*}
    \Gamma' :=&~ \Gamma \\
    ,&~ x_s: M_\Trm[1 + \ell/\ell,g/g,\Ty \ell/U_T,\boxit{u_s^\id}/x_t] \at{l_2} \\
    ,&~ x_t:M_\Trm[1 + \ell'/\ell,(g,x:\Elt{\ell}{u_s^\id})/g,\Ty{\ell'}/U_T,\boxit{u_t^\id}/x_t]\at {l_2}
  \end{align*}
  Notice that this premise for the encoding of $\Pi$ is almost identical to the
  premise in $G_\Typ$ above, with necessary adjustment to return the proper motive
  $M_\Trm$ instead. %
  
\item
  \[
    \lttyping[\Psi,g:\Ctx][\Gamma][L,\ell]\metalevel{t'_\Se}{M_\Trm[{2 + \ell}/\ell,g/g, \Ty{1 + \ell}/U_T,\boxit{\Ty \ell}/x_T]}{l_2}
  \]
\item
  \[
    \lttyping[\Psi, g : \Ctx]\metalevel{t_\ze}{M_\Trm[0/\ell,g/g,\Nat/U_T,\boxit \ze/x_t]}{l_2}    
  \]
\item 
  \[
    \lttyping[\Psi, g : \Ctx,u_t: \DTrm[g]\codelevel\Nat 0][\Gamma']\metalevel{t_\tsucc}{M_\Trm[0/\ell,g/g,\Nat/U_T,\boxit{(\su{u_t^\id})}/x_t]}{l_2}    
  \]
  where $\Gamma' := \Gamma, x_t:M_\Trm[0/\ell,g/g,\Nat/U_T,\boxit{u_t^\id}/x_t] \at{l_2}$.
\item We carefully explain this premise for the code of the elimination of natural
  numbers. %
  Recall that the syntax is $\ELIMN l{x.M}{s}{x,y.s'}t$. %
  We use corresponding global variables to capture the sub-structures. 
  \[
    \lttyping[\Psi'][\Gamma'][L,\ell]\metalevel{t_{\telimn}}{M_\Trm[\ell/\ell,g/g,U_M^{\id_g,u_t^\id/x}/U_T,\ELIMN \ell{x.U_M^{\id_{g,x}}}{u_s^{\id_g}}{x,y. u_{s'}^{\id_{g,x,y}}}{u_t^{\id_g}}/x_t]}{l_2}
  \]
  where
  \begin{align*}
    \Psi' :=&~ \Psi \\
    ,&~g : \Ctx \\
    ,&~U_M : \DTyp[g,x:\Nat\at 0]\codelevel \ell
       \tag{the global variable for the motive; it lives at layer $\codelevel$ as it is a sub-structure}\\
    ,&~u_s : \DTrm[g]\codelevel{U_M^{\id_g,\ze/x}}\ell
       \tag{the code for the base case; its type refers to the code of the motive with
       $x$ for $\ze$}\\
    ,&~u_{s'} : \DTrm[g, x : \Nat\at 0, y : U_M^{\id_{g,x}} \at
       \ell]\codelevel{U_M^{\id_g,\su{x}/x}}\ell
       \tag{the code for the step case; the local context is extended with the
       predecessor and the recursive call} \\
    ,&~u_t : \DTrm[g]\codelevel{\Nat}0
       \tag{the code for the scrutinee}
  \end{align*}
  and
  \begin{align*}
    \Gamma' :=&~\Gamma \\
    ,&~x_M : M_\Typ[\ell/\ell,(g,x:\Nat\at 0)/g,\boxit{U_M^{\id_{g,x}}}/x_T] \at{l_1}
       \tag{since the motive is a sub-structure, a recursive call is available}\\
    ,&~x_s : M_\Trm[\ell/\ell,g/g,U_M^{\id_g,\ze/x}/U_T,\boxit{u_s^{\id_g}}/x_t] \at{l_2}
       \tag{the recursive call for the base case; recall that $U_T$ is the type of
       $x_t$,}\\
    \tag{which in this case is also the type of $u_s$} \\
    ,&~x_{s'} : M_\Trm[\ell/\ell,(g, x : \Nat\at 0, y : U_M^{\id_{g,x}} \at
       \ell)/g,U_M^{\id_g,\su x/x}/U_T,\boxit{u_{s'}^{\id_{g,x,y}}}/x_t] \at{l_2}
       \tag{the recursive call for the step case; similar logic applies but more longer}\\
    ,&~x_t : M_\Trm[\ell/\ell,g/g,\Nat/U_T,\boxit{u_t^\id}/x_t] \at{l_2}
       \tag{the recursive call for the scrutinee}
  \end{align*}
\item
  \[
    \lttyping[\Psi'][\Gamma'][L,\ell,\ell']\metalevel{t_\lambda}{M_\Trm[\ell\sqcup \ell'/\ell,g/g,\PI
      \ell{\ell'}x{U_S^{\id_g}}{U_T^{\id_{g,x}}}/U_T, \boxit{\LAM \ell{\ell'} x{U_S^{\id_g}}{u_t^{\id_{g,x}}}}/x_t]}{l_2}
  \]
  where
  \begin{align*}
    \Psi' :=&~ \Psi \\
    ,&~g : \Ctx \\
    ,&~U_S : \DTyp[g]\codelevel \ell \\
    ,&~U_T : \DTyp[g, x : U_S^{\id_g} \at \ell]\proglevel{\ell'} \\
    ,&~u_t : \DTrm[g, x : U_S^{\id_g} \at \ell]\codelevel{U_T^{\id_{g,x}}}{\ell'}
  \end{align*}
  and
  \begin{align*}
    \Gamma' :=&~\Gamma \\
    ,&~x_S : M_\Typ[\ell/\ell,g/g,\boxit{U_S^{\id_g}}/x_T] \at{l_1} \\
    ,&~x_t : M_\Trm[\ell'/\ell,(g,x : U_S^{\id_g} \at \ell)/g,U_T^{\id_{g,x}}/U_T,\boxit{u_t^{\id_{g,x}}}/x_t] \at{l_2}
  \end{align*}
  Note that here $U_T$ is at layer $\proglevel$. %
  This is because the return type of not a sub-structure in a function abstraction
  $\LAM l{l'} x S t$, and therefore it must be captured externally from the indexing
  arguments of the recursive principle. %
  Since it is not a sub-structure, there also is not a recursive call for it. %
  It is possible to include the return type as a sub-structure, e.g. $\LAM{l}{l'}x S{(t
    : T)}$ but we decided to show this alternative to demonstrate various design spaces.
  
\item Finally,
  \[
    \lttyping[\Psi'][\Gamma'][L,\ell,\ell']\metalevel{t_\tapp}{M_\Trm[\ell'/\ell,g/g,U_T^{\id_g,u_t^{\id_g}/x}/U_T,\boxit{(\APP{u_t^{\id_g}}
        \ell{\ell'} x{U_S^{\id_{g}}}{U_T^{\id_{g,x}}}{u_s^{\id_g}})}/x_t]}{l_2}
  \]
  where
  \begin{align*}
    \Psi' :=&~ \Psi \\
    ,&~g : \Ctx \\
    ,&~U_S : \DTyp[g]\codelevel \ell \\
    ,&~U_T : \DTyp[g, x : U_S^{\id_g} \at \ell]\codelevel{\ell'} \\
    ,&~u_t : \DTrm[g]\codelevel{\PI \ell{\ell'}x{U_S^{\id_{g}}}{U_T^{\id_{g,x}}}}{\ell \sqcup \ell'} \\
    ,&~u_s : \DTrm[g]\codelevel{U_S^{\id_{g}}}{\ell}
  \end{align*}
  and
  \begin{align*}
    \Gamma' :=&~\Gamma \\
    ,&~x_S : M_\Typ[\ell/\ell,g/g,\boxit{U_S^{\id_g}}/x_T] \at{l_1} \\
    ,&~x_T : M_\Typ[\ell/\ell,(g,x : U_S^{\id_g} \at \ell)/g,\boxit{U_T^{\id_{g,x}}}/x_T] \at{l_1} \\
    ,&~x_t : M_\Trm[\ell \sqcup \ell'/\ell,g/g,\PI \ell{\ell'}x{U_S^{\id_{g}}}{U_T^{\id_{g,x}}}/U_T,\boxit{u_t^{\id_{g}}}/x_t] \at{l_2} \\
    ,&~x_s : M_\Trm[\ell/\ell,g/g,U_S^{\id_{g}}/U_T,\boxit{u_s^{\id_{g}}}/x_t] \at{l_2}
  \end{align*}
  This premise shows why we must use a more verbose syntax for application, i.e. $\APP
  t l{l'}xSTs$. %
  In the global context, we must introduce the global variables for the input and
  output types. %
  However, a vanilla function application $t~s$ has no such information at all. %
  Since the current syntax has both input and output types as sub-structures, we can
  also allow their recursive calls. 
\end{itemize}
All premises above conclude the group for terms, which we name $G_\Trm$. %
We collectively use $G_A$ for all three groups above, i.e. $G_A :=
G_M~G_\Typ~G_\Trm$. %
Then we have the typing rule for the recursive principles as follows:
\begin{mathpar}
  \inferrule
  {G_A \\ \typing[L]{l'}\Level \\ \lpjudge \proglevel \Delta \\ \lttyping \metalevel t{\CTyp[\Delta]{l'}}{0}}
  {\lttyping \metalevel {\ELIMTYP{l_1}{l_2}M b{l'}\Delta t}{M_\Typ[l'/\ell,\Delta/g,t/x_T]}{l_1}}

  \inferrule
  {G_A \\ \typing[L]{l'}\Level \\ \lpjudge \proglevel \Delta \\ \lttypwf[\Psi][\Delta]\proglevel T{l'} \\
    \lttyping \metalevel t{\CTrm[\Delta]T{l'}}{0}}
  {\lttyping \metalevel {\ELIMTRM{l_1}{l_2}M b{l'}\Delta T t}{M_\Trm[l'/\ell,\Delta/g,T/U_T,t/x_t]}{l_2}}
\end{mathpar}

\subsection{More Congruence Rules for Typing}

The congruence rules for the additional typing rules are naturally derived from
the typing rules above.
\begin{mathpar}
  \inferrule
  {\lpjudge{\typeof i}\Gamma \\ u : \DTrm[\Delta]{i'}T l \in \Psi \\ i' \in \{\varlevel, \codelevel\} \\ i \in
    \{\varlevel, \codelevel, \proglevel, \metalevel\} \\ i' \le i \\ \ltsubeq i \delta{\delta'} \Delta}
  {\lttyequiv i{u^\delta}{u^{\delta'}}{T[\delta]}{l}}

  \inferrule
  {\lpjudge \metalevel\Gamma \\ \lttyequiv[\Psi, g : \Ctx]\metalevel t{t'}{T}{l} \\ \tyequiv[L]l{l'}\Level}
  {\lttyequiv \metalevel{\CLAM l g t}{\CLAM{l'} g{t'}}{\CPI g l T}{l}}

  \inferrule
  {\lttyequiv \metalevel{t}{t'}{\CPI g l T}{l} \\ \lpequiv \proglevel \Delta{\Delta'}}
  {\lttyequiv \metalevel{\CAPP t \Delta}{\CAPP{t'}{\Delta'}}{T[\Delta/g]}{l}}

  \inferrule
  {\lpjudge \metalevel\Gamma \\ \lttyequiv[\Psi, U : \DTyp[\Delta] \proglevel{l_1}]\metalevel{t}{t'}{T}{l_2} \\ \tyequiv[L]{l_1}{l_3}\Level \\ \tyequiv[L]{l_2}{l_4}\Level}
  {\lttyequiv \metalevel{\TLAM{l_1}{l_2}U t}{\TLAM{l_3}{l_4}U{t'}}{\TPI U[\Delta]{l_1}{l_2}
      T}{l_2}}

  \inferrule
  {\lttyequiv \metalevel{t}{t'}{\TPI U[\Delta] l{l'}{T''}}{l'} \\ \lttypeq[\Psi][\Delta] \proglevel{T}{T'}l}
  {\lttyequiv \metalevel{\TAPP t{T}}{\TAPP{t'}{T'}}{T''[T/U]}{l'}}
\end{mathpar}

The following rules are related to meta-programming and intensional analysis. 
\begin{mathpar}
  \inferrule
  {\lpjudge \metalevel \Gamma \\ \lttypeq[\Psi][\Delta]\codelevel T{T'}l}
  {\lttyequiv \metalevel{\boxit T}{\boxit{T'}}{\CTyp[\Delta] l}{0}}

  \inferrule
  {\lpjudge \metalevel \Gamma \\ \lttyequiv[\Psi][\Delta]\codelevel t{t'} T l}
  {\lttyequiv \metalevel{\boxit t}{\boxit{t'}}{\CTrm[\Delta] T l}{0}}

  \inferrule
  {\lpjudge \metalevel \Gamma \\ \tyequiv[L]{l_1}{l_3}\Level \\ \tyequiv[L]{l_2}{l_4}\Level \\ \lpequiv \proglevel \Delta{\Delta'} \\
    \lttyequiv \metalevel{t}{t'}{\CTyp[\Delta]{l_2}}{0} \\
    \lttypeq[\Psi][\Gamma,x_T : \CTyp[\Delta]{l_2} \at{0}]\metalevel{M}{M'}{l_1} \\
  \lttyequiv[\Psi, U : \DTyp[\Delta]\codelevel {l_2}]\metalevel{t_1}{t_2}{M[\boxit{U^\id}/x_T]}{l_1}}
  {\lttyequiv \metalevel{\LETBTYP{l_1}{l_2} \Delta{x_T.M}{U}{t_1}t}{\LETBTYP{l_3}{l_4}{\Delta'}{x_T.M'}{U}{t_2}{t'}}{M[t/x_T]}{l_1}}

  \inferrule
  {\lpjudge \metalevel \Gamma \\ \tyequiv[L]{l_1}{l_3}\Level \\ \tyequiv[L]{l_2}{l_4}\Level \\
    \lpequiv \proglevel \Delta{\Delta'} \\
    \lttypeq \proglevel T{T'}{l_2} \\
    \lttyequiv \metalevel{t}{t'}{\CTrm[\Delta]T{l_2}}{0} \\
    \lttypeq[\Psi][\Gamma,x_T : \CTrm[\Delta]T{l_2} \at{0}]\metalevel{M}{M'}{l_1} \\
    \lttyequiv[\Psi, u : \DTrm[\Delta]\codelevel T{l_2}]\metalevel{t_1}{t_2}{M[\boxit{u^\id}/x_t]}{l_1}}
  {\lttyequiv \metalevel{\LETBTRM{l_1}{l_2} \Delta{T}{x_t.M}{u}{t_1}t}{\LETBTRM{l_3}{l_4}{\Delta'}{T'}{x_T.M'}{u}{t_2}{t'}}{M[t/x_t]}{l_1}}
\end{mathpar}
We omit the congruence rules for the recursive principles for code as they are
conceptually simple but too long. %
We simply let equivalence to propagate inwards to all the sub-terms of the recursive
principles. 

\subsection{Computation Rules}

Finally, we list all the computation rules. %
In the rules below, we let $\compt i$. %
We first list the $\beta$ rules for
natural numbers:
\begin{mathpar}
  \inferrule
  {\typing[L]l\Level \\ \lttypwf[\Psi][\Gamma, x : \Nat \at 0] i M l \\
    \lttyping i s {M[\ze/x]}l \\
    \lttyping[\Psi][\Gamma, x : \Nat \at 0, y : M \at l] i {s'}{M[\su x/x]}l}
  {\lttyequiv i{s}{\ELIMN l{x.M}s{x,y. s'}\ze}{M[\ze/x]}{l}}

  \inferrule
  {\typing[L]l\Level \\ \lttypwf[\Psi][\Gamma, x : \Nat \at 0] i M l \\
    \lttyping i s {M[\ze/x]}l \\
    \lttyping[\Psi][\Gamma, x : \Nat \at 0, y : M \at l] i {s'}{M[\su x/x]}l \\
    \lttyping i t \Nat 0}
  {\lttyequiv i{s'[t/x,\ELIMN l{x.M}s{x,y. s'}t/y]}{\ELIMN l{x.M}s{x,y. s'}{(\su t)}}{M[\su t/x]}{l}}
\end{mathpar}
Then we have the $\beta$ and $\eta$ rules for dependent functions:
\begin{mathpar}
  \inferrule
  {\typing[L]l\Level \\ \typing[L]{l'}\Level \\ \lttypwf i S{l}
    \\
    \lttypwf[\Psi][\Gamma, x : S \at{l}] i{T}{l'} \\
    \lttyping[\Psi][\Gamma, x : S \at{l}] i t{T}{l'} \\ \lttyping i s{S}{l}}
  {\lttyequiv i{t[s/x]}{\APP {\LAM {l}{l'} x S t} {l}{l'} x{S} T s}{T[s/x]}{l'}}

  \inferrule
  {\typing[L]l\Level \\ \typing[L]{l'}\Level \\ \lttypwf i S{l} \\
    \lttypwf[\Psi][\Gamma, x : S \at{l}] i{T}{l'} \\
    \lttyping i t{\PI {l}{l'}x S T}{l \sqcup l'}}
  {\lttyequiv i{\LAM {l}{l'} x{S} {\APP t {l}{l'}x S T x}}{t}{\PI {l}{l'} x S T}{l \sqcup l'}}
\end{mathpar}
In the $\eta$ rule, on the right hand side, all $t$, $S$ and $T$ should be properly
locally weakened.

Finally we have $\beta$ and $\eta$ rules for universe-polymorphic functions:
\begin{mathpar}
  \inferrule
  {\lpjudge \metalevel \Gamma \\ \lttyping[\Psi][\Gamma][L, \vect \ell]\metalevel{t}{T}{l} \\ \typing[L,\vect\ell]l\Level \\
    |\vect\ell| = |\vect l| > 0 \\ \forall l' \in \vect l ~.~ \typing[L]{l'}\Level}
  {\lttyequiv \metalevel{t[\vect l/\vect \ell]}{(\ULAM l \ell t)~\$~\vect l}{T[\vect l/\vect \ell]}{l[\vect l/\vect \ell]}}

  \inferrule
  {\lttyping \metalevel{t}{\UPI \ell l T}{\omega}}
  {\lttyequiv \metalevel{\ULAM l \ell{(t~\$~\vect\ell)}}{t}{\UPI \ell l T}{\omega}}
\end{mathpar}
Similarly, in the $\eta$ rule, the universe variables appearing in $t$ must also be
properly weakened. 
This concludes all the rules for the MLTT portion of \delamlang. 

Then we move on to considering the computation rules for the extended types. %
Let us finish considering all meta-function types.
\begin{mathpar}
  \inferrule
  {\lpjudge \metalevel \Gamma \\ \lttyping[\Psi, g : \Ctx]\metalevel t{T}{l} \\ \typing[L]l\Level \\ \lpjudge \proglevel \Delta}
  {\lttyequiv \metalevel{t[\Delta/g]}{\CAPP{(\CLAM l g t)}\Delta}{T[\Delta/g]}{l}}

  \inferrule
  {\lttyping \metalevel{t}{\CPI g l T}{l}}
  {\lttyequiv \metalevel{\CLAM l g {(\CAPP t g)}}{t}{\CPI g l T}{l}}

  \inferrule
  {\lpjudge \metalevel \Gamma \\ \lttyping[\Psi, U : \DTyp[\Delta] \proglevel l]\metalevel{t}{T'}{l'} \\ \typing[L]{l}\Level \\ \typing[L]{l'}\Level \\ \lttypwf[\Psi][\Delta] \proglevel{T} l}
  {\lttyequiv \metalevel{t[T/U]}{\TAPP{(\TLAM{l}{l'}U t)}{T}}{T'[T/U]}{l'}}

  \inferrule
  {\lttyping \metalevel{t}{\TPI U[\Delta] l{l'}{T'}}{l'}}
  {\lttyequiv \metalevel{\TLAM{l}{l'}U{(\TAPP t{U^\id})}}{t}{\TPI U[\Delta] l{l'}{T'}}{l'}}
\end{mathpar}
Now we consider the contextual types. %
They only have $\beta$ rules. %
Let us consider $\tletbox$ first.
\begin{mathpar}
  \inferrule
  {\lpjudge \metalevel \Gamma \\ \typing[L]{l'}\Level \\ \typing[L]l\Level \\ \lpjudge \proglevel \Delta \\
    \lttypwf[\Psi][\Delta]\codelevel T l \\
    \lttypwf[\Psi][\Gamma,x_T : \CTyp[\Delta]l \at{0}]\metalevel{M}{l'} \\
  \lttyping[\Psi, U : \DTyp[\Delta]\codelevel l]\metalevel{t'}{M[\boxit U^\id/x_T]}{l'}}
  {\lttyequiv \metalevel{t'[T/U]}{\LETBTYP{l'}l \Delta{x_T.M}{U}{t'}{(\boxit T)}}{M[\boxit T/x_T]}{l'}}

  \inferrule
  {\lpjudge \metalevel \Gamma \\ \typing[L]{l'}\Level \\ \typing[L]l\Level \\ \lpjudge \proglevel \Delta \\
    \lttypwf[\Psi][\Delta]\proglevel T l \\
    \lttyping[\Psi][\Delta]\codelevel t T l \\
    \lttypwf[\Psi][\Gamma,x_t : \CTrm[\Delta] T l \at{0}]\metalevel{M}{l'} \\
  \lttyping[\Psi, u : \DTrm[\Delta]\codelevel T l]\metalevel{t'}{M[\boxit u^\id/x_t]}{l'}}
  {\lttyequiv \metalevel{t'[t/u]}{\LETBTRM{l'}l \Delta T{x_t.M}{u}{t'}{(\boxit t)}}{M[\boxit t/x_t]}{l'}}
\end{mathpar}
We can also give the $\beta$ rules for the recursive principles for code. %
There are too many to list them all, and moreover they follow the same pattern, so we
just list a selected few of them. %
We begin with something easy:
\begin{mathpar}
  \inferrule
  {G_A \\ \lpjudge \metalevel \Gamma \\ \lpjudge \proglevel \Delta}
  {\lttyequiv \metalevel {t_\Nat[\Delta/g]}{\ELIMTYP{l_1}{l_2}M b{0}\Delta{(\boxit
        \Nat)}}{M[0/\ell,\Delta/g,\boxit \Nat/x_T]}{l_1}}
\end{mathpar}
In this case, we provide $\Nat$ to the recursive principle for code of types. %
It hits the base case described by $t_\Nat$, and thus the whole program is reduced to
$t_\Nat$ with $g$ for $\Delta$. %
Note that the universe level for $\Nat$ must be $0$ as specified by the typing
judgment at layer $\codelevel$. %
The recursive principle for code of terms behaves very similarly when encountering the
code of $\Nat$. %
Instead, it picks the right branch $t'_\Nat$ and returns the right motive instead:
\begin{mathpar}
  \inferrule
  {G_A \\ \lpjudge \metalevel \Gamma \\ \lpjudge \proglevel \Delta}
  {\lttyequiv \metalevel {t'_\Nat[\Delta/g]}{\ELIMTRM{l_1}{l_2}M b{1}\Delta{\Ty 0}{(\boxit
        \Nat)}}{M_\Trm[1/\ell,\Delta/g,\Ty 0/U_T,\boxit \Nat/x_t]}{l_2}}
\end{mathpar}
In order to have the code to be $\boxit\Nat$ as a term, this code must have type $\Ty
0$, which lives at universe level $1$. %
Hence the indices are forced by the typing rules at layer $\codelevel$. 

Then we specify the variable case:
\begin{mathpar}
  \inferrule
  {G_A \\ \lpjudge \metalevel \Gamma \\ \typing[L]{l'}\Level \\ \lpjudge \proglevel \Delta \\ \lttypwf[\Psi][\Delta]\proglevel T{l'}
    \\
    x : T \at{l'} \in \Delta}
  {\lttyequiv \metalevel {t_x[l'/\ell,\Delta/g,T/U_T,x/u_x]}{\ELIMTRM{l_1}{l_2}M b{l'}\Delta{T}{(\boxit
        x)}}{M_\Trm[l'/\ell,\Delta/g,T/U_T,\boxit x/x_t]}{l_2}}
\end{mathpar}
The subtlety here is that $u_x$ can only receive a variable as it is typed at layer
$\varlevel$, but it is fine as $x$ is precisely just a variable. 

Then let us consider a more complex case of $\Pi$ types.
\begin{mathpar}
  \inferrule
  {G_A \\ \lpjudge \metalevel \Gamma \\ \typing[L]{l}\Level \\ \typing[L]{l'}\Level \\ \lpjudge \proglevel \Delta \\
    \lttypwf[\Psi][\Delta] \codelevel S l  \\ \lttypwf[\Psi][\Delta, x : S \at l]\codelevel{T}{l'} \\
    t = \boxit{\PI{l}{l'}x S T} \\
    s_S = \ELIMTYP{l_1}{l_2}M b{l}\Delta{(\boxit{S})} \\
    s_T = \ELIMTYP{l_1}{l_2}M b{l'}{(\Delta, x : S \at l)}{(\boxit{T})}}
  {\lttyequiv \metalevel {t_\Pi[l/\ell,l'/\ell',\Delta/g,S/U_S,T/U_T,s_S/x_S,
      s_T/x_T]}{\ELIMTYP{l_1}{l_2}M b{(l \sqcup
        l')}\Delta{t}}{M_\Trm[l\sqcup l'/\ell,\Delta/g,t/x_T]}{l_1}}
\end{mathpar}
Notice how $s_S$ and $s_T$ recurse down the sub-structures, i.e. $S$ and $T$ with the
proper universe levels and local contexts. %
We end our discussion by given the $\beta$ rules for code of function abstractions and
applications, as they appear to be rather complex, but their essence is fundamentally simple.
\begin{mathpar}
  \inferrule
  {G_A \\ \lpjudge \metalevel \Gamma \\ \lpjudge \proglevel \Delta \\ \typing[L]l\Level  \\
    \typing[L]{l'}\Level \\ \lttypwf[\Psi][\Delta] \codelevel S l \\
    \lttyping[\Psi][\Delta, x : S \at l] \codelevel t{T}{l'} \\
    l_\Pi = l \sqcup l' \\
    T_\Pi = \PI{l}{l'}x{S}{T} \\
    t' = \boxit{\LAM{l}{l'}x{S}{t}} \\
    s_S = \ELIMTYP{l_1}{l_2}M b{l}\Delta{(\boxit{S})} \\
    s_t = \ELIMTRM{l_1}{l_2}M b{l'}{(\Delta, x : S \at l)}{T}{(\boxit{t})} \\
    \delta = s_S/x_S,s_t/x_t}
  {\lttyequiv \metalevel {t_\lambda[l/\ell,l'/\ell',\Delta/g,S/U_S,T/U_T,t/u_t,\delta]}{\ELIMTRM{l_1}{l_2}M b{l_\Pi}\Delta{T_\Pi}{t'}}{M_\Trm[l_\Pi/\ell,\Delta/g,T_\Pi/U_T,t'/x_t]}{l_2}}
\end{mathpar}
Similarly, the recursive principle for code of terms picks the right branch
($t_\lambda$) with variables properly substituted. %
Since $S$ is also a sub-structure, the recursive call $s_S$ invokes the recursive
principle for code of types instead, hence making the recursive principles mutually
defined.

Last, we give the case for function applications.
\begin{mathpar}
  \inferrule
  {G_A \\ \lpjudge \metalevel \Gamma \\ \lpjudge \proglevel \Delta \\ \typing[L]l\Level  \\ \typing[L]{l'}\Level
    \\ \lttypwf[\Psi][\Delta]\codelevel S l \\ \lttypwf[\Psi][\Delta, x : S \at l] \codelevel{T}{l'} \\
    \lttyping[\Psi][\Delta] \codelevel t{\PI l{l'} x S T}{l \sqcup l'} \\
    \lttyping[\Psi][\Delta] \codelevel s S l \\
    T_\tapp = T[s/x] \\
    t' = \boxit{(\APP t l{l'} x S T s)} \\
    s_S = \ELIMTYP{l_1}{l_2}M b{l}\Delta{(\boxit{S})} \\
    s_T = \ELIMTYP{l_1}{l_2}M b{l'}{(\Delta, x : S \at l)}{(\boxit{T})} \\
    s_t = \ELIMTRM{l_1}{l_2}M b{(l \sqcup l')}{\Delta}{(\PI{l}{l'}xST)}{(\boxit{t})} \\
    s_s = \ELIMTRM{l_1}{l_2}M b{l}{\Delta}{S}{(\boxit{s})} \\
    \sigma = \Delta/g,S/U_S,T/U_T,t/u_t,s/u_s \\
    \delta = s_S/x_S,s_T/x_T,s_t/x_t,s_s/x_s}
  {\lttyequiv \metalevel {t_\tapp[l/\ell,l'/\ell',\sigma,\delta]}{\ELIMTRM{l_1}{l_2}M b{l'}\Delta{T_\tapp}{t'}}{M_\Trm[l'/\ell,\Delta/g,T_\tapp/U_T,t'/x_t]}{l_2}}
\end{mathpar}
Similar to above, we can do recursion on all sub-structures, including $S$ and $T$,
which are handled by the recursive principles for code of types. %
It is not only convenient to put $S$ and $T$ in the syntax of a function
application, but also necessary. %
If we look at $s_t$ and $s_s$, the recursive calls on the function and the argument,
we see that we must supply their types, i.e. $\PI{l}{l'}xST$ and $S$, respectively. %
This information, unfortunately, cannot be recovered, if we employed the more common
syntax of $t~s$. %
In practice, the $\Pi$ type can be filled in by a type inference algorithm when we do
not care, so it does not truly make the type theory more difficult, but rather enables
the recursion on code of function applications.

At this point, we conclude all rules for \delamlang. %
Next, we shall carefully define all syntactic operations and examine the syntactic
properties of \delamlang. %
Then we work out the semantics by following \Cref{sec:cv} and
\citet{abel_decidability_2017}, from which we conclude the convertibility problem of
\delamlang is decidable.

\subsection{A Note on Layer $\varlevel$ Rules}

To summarize, only the following rules can be indexed by layer $\varlevel$:
\begin{itemize}
\item the typing rule for local variables and its congruence;
\item the local substitution rules and their equivalence rules;
\item the typing rule for global variables and its congruence;
\item all conversion rules for terms and their equivalence.
\end{itemize}
In particular, we are not even obliged to include symmetry and transitivity, because
they can be derived from existing rules.

\section{Syntactic Operations and Properties of \delamlang}\labeledit{sec:dt:syn}

In the previous section, we have introduced all judgments of \delamlang, but we have
left out some details. %
For one, we have not defined the substitution operations yet, though they are very
intuitive. %
For the sake of completeness, we will give their definitions. %
Then we examine the syntactic properties of \delamlang before entering the semantic
zone. %

\subsection{Substitution Operations}

In \Cref{sec:dt:ulevel}, we have given the definition of substitutions for universe
levels and how to apply one to a universe level. %
Applying a substitution for universe levels to types and terms simply propagate the
substitution downwards.
\begin{longtable}{R@{$[\phi]~:=~$}L}
  \Nat & \Nat \\
  \PI{l}{l'}x S T & \PI{l[\phi]}{l'[\phi]}x{S[\phi]}{(T[\phi])} \\
  \Ty l & \Ty{l[\phi]} \\
  \UPI \ell l T &
  \UPI\ell{l[\phi,\vect\ell/\vect\ell]}{(T[\phi,\vect\ell/\vect\ell])} \\
  \Elt l t & \Elt{l[\phi]}{(t[\phi])} \\
  U^\delta & U^{\delta[\phi]} \\
  \CPI g l T & \CPI g{l[\phi]}{(T[\phi])} \\
  \TPI U l{l'}T & \TPI U[\Gamma[\phi]]{l[\phi]}{l'[\phi]}{(T[\phi])} \\
  \CTyp l & \CTyp[\Gamma[\phi]]{l[\phi]} \\
  \CTrm T l & \CTrm[\Gamma[\phi]]{T[\phi]}{l[\phi]} \\[5pt]
  \cdot & \cdot \\
  g & g \\
  \Gamma, x : T \at l & \Gamma[\phi], x : T[\phi] \at{l[\phi]} \\[5pt]
  \cdot^k_{g?} & \cdot^k_{g?} \\
  \wk^k_g & \wk^k_g \\
  \delta, t/x & \delta[\phi], t[\phi]/x \\[5pt]
  x & x \\
  \Nat & \Nat \\
  \PI{l}{l'}x s t & \PI{l[\phi]}{l'[\phi]}x{s[\phi]}{(t[\phi])} \\
  \Ty l & \Ty{l[\phi]} \\
  \ze & \ze \\
  \su t & \su{(t[\phi])} \\
  \ELIMN l{x.M}{s}{x,y.s'}{t} &
  \ELIMN{l[\phi]}{x.M[\phi]}{(s[\phi])}{x,y.s'[\phi]}{(t[\phi])}
  \\
  \LAM{l}{l'}x{S}t & \LAM{l[\phi]}{l'[\phi]}x{S[\phi]}{(t[\phi])} \\
  \APP t{l}{l'}x{S}T s & \APP{t[\phi]}{l[\phi]}{l'[\phi]}x{S[\phi]}{T[\phi]}{(s[\phi])}
  \\
  \ULAM l \ell t & \ULAM{l[\phi,\vect\ell/\vect\ell]}\ell{t[\phi,\vect\ell/\vect\ell]}
  \\
  \UAPP t l & (t[\phi])~\$~(\vect l[\phi]) \\
  u^\delta & u^{\delta[\phi]} \\
  \CLAM l g t & \CLAM{l[\phi]}g{(t[\phi])} \\
  \CAPP t \Gamma & \CAPP{t[\phi]}{(\Gamma[\phi])} \\
  \TLAM l{l'} U t & \TLAM{l[\phi]}{l'[\phi]} U{(t[\phi])} \\
  \TAPP t T & \TAPP{t[\phi]}{(T[\phi])} \\
  \boxit T & \boxit{(T[\phi])} \\
  \boxit t & \boxit{(t[\phi])} \\
  \LETBTYP{l'}l\Gamma{x_T. M}{U}{t'}t & \LETBTYP{l'[\phi]}{(l[\phi])}{(\Gamma[\phi])}{x_T. M[\phi]}{U}{t'[\phi]}{(t[\phi])} \\
  \LETBTRM{l'}l\Gamma T{x_t. M}{u}{t'}t &
  \LETBTRM{l'[\phi]}{(l[\phi])}{(\Gamma[\phi])}{(T[\phi])}{x_t. M[\phi]}{u}{t'[\phi]}{(t[\phi])}
  \\
  \ELIMTYP{l_1}{l_2}Mbl\Gamma t & \ELIMTYPn{l_1[\phi]}{l_2[\phi]}{(\vect M[\phi])}{(\vect b[\phi])}{(l[\phi])}{(\Gamma[\phi])}{(t[\phi])} \\
  \ELIMTRM{l_1}{l_2}Mbl\Gamma T t & \ELIMTRMn{l_1[\phi]}{l_2[\phi]}{(\vect
    M[\phi])}{(\vect b[\phi])}{(l[\phi])}{(\Gamma[\phi])}{(T[\phi])}{(t[\phi])}
  \\[5pt]
  \vect M & (\ell,g,x_T. M_\Typ[\phi,\ell/\ell])~(\ell,g,U_T,x_t. M_\Trm[\phi,\ell/\ell]) \\[5pt]
  (g. t_\Nat) & (g. t_\Nat[\phi]) \\
  (\ell,\ell',g, U_S,U_T,x_S,x_T. t_\Pi) & (\ell,\ell',g, U_S,U_T,x_S,x_T. t_\Pi[\phi,\ell/\ell,\ell'/\ell'])  \\
  (\ell,g. t_{\Se}) & (\ell,g. t_{\Se}[\phi,\ell/\ell]) \\
  (\ell,g,u_t,x_t. t_\tEl) & (\ell,g,u_t,x_t. t_\tEl[\phi,\ell/\ell]) \\[5pt]
  (\ell,g,U_T,u_x. t_x) & (\ell,g,U_T,u_x. t_x[\phi,\ell/\ell]) \\
  (g.t'_\Nat) & (g.t'_\Nat[\phi]) \\
  (\ell,\ell',g,u_s,u_t,x_s,x_t.t'_\Pi) & (\ell,\ell',g,u_s,u_t,x_s,x_t.t'_\Pi[\phi,\ell/\ell,\ell'/\ell']) \\
  (\ell,g.t'_\Se) & (\ell,g.t'_\Se[\phi,\ell/\ell]) \\
  (g.t_\ze) & (g.t_\ze[\phi]) \\
  (g,u_t,x_t.t_{\tsucc}) & (g,u_t,x_t.t_{\tsucc}[\phi]) \\
  (\ell,g,U_M,u_s,u_{s'},u_t,x_M,x_s,x_{s'},x_t.t_{\telimn}) & (\ell,g,U_M,u_s,u_{s'},u_t,x_M,x_s,x_{s'},x_t.t_{\telimn}[\phi,\ell/\ell]) \\
  (\ell,\ell',g,U_S,U_T,u_t,x_S,x_t.t_\lambda) & (\ell,\ell',g,U_S,U_T,u_t,x_S,x_t.t_\lambda[\phi,\ell/\ell,\ell'/\ell']) \\
  (\ell,\ell',g,U_S,U_T,u_t,u_s,x_S,x_T,x_t,x_s.t_\tapp) &
  (\ell,\ell',g,U_S,U_T,u_t,u_s,x_S,x_T,x_t,x_s.t_\tapp[\phi,\ell/\ell,\ell'/\ell']) \\
\end{longtable}
The composition operation and the identity substitution are defined intuitively as:
\begin{longtable}{R@{$~\circ~\phi~:=~$}L}
  \cdot & \cdot \\
  (\phi',l/\ell) & (\phi' \circ \phi), l[\phi]/\ell \\
\end{longtable}
\begin{longtable}{R@{$~:=~$}L}
  \id_\cdot & \cdot \\
  \id_{L,\ell} & \id_L,\ell/\ell \\
\end{longtable}
The presentation of the identity substitution is simpler as we do not consider
weakenings for universe contexts.  %
We sometimes omit the subscript when it can be inferred from the textual context. %
We will also need to apply a universe substitution to global context, which does not
need to be mutually defined:
\begin{longtable}{R@{$[\phi]~:=~$}L}
  \cdot & \cdot \\
  (\Psi, g : \Ctx) & \Psi[\phi], g : \Ctx \\
  (\Psi, U : \DTyp i l) & \Psi[\phi], U : \DTyp[\Gamma[\phi]]{i}{l[\phi]} \\
  (\Psi, u : \DTrm i T l) & \Psi[\phi], u : \DTrm[\Gamma[\phi]]{i}{T[\phi]}{l[\phi]}
\end{longtable}

Then we give the the application of a local substitution. %
Following \Cref{sec:cv}, we need two auxiliary definitions to query a local
substitution in order to define its composition. %
We repeat their definitions as follows:
\begin{align*}
  \widehat{\cdot_{g?}^k} &:= k \\
  \widehat{\wk_{g}^k} &:= k \\
  \widehat{\delta, t/x} &:= \widehat{\delta} \\
  \widecheck{\cdot_{g?}^k} &:= g? \\
  \widecheck{\wk_g^k} &:= g \\
  \widecheck{\delta, t/x} &:= \widecheck{\delta}
\end{align*}
Then we give the application of local substitutions:
\begin{longtable}{R@{$[\delta]~:=~$}L}
  \Nat & \Nat \\
  \PI{l}{l'}x S T & \PI{l}{l'}x{S[\delta]}{(T[\delta,x/x])} \\
  \Ty l & \Ty{l} \\
  \UPI \ell l T &
  \UPI\ell{l}{(T[\delta])} \\
  \Elt l t & \Elt{l}{(t[\delta])} \\
  U^{\delta'} & U^{\delta' \circ \delta} \\
  \CPI g l T & \CPI g{l}{(T[\delta])} \\
  \TPI U l{l'}T & \TPI U[\Gamma]{l}{l'}{(T[\delta])} \\
  \CTyp l & \CTyp{l} \\
  \CTrm T l & \CTrm{T}{l} \\[5pt]
  x & \delta(x) \hfill \text{(lookup of $x$ in $\delta$)} \\
  \Nat & \Nat \\
  \PI{l}{l'}x s t & \PI{l}{l'}x{s[\delta]}{(t[\delta,x/x])} \\
  \Ty l & \Ty{l} \\
  \ze & \ze \\
  \su t & \su{(t[\delta])} \\
  \ELIMN l{x.M}{s}{x,y.s'}{t} &
  \ELIMN{l}{x.M[\delta,x/x]}{(s[\delta])}{x,y.s'[\delta,x/x,y/y]}{(t[\delta])}
  \\
  \LAM{l}{l'}x{S}t & \LAM{l}{l'}x{S[\delta]}{(t[\delta,x/x])} \\
  \APP t{l}{l'}x{S}T s & \APP{t[\delta]}{l}{l'}x{S[\delta]}{T[\delta,x/x]}{(s[\delta])}
  \\
  \ULAM l \ell t & \ULAM{l}\ell{t[\delta]}
  \\
  \UAPP t l & \UAPP{t[\delta]}{l} \\
  u^{\delta'} & u^{\delta' \circ \delta} \\
  \CLAM l g t & \CLAM{l}g{(t[\delta])} \\
  \CAPP t \Gamma & \CAPP{t[\delta]}{\Gamma} \\
  \TLAM l{l'} U t & \TLAM{l}{l'} U{(t[\delta])} \\
  \TAPP t T & \TAPP{t[\delta]}{T} \\
  \boxit T & \boxit{T} \\
  \boxit t & \boxit{t} \\
  \LETBTYP{l'}l\Gamma{x_T. M}{U}{t'}t & \LETBTYP{l'}{l}{\Gamma}{x_T. M[\delta,x_T/x_T]}{U}{t'[\delta]}{(t[\delta])} \\
  \LETBTRM{l'}l\Gamma T{x_t. M}{u}{t'}t &
  \LETBTRM{l'}{l}{\Gamma}{T}{x_t. M[\delta,x_t/x_t]}{u}{t'[\delta]}{(t[\delta])}
  \\
  \ELIMTYP{l_1}{l_2}Mbl\Gamma t & \ELIMTYPn{l_1}{l_2}{(\vect M[\delta])}{(\vect b[\delta])}{l}{\Gamma}{(t[\delta])} \\
  \ELIMTRM{l_1}{l_2}Mbl\Gamma T t & \ELIMTRMn{l_1}{l_2}{(\vect
    M[\delta])}{(\vect b[\delta])}{l}{\Gamma}{T}{(t[\delta])}
  \\[5pt]
  \vect M & (\ell,g,x_T. M_\Typ[\delta,x_T/x_T])~(\ell,g,U_T,x_t. M_\Trm[\delta,x_t/x_t]) \\[5pt]
  (g. t_\Nat) & (g. t_\Nat[\delta]) \\
  (\ell,\ell',g, U_S,U_T,x_S,x_T. t_\Pi) & (\ell,\ell',g, U_S,U_T,x_S,x_T. t_\Pi[\delta,x_S/x_S,x_T/x_T])  \\
  (\ell,g. t_{\Se}) & (\ell,g. t_{\Se}[\delta]) \\
  (\ell,g,u_t,x_t. t_\tEl) & (\ell,g,u_t,x_t. t_\tEl[\delta,x_t/x_t]) \\[5pt]
  (\ell,g,U_T,u_x. t_x) & (\ell,g,U_T,u_x. t_x[\delta]) \\
  (g.t'_\Nat) & (g.t'_\Nat[\delta]) \\
  (\ell,\ell',g,u_s,u_t,x_s,x_t.t'_\Pi) & (\ell,\ell',g,u_s,u_t,x_s,x_t.t'_\Pi[\delta,x_s/x_s,x_t/x_t]) \\
  (\ell,g.t'_\Se) & (\ell,g.t'_\Se[\delta]) \\
  (g.t_\ze) & (g.t_\ze[\delta]) \\
  (g,u_t,x_t.t_{\tsucc}) & (g,u_t,x_t.t_{\tsucc}[\delta,x_t/x_t]) \\
  (\ell,g,U_M,u_s,u_{s'},u_t,x_M,x_s,x_{s'},x_t.t_{\telimn}) & \\
  \multicolumn{2}{R}{(\ell,g,U_M,u_s,u_{s'},u_t,x_M,x_s,x_{s'},x_t.t_{\telimn}[\delta,x_M/x_M,x_s/x_s,x_{s'}/x_{s'},x_t/x_t])} \\
  (\ell,\ell',g,U_S,U_T,u_t,x_S,x_t.t_\lambda) & (\ell,\ell',g,U_S,U_T,u_t,x_S,x_t.t_\lambda[\delta,x_S/x_S,x_t/x_t]) \\
  (\ell,\ell',g,U_S,U_T,u_t,u_s,x_S,x_T,x_t,x_s.t_\tapp) & \\
  \multicolumn{2}{R}{(\ell,\ell',g,U_S,U_T,u_t,u_s,x_S,x_T,x_t,x_s.t_\tapp[\delta,x_S/x_S,x_T/x_T,x_{t}/x_{t},x_s/x_s])} \\
\end{longtable}
where composition is defined in the same way as \Cref{sec:cv}:
\begin{align*}
  \wk_g^k \circ \delta &:= \wk_g^{\widehat\delta} \\
  \cdot^k \circ \delta &:= \cdot_{\widecheck{\delta}}^{\widehat\delta} \\
  \cdot_{g}^k \circ \delta &:= \cdot_{g}^{\widehat\delta} \\
  (\delta', t/x) \circ \delta &:= (\delta' \circ \delta), t[\delta]/x
\end{align*}
The identity local substitution is characterized as a generalization of local
weakening $\wk$.
\begin{align*}
  \wk^k_{\cdot} &:=  \cdot^k \\
  \wk^k_{g} &:=  \wk_g^k \\
  \wk^k_{\Gamma, x : T \at l} &:= \wk^{1 + k}_{\Gamma}, x/x
\end{align*}
Identity is just $\id_\Gamma := \wk^0_\Gamma$. 

Then we give the global substitutions.
\begin{alignat*}{2}
  \sigma &:=&& \cdot \sep \sigma, \Gamma/g \sep \sigma, T/U  \sep \sigma, t/u
               \tag{Global substitutions}
\end{alignat*}
Then we define the typing rules as:
\begin{mathpar}
  \inferrule*
  {\judge[L] \Psi}
  {\ptyping{\cdot}{\cdot}}

  \inferrule*
  {\ptyping{\sigma}{\Phi} \\ \lpjudge \proglevel \Gamma}
  {\ptyping{\sigma, \Gamma/g}{\Phi, g: \Ctx}}

  \inferrule*
  {\ptyping{\sigma}{\Phi} \\ \lpjudge[\Phi] \proglevel \Gamma \\\\ \typing[L]l\Level \\ i \in \{\codelevel, \proglevel\}
    \\
    \lttypwf[\Psi][\Gamma[\sigma]] i{T}l}
  {\ptyping{\sigma, T/U}{\Phi, u : \DTyp i l}}
  
  \inferrule*
  {\ptyping{\sigma}{\Phi} \\ \lttypwf[\Phi] \proglevel T l \\ \typing[L]l\Level \\
    i \in \{\varlevel, \codelevel\} \\
    \lttyping[\Psi][\Gamma[\sigma]] i t{T[\sigma]}l}
  {\ptyping[\Psi]{\sigma, t/u}{\Phi, u : \DTrm{i} T l}}
\end{mathpar}

Then we consider the cases for application:
\begin{longtable}{R@{$[\sigma]~:=~$}L}
  \Nat & \Nat \\
  \PI{l}{l'}x S T & \PI{l}{l'}x{S[\sigma]}{(T[\sigma])} \\
  \Ty l & \Ty{l} \\
  \UPI \ell l T &
  \UPI\ell{l}{(T[\sigma])} \\
  \Elt l t & \Elt{l}{(t[\sigma])} \\
  U^{\delta} & \sigma(U)[\delta[\sigma]]
  \hfill \text{(lookup of $U$ in $\sigma$)}\\
  \CPI g l T & \CPI g{l}{(T[\sigma,g/g])} \\
  \TPI U l{l'}T & \TPI U[\Gamma[\sigma]]{l}{l'}{(T[\sigma,U^\id/U])} \\
  \CTyp l & \CTyp[\Gamma[\sigma]]{l} \\
  \CTrm T l & \CTrm[\Gamma[\sigma]]{T[\sigma]}{l} \\[5pt]
  \cdot & \cdot \\
  g & \sigma(g) \\
  \Gamma, x : T \at l & \Gamma[\sigma], x : T[\sigma] \at{l} \\[5pt]
  \wk_g^k & \wk^k_{\sigma(g)} \\
  \multicolumn{2}{r}{(a local weakening extending the local context by length $\metalevel$)}\\
  \cdot^k  & \cdot^k \\
  \cdot_{g}^k & \cdot^{|\Gamma| + \metalevel}
  \hfill \text{(if $\sigma(g) = \Gamma$ and $\Gamma$ ends with a $\cdot$)} \\
  \cdot_{g}^k  & \cdot_{g'}^{|\Gamma| + \metalevel}
  \hfill \text{(if $\sigma(g) = \Gamma$ and $\Gamma$ ends with a $g'$)} \\
  (\delta, t/x) & (\delta[\sigma]), t[\sigma]/x \\[5pt]
  x & x
  \hfill \text{(no effect on local variables)}\\
  \Nat & \Nat \\
  \PI{l}{l'}x s t & \PI{l}{l'}x{s[\sigma]}{(t[\sigma])} \\
  \Ty l & \Ty l \\
  \ze & \ze \\
  \su t & \su{(t[\sigma])} \\
  \ELIMN l{x.M}{s}{x,y.s'}{t} &
  \ELIMN{l}{x.M[\sigma]}{(s[\sigma])}{x,y.s'[\sigma]}{(t[\sigma])}
  \\
  \LAM{l}{l'}x{S}t & \LAM{l}{l'}x{S[\sigma]}{(t[\sigma])} \\
  \APP t{l}{l'}x{S}T s & \APP{t[\sigma]}{l}{l'}x{S[\sigma]}{T[\sigma]}{(s[\sigma])}
  \\
  \ULAM l \ell t & \ULAM{l}\ell{t[\sigma]}
  \\
  \UAPP t l & \UAPP{t[\sigma]}l \\
  u^{\delta} & \sigma(u)[\delta[\sigma]]
  \hfill \text{(lookup of $u$ in $\sigma$)} \\
  \CLAM l g t & \CLAM{l}g{(t[\sigma,g/g])} \\
  \CAPP t \Gamma & \CAPP{(t[\sigma])}{(\Gamma[\sigma])} \\
  \TLAM l{l'} U t & \TLAM{l}{l'} U{(t[\sigma,U^\id/U])} \\
  \TAPP t T & \TAPP{(t[\sigma])}{(T[\sigma])} \\
  \boxit T & \boxit{(T[\sigma])} \\
  \boxit t & \boxit{(t[\sigma])} \\
  \LETBTYP{l'}l\Gamma{x_T. M}{U}{t'}t & \LETBTYP{l'}{l}{\Gamma}{x_T. M[\sigma]}{U}{t'[\sigma,U^\id/U]}{(t[\sigma])} \\
  \LETBTRM{l'}l\Gamma T{x_t. M}{u}{t'}t &
  \LETBTRM{l'}{l}{\Gamma}{T}{x_t. M[\sigma]}{u}{t'[\sigma,u/u]}{(t[\sigma])}
  \\
  \ELIMTYP{l_1}{l_2}Mbl\Gamma t & \ELIMTYPn{l_1}{l_2}{(\vect M[\sigma])}{(\vect b[\sigma])}{l}{(\Gamma[\sigma])}{(t[\sigma])} \\
  \ELIMTRM{l_1}{l_2}Mbl\Gamma T t & \ELIMTRMn{l_1}{l_2}{(\vect
    M[\sigma])}{(\vect b[\sigma])}{l}{(\Gamma[\sigma])}{(T[\sigma])}{(t[\sigma])}
  \\[5pt]
  \vect M & (\ell,g,x_T. M_\Typ[\sigma,g/g])~(\ell,g,U_T,x_t. M_\Trm[\sigma,g/g,U_T^{\id}/U_T]) \\[5pt]
  (g. t_\Nat) & (g. t_\Nat[\sigma,g/g]) \\
  (\ell,\ell',g, U_S,U_T,x_S,x_T. t_\Pi) & (\ell,\ell',g, U_S,U_T,x_S,x_T. t_\Pi[\sigma,g/g,U_S^\id/U_S,U_T^\id/U_T])  \\
  (\ell,g. t_{\Se}) & (\ell,g. t_{\Se}[\sigma,g/g]) \\
  (\ell,g,u_t,x_t. t_\tEl) & (\ell,g,u_t,x_t. t_\tEl[\sigma,g/g,u_t^\id/u_t]) \\[5pt]
  (\ell,g,U_T,u_x. t_x) & (\ell,g,U_T,u_x. t_x[\sigma,g/g,U_T^\id/U_T,u_x^\id/u_x]) \\
  (g.t'_\Nat) & (g.t'_\Nat[\sigma,g/g]) \\
  (\ell,\ell',g,u_s,u_t,x_s,x_t.t'_\Pi) & (\ell,\ell',g,u_s,u_t,x_s,x_t.t'_\Pi[\sigma,g/g,u_s^\id/u_s,u_t^\id/u_t]) \\
  (\ell,g.t'_\Se) & (\ell,g.t'_\Se[\sigma,g/g]) \\
  (g.t_\ze) & (g.t_\ze[\sigma,g/g]) \\
  (g,u_t,x_t.t_{\tsucc}) & (g,u_t,x_t.t_{\tsucc}[\sigma,g/g,u_t^\id/u_t]) \\
  (\ell,g,U_M,u_s,u_{s'},u_t,x_M,x_s,x_{s'},x_t.t_{\telimn}) & \\
  \multicolumn{2}{R}{(\ell,g,U_M,u_s,u_{s'},u_t,x_M,x_s,x_{s'},x_t.t_{\telimn}[\sigma,g/g,U_M^\id/U_M,u_s^\id/u_s,u_{s'}^\id/u_{s'},u_t^\id/u_t])} \\
  (\ell,\ell',g,U_S,U_T,u_t,x_S,x_t.t_\lambda) & (\ell,\ell',g,U_S,U_T,u_t,x_S,x_t.t_\lambda[\sigma,g/g,U_S^\id/U_S,U_T^\id/U_T,u_t^\id/u_t]) \\
  (\ell,\ell',g,U_S,U_T,u_t,u_s,x_S,x_T,x_t,x_s.t_\tapp) & \\
  \multicolumn{2}{R}{(\ell,\ell',g,U_S,U_T,u_t,u_s,x_S,x_T,x_t,x_s.t_\tapp[\sigma,g/g,U_S^\id/U_S,U_T^\id/U_T,u_{t}^\id/u_{t},u_s^\id/u_s])} \\
\end{longtable}
Following \Cref{sec:cv}, we give the identity global substitution as a special case of
global weakenings, and composition. %
\begin{align*}
  \wk^k_{\cdot} &:=  \cdot \\
  \wk^k_{\Psi, g : \Ctx} &:=  \wk_{\Psi}^{1+k}, g/g \\
  \wk^k_{\Psi, U : \DTyp i l} &:= \wk^{1 + k}_{\Psi}, U^{\id_\Gamma}/U \\
  \wk^k_{\Psi, u : \DTrm i T l} &:= \wk^{1 + k}_{\Psi}, u^{\id_\Gamma}/u
\end{align*}
As a special case, we have
\begin{align*}
  \id_{\Psi} := \wk^0_{\Psi}
\end{align*}
Moreover, we have composition
\begin{align*}
  \cdot \circ \sigma' &:= \cdot \\
  (\sigma, \Gamma/g) \circ \sigma' &:= (\sigma \circ \sigma'), \Gamma[\sigma']/g \\
  (\sigma, T/U) \circ \sigma' &:= (\sigma \circ \sigma'), T[\sigma']/U \\
  (\sigma, t/u) \circ \sigma' &:= (\sigma \circ \sigma'), t[\sigma']/u 
\end{align*}

\subsection{Properties of Substitutions}

In the next step, we examine the algebraic properties of all substitutions. %
In the lemmas below, we always assume well-formedness or well-typedness of the
subjects in the lemmas, unless the lemmas are about typing. %
For conciseness, we do not spell out the conditions as they are routine.
\begin{lemma} $ $
  \begin{itemize}
  \item $\typing[L]{\id_L}{L}$
  \item If $\typing[L]{\phi}{L'}$ and $\typing[L']{\phi'}{L''}$, then
    $\typing[L]{\phi' \circ \phi}{L''}$.
  \end{itemize}
\end{lemma}
\begin{proof}
  Analyze the definition of identity and composition.
\end{proof}
\begin{lemma}[Algebra of Universe Substitutions] $ $
  \begin{itemize}
  \item $l[\phi][\phi'] = l[\phi \circ \phi']$ ($T$, $\Gamma$, $\delta$, $t$, $\Psi$ resp.)
  \item $l[\id] = l$ ($T$, $\Gamma$, $\delta$, $t$, $\Psi$ resp.)
  \item $\id \circ \phi = \phi$ and $\phi \circ \id = \phi$
  \item $(\phi_1 \circ \phi_2) \circ \phi_3 = \phi_1 \circ (\phi_2 \circ \phi_3)$
  \end{itemize}
\end{lemma}
\begin{proof}
  The proofs are routine; the first two statements are proved by induction on $l$
  first, and then mutual induction on all applications of universe substitutions. %
  The last two we analyze the definition of composition. 
\end{proof}

Similar lemmas hold for local and global substitutions.
\begin{lemma} $ $
  \begin{itemize}
  \item $\ltsubst[\Psi][\Gamma,\Delta]i{\wk_{\Gamma}^{|\Delta|}}{\Gamma}$
  \item $\ltsubst i{\id_{\Gamma}}{\Gamma}$
  \end{itemize}
\end{lemma}
Note that local substitutions permit $i = \varlevel$, so we have $\ltsubst
\varlevel{\id_{\Gamma}}{\Gamma}$. %
This intuitively makes sense, as all terms in $\id_{\Gamma}$ are just local variables. 
\begin{lemma}[Algebra of Local Substitutions] $ $
  \begin{itemize}
  \item $T[\delta][\delta'] = T[\delta \circ \delta']$ ($t$ resp.)
  \item $T[\id] = T$ ($t$ resp.)
  \item $\id \circ \delta = \delta$ and $\delta \circ \id = \delta$
  \item $(\delta_1 \circ \delta_2) \circ \delta_3 = \delta_1 \circ (\delta_2 \circ \delta_3)$
  \end{itemize}
\end{lemma}
\begin{proof}
  The first statement is mutually proved with associativity and by mutual
  induction. The second statement is mutually proved with right identity and also by
  mutual induction. %
  When proving right identity, we realize that all extended local substitutions under
  binders are identities.
\end{proof}

Then we reason about global substitutions.
\begin{lemma} $ $
  \begin{itemize}
  \item $\ptyping[\Psi,\Phi] {\wk_{\Psi}^{|\Phi|}}{\Psi}$
  \item $\ptyping {\id_{\Psi}}{\Psi}$
  \end{itemize}
\end{lemma}
\begin{lemma}[Algebra of Global Substitutions] $ $
  \begin{itemize}
  \item $T[\sigma][\sigma'] = T[\sigma \circ \sigma']$ ($\Gamma$, $\delta$, $t$ resp.)
  \item $T[\id] = T$ ($\Gamma$, $\delta$, $t$ resp.)
  \item $\id \circ \sigma = \sigma$ and $\sigma \circ \id = \sigma$
  \item $(\sigma_1 \circ \sigma_2) \circ \sigma_3 = \sigma_1 \circ (\sigma_2 \circ \sigma_3)$
  \end{itemize}
\end{lemma}
\begin{proof}
  The first two statements require mutual inductions on the applications of global
  substitutions. %
  Right identity is a natural consequence of the second statement. %
  Left identity is proved by simply looking at the definition of the identity global
  substitution. %
  Associativity is routine. 
\end{proof}

Finally, we conclude how all these kinds of substitutions interact.
\begin{lemma}[Acting on Weakenings]\labeledit{lem:dt:wkact} $ $
  \begin{itemize}
  \item $\wk^k_\Gamma[\phi] = \wk^k_{\Gamma[\phi]}$
  \item $\wk^k_\Psi[\phi] = \wk^k_{\Psi[\phi]}$
  \item $\wk^k_\Gamma[\sigma] = \wk^k_{\Gamma[\sigma]}$
  \end{itemize}  
\end{lemma}
\begin{proof}
  The first two statements are pretty straightforward as the lengths of the contexts
  are not altered. %
  The last one requires a bit more thought. %
  We proceed by induction on $\Gamma$.
  \begin{itemize}[label=Case]
  \item
    \begin{align*}
      \wk^k_\cdot[\sigma]
      &= \cdot^k[\sigma] \\
      &= \cdot^k
    \end{align*}
  \item
    \begin{align*}
      \wk^k_g[\sigma]
      &= \wk^k_{\sigma(g)}
    \end{align*}
    which already matches the definition of $\wk^k_{g[\sigma]}$.
  \item
    \begin{align*}
      \wk^k_{\Gamma,x: T \at l}[\sigma]
      &= (\wk^{1 + \metalevel}_{\Gamma}, x/x)[\sigma] \\
      &= \wk^{1 + \metalevel}_{\Gamma}[\sigma], x/x \\
      &= \wk^{1 + \metalevel}_{\Gamma[\sigma]}, x/x
        \byIH \\
      &= \wk^{k}_{\Gamma[\sigma], x : T[\sigma] \at l}
    \end{align*}
  \end{itemize}
\end{proof}
\begin{corollary}[Acting on Identities] $ $
  \begin{itemize}
  \item $\id_\Gamma[\phi] = \id_{\Gamma[\phi]}$
  \item $\id_\Psi[\phi] = \id_{\Psi[\phi]}$
  \item $\id_\Gamma[\sigma] = \id_{\Gamma[\sigma]}$
  \end{itemize}
\end{corollary}
\begin{lemma}[Interactions between Different Substitutions]\labeledit{lem:dt:interact} $ $
  \begin{itemize}
  \item $T[\delta][\phi] = T[\phi][\delta[\phi]]$ ($t$, resp.)
  \item $(\delta \circ \delta') [\phi] = (\delta[\phi]) \circ (\delta'[\phi])$
  \item $T[\sigma][\phi] = T[\phi][\sigma[\phi]]$ ($\Gamma$, $\delta$, $t$, resp.)
  \item $T[\delta][\sigma] = T[\sigma][\delta[\sigma]]$ ($t$, resp.)
  \item $(\delta \circ \delta') [\sigma] = (\delta[\sigma]) \circ (\delta'[\sigma])$
  \end{itemize}
\end{lemma}
\begin{proof}
  The first two statements are mutually proved. The last two statements are also
  mutually proved. %

  Most of them can be done by simply following the IHs. %
  We give a few examples.
  \begin{itemize}[label=Case]
  \item
    \begin{align*}
      (\LAM{l}{l'}x{S}t)[\delta][\phi]
      &= \LAM{l}{l'}x{S[\delta][\phi]}{(t[\delta,x/x][\phi])} \\
      &= \LAM{l}{l'}x{S[\phi][\delta[\phi]]}{(t[\phi][(\delta,x/x)[\phi]])}
      \byIH \\
      &= \LAM{l}{l'}x{S[\phi][\delta[\phi]]}{(t[\phi][(\delta[\phi],x/x)])} \\
      &= (\LAM{l}{l'}x{S}t)[\phi][\delta[\phi]]
    \end{align*}
  \item
    \begin{align*}
      (\ell,g,x_T. M_\Typ)[\delta][\phi]
      &= (\ell,g,x_T. M_\Typ[\delta,x_T/x_T][\phi,\ell/\ell]) \\
      &= (\ell,g,x_T. M_\Typ[\phi,\ell/\ell][(\delta,x_T/x_T)[\phi,\ell/\ell]])
      \byIH \\
      &= (\ell,g,x_T. M_\Typ[\phi,\ell/\ell][(\delta[\phi],x_T/x_T)])
      \tag{$\delta[\phi,\ell/\ell] = \delta[\phi]$ due to weakening of universe variables}\\
      &= (\ell,g,x_T. M_\Typ)[\phi][\delta[\phi]]
    \end{align*}
  \item
    \begin{align*}
      (\ell,g,u_t,x_t. t_\tEl)[\sigma][\phi]
      &= (\ell,g,u_t,x_t. t_\tEl[\sigma,g/g,u_t^\id/u_t][\phi,\ell/\ell]) \\
      &= (\ell,g,u_t,x_t. t_\tEl[\phi,\ell/\ell][(\sigma,g/g,u_t^\id/u_t)[\phi,\ell/\ell]])
      \byIH \\
      &= (\ell,g,u_t,x_t. t_\tEl[\phi,\ell/\ell][(\sigma[\phi],g/g,u_t^\id/u_t)])
        \tag{$\sigma$ is universe weakened; \Cref{lem:dt:wkact}} \\
      &= (\ell,g,u_t,x_t. t_\tEl)[\phi][\sigma[\phi]]
    \end{align*}
  \item
    \begin{align*}
      \wk^k_g[\sigma][\phi]
      &= \wk^k_{\sigma(g)}[\phi] \\
      &= \wk^k_{\sigma(g)[\phi]} 
      \tag{by \Cref{lem:dt:wkact}} \\
      &= \wk^k_{g[\sigma[\phi]]}  \\
      &= \wk^k_g[\phi][\sigma[\phi]] 
    \end{align*}
    
  \item
    Then we consider $\cdot^k_g$, and we case analyze $\Gamma := \sigma(g)$:
    \begin{itemize}[label=Subcase]
    \item If $\Gamma$ ends with $\cdot$. 
      \begin{align*}
        \cdot^k_g[\sigma][\phi]
        = \cdot^{|\Gamma| + \metalevel} [\phi] = \cdot^{|\Gamma| + \metalevel}
      \end{align*}
    \item If $\Gamma$ ends with $g'$. 
      \begin{align*}
        \cdot^k_g[\sigma][\phi]
        = \cdot^{|\Gamma| + \metalevel}_{g'} [\phi] = \cdot^{|\Gamma| + \metalevel}_{g'}
      \end{align*}
    \end{itemize}
    
  \item
    \begin{align*}
      (\ell,g,u_t,x_t. t_\tEl)[\delta][\sigma]
      &= (\ell,g,u_t,x_t. t_\tEl[\delta,x_t/x_t][\sigma,u_t^\id/u_t]) \\
      &= (\ell,g,u_t,x_t. t_\tEl[\sigma,u_t^\id/u_t][(\delta,x_t/x_t)[\sigma,u_t^\id/u_t]])
      \byIH \\
      &= (\ell,g,u_t,x_t. t_\tEl[\sigma,u_t^\id/u_t][\delta[\sigma],x_t/x_t])
      \tag{$\delta$ is globally weakened} \\
      &= (\ell,g,u_t,x_t. t_\tEl)[\sigma][\delta[\sigma]]
    \end{align*}
    
  \item Now we consider $\wk^k_g \circ \delta$. %
    This composition basically cancels out all terms from $\delta$ and leave a
    weakening behind. %
    In this case, we know that $\delta$ must end with $\wk^{k'}_g$ for some $g$. %
    Moreover, in order to compose, we have that $|\delta| = \metalevel$. %
    Therefore, 
    \begin{align*}
      (\wk^k_g \circ \delta) [\sigma]
      &= \wk^{\widehat \delta}_g[\sigma] \\
      &= \wk^{k'}_{\sigma(g)}
    \end{align*}
    Moreover,
    \begin{align*}
      (\wk^k_g[\sigma]) \circ (\delta[\sigma])
      &= \wk^{|\delta|}_{\sigma(g)} \circ (\delta[\sigma]) \\
      &= \wk^{k'}_{\sigma(g)}
        \tag{$\wk^{|\delta|}_{\sigma(g)}$ projects away all leading terms kept by
        $\delta$ so only $\wk^{k'}_{g}[\sigma]$ is left}
    \end{align*}
    Therefore two expressions are equal. %
    Similar reasoning holds for the case of $\cdot^k_g \circ \delta$. 
    
  \item Then we consider $\cdot^k$ and case analyze $\widecheck \delta$.
    \begin{itemize}[label=Subcase]
    \item If $\widecheck\delta$ is not a context variable, then $\delta$ must end
      with $\cdot^{k'}$ and also $|\delta| = \metalevel$ due to well-typedness.
      \begin{align*}
        (\cdot^k \circ \delta)[\sigma]
        &= \cdot^{\widehat \delta}[\sigma] \\
        &= \cdot^{k'} \\
        &= \cdot^{k'}[\sigma] \\
        &= (\cdot^k[\sigma]) \circ (\delta[\sigma])
      \end{align*}
    \item If $\widecheck\delta$ is context variable $g$, then $\delta$ must end
      with $\cdot^{k'}_g$ and also $|\delta| = \metalevel$ due to well-typedness.
      \begin{align*}
        (\cdot^k \circ \delta)[\sigma]
        &= \cdot^{\widehat \delta}_g[\sigma] \\
        &= \cdot^{k'}_g[\sigma]
      \end{align*}
      Moreover,
      \begin{align*}
        \cdot^k[\sigma] \circ (\delta[\sigma])
        &= \cdot^k \circ (\delta[\sigma]) \\
        &= \cdot^{\widehat{\delta[\sigma]}}_{\widecheck{\delta[\sigma]}} \\
        &= \cdot^{|\sigma(g)| + k'}_{\widecheck{\delta[\sigma]}}
      \end{align*}
      Then we consider whether $\sigma(g)$ ends with another context variable or
      not.
      \begin{itemize}[label=Subsubcase]
      \item If $\sigma(g)$ ends with $\cdot$, then 
        \begin{align*}
          (\cdot^k \circ \delta)[\sigma]
          &= \cdot^{k'}_g[\sigma] \\
          &= \cdot^{|\sigma(g)| + k'}
        \end{align*}
        and also $\delta[\sigma]$ must also end with no global variable.
      \item If $\sigma(g)$ ends with some $g'$, then 
        \begin{align*}
          (\cdot^k \circ \delta)[\sigma]
          &= \cdot^{k'}_g[\sigma] \\
          &= \cdot^{|\sigma(g)| + k'}_{g'}
        \end{align*}
        Then $\widecheck{\delta[\sigma]}$ must also return $g'$. 
      \end{itemize}
    \end{itemize}
  \end{itemize}
\end{proof}

% ----------
\begin{lemma}[Universe Substitutions] $ $
  \begin{itemize}
  \item If $\judge[L']\Psi$ and $\typing[L]\phi{L'}$, then $\judge[L]{\Psi[\phi]}$.
  \item If $\lpjudge[\Psi][L']i\Gamma$ and $\typing[L]\phi{L'}$, then $\lpjudge[\Psi[\phi]]i{\Gamma[\phi]}$.
  \item If $\lpequiv[\Psi][L']i\Gamma\Delta$ and $\typing[L]\phi{L'}$, then $\lpequiv[\Psi[\phi]]i{\Gamma[\phi]}{\Delta[\phi]}$.
  \item If $\lttypwf[\Psi][\Gamma][L']i T l$ and $\typing[L]\phi{L'}$, then
    $\lttypwf[\Psi[\phi]][\Gamma[\phi]] i{T[\phi]}{l[\phi]}$.
  \item If $\lttypeq[\Psi][\Gamma][L']i T{T'} l$ and $\typing[L]\phi{L'}$, then
    $\lttypeq[\Psi[\phi]][\Gamma[\phi]] i{T[\phi]}{T'[\phi]} {l[\phi]}$.
  \item If $\lttyping[\Psi][\Gamma][L']i t T l$ and $\typing[L]\phi{L'}$, then
    $\lttyping[\Psi[\phi]][\Gamma[\phi]] i{t[\phi]}{T[\phi]} {l[\phi]}$.
  \item If $\lttyequiv[\Psi][\Gamma][L']i t{t'} T l$ and $\typing[L]\phi{L'}$, then
    $\lttyequiv[\Psi[\phi]][\Gamma[\phi]] i{t[\phi]}{t'[\phi]}{T[\phi]} {l[\phi]}$.
  \item If $\ltsubst[\Psi][\Gamma][L'] i{\delta}{\Gamma'}$ and $\typing[L]\phi{L'}$, then
    $\ltsubst[\Psi[\phi]][\Gamma[\phi]] i{\delta[\phi]}{\Gamma'[\phi]}$.
  \item If $\ltsubeq[\Psi][\Gamma][L'] i{\delta}{\delta'}{\Gamma'}$ and $\typing[L]\phi{L'}$, then
    $\ltsubeq[\Psi[\phi]][\Gamma[\phi]] i{\delta[\phi]}{\delta'[\phi]}{\Gamma'[\phi]}$.
  \end{itemize}  
\end{lemma}
\begin{proof}
  Do a mutual induction. Most rules do not alter the universe context at all so they
  are discharged naturally. %
  When encountering the only changing cases, i.e. universe-polymorphic functions and
  branches of the recursive principles for code, we extend the unvierse substitutions
  with sufficient new universe variables before applying the IHs. %
  When computation rules are encountered, we apply the composition lemma above before
  applying the IHs.

  We consider one complex computational rule in detail to illustrate the proof of the
  lemma:
  \begin{mathpar}
    \inferrule
    {G_A \\ \lpjudge \metalevel \Gamma \\ \typing[L']{l}\Level \\ \typing[L']{l'}\Level \\ \lpjudge[\Psi][L'] \proglevel \Delta \\
      \lttypwf[\Psi][\Delta][L'] \codelevel S l  \\ \lttypwf[\Psi][\Delta, x : S \at l][L']\codelevel{T}{l'} \\
      t = \boxit{\PI{l}{l'}x S T} \\
      s_S = \ELIMTYP{l_1}{l_2}M b{l}\Delta{(\boxit{S})} \\
      s_T = \ELIMTYP{l_1}{l_2}M b{l'}{(\Delta, x : S \at l)}{(\boxit{T})}}
    {\lttyequiv[\Psi][\Gamma][L'] \metalevel {t_\Pi[l/\ell,l'/\ell',\Delta/g,S/U_S,T/U_T,s_S/x_S,
        s_T/x_T]}{\ELIMTYP{l_1}{l_2}M b{(l \sqcup
          l')}\Delta{t}}{M_\Typ[l\sqcup l'/\ell,\Delta/g,t/x_T]}{l_1}}
  \end{mathpar}
  On the right hand side, $\phi$ simply propagates. %
  By IH, we can show that all premises with the universe-substituted sub-terms are
  well-formed. %
  For example, we have
  \[
    \lttypwf[\Psi[\phi]][\Delta[\phi]]\codelevel{S[\phi]}{l[\phi]}
  \]
  and
  \[
    \lttypwf[\Psi[\phi]][\Delta[\phi], x : S[\phi] \at {l[\phi]}]\codelevel{T[\phi]}{l'[\phi]}
  \]
  Now we consider the left hand side.
  \begin{align*}
    & t_\Pi[l/\ell,l'/\ell',\Delta/g,S/U_S,T/U_T,s_S/x_S,
    s_T/x_T][\phi] \\
    =~ & t_\Pi[l/\ell,l'/\ell'][\phi][\Delta[\phi]/g,S[\phi]/U_S,T[\phi]/U_T,s_S[\phi]/x_S,
      s_T[\phi]/x_T]
      \tag{by \Cref{lem:dt:interact}} \\
    =~ & t_\Pi[\phi,\ell/\ell,\ell'/\ell'][l[\phi]/\ell,l'[\phi]/\ell'][\Delta[\phi]/g,S[\phi]/U_S,T[\phi]/U_T,s_S[\phi]/x_S,
      s_T[\phi]/x_T]
      \tag{by naturality of substitutions}
  \end{align*}
  Note that $t_\Pi[\phi,\ell/\ell,\ell'/\ell']$ is precisely how $\phi$ should be
  propagated in $t_\Pi$. %

  Now we consider the type. %
  A similar reasoning applies:
  \begin{align*}
    M_\Typ[l \sqcup l'/\ell, \Delta/g,t/x_T][\phi]
    &= M_\Typ[l \sqcup l'/\ell][\phi][\Delta[\phi]/g,t[\phi]/x_T]
      \tag{by \Cref{lem:dt:interact}} \\
    &= M_\Typ[\phi,\ell/\ell][l[\phi] \sqcup l'[\phi]/\ell][\Delta[\phi]/g,t[\phi]/x_T]
      \tag{by naturality of substitutions}
  \end{align*}
  This equality verifies the resulting type is correct. %
  The same rule ensures the equivalence lives in the universe level $l_1[\phi]$. 
\end{proof}

Then we consider the properties of local substitutions.
\begin{lemma}[Partial Presupposition]\labeledit{lem:dt:presup1} $ $
  \begin{itemize}
  \item If $\lpjudge i \Gamma$, then $\judge[L]\Psi$. 
  \item If $\lpequiv i \Gamma \Delta$, then $\lpjudge i \Gamma$ and $\lpjudge i
    \Delta$. 
  \item If $\ltsubst i \delta \Delta$, then $\lpjudge i \Gamma$. 
  \item If $\ltsubeq i{\delta}{\delta'}\Delta$, then $\ltsubst i \delta \Delta$ and
    $\ltsubst i{\delta'} \Delta$.
  \item If $\ptyping \sigma \Phi$, then $\judge[L]\Psi$ and $\judge[L]\Phi$.
  \end{itemize}
\end{lemma}
\begin{proof}
  Induction on their respective premises. Note that in the second statement, the
  definition of $\lpequiv i \Gamma \Delta$ is adjusted so that a simple induction
  would suffice. %
  The third statement requires the extra premises added to the step case of the equivalence judgment.
\end{proof}
In fact the lemma above has given full presupposition for local contexts and their
equivalence. %
Therefore, in the forthcoming full presupposition lemma, we do not have to state these
cases.

\begin{lemma}[Symmetry and Transitivity of Local Substitutions] $ $
  \begin{itemize}
  \item If $\ltsubeq i{\delta}{\delta'}\Delta$, then $\ltsubeq i{\delta'}{\delta}\Delta$.
  \item If $\ltsubeq i{\delta}{\delta'}\Delta$ and $\ltsubeq
    i{\delta}{\delta''}\Delta$, then $\ltsubeq i{\delta}{\delta''}\Delta$.
  \end{itemize}
\end{lemma}
\begin{proof}
  By induction.
\end{proof}

\begin{lemma}[Reflexivity] $ $
  \begin{itemize}
  \item If $\lttypwf i T l$, then $\lttypeq i T T l$.
  \item If $\lttyping i t T l$, then $\lttyequiv i t t T l$.
  \item If $\ltsubst i \delta \Delta$, then $\ltsubeq i \delta \delta \Delta$. 
  \item If $\lpjudge i \Gamma$, then $\lpequiv i \Gamma \Gamma$.
  \end{itemize}
\end{lemma}
\begin{proof}
  The first two statements are proved by symmetry and then transitivity. %
  The third (fourth) statement is a natural consequence of the second (first, resp.) statement. 
\end{proof}

\begin{lemma}[Local Substitutions]$ $
  \begin{itemize}
  \item If $\lttypwf[\Psi][\Gamma']i T l$ and $\ltsubst i{\delta}{\Gamma'}$, then
    $\lttypwf i{T[\delta]} l$.
  \item If $\lttypeq[\Psi][\Gamma']i T{T'} l$ and $\ltsubst i{\delta}{\Gamma'}$, then
    $\lttypeq i{T[\delta]}{T'[\delta]} l$.
  \item If $\lttyping[\Psi][\Gamma']i t T l$ and $\ltsubst i{\delta}{\Gamma'}$, then
    $\lttyping i{t[\delta]}{T[\delta]} l$.
  \item If $\lttyequiv[\Psi][\Gamma']i t{t'} T l$ and $\ltsubst i{\delta}{\Gamma'}$, then
    $\lttyequiv i{t[\delta]}{t'[\delta]}{T[\delta]} l$.
  \item If $\ltsubst[\Psi][\Gamma'] i{\delta'}{\Gamma''}$ and $\ltsubst i{\delta}{\Gamma'}$, then
    $\ltsubst i{\delta' \circ \delta}{\Gamma''}$.
  \item If $\ltsubeq[\Psi][\Gamma'] i{\delta'}{\delta''}{\Gamma''}$ and $\ltsubst i{\delta}{\Gamma'}$, then
    $\ltsubeq i{\delta' \circ \delta}{\delta'' \circ \delta}{\Gamma''}$.
  \end{itemize}
\end{lemma}
\begin{proof}
  Do a mutual induction. %
  Many cases go through naturally if their premises do not alter the local context. %
  In the base cases, we use partial presupposition above to obtain $\lpjudge i
  \Gamma$ and $\judge[L]\Psi$. %
  The global variable cases depend on the composition of local substitutions. %
  The local variable case depends on the reflexivity lemma. 
  
  We consider a few cases:
  \begin{itemize}[label=Case]
  \item
    \begin{mathpar}
      \inferrule
      {\lttyequiv[\Psi][\Gamma'] \metalevel{t}{t'}{\TPI U[\Delta] l{l'}{T''}}{l'} \\ \lttypeq[\Psi][\Delta] \proglevel{T}{T'}l}
      {\lttyequiv[\Psi][\Gamma'] \metalevel{\TAPP t{T}}{\TAPP{t'}{T'}}{T''[T/U]}{l'}}
    \end{mathpar}
    \begin{align*}
      & \lttyequiv \metalevel{t[\delta]}{t'[\delta]}{\TPI U[\Delta] l{l'}{(T''[\delta])}}{l'}
        \byIH \\
      & \lttyequiv \metalevel{\TAPP{t[\delta]}{T}}{\TAPP{t'[\delta]}{T'}}{T''[\delta][T/U]}{l'}
        \tag{by the same congruence rule} 
    \end{align*}
    Notice that
    \[
      T''[\delta][T/U] = T''[T/U][\delta[T/U]] = T''[T/U][\delta]
    \]
    where the second equation holds because $\delta$ is not typed in a context with
    $U$.
    
  \item
    \begin{mathpar}
      \inferrule
      {G_A \\ \lpjudge \metalevel {\Gamma'} \\ \lpjudge \proglevel \Delta \\ \typing[L]l\Level  \\
        \typing[L]{l'}\Level \\ \lttypwf[\Psi][\Delta] \codelevel S l \\
        \lttyping[\Psi][\Delta, x : S \at l] \codelevel t{T}{l'} \\
        l_\Pi = l \sqcup l' \\
        T_\Pi = \PI{l}{l'}x{S}{T} \\
        t' = \boxit{\LAM{l}{l'}x{S}{t}} \\
        s_S = \ELIMTYP{l_1}{l_2}M b{l}\Delta{(\boxit{S})}[\delta] \\
        s_t = \ELIMTRM{l_1}{l_2}M b{l'}{(\Delta, x : S \at l)}{T}{(\boxit{t})}[\delta] \\
        \delta' = s_S/x_S,s_t/x_t}
      {\lttyequiv[\Psi][\Gamma'] \metalevel{t_\lambda[l/\ell,l'/\ell',\Delta/g,S/U_S,T/U_T,t/u_t,\delta']}{\ELIMTRM{l_1}{l_2}M b{l_\Pi}\Delta{T_\Pi}{t'}}{M_\Trm[l_\Pi/\ell,\Delta/g,T_\Pi/U_T,t'/x_t]}{l_2}}
    \end{mathpar}
    In this case, we first apply IHs so that $\delta$ is propagated into all premises
    in $G_A$ and we must reason about the left hand side and the result type.
    \begin{align*}
      & t_\lambda[\delta,x_S/x_S,x_t/x_t][l/\ell,l'/\ell',\Delta/g,S/U_S,T/U_T,t/u_t,\delta'] \\
      =~& t_\lambda[l/\ell,l'/\ell',\Delta/g,S/U_S,T/U_T,t/u_t][\delta[l/\ell,l'/\ell',\Delta/g,S/U_S,T/U_T,t/u_t],x_S/x_S,x_t/x_t][\delta']
          \tag{by \Cref{lem:dt:interact}} \\
      =~&t_\lambda[l/\ell,l'/\ell',\Delta/g,S/U_S,T/U_T,t/u_t][\delta,x_S/x_S,x_t/x_t][\delta']
          \tag{$\delta$ has no those variables} \\
      =~&t_\lambda[l/\ell,l'/\ell',\Delta/g,S/U_S,T/U_T,t/u_t][s_S/x_S,s_t/x_t][\delta]
          \tag{naturality of local substitutions}
    \end{align*}
    On the return type, we have
    \begin{align*}
      & M_\Trm[\delta,x_t/x_t][l_\Pi/\ell,\Delta/g,T_\Pi/U_T,t'/x_t] \\
      =~& M_\Trm[l_\Pi/\ell,\Delta/g,T_\Pi/U_T][\delta,x_t/x_t][t'/x_t]
          \tag{by \Cref{lem:dt:interact} similarly} \\
      =~& M_\Trm[l_\Pi/\ell,\Delta/g,T_\Pi/U_T][t'/x_t][\delta]
          \tag{naturality of local substitutions}
    \end{align*}
    Both equations conclude this case. 
  \end{itemize}
\end{proof}

We also need a similar lemma about equivalent local substitutions.
\begin{lemma}[Equivalent Local Substitutions]\labeledit{lem:dt:lsubst-eq}$ $
  \begin{itemize}
  \item If $\lttypwf[\Psi][\Gamma']i T l$ and $\ltsubeq i{\delta_1}{\delta_2}{\Gamma'}$, then
    $\lttypeq i{T[\delta_1]}{T[\delta_2]} l$.
  \item If $\lttypeq[\Psi][\Gamma']i T{T'} l$ and $\ltsubeq i{\delta_1}{\delta_2}{\Gamma'}$, then
    $\lttypeq i{T[\delta_1]}{T'[\delta_2]} l$.
  \item If $\lttyping[\Psi][\Gamma']i t T l$ and $\ltsubeq i{\delta_1}{\delta_2}{\Gamma'}$, then
    $\lttyequiv i{t[\delta_1]}{t[\delta_2]}{T[\delta_1]} l$.
  \item If $\lttyequiv[\Psi][\Gamma']i t{t'} T l$ and $\ltsubeq i{\delta_1}{\delta_2}{\Gamma'}$, then
    $\lttyequiv i{t[\delta_1]}{t'[\delta_2]}{T[\delta_1]} l$.
  \item If $\ltsubst[\Psi][\Gamma'] i{\delta}{\Gamma''}$ and $\ltsubeq i{\delta_1}{\delta_2}{\Gamma'}$, then
    $\ltsubeq i{\delta \circ \delta_1}{\delta \circ \delta_2}{\Gamma''}$.
  \item If $\ltsubeq[\Psi][\Gamma'] i{\delta}{\delta'}{\Gamma''}$ and $\ltsubeq i{\delta_1}{\delta_2}{\Gamma'}$, then
    $\ltsubeq i{\delta \circ \delta_1}{\delta' \circ \delta_2}{\Gamma''}$.
  \end{itemize}
\end{lemma}
\begin{proof}
  We proceed by mutual induction. %
  When we encounter well-formedness of types and well-typedness of terms, we conclude
  by IHs and respective congruence rules. %
  What is more difficult are the asymmetric equivalence rules. %
  We must apply IHs a bit more carefully to obtain the conclusions. %
  We elaborate on a few cases.
  \begin{itemize}[label=Case]
  \item
    \begin{mathpar}
      \inferrule
      {\lttyequiv[\Psi][\Gamma'] i t{t'}T l \\ \lttyequiv[\Psi][\Gamma'] i{t'}{t''}T l}
      {\lttyequiv[\Psi][\Gamma'] i t{t''}T l}
    \end{mathpar}
    \begin{align*}
      & \ltsubst i{\delta_1}{\Gamma'}
        \tag{by presupposition} \\
      & \lttyequiv i {t[\delta_1]}{t'[\delta_1]}{T[\delta_1]} l
        \tag{by local substitution lemma} \\
      & \lttyequiv i{t'[\delta_1]}{t''[\delta_2]}{T[\delta_1]} l
        \byIH \\
      & \lttyequiv i {t[\delta_1]}{t''[\delta_2]}{T[\delta_1]} l
        \tag{by transitivity}
    \end{align*}
  \item
    \begin{mathpar}
      \inferrule
      {\lttyping[\Psi][\Gamma'] i t{T'}l \\ \lttypeq[\Psi][\Gamma']{\typeof i}T{T'}l}
      {\lttyping[\Psi][\Gamma'] i t T l}
    \end{mathpar}
    \begin{align*}
      & \ltsubst i{\delta_1}{\Gamma'}
        \tag{by presupposition} \\
      & \lttypeq{\typeof i}{T[\delta_1]}{T'[\delta_1]}l
        \tag{by local substitution lemma} \\
      & \lttyequiv i {t[\delta_1]}{t[\delta_2]}{T'[\delta_1]}l
        \byIH \\
      & \lttyequiv i {t[\delta_1]}{t[\delta_2]}{T[\delta_1]}l
        \tag{by conversion}
    \end{align*}
    
  \item
    \begin{mathpar}
      \inferrule
      {\lpjudge \metalevel{\Gamma'} \\ \lttyping[\Psi, U : \DTyp[\Delta] \proglevel l][\Gamma']\metalevel{t}{T'}{l'} \\ \typing[L]{l}\Level \\ \typing[L]{l'}\Level \\ \lttypwf[\Psi][\Delta] \proglevel{T} l}
      {\lttyequiv[\Psi][\Gamma'] \metalevel{t[T/U]}{\TAPP{(\TLAM{l}{l'}U t)}{T}}{T'[T/U]}{l'}}
    \end{mathpar}
    \begin{align*}
      & \ltsubst \metalevel{\delta_1}{\Gamma'}
        \tag{by presupposition} \\
      & \lttyequiv \metalevel{t[T/U][\delta_1]}{\TAPP{(\TLAM{l}{l'}U
        {t[\delta_1]})}{T}}{T'[T/U][\delta_1]}{l'}
        \tag{by local substitution lemma} \\
      & \lttyequiv[\Psi, U : \DTyp[\Delta] \proglevel
        l]\metalevel{t[\delta_1]}{t[\delta_2]}{T'[\delta_1]}{l'}
        \byIH \\
      & \lttyequiv \metalevel{\TAPP{(\TLAM{l}{l'}U {t[\delta_1]})}{T}}{\TAPP{(\TLAM{l}{l'}U
        {t[\delta_2]})}{T}}{T'[T/U][\delta_1]}{l'}
        \tag{by congruence; note that $T'[T/U][\delta_1] = T'[\delta_1][T/U]$ as
        $U$ is not in $\delta_1$} \\
      & \lttyequiv \metalevel{t[T/U][\delta_1]}{\TAPP{(\TLAM{l}{l'}U
        {t[\delta_2]})}{T}}{T'[T/U][\delta_1]}{l'}
        \tag{by transitivity} 
    \end{align*}
    
  \item
    \begin{mathpar}
      \inferrule
      {\lpjudge \metalevel{\Gamma'} \\ \typing[L]{l'}\Level \\ \typing[L]l\Level \\ \lpjudge \proglevel \Delta \\
        \lttypwf[\Psi][\Delta]\codelevel T l \\
        \lttypwf[\Psi][\Gamma',x_T : \CTyp[\Delta]l \at{0}]\metalevel{M}{l'} \\
        \lttyping[\Psi, U : \DTyp[\Delta]\codelevel l][\Gamma']\metalevel{t'}{M[\boxit U^\id/x_T]}{l'}}
      {\lttyequiv[\Psi][\Gamma'] \metalevel{t'[T/U]}{\LETBTYP{l'}l \Delta{x_T.M}{U}{t'}{(\boxit T)}}{M[\boxit T/x_T]}{l'}}
    \end{mathpar}
    \begin{align*}
      & \ltsubst \metalevel{\delta_1}{\Gamma'}
        \tag{by presupposition} \\
      & \lpjudge \metalevel{\Gamma}
        \tag{by presupposition} \\        
      & \lttyequiv \metalevel{t'[T/U][\delta_1]}{\LETBTYP{l'}l
        \Delta{x_T.M[\delta_1,x_T/x_T]}{U}{t'[\delta_1]}{(\boxit T)}}{M[\boxit
        T/x_T][\delta_1]}{l'}
        \tag{by local substitution lemma} \\
      & \lttypeq[\Psi][\Gamma,x_T : \CTyp[\Delta]l \at{0}]\metalevel{M[\delta_1,x_T/x_T]}{M[\delta_2,x_T/x_T]}{l'}
        \byIH \\
      & \lttyequiv[\Psi, U : \DTyp[\Delta]\codelevel
        l]\metalevel{t'[\delta_1]}{t'[\delta_2]}{M[\delta_1,x_T/x_T][\boxit U^\id/x_T]}{l'}
        \byIH \\
      & \quad {\LETBTYP{l'}l
        \Delta{x_T.M[\delta_1,x_T/x_T]}{U}{t'[\delta_1]}{(\boxit T)}} \\
      & \approx {\LETBTYP{l'}l
        \Delta{x_T.M[\delta_2,x_T/x_T]}{U}{t'[\delta_2]}{(\boxit T)}}:
        {M[\delta_1,x_T/x_T][\boxit T/x_T]}
        \tag{by congruence} \\
      & \lttyequiv \metalevel{t'[T/U][\delta_1]}{\LETBTYP{l'}l
        \Delta{x_T.M[\delta_2,x_T/x_T]}{U}{t'[\delta_2]}{(\boxit T)}}{M[\delta_1,x_T/x_T][\boxit
        T/x_T]}{l'}
        \tag{by transitivity}
    \end{align*}
    
  \item Finally we consider a $\eta$ rule.
    \begin{mathpar}
      \inferrule
      {\lttyping[\Psi][\Gamma'] \metalevel{t}{\CPI g l T}{l}}
      {\lttyequiv[\Psi][\Gamma'] \metalevel{\CLAM l g {(\CAPP t g)}}{t}{\CPI g l T}{l}}
    \end{mathpar}
    \begin{align*}
      & \ltsubst \metalevel{\delta_1}{\Gamma'}
        \tag{by presupposition} \\
      & \lttyping \metalevel{t[\delta_1]}{\CPI g l{(T[\delta_1])}}{l}
        \tag{by local substitution lemma} \\
      & \lttyequiv \metalevel{\CLAM l g {(\CAPP{(t[\delta_1])} g)}}{t[\delta_1]}{\CPI g l{(T[\delta_1])}}{l}
        \tag{by the same $\eta$ rule} \\
      & \lttyequiv \metalevel{t[\delta_1]}{t[\delta_2]}{\CPI g l{(T[\delta_1])}}{l}
        \byIH \\
      & \lttyequiv \metalevel{\CLAM l g {(\CAPP{(t[\delta_1])} g)}}{t[\delta_2]}{\CPI g
        l{(T[\delta_1])}}{l}
        \tag{by transitivity}
    \end{align*}
  \end{itemize}
\end{proof}
\begin{remark}
  Note that the statement of this lemma is left-biased. %
  For example, when considering terms, the types are substituted by $\delta_1$. %
  This bias causes the whole formulation of equivalence judgment between terms to be
  left biased as well. %
  Otherwise, this lemma cannot be easily justified as above and requires the global
  substitution lemma, which definitely causes issues as the latter depends on this
  very lemma in the global variable cases. %

  A visible effect of this left bias is especially evident in the computation rules. %
  For example, if we define the following $\beta$ rule instead, then the lemma above
  suddenly becomes unprovable at this stage:
  \begin{mathpar}
    \inferrule
    {\lpjudge \metalevel{\Gamma'} \\ \lttyping[\Psi, U : \DTyp[\Delta] \proglevel l][\Gamma']\metalevel{t}{T'}{l'} \\ \typing[L]{l}\Level \\ \typing[L]{l'}\Level \\ \lttypwf[\Psi][\Delta] \proglevel{T} l}
    {\lttyequiv[\Psi][\Gamma'] \metalevel{\TAPP{(\TLAM{l}{l'}U t)}{T}}{t[T/U]}{T'[T/U]}{l'}}
  \end{mathpar}
  Let us work on the proof to see what happens. %
  We now must prove
  \[
    \lttyequiv \metalevel{\TAPP{(\TLAM{l}{l'}U{(t[\delta_1])})}{T}}{t[T/U][\delta_2]}{T'[T/U][\delta_1]}{l'}
  \]
  The only way to introduce $\delta_2$ on the right hand side now is to apply the
  local substitution lemma, which yields
  \[
    \lttyequiv
    \metalevel{\TAPP{(\TLAM{l}{l'}U{(t[\delta_2])})}{T}}{t[T/U][\delta_2]}{T'[T/U][{\color{red}\delta_2}]}{l'}
  \]
  Notice that the return type yields a local substitution $\delta_2$, instead of
  $\delta_1$ as required by the goal. %
  At this stage, however, we are not able to prove the equivalence between
  $T'[T/U][\delta_1]$ and $T'[T/U][\delta_2]$ as we are missing the global
  substitution lemma to justify that $T'[T/U]$ remains well-formed.

  Similarly, we are not able to prove a lemma if some $\eta$ rules are flipped
  either. %
  Consider the following ``innocent'' $\eta$ rule:
  \begin{mathpar}
    \inferrule
    {\lttyping[\Psi][\Gamma'] \metalevel{t}{\CPI g l T}{l}}
    {\lttyequiv[\Psi][\Gamma'] \metalevel{t}{\CLAM l g {(\CAPP t g)}}{\CPI g l T}{l}}
  \end{mathpar}
  Since there is only one premise, we have no choice but to eventually use IH to
  obtain
  \[
    \lttyequiv \metalevel{t[\delta_1]}{t[\delta_2]}{\CPI g l{(T[\delta_1])}}{l}
  \]
  This leaves us \emph{to prove}
  \[
    \lttyequiv \metalevel{t[\delta_2]}{\CLAM l g {(\CAPP{(t[\delta_2])} g)}}{\CPI g l {(T[{\color{red}\delta_1}])}}{l}
  \]
  Notice how the equivalence itself talks about $\delta_2$ exclusively while the type
  refers to $\delta_1$. %
  This asymmetry forces us to flip the equivalence to obtain a better proof. 
\end{remark}

Then we move on to the global substitution lemma. %
We must first establish a number of other lemmas. %
The lifting lemma is one of the guiding lemmas of the layering principle, where we
require that well-formedness can be carried over to higher layers. 
\begin{lemma}[Lifting]\labeledit{lem:dt:lift} If $i \le i'$, and
  \begin{itemize}
  \item $\lpjudge i \Gamma$, then $\lpjudge{i'} \Gamma$;
  \item $\lpequiv i \Gamma \Delta$, then $\lpequiv{i'} \Gamma \Delta$;
  \item $\lttypwf i T l$, then $\lttypwf{i'} T l$;
  \item $\lttypeq i T{T'} l$, then $\lttypeq{i'} T{T'} l$;
  \item $\lttyping i t T l$, then $\lttyping{i'} t T l$.
  \item $\lttyequiv i t{t'} T l$, then $\lttyequiv{i'} t{t'} T l$;
  \item $\ltsubst i{\delta}{\Delta}$, then $\ltsubst{i'}{\delta}{\Delta}$;
  \item $\ltsubeq i{\delta}{\delta'}{\Delta}$, then $\ltsubeq{i'}{\delta}{\delta'}{\Delta}$.
  \end{itemize}
\end{lemma}
\begin{proof}
  First, we realize that the $\ttypeof$ function is monotonic, i.e. $\typeof i \le
  \typeof{i'}$. %
  We proceed by a mutual induction. %
  Most cases are obvious by IHs. %
  Notice that there are cases where we have premises like $\lpjudge{\typeof i}
  \Gamma$, so we must apply IH to obtain $\lpjudge{\typeof{i'}} \Gamma$ with the
  monotonicity property above. %
  It works similarly for the conversion rule, where we have $\lttypeq{\typeof i}T
  {T'} l$. %
  In the cases of global variables, the transitivity of $\le$ eventually complete the
  proof of this lemma. %
  We elaborate on one case:
  \begin{mathpar}
    \inferrule
    {\lpjudge{\typeof i}\Gamma \\ u : \DTrm[\Delta]{i''}T l \in \Psi \\ i'' \in \{\varlevel, \codelevel\} \\ i \in
      \{\varlevel, \codelevel, \proglevel, \metalevel\} \\ i'' \le i \\ \ltsubst i \delta \Delta}
    {\lttyping i{u^\delta}{T[\delta]}{l}}
  \end{mathpar}
  \begin{align*}
    & \lpjudge{\typeof{i'}}\Gamma
      \byIH \\
    & \ltsubst{i'} \delta \Delta
      \byIH \\
    & i'' \le i \le i' \\
    & \lttyping{i'}{u^\delta}{T[\delta]}{l}
      \tag{by the same rule}
  \end{align*}
\end{proof}

The inverse of lifting sometimes is possible
\begin{lemma}[Unlifting] $ $
  \begin{itemize}
  \item $\lttypwf \proglevel T l$, then $\lttypwf \codelevel T l$;
  \item $\lttyping \proglevel t T l$, then $\lttyping \codelevel t T l$.
  \item $\ltsubst \proglevel{\delta}{\Delta}$, then $\ltsubst \codelevel{\delta}{\Delta}$;
  \end{itemize}
\end{lemma}
\begin{proof}
  Induction. %
  Notice that $\typeof \proglevel = \typeof \codelevel = \proglevel$. 
\end{proof}
The unlifting lemma says that the typing at layer $\proglevel$ can be unlifted back to layer
$\codelevel$. %

As another guiding lemma, we have the static code lemma, which states that code at
layer $\varlevel$ and $\codelevel$ has no computational behavior.
\begin{lemma}[Static Code]\labeledit{lem:dt:static} If $i \in \{\varlevel, \codelevel\}$,
  \begin{itemize}
  \item $\lttypeq i T{T'} l$, then $T = T'$;
  \item $\lttyequiv i t{t'} T l$, then $t = t'$;
  \item $\ltsubeq i{\delta}{\delta'}{\Delta}$, then $\delta = \delta'$.
  \end{itemize}
  All equalities above are quotient over the equivalence of universe levels.
\end{lemma}
\begin{proof}
  Mutual induction. %
  We are not concerned about the equivalence of types due to the conversion rule.
\end{proof}
We emphasize again that the equalities hold modulo the equivalence of universe
levels. %
For example, $\Ty{\ell \sqcup \ell'}$ and $\Ty{\ell' \sqcup \ell}$ as code are
considered \emph{equal}, though their universe levels are not exactly syntactically
identical. %
This is fine as we know how to decide the equality between two universe levels as
shown in \Cref{sec:dt:ulevel}. 

\begin{lemma}[Global Substitutions]$ $
  \begin{itemize}
  \item If $\lpjudge[\Phi] i \Gamma$, $i \in \{\proglevel, \metalevel\}$ and $\ptyping{\sigma}{\Phi}$, then $\lpjudge i {\Gamma[\sigma]}$.
  \item If $\lpequiv[\Phi] i \Gamma \Delta$, $i \in \{\proglevel, \metalevel\}$ and $\ptyping{\sigma}{\Phi}$, then $\lpequiv i {\Gamma[\sigma]}{\Delta[\sigma]}$.
  \item If $\lttypwf[\Phi]i T l$ and $\ptyping\sigma\Phi$, then
    $\lttypwf[\Psi][\Gamma[\sigma]] i{T[\sigma]} l$.
  \item If $\lttypeq[\Phi]i T{T'} l$ and $\ptyping\sigma\Phi$, then
    $\lttypeq[\Psi][\Gamma[\sigma]] i{T[\sigma]}{T'[\sigma]} l$.
  \item If $\lttyping[\Phi]i t T l$ and $\ptyping\sigma\Phi$, then
    $\lttyping[\Psi][\Gamma[\sigma]] i{t[\sigma]}{T[\sigma]} l$.
  \item If $\lttyequiv[\Phi]i t{t'} T l$ and $\ptyping\sigma\Phi$, then
    $\lttyequiv[\Psi][\Gamma[\sigma]] i{t[\sigma]}{t'[\sigma]}{T[\sigma]} l$.
  \item If $\ltsubst[\Phi] i{\delta}{\Delta}$ and $\ptyping\sigma\Phi$, then
    $\ltsubst[\Psi][\Gamma[\sigma]] i{\delta[\sigma]}{\Delta[\sigma]}$.
  \item If $\ltsubeq[\Phi] i{\delta}{\delta'}{\Delta}$ and $\ptyping\sigma\Phi$, then
    $\ltsubeq[\Psi][\Gamma[\sigma]] i{\delta[\sigma]}{\delta'[\sigma]}{\Delta[\sigma]}$.
  \end{itemize}
\end{lemma}
\begin{proof}
  We proceed by a mutual induction. %
  Notice that in the first two statements, $i \in \{\proglevel, \metalevel\}$, namely the range of the
  $\ttypeof$ function. %
  This ensures a lookup $\sigma(g)$ of a context variable $g$ to be well-formed at
  layer $i$, due to \Cref{lem:dt:lift}.  %
  Most cases can be discharged by IHs directly. %
  The complex cases are the computation rules and the global variable cases. %
  
  We consider a few cases:
  \begin{itemize}[label=Case]
  \item
    \begin{mathpar}
      \inferrule
      {\lpjudge[\Phi]{\typeof i}\Gamma \\ u : \DTrm[\Delta]{i'}T l \in \Phi \\ i' \in \{\varlevel, \codelevel\} \\ i \in
        \{\varlevel, \codelevel, \proglevel, \metalevel\} \\ i' \le i \\ \ltsubeq[\Phi] i \delta{\delta'} \Delta}
      {\lttyequiv[\Phi] i{u^\delta}{u^{\delta'}}{T[\delta]}{l}}
    \end{mathpar}
    \begin{align*}
      & \lttyping[\Psi][\Delta[\sigma]]{i'}{\sigma(u)}{T[\sigma]}l
        \tag{by lookup} \\
      & \lttyping[\Psi][\Delta[\sigma]]{i}{\sigma(u)}{T[\sigma]}l
        \tag{by lifting} \\
      & \ltsubeq[\Psi][\Gamma[\sigma]] i
        {\delta[\sigma]}{\delta'[\sigma]}{\Delta[\sigma]}
        \byIH \\
      &
        \lttyequiv[\Psi][\Gamma[\sigma]]{i}{\sigma(u)[\delta[\sigma]]}{\sigma(u)[\delta'[\sigma]]}{T[\sigma][\delta[\sigma]]}l
        \tag{by equivalent local substitution lemma}
    \end{align*}
    Notice that
    \[
      T[\sigma][\delta[\sigma]] = T[\delta][\sigma]
    \]
    
  \item
    \begin{mathpar}
      \inferrule
      {\lttyequiv[\Phi] \metalevel{t}{t'}{\UPI \ell l T}{\omega} \\ |\vect\ell| = |\vect l| = |\vect
        l'| > 0 \\
        \forall 0 \le n < |\vect l| ~.~ \tyequiv[L]{\vect l(n)}{\vect l'(n)}\Level}
      {\lttyequiv[\Phi] \metalevel{t~\$~\vect l}{t'~\$~\vect l'}{T[\vect l/\vect \ell]}{l[\vect l/\vect \ell]}}
    \end{mathpar}
    \begin{align*}
      & \lttyequiv[\Psi][\Gamma[\sigma]] \metalevel{t[\sigma]}{t'[\sigma]}{\UPI \ell l{(T[\sigma])}}{\omega}
        \byIH \\
      & \lttyequiv[\Psi][\Gamma[\sigma]] \metalevel{(t[\sigma])~\$~\vect l}{(t'[\sigma])~\$~\vect l'}{T[\sigma][\vect l/\vect \ell]}{l[\vect l/\vect \ell]}
    \end{align*}
    Note that
    \[
      T[\sigma][\vect l/\vect \ell] = T[\vect l/\vect \ell][\sigma[\vect l/\vect \ell]] = T[\vect l/\vect \ell][\sigma]
    \]
    because all $\vect\ell$ do not occur in $\sigma$.
    
  \item
    \begin{mathpar}
      \inferrule
      {G_A \\ \lpjudge \metalevel{\Gamma} \\ \typing[L]{l}\Level \\ \typing[L]{l'}\Level \\ \lpjudge[\Phi] \proglevel \Delta \\
        \lttypwf[\Phi][\Delta] \codelevel S l  \\ \lttypwf[\Phi][\Delta, x : S \at l]\codelevel{T}{l'} \\
        t = \boxit{\PI{l}{l'}x S T} \\
        s_S = \ELIMTYP{l_1}{l_2}M b{l}\Delta{(\boxit{S})} \\
        s_T = \ELIMTYP{l_1}{l_2}M b{l'}{(\Delta, x : S \at l)}{(\boxit{T})}}
      {\lttyequiv[\Phi] \metalevel {t_\Pi[l/\ell,l'/\ell',\Delta/g,S/U_S,T/U_T,s_S/x_S,
          s_T/x_T]}{\ELIMTYP{l_1}{l_2}M b{(l \sqcup
            l')}\Delta{t}}{M_\Typ[l\sqcup l'/\ell,\Delta/g,t/x_T]}{l_1}}
    \end{mathpar}
    We first proceed by using IHs on the premises, which include the following
    judgments:
    \begin{align*}
      & \lpjudge \proglevel {\Delta[\sigma]} \\
      & \lttypwf[\Psi][\Delta[\sigma]] \codelevel{S[\sigma]} l  \\
      & \lttypwf[\Psi][\Delta[\sigma], x : S[\sigma] \at l]\codelevel{T[\sigma]}{l'}
    \end{align*}
    By using the same $\beta$ rule, we must check the resulting left hand side and the
    result type are equal to the target goal. %
    Let us first consider the left hand side:
    \begin{align*}
      & t_\Pi[\sigma,g/g,U_S^\id/U_S,U_T^\id/U_T][l/\ell,l'/\ell',\Delta[\sigma]/g,S[\sigma]/U_S,T[\sigma]/U_T,s_S[\sigma]/x_S,
        s_T[\sigma]/x_T] \\
      =~& t_\Pi[l/\ell,l'/\ell'][\sigma,g/g,U_S^\id/U_S,U_T^\id/U_T][\Delta[\sigma]/g,S[\sigma]/U_S,T[\sigma]/U_T,s_S[\sigma]/x_S,
          s_T[\sigma]/x_T]
          \tag{by \Cref{lem:dt:interact}; $\ell$ and $\ell'$ do not occur in $\sigma$}\\
      =~& t_\Pi[l/\ell,l'/\ell'][\sigma,\Delta[\sigma]/g,S[\sigma]/U_S,T[\sigma]/U_T][s_S[\sigma]/x_S,
          s_T[\sigma]/x_T] \\
      =~& t_\Pi[l/\ell,l'/\ell'][\Delta/g,S/U_S,T/U_T][\sigma][s_S[\sigma]/x_S,
          s_T[\sigma]/x_T]
          \tag{by naturality} \\
      =~& t_\Pi[l/\ell,l'/\ell'][\Delta/g,S/U_S,T/U_T][s_S/x_S, s_T/x_T][\sigma]
          \tag{by \Cref{lem:dt:interact}}
    \end{align*}
    Then we consider the result type in a similar way:
    \begin{align*}
      & M_\Typ[\sigma,g/g][l\sqcup l'/\ell,\Delta[\sigma]/g,t[\sigma]/x_T] \\
      =~& M_\Typ[l\sqcup l'/\ell][\sigma,g/g][\Delta[\sigma]/g,t[\sigma]/x_T] \\
      =~& M_\Typ[l\sqcup l'/\ell][\Delta/g][\sigma][t[\sigma]/x_T] \\
      =~& M_\Typ[l\sqcup l'/\ell][\Delta/g][t/x_T][\sigma]
    \end{align*}
    Both equations allow us to conclude the goal. 
  \end{itemize}
\end{proof}

Next, we consider the effect of equivalent global substitutions on the judgments. %
We first define the equivalence relation between global substitutions:
\begin{mathpar}
  \inferrule*
  {\judge[L] \Psi}
  {\ptyequiv{\cdot}{\cdot}{\cdot}}

  \inferrule*
  {\ptyequiv{\sigma}{\sigma'}{\Phi} \\ \lpequiv \proglevel \Gamma\Delta}
  {\ptyequiv{\sigma, \Gamma/g}{\sigma', \Delta/g}{\Phi, g: \Ctx}}

  \inferrule*
  {\ptyequiv{\sigma}{\sigma'}{\Phi} \\ \lpjudge[\Phi] \proglevel \Gamma \\ \typing[L]l\Level \\ i \in \{\codelevel, \proglevel\}
    \\ \lttypwf[\Psi][\Gamma[\sigma]] i{T}l \\ \lttypwf[\Psi][\Gamma[\sigma]] i{T'}l \\
    \lttypeq[\Psi][\Gamma[\sigma]] i{T}{T'}l}
  {\ptyequiv{\sigma, T/U}{\sigma', T'/U}{\Phi, u : \DTyp i l}}
  
  \inferrule*
  {\ptyequiv{\sigma}{\sigma'}{\Phi} \\ \lttypwf[\Phi] \proglevel T l \\ \typing[L]l\Level \\
    i \in \{\varlevel, \codelevel\} \\ \lttyping[\Psi][\Gamma[\sigma]] i t{T[\sigma]}l \\
    \lttyping[\Psi][\Gamma[\sigma]] i{t'}{T[\sigma]}l \\
    \lttyequiv[\Psi][\Gamma[\sigma]] i t{t'}{T[\sigma]}l}
  {\ptyequiv[\Psi]{\sigma, t/u}{\sigma', t'/u}{\Phi, u : \DTrm{i} T l}}
\end{mathpar}
We can then consider similar properties of this equivalence relation.
\begin{lemma}[Presupposition]
  If $\ptyequiv{\sigma}{\sigma'}{\Phi}$, then $\ptyping{\sigma}{\Phi}$ and $\ptyping{\sigma'}{\Phi}$.
\end{lemma}
\begin{proof}
  Induction. 
\end{proof}

\begin{lemma}[Equivalent Global Substitutions]\labeledit{lem:dt:gsubst-eq}$ $
  \begin{itemize}
  \item If $\lpjudge[\Phi] i \Gamma$, $i \in \{\proglevel, \metalevel\}$ and $\ptyequiv{\sigma}{\sigma'}{\Phi}$, then $\lpequiv i {\Gamma[\sigma]}{\Gamma[\sigma']}$.
  \item If $\lpequiv[\Phi] i \Gamma \Delta$, $i \in \{\proglevel, \metalevel\}$ and $\ptyequiv{\sigma}{\sigma'}{\Phi}$, then $\lpequiv i {\Gamma[\sigma]}{\Delta[\sigma']}$.
  \item If $\lttypwf[\Phi]i T l$ and $\ptyequiv{\sigma}{\sigma'}{\Phi}$, then
    $\lttypeq[\Psi][\Gamma[\sigma]] i{T[\sigma]}{T[\sigma']} l$.
  \item If $\lttypeq[\Phi]i T{T'} l$ and $\ptyequiv{\sigma}{\sigma'}{\Phi}$, then
    $\lttypeq[\Psi][\Gamma[\sigma]] i{T[\sigma]}{T'[\sigma']} l$.
  \item If $\lttyping[\Phi]i t T l$ and $\ptyequiv{\sigma}{\sigma'}{\Phi}$, then
    $\lttyequiv[\Psi][\Gamma[\sigma]] i{t[\sigma]}{t[\sigma']}{T[\sigma]} l$.
  \item If $\lttyequiv[\Phi]i t{t'} T l$ and $\ptyequiv{\sigma}{\sigma'}{\Phi}$, then
    $\lttyequiv[\Psi][\Gamma[\sigma]] i{t[\sigma]}{t'[\sigma']}{T[\sigma]} l$.
  \item If $\ltsubst[\Phi] i{\delta}{\Delta}$ and $\ptyequiv{\sigma}{\sigma'}{\Phi}$, then
    $\ltsubeq[\Psi][\Gamma[\sigma]] i{\delta[\sigma]}{\delta[\sigma']}{\Delta[\sigma]}$.
  \item If $\ltsubeq[\Phi] i{\delta}{\delta'}{\Delta}$ and $\ptyequiv{\sigma}{\sigma'}{\Phi}$, then
    $\ltsubeq[\Psi][\Gamma[\sigma]] i{\delta[\sigma]}{\delta'[\sigma']}{\Delta[\sigma]}$.
  \end{itemize}
\end{lemma}
\begin{proof}
  We apply mutual induction. %
  This lemma is much less sensitive to the exact statement of rules compared to
  \Cref{lem:dt:lsubst-eq}. %
  Since now we have proved the global substitution lemma, we could use conversion
  rules whenever necessary.
\end{proof}

At this point, we have concluded that all substitutions are coherent with
well-formedness and typing judgments. %
Next, we shall move towards the full presupposition lemma and end our discussion on
syntactic properties with it.

\subsection{Context Equivalence and Presupposition}

In order to establish presupposition, we must concern ourselves with the asymmetry in
the congruence rules of the equivalence judgments. %
Presupposition, intuitively, requires us to show that this asymmetry ``does not
matter''. %
This intuition is formalized by the context equivalence lemma. %
In fact, we need two such lemmas, as we need to show one for local contexts and one
for global contexts. %
In light of that, let us proceed with the lemma for local contexts first.
% \begin{lemma}[Local Context Equivalence Weakening]\labeledit{lem:dt:lctxeqwk} $ $
%   \begin{itemize}
%   \item If $\lpequiv i\Delta \Gamma$ and $\judge[L]{\Psi,\Phi}$, then $\lpequiv[\Psi,\Phi]i\Delta\Gamma$.
%   \item If $\lpequiv i\Delta \Gamma$, then $\lpequiv[\Psi][L,L']i\Delta\Gamma$.
%   \end{itemize}
% \end{lemma}
% \begin{proof}
%   Induction. 
% \end{proof}

\begin{lemma}[Local Context Equivalence]\labeledit{lem:dt:lctxeq1}$ $
  \begin{itemize}
  \item If $\lttypwf[\Psi][\Delta]i T l$ and $\lpequiv{\typeof i}\Delta\Gamma$, then
    $\lttypwf i{T} l$.
  \item If $\lttypeq[\Psi][\Delta]i T{T'} l$ and $\lpequiv{\typeof i}\Delta\Gamma$, then
    $\lttypeq i{T}{T'} l$.
  \item If $\lttyping[\Psi][\Delta]i t T l$ and $\lpequiv{\typeof i}\Delta\Gamma$, then
    $\lttyping i{t}{T} l$.
  \item If $\lttyequiv[\Psi][\Delta]i t{t'} T l$ and $\lpequiv{\typeof i}\Delta\Gamma$, then
    $\lttyequiv i{t}{t'}{T} l$.
  \item If $\ltsubst[\Psi][\Delta] i{\delta}{\Gamma'}$ and $\lpequiv{\typeof i}\Delta\Gamma$, then
    $\ltsubst i{\delta}{\Gamma'}$.
  \item If $\ltsubeq[\Psi][\Delta] i{\delta}{\delta'}{\Gamma'}$ and $\lpequiv{\typeof i}\Delta\Gamma$, then
    $\ltsubeq i{\delta}{\delta'}{\Gamma'}$.
  \end{itemize}
\end{lemma}
\begin{proof}
  We start by mutual induction. %
  The base case is the local variable cases, where we simply apply the conversion rule
  to the equivalence given by $\lpequiv{\typeof i}\Delta\Gamma$. %
  We might also use presupposition (\Cref{lem:dt:presup1}) to derive $\lpjudge{\typeof
    i}\Delta$. %
  Otherwise, most cases can be handled by IHs. %
  In cases where local contexts are extended with variables, we shall carefully use
  IHs to obtain the necessary premises to extend the equivalence of $\Delta$ and
  $\Gamma$ as well. %
  % Also make use of \Cref{lem:dt:lctxeqwk} whenever necessary. %

  We consider a few cases:
  \begin{itemize}[label=Case]
  \item
    \begin{mathpar}
      \inferrule
      {\tyequiv[L]{l_1}{l_3}\Level  \\ \tyequiv[L]{l_2}{l_4}\Level \\ \lttypwf[\Psi][\Delta] i S{l_1} \\ \lttypeq[\Psi][\Delta] i
        S{S'}{l_1} \\
        \lttyequiv[\Psi][\Delta, x : S \at {l_1}] i t{t'}{T}{l_2}}
      {\lttyequiv[\Psi][\Delta] i{\LAM {l_1}{l_2} x S t}{\LAM {l_3}{l_4} x{S'}{t'}}{\PI l{l'} x S T}{l \sqcup l'}}
    \end{mathpar}
    In this case, the crucial part is to be able to invoke IH on $t \approx t'$. %
    We proceed as follows:
    \begin{align*}
      & \lttypwf i S{l_1}
      \byIH \\
      & \lpequiv{\typeof i}{\Delta, x : S \at {l_1}}{\Gamma, x : S \at {l_1}}
        \tag{by lifting and step case of the equivalence} \\
      & \lttyequiv[\Psi][\Gamma, x : S \at {l_1}] i t{t'}{T}{l_2}
        \byIH
    \end{align*}
    Then IHs will allow us to conclude the rest. 
  \item
    \begin{mathpar}
      \inferrule
      {\lpjudge \metalevel \Delta \\ \typing[L]{l'}\Level \\ \typing[L]l\Level \\ \lpjudge \proglevel{\Delta'} \\
        \lttypwf[\Psi][\Delta']\codelevel T l \\
        \lttypwf[\Psi][\Delta,x_T : \CTyp[\Delta']l \at{0}]\metalevel{M}{l'} \\
        \lttyping[\Psi, U : \DTyp[\Delta']\codelevel l][\Delta]\metalevel{t'}{M[\boxit U^\id/x_T]}{l'}}
      {\lttyequiv[\Psi][\Delta] \metalevel{t'[T/U]}{\LETBTYP{l'}l{\Delta'}{x_T.M}{U}{t'}{(\boxit T)}}{M[\boxit T/x_T]}{l'}}
    \end{mathpar}
    \begin{align*}
      & \lpjudge \metalevel \Gamma
        \tag{by presupposition (\Cref{lem:dt:presup1})} \\
      & \lpequiv \metalevel{\Delta,x_T : \CTyp[\Delta']l \at{0}}{\Gamma,x_T :
        \CTyp[\Delta']l \at{0}}
        \tag{note that well-formedness of $\CTyp[\Delta']l$ does not depend on
        $\Delta$ or $\Gamma$} \\
      & \lttypwf[\Psi][\Gamma,x_T : \CTyp[\Delta']l \at{0}]\metalevel{M}{l'}
        \byIH \\
      & \lpequiv[\Psi, U : \DTyp[\Delta']\codelevel l] \metalevel{\Delta}{\Gamma}
        \tag{by global weakening} \\
      & \lttyping[\Psi, U : \DTyp[\Delta']\codelevel l]\metalevel{t'}{M[\boxit U^\id/x_T]}{l'}
        \byIH \\
      & \lttyequiv \metalevel{t'[T/U]}{\LETBTYP{l'}l{\Delta'}{x_T.M}{U}{t'}{(\boxit T)}}{M[\boxit T/x_T]}{l'}
    \end{align*}
  \end{itemize}
\end{proof}
As a corollary, we can prove the following lemma.
\begin{lemma}[Symmetry and Transitivity of Local Contexts]$ $
  \begin{itemize}
  \item If $\lpequiv i \Gamma \Delta$, then $\lpequiv i \Delta \Gamma$.
  \item If $\lpequiv i{\Gamma_1}{\Gamma_2}$ and $\lpequiv i{\Gamma_2}{\Gamma_3}$, then
    $\lpequiv i{\Gamma_1}{\Gamma_3}$.
  \end{itemize}
\end{lemma}
\begin{proof}
  Induction. Note that transitivity replies on the local context equivalence lemma. 
\end{proof}
A similar lemma replaces the codomain local contexts of local substitutions. %
This variant is much simpler just by conversion rules.
\begin{lemma}[Local Context Conversion]\labeledit{lem:dt:lctxeq2} $ $
  \begin{itemize}
  \item If $\ltsubst i{\delta}{\Gamma'}$ and $\lpequiv{\typeof i}{\Gamma'}{\Delta}$, then
    $\ltsubst i{\delta}{\Delta}$.
  \item If $\ltsubeq i{\delta}{\delta'}{\Gamma'}$ and $\lpequiv{\typeof i}{\Gamma'}{\Delta}$, then
    $\ltsubeq i{\delta}{\delta'}{\Delta}$.
  \end{itemize}
\end{lemma}
\begin{proof}
  By induction. Propagate conversion rules together with the local substitution lemma in the step case. 
\end{proof}

Then we work on the global context equivalence lemma. %
To state this lemma, we should first specify what does that mean for two global
contexts are equivalent.
\begin{mathpar}
  \inferrule
  { }
  {\gequiv\cdot\cdot}

  \inferrule
  {\gequiv\Psi\Phi}
  {\gequiv{\Psi, g : \Ctx}{\Phi, g : \Ctx}}

  \inferrule
  {\gequiv\Psi\Phi \\ \lpjudge \proglevel \Gamma \\ \lpjudge[\Phi] \proglevel\Delta \\ \lpequiv \proglevel \Gamma\Delta \\ \typing[L]l\Level \\ i \in \{\codelevel, \proglevel\}}
  {\gequiv{\Psi, U : \DTyp i l}{\Phi, U : \DTyp[\Delta] i l}}

  \inferrule
  {\gequiv\Psi \Phi % \\ \lpjudge \proglevel \Gamma \\ \lpjudge[\Phi] \proglevel\Delta
    \\ \lpequiv \proglevel
    \Gamma\Delta \\ \lttypwf \proglevel T l \\ \lttypwf[\Phi][\Delta] \proglevel{T'} l \\ \lttypeq \proglevel T{T'} l \\ \typing[L]l\Level \\ i \in \{\varlevel, \codelevel\}}
  {\gequiv{\Psi, u : \DTrm i T l}{\Phi, u : \DTrm[\Delta] i{T'} l}}
\end{mathpar}
Essentially, the equivalence of global contexts are just point-wise equivalence of
types within. %
We can reconstruct the well-formedness of both components from the premises:
\begin{lemma}[Presupposition of Equivalence of Global Contexts]\labeledit{lem:dt:presup-gctx}
  If $\gequiv \Psi\Phi$, then $\judge[L]\Psi$ and $\judge[L]\Phi$. 
\end{lemma}
For the global context equivalence lemma, we would like to take a shortcut by taking
advantage of the global substitution lemma. %
\begin{lemma}\labeledit{lem:dt:lctxeq-id}
  If $\lpequiv{\typeof i}\Gamma \Delta$, then $\ltsubst i{\id}\Delta$. 
\end{lemma}
\begin{lemma}\labeledit{lem:dt:gctxeq-id}
  If $\gequiv \Psi\Phi$, then $\ptyping{\id}\Phi$. 
\end{lemma}
\begin{proof}
  We proceed by induction. %
  In each step case, notice that weakening is used implicitly. %
  Use \Cref{lem:dt:lctxeq-id} to derive $\ltsubst \proglevel{\id}\Delta$ whenever necessary.
\end{proof}

\begin{lemma}[Global Context Equivalence]$ $
  \begin{itemize}
  \item If $\lpjudge[\Phi] i \Gamma$ and $\gequiv\Phi\Psi$, then $\lpjudge i {\Gamma}$.
  \item If $\lpequiv[\Phi] i \Gamma \Delta$ and $\gequiv\Phi\Psi$, then $\lpequiv i {\Gamma}{\Delta}$.
  \item If $\lttypwf[\Phi]i T l$ and $\gequiv\Phi\Psi$, then
    $\lttypwf[\Psi][\Gamma] i{T} l$.
  \item If $\lttypeq[\Phi]i T{T'} l$ and $\gequiv\Phi\Psi$, then
    $\lttypeq[\Psi][\Gamma] i{T}{T'} l$.
  \item If $\lttyping[\Phi]i t T l$ and $\gequiv\Phi\Psi$, then
    $\lttyping[\Psi][\Gamma] i{t}{T} l$.
  \item If $\lttyequiv[\Phi]i t{t'} T l$ and $\gequiv\Phi\Psi$, then
    $\lttyequiv[\Psi][\Gamma] i{t}{t'}{T} l$.
  \item If $\ltsubst[\Phi] i{\delta}{\Delta}$ and $\gequiv\Phi\Psi$, then
    $\ltsubst[\Psi][\Gamma] i{\delta}{\Delta}$.
  \item If $\ltsubeq[\Phi] i{\delta}{\delta'}{\Delta}$ and $\gequiv\Phi\Psi$, then
    $\ltsubeq[\Psi][\Gamma] i{\delta}{\delta'}{\Delta}$.
  \end{itemize}
\end{lemma}
\begin{proof}
  We have $\ptyping\id\Phi$ due to \Cref{lem:dt:gctxeq-id}. %
  Then by the global substitution lemma, we have our goal by knowing that a global
  identity substitution $\id$ does no action. 
\end{proof}

Finally, we prove the presupposition lemma, which is the last guiding lemma of the
layering principle.
\begin{lemma}[Presupposition] $ $
  \begin{itemize}
  \item If $\lttypwf i T l$, then $\lpjudge{\typeof i}\Gamma$ and $\typing[L]l\Level$ or $i = \metalevel \wedge l = \omega$. 
  \item If $\lttypeq i T{T'} l$, then $\lpjudge{\typeof i}\Gamma$, $\lttypwf i T l$,
    $\lttypwf i{T'}l$ and $\typing[L]l\Level$ or $i = \metalevel \wedge l = \omega$.
  \item If $\lttyping i t T l$, then $\lpjudge{\typeof i}\Gamma$, \newline
    $\lttypwf{\typeof i} T l$ and $\typing[L]l\Level$ or $i = \metalevel \wedge l = \omega$.
  \item If $\lttyequiv i t{t'} T l$, then $\lpjudge{\typeof i}\Gamma$, $\lttyping i t
    T l$, $\lttyping i{t'} T l$, $\lttypwf{\typeof i} T l$
     and $\typing[L]l\Level$ or $i = \metalevel \wedge l = \omega$.
  \item If $\ltsubst i{\delta}{\Delta}$, then $\lpjudge{\typeof i}\Delta$.
  \item If $\ltsubeq i{\delta}{\delta'}{\Gamma'}$, then $\lpjudge{\typeof i}\Delta$.
  \end{itemize}  
\end{lemma}
Notice that in the statement of the lemma, we sometimes conclude $i = \metalevel \wedge l =
\omega$. %
The only occasion when $\omega$ is used is when universe polymorphic functions are
involved. %
In that case, we know for sure that $i = \metalevel$. %
In any other cases, we obtain $\typing[L]l\Level$, which excludes $l = \omega$.
\begin{proof}
  We proceed by a mutual induction. %
  In certain congruence rules, we must apply
  \Cref{lem:dt:lsubst-eq,lem:dt:gsubst-eq} to resolve the asymmetry in the rules. %
  Otherwise, we simply apply the substitution lemmas whenever necessary. %
  Note that our rules are stated with redundant premises to make sure this lemma
  eventually checks out. %
  \begin{mathpar}
    \inferrule
    {\lttyequiv \metalevel{t}{t'}{\TPI U[\Delta] l{l'}{T''}}{l'} \\ \lttypeq[\Psi][\Delta] \proglevel{T}{T'}l}
    {\lttyequiv \metalevel{\TAPP t{T}}{\TAPP{t'}{T'}}{T''[T/U]}{l'}}
  \end{mathpar}
  \begin{align*}
    & \lttyping \metalevel{t}{\TPI U[\Delta] l{l'}{T''}}{l'}
      \byIH \\
    & \lttyping \metalevel{t'}{\TPI U[\Delta] l{l'}{T''}}{l'}
      \byIH \\
    & \lttypwf[\Psi][\Delta] \proglevel{T}l
      \byIH \\
    & \lttypwf[\Psi][\Delta] \proglevel{T'}l
      \byIH \\
    & \lttyping \metalevel{\TAPP{t'}{T'}}{T''[T'/U]}{l'} \\
    & \lttypwf \metalevel{\TPI U[\Delta] l{l'}{T''}}{l'}
      \byIH \\
    & \lttypwf[\Psi, U : \DTyp[\Delta] \proglevel l] \metalevel{T''}{l'}
      \tag{by inversion}
  \end{align*}
  Notice that
  \[
    T/U \approx T'/U
  \]
  We then have the goal by \Cref{lem:dt:gsubst-eq}. 
\end{proof}

\subsection{Coverage and Progress of Recursive Principles}\labeledit{sec:dt:syn:cov}

Before moving to the semantics, let us pause a second and think about the recursive
principles: is it guaranteed to always pick a case from the branches? %
In this section, we would like to positively answer this question. %
The ingredient lies in the typing judgment at layer $\codelevel$ and how the recursive
principle is formulated. %
Recall that the recursive principle for code of terms is
\begin{mathpar}
  \inferrule
  {G_A \\ \typing[L]{l'}\Level \\ \lpjudge \proglevel \Delta \\ \lttypwf[\Psi][\Delta]\proglevel T{l'} \\
    \lttyping \metalevel t{\CTrm[\Delta]T{l'}}{0}}
  {\lttyping \metalevel {\ELIMTRM{l_1}{l_2}M b{l'}\Delta T t}{M_\Trm[l'/\ell,\Delta/g,T/U_T,t/x_t]}{l_2}}
\end{mathpar}
In this rule, we see that the type of $t$ is indexed by $l$, $\Delta$ and $T$, both of
which live at layer $\proglevel$. %
When $t = \boxit{t'}$, then $t'$ must be typed at layer $\codelevel$. %
Then coverage is provided by the exhaustiveness of the branches which should enumerate
all possible types and terms at layer $\codelevel$. %
This is simple as we simply check the syntax at layer $\codelevel$ and can confirm that the
branches are indeed exhaustive. %
Progress, on the other hand, requires both $\Delta$ and $T$ are in the right form
prescribed by the $\beta$ rules. %
For example, the following rule gives the $\beta$ rule for the $\Nat$ case as a term:
\begin{mathpar}
  \inferrule
  {G_A \\ \lpjudge \metalevel \Gamma \\ \lpjudge \proglevel \Delta}
  {\lttyequiv \metalevel {t'_\Nat[\Delta/g]}{\ELIMTRM{l_1}{l_2}M b{1}\Delta{\Ty 0}{(\boxit
        \Nat)}}{M_\Trm[1/\ell,\Delta/g,\Ty 0/U_T,\boxit \Nat/x_t]}{l_2}}
\end{mathpar}
where $T = \Ty 0$ is required. %
In this section, we show that when $t= \boxit{t'}$ where $t'$ is of some concrete form
prescribed by a $\beta$ rule, then the indices must have the right form. %

We first consider the well-formed types at layer $\codelevel$:
\begin{lemma}
  If $\lttypwf \codelevel \Nat l$, then $\tyequiv[L]{l}{0}\Level$. 
\end{lemma}
\begin{proof}
  By induction. The only applicable rules are the well-formedness rule and the
  conversion rule. %
\end{proof}
Similar lemmas can be stated and proved.
\begin{lemma}
  If $\lttypwf \codelevel{\PI l{l'}x S T}{l''}$, then $\lttypwf \codelevel S{l}$ and
  $\lttypwf[\Psi][\Gamma, x : S \at l] \codelevel T{l'}$ and they are sub-derivations of the
  assumption; moreover, $\tyequiv[L]{l''}{l \sqcup l'}\Level$. 
\end{lemma}
\begin{proof}
  Induction. 
\end{proof}
That the judgments in the conclusion are sub-derivations ensures the well-foundedness
of the recursion. %
Effectively, the recursive principles recurse on the structures of the typing
derivations, so they are the most general principles that can be formulated on top the
syntax of MLTT.
\begin{lemma}
  If $\lttypwf \codelevel{\Elt l t}{l'}$, then $\lttyping \codelevel t{\Ty l}{1 + l}$ as a
  sub-derivation and $\tyequiv[L]{l'}{l}\Level$. 
\end{lemma}
\begin{lemma}
  If $\lttypwf \codelevel{\Ty l}{l'}$, then $\tyequiv[L]{l'}{1 + l}\Level$.
\end{lemma}
Now we have exhausted all possible cases for types, so we move on to terms.

\begin{lemma}
  If $\lttyping \codelevel x T l$, then $x : T' \at{l'} \in \Gamma$ and $\lttypeq \proglevel{T}{T'}{l}$
  and $\tyequiv[L]{l}{l'}\Level$. 
\end{lemma}
The statement of lemmas for terms need to also consider the equivalence of types. %
The equivalence of types is implicitly handled by when evaluating the recursive
principles: since the equivalence is at layer $\proglevel$, computation applies, so the
equivalence can be acknowledged by the conversion checking algorithm. %
\begin{lemma}
  If $\lttyping \codelevel \Nat{T}{l}$, then $\lttypeq \proglevel{T}{\Ty 0}{l}$ and $\tyequiv[L]{l}{1}\Level$. 
\end{lemma}
\begin{lemma}
  If $\lttyping \codelevel{\PI l{l'}x s t}{T}{l''}$, then as sub-derivations $\lttyping \codelevel s{\Ty
    l}{1 + l}$ and
  $\lttyping[\Psi][\Gamma, x : \Elt l s \at l] \codelevel{t}{\Ty{l'}}{1 + l'}$, $\lttypeq
  \proglevel{T}{\Ty{l \sqcup l'}}{l''}$ and $\tyequiv[L]{l''}{1 + {(l \sqcup l')}}\Level$. 
\end{lemma}
\begin{lemma}
  If $\lttyping \codelevel{\Ty l}{T}{l'}$, then $\lttypeq \proglevel{T}{\Ty{1 + l}}{l'}$ and
  $\tyequiv[L]{l'}{2 + l}\Level$.
\end{lemma}

\begin{lemma}
  If $\lttyping \codelevel \ze T l$, then $\lttypeq \proglevel{T}{\Nat}{0}$
  and $\tyequiv[L]{l}{0}\Level$.   
\end{lemma}
\begin{lemma}
  If $\lttyping \codelevel{\su t}{T}{l}$, then as a sub-derivation $\lttyping \codelevel t{\Nat}{0}$,
  $\lttypeq \proglevel{T}{\Nat}{0}$ and
  $\tyequiv[L]{l}{0}\Level$.
\end{lemma}
\begin{lemma}
  If $\lttyping \codelevel{\ELIMN l{x.M}{s}{x,y.s'}t}{T}{l'}$, then as sub-derivations
  \begin{itemize}
  \item $\lttypwf[\Psi][\Gamma, x : \Nat \at 0] \codelevel M l$,
  \item $\lttyping \codelevel s {M[\ze/x]}l$,
  \item $\lttyping[\Psi][\Gamma, x : \Nat \at 0, y : M \at l] \codelevel {s'}{M[\su x/x]}l$
  \item $\lttyping \codelevel t \Nat 0$;
  \end{itemize}
  moreover $\lttypeq \proglevel{T}{M[t/x]}{l'}$ and
  $\tyequiv[L]{l'}{l}\Level$.
\end{lemma}
\begin{lemma}
  If $\lttyping \codelevel{\LAM{l}{l'}x S t}{T'}{l''}$, then as sub-derivations
  \begin{itemize}
  \item $\lttypwf \codelevel S l$,
  \item $\lttyping[\Psi][\Gamma, x : S \at l] \codelevel t{T}{l'}$;
  \end{itemize}
  moreover $\lttypeq \proglevel{T'}{\PI{l}{l'}x S T}{l''}$ and
  $\tyequiv[L]{l''}{l \sqcup l'}\Level$.
\end{lemma}
Notice that in this case, we do not have the well-formedness of $T$ as a
sub-derivation. %
This is reflected in the premises for the $t_\lambda$ branch that the global variable $U_T$
representing $T$ lives at layer $\proglevel$. %
In general, a global assumption can live at layer $\codelevel$ and have a recursive call only
if it has a sub-derivation in the typing judgment.
\begin{lemma}
  If $\lttyping \codelevel{\APP t {l}{l'}x S T s}{T'}{l''}$, then as sub-derivations
  \begin{itemize}
  \item $\lttypwf \codelevel S l$,
  \item $\lttypwf[\Psi][\Gamma, x : S \at l] \codelevel{T}{l'}$,
  \item $\lttyping \codelevel t{\PI l{l'} x S T}{l \sqcup l'}$,
  \item $\lttyping \codelevel s S l$;
  \end{itemize}
  moreover $\lttypeq \proglevel{T'}{T[s/x]}{l''}$ and
  $\tyequiv[L]{l''}{l'}\Level$.
\end{lemma}
There is no other possible terms at layer $\codelevel$. %
These lemmas give us a syntactic account of coverage and progress of the recursive
principles. %
In the next section, we give a more rigorous semantic account.

\section{Reduction and Convertibility}

We have finished syntactic verification for \delamlang. %
In this section, let us consider its dynamics by providing the reduction rules for
types and terms and the convertibility checking algorithm between two terms. %
The reduction relations to be given compute the weak head normal forms for
types and terms, respectively, and are sub-relations for equivalence judgments for
types and terms. %
We first give the syntax for weak head normal forms and neutral forms, and then give
the rules for reduction. %
We will need the reduction relations to write down the Kripke logical relations in the
next section as well as in the convertibility checking algorithm. %

\subsection{Weak Head Normal Forms}

The following gives the syntax for weak head normal forms and neutral forms for types
and terms. %
As usual, we use capital case for types and lower case for terms. 

\begin{alignat*}{2}
  W &:= && \ \Nat \sep \PI{l}{l'}x S T \sep \Ty l \sep \UPI \ell l T 
           \tag{Weak head normal form for types}\\
    & &&  \sep \CPI g l T \sep \TPI U l{l'}T \sep \CTyp l \sep \CTrm T l  \\
  V & := &&\ U^\delta \sep \Elt l \nu
  \tag{Neutral form for types}\\
  w &:= &&\ \nu \sep \Nat \sep \PI{l}{l'}x s t \sep \Ty l \sep \ze \sep \su t \sep \LAM l{l'} x S t
           \tag{Weak head normal form for terms ($\Nf$)} \\
    & && \sep \ULAM l \ell t \sep \CLAM l g t \sep \TLAM l{l'} U t \sep \boxit T \sep
         \boxit t \\
  \nu &:= &&\ x \sep u^\delta \sep \ELIMN l{x.M}{s}{x, y.s'}{\nu} \sep \APP \nu l{l'} x S T
           s \sep \UAPP \nu l
           \tag{Neutral form for terms ($\Ne$)} \\
    & && \sep \CAPP \nu \Gamma \sep \TAPP \nu T \sep \LETBTYP{l'}l\Gamma{x_T. M}{U}{t'}\nu
         \sep \LETBTRM{l'}l\Gamma T{x_t. M}{u}{t'}\nu \\
    & && \sep \ELIMTYP{l_1}{l_2}Mbl\Gamma \nu \sep \ELIMTYP{l_1}{l_2}Mbl\Gamma {(\boxit
         {U^\delta})} \\
    & && \sep \ELIMTRM{l_1}{l_2}Mbl\Gamma T \nu \sep \ELIMTRM{l_1}{l_2}Mbl\Gamma T {(\boxit{u^\delta})}
\end{alignat*}

Notice that for the recursive principles, we block on global variables, following
\citet{hu2024layered}. 

\subsection{Reduction Relations}

There are two required reduction relations, one for types and one for terms. %
The one for types simply compute types from encodings via $\tEl$. %
Unlike~\citet{abel_decidability_2017} who employed typed reductions, we deliberately
use untyped reductions and use preservation later to make sure the reductions are
well-defined. %
This deviation from \citet{abel_decidability_2017} requires us to establish enough
syntactic theorems before hand. %
It is particularly important to use untyped reductions because of the way in which the logical relations
relate terms. %
We let $\reds$ to be the reflexive transitive closure of $\redd$. 
\begin{mathpar}
  {{\Elt 0 \Nat} \redd \Nat}

  {{\Elt{1 + l}{\Ty l}} \redd {\Ty l}}
  
  {{\Elt{l \sqcup l'}{\PI l{l'} x s t}}\redd{\PI l{l'} x{\Elt l s}{\Elt{l'} t}}}

  \inferrule
  {t \redd t'}
  {{\Elt{l}t} \redd {\Elt l{t'}}}
\end{mathpar}

The reduction rules for terms are simply the $\beta$ equivalence rules. %
\begin{mathpar}
  {{\ELIMN l{x.M}s{x,y. s'}\ze} \redd {s}}

  {{\ELIMN l{x.M}s{x,y. s'}{(\su t)}} \redd {s'[t/x,\ELIMN l{x.M}s{x,y. s'}t/y]}}

  {{\APP {\LAM {l}{l'} x S t} {l}{l'} x{S} T s} \redd {t[s/x]}}

  {{(\ULAM l \ell t)~\$~\vect l} \redd {t[\vect l/\vect \ell]}}

  {{\CAPP{(\CLAM l g t)}\Delta} \redd {t[\Delta/g]}}

  {{\TAPP{(\TLAM{l}{l'}U t)}{T}} \redd {t[T/U]}}

  {{\LETBTYP{l'}l \Delta{x_T.M}{U}{t'}{(\boxit T)}} \redd {t'[T/U]}}

  {{\LETBTRM{l'}l \Delta T{x_t.M}{u}{t'}{(\boxit t)}} \redd {t'[t/u]}}
\end{mathpar}

The reduction rules for recursors follow the same principle. %
We write down only one rule as an example and omit the rest as they are just
$\beta$ rules:
\begin{mathpar}
  {{\ELIMTYP{l_1}{l_2}M b{0}\Delta{(\boxit
        \Nat)}} \redd {t_\Nat[\Delta/g]}}
\end{mathpar}

The congruence rules reduce the terms at the weak head positions to discover further
redices. %
There are at least one congruence rules for all elimination forms. %
\begin{mathpar}
  \inferrule
  {t \redd t'}
  {{\ELIMN l{x.M}s{x,y. s'}t} \redd {\ELIMN l{x.M}s{x,y. s'}{t'}}}

  \inferrule
  {t \redd t'}
  {{\APP t {l}{l'} x{S} T s} \redd {\APP{t'}{l}{l'} x{S} T s}}
\end{mathpar}
For elimination forms for meta-programming, we have
\begin{mathpar}
  \inferrule
  {t \redd t'}
  {{t~\$~\vect l} \redd {t'~\$~\vect l}}

  \inferrule
  {t \redd t'}
  {{\CAPP{t}\Delta} \redd {\CAPP{t'}\Delta}}

  \inferrule
  {t \redd t'}
  {{\TAPP{t}{T}} \redd {\TAPP{t'}{T}}}

  \inferrule
  {t \redd t'}
  {{\LETBTYP{l'}l \Delta{x_T.M}{U}{s}{t}} \redd {\LETBTYP{l'}l \Delta{x_T.M}{U}{s}{t'}}}

  \inferrule
  {t \redd t'}
  {{\LETBTRM{l'}l \Delta T{x_t.M}{u}{s}{t}} \redd {\LETBTRM{l'}l \Delta T{x_t.M}{u}{s}{t'}}}
\end{mathpar}
There are also congruence rules for the recursive principles. %
For the recursive principle for code of terms, we choose to reduce the type to weak
head normal form first and then reduce the term itself. %
This order is arbitrary and can be flipped. %
We simply fix a choice here.
\begin{mathpar}
  \inferrule
  {t \redd t'}
  {{\ELIMTYP{l_1}{l_2}M b{l'}\Delta t} \redd {\ELIMTYP{l_1}{l_2}M b{l'}\Delta {t'}}}

  \inferrule
  {T \redd T'}
  {{\ELIMTRM{l_1}{l_2}M b{l'}\Delta T t} \redd {\ELIMTRM{l_1}{l_2}M b{l'}\Delta {T'}{t}}}

  \inferrule
  {t \redd t'}
  {{\ELIMTRM{l_1}{l_2}M b{l'}\Delta W t} \redd {\ELIMTRM{l_1}{l_2}M b{l'}\Delta W {t'}}}
\end{mathpar}
We first verify the fact that reductions are just sub-relations of the
equivalence judgments:
\begin{lemma}[Soundness]
  Given $\compt i$,
  \begin{itemize}
  \item if $\lttypwf i T l$ and $T \redd T'$, then $\lttypeq i T{T'} l$;
  \item if $\lttyping i t T l$ and $t \redd t'$, then $\lttyequiv i t{t'} T l$.
  \end{itemize}
\end{lemma}
\begin{proof}
  We proceed by mutual induction on the typing judgments and then invert the reduction
  relations. %
  We use $\compt i$ to make sure computation rules are available. %
  We select a few cases for discussion:
  \begin{itemize}[label=Case]
  \item
    The following is the only possible rule for types:
    \begin{mathpar}
      \inferrule
      {\typing[L] l \Level \\ \lttyping i t{\Ty l}{1 + l}}
      {\lttypwf i{\Elt l t}{l}}
    \end{mathpar}
    Inversion of $\Elt l t \redd T'$ gives four possible subcases. %
    We only consider two:
    \begin{itemize}[label=Subcase]
    \item
      \[
        {\Elt{l_1 \sqcup l_2}{\PI {l_1}{l_2} x s t}}\redd{\PI {l_1}{l_2} x{\Elt{l_1} s}{\Elt{l_2} t}}
      \]
      Then we know
      \[
        \lttyping i{\PI {l_1}{l_2} x s t}{\Ty l}{1 + l}
      \]
      We further do an inner induction on the typing judgment above after generalizing
      $\Ty l$ to some arbitrary $T$. %
      There are only three cases to consider:
      \begin{itemize}[label=Subsubcase]
      \item
        \begin{mathpar}
          \inferrule
          {\typing[L]{l_1}\Level  \\ \typing[L]{l_2}\Level \\
            \lttyping i s{\Ty{l_1}}{1 + l_1} \\ \lttyping[\Psi][\Gamma, x : \Elt{l_1} s \at{l_1}] i t{\Ty{l_2}}{1 + l_2}}
          {\lttyping i{\PI {l_1}{l_2} x s t}{\Ty{l_1 \sqcup l_2}}{1 + {(l_1 \sqcup l_2)}}}
        \end{mathpar}
        Then this case we derive the goal immediately from the $\tEl$ rule for $\Pi$
        types. 
      \item
        \begin{mathpar}
          \inferrule
          {\lttyping i {\PI {l_1}{l_2} x s t}{T'}{l'} \\ \lttypeq{\typeof i}{T}{T'}{l'}}
          {\lttyping i {\PI {l_1}{l_2} x s t}{T}{l'}}
        \end{mathpar}
        In this case, we simply apply the inner IH to obtain the goal. 
      \item
        \begin{mathpar}
          \inferrule
          {\lttyping i {\PI {l_1}{l_2} x s t} T{l'} \\ \tyequiv[L]l{l'}\Level}
          {\lttyping i {\PI {l_1}{l_2} x s t} T l}
        \end{mathpar}
        Similarly, we use the inner IH to obtain the goal. 
      \end{itemize}
      In general, when we know the form of a term, an inner induction must reveal only
      three cases to consider. %
      This pattern appears a lot when we consider cases for types.
      
    \item
      We have this case
      \begin{mathpar}
        \inferrule
        {t \redd t'}
        {{\Elt{l}t} \redd {\Elt l{t'}}}
      \end{mathpar}
      Then by IH, we have
      \[
        \lttyequiv i t{t'}{\Ty l}{1 + l}
      \]
      We obtain the goal by the congruence rule for $\tEl$.
      
    \end{itemize}
    
  \item
    \begin{mathpar}
      \inferrule
      {\lttyping i t{T'}l \\ \lttypeq{\typeof i}T{T'}l}
      {\lttyping i t T l}
    \end{mathpar}
    By IH, we have
    \[
      \lttyequiv i{t}{t'}{T'}l
    \]
    We obtain the goal by the conversion rule.
    
  \item
    \begin{mathpar}
      \inferrule
      {\lttyping i tT{l'} \\ \tyequiv[L]l{l'}\Level}
      {\lttyping i t T l}
    \end{mathpar}
    By IH, we have
    \[
      \lttyequiv i{t}{t'}{T}{l'}
    \]
    We obtain the goal by the conversion rule.

  \item
    For the recursive principle for natural numbers, there are three subcases after
    inverting the reduction premise. %
    We apply $\beta$ rules or the congruence rule properly.

    The same principle applies for the recursive principles for code, but a bit more
    complex. %
    We will use theorems from \Cref{sec:dt:syn:cov} in combination of the congruence
    rules to obtain our goals. 
    
  \end{itemize}
\end{proof}

As a corollary,
\begin{lemma}[Preservation] $ $
  \begin{itemize}
  \item If $\lttypwf i T l$ and $T \redd T'$, then $\lttypwf i {T'} l$.
  \item If $\lttyping i t T l$ and $t \redd t'$, then $\lttyping i {t'} T l$.
  \end{itemize}
\end{lemma}
\begin{proof}
  We analyze $i$. %
  If $i = \varlevel$, then there is no applicable reduction rule. %
  If $i = \codelevel$, then we use lifting to lift $i$ to $\proglevel$, then use the soundness,
  presupposition and unlifting lemmas to obtain the goals. %
  Otherwise, we use the soundness lemma and the presupposition lemma. %
\end{proof}

The substitution lemmas require the well-formedness of types and the well-typedness of
terms to make use of algebraic laws of substitutions.
\begin{lemma}[Universe Substitutions] $ $
  Given $\compt i$,
  \begin{itemize}
  \item if $\lttypwf[\Psi][\Gamma][L'] i T l$, $T \redd T'$ and $\typing[L]{\phi}{L'}$, then $T[\phi] \redd T'[\phi]$;
  \item if $\lttyping[\Psi][\Gamma][L'] i t T l$, $t \redd t'$ and
    $\typing[L]{\phi}{L'}$, then $t[\phi] \redd t'[\phi]$.
  \end{itemize}
\end{lemma}

\begin{lemma}[Local Substitutions] $ $
  Given $\compt i$,
  \begin{itemize}
  \item if $\lttypwf[\Psi][\Delta] i T l$, $T \redd T'$ and $\ltsubst i \delta \Delta$, then $T[\delta] \redd T'[\delta]$;
  \item if $\lttyping[\Psi][\Delta] i t T l$, $t \redd t'$ and
    $\ltsubst i \delta \Delta$, then $t[\delta] \redd t'[\delta]$.
  \end{itemize}
\end{lemma}

\begin{lemma}[Global Substitutions] $ $
  Given $\compt i$,
  \begin{itemize}
  \item if $\lttypwf[\Phi] i T l$, $T \redd T'$ and $\ptyping \sigma \Phi$, then $T[\sigma] \redd T'[\sigma]$;
  \item if $\lttyping[\Phi] i t T l$, $t \redd t'$ and
    $\ptyping \sigma \Phi$, then $t[\sigma] \redd t'[\sigma]$.
  \end{itemize}
\end{lemma}

All lemmas above also work for the reflexive transitive closure versions of
reduction. 

\begin{lemma}[Determinacy] $ $
  \begin{itemize}
  \item If $T \redd T'$ and $T \redd T'$, then $T' = T''$. 
  \item If $t \redd t'$ and $t \redd t''$, then $t' = t''$. 
  \end{itemize}
\end{lemma}

If a multi-step reduction reaches a normal form, then we know this normal form is also
uniquely determined:
\begin{lemma}[Determinacy] $ $
  \begin{itemize}
  \item If $T \reds W$ and $T \reds W'$, then $W = W'$. 
  \item If $t \reds w$ and $t \reds w'$, then $w = w'$. 
  \end{itemize}
\end{lemma}
\begin{proof}
  Induction. Use the fact that weak head normal forms do not reduce and determinacy of
  single-step reduction. 
\end{proof}

Due to preservation, we often are interested in keeping track of well-formedness and
well-typedness of types and terms. %
Therefore it is convenient to give the following convenient auxiliary judgments:
\begin{mathpar}
  \inferrule
  {\lttypwf i T l \\ T \redd T'}
  {\ttypred i T{T'} l}

  \inferrule
  {\lttypwf i T l \\ T \reds T'}
  {\ttypreds i T{T'} l}

  \inferrule
  {\lttyping i t T l \\ t \redd t'}
  {\ttrmred i t{t'} T l}

  \inferrule
  {\lttyping i t T l \\ t \reds t'}
  {\ttrmreds i t{t'} T l}
\end{mathpar}

\subsection{Convertibility Checking}

The convertibility checking is standard: we first reduce types or terms to their
weak head normal forms using reduction, and then recursively compare the
sub-structures. %
Either we detect a mismatch which causes a failure, or everything checks out and the
convertibility is verified.

Following this line, we give the following judgments for convertibility checking. %
Here we always quantify $\compt i$. %
The layering index $i$ restricts only types (i.e. those in MLTT or in \delamlang), but
not terms. %
In other words, it is possible for convertibility checking to relate at layer $\proglevel$ two
terms only well-typed at layer $\metalevel$, as long as these two terms have type well-formed
at layer $\proglevel$ (i.e. MLTT). %
This is a critical property to establish a relation between the logical relations at
both layers. %
\begin{itemize}
\item $\tconvtyp i T{T'} l$ denotes that $T$ and $T'$ are convertible at universe level $l$. 
\item $\tconvtypnf i W{W'} l$ denotes that $W$ and $W'$ are convertible normal types.
\item $\tconvtypne i V{V'}l$ denotes that $V$ and $V'$ are convertible neutral types.
\item $\tconvctx i \Gamma \Delta$ denotes that $\Gamma$ and $\Delta$ are convertible contexts. %
  This judgment is defined by using $\tconvtyp i T{T'} l$ pairwise.
\item $\tconvtrm i t{t'}T l$ denotes that $t$ and $t'$ of type $T$ are convertible.
\item $\tconvtrmnf i w{w'}W l$ denotes that $w$ and $w'$ are convertible normal terms of a normal
  type $W$.
\item $\tconvtrmne i \nu{\nu'}T l$ denotes that $\nu$ and $\nu'$ are convertible
  neutral terms. %
  $T$ is the result of inference. %
\item $\tconvtrmnee i \nu{\nu'}W l$ denotes that $\nu$ and $\nu'$ are convertible neutral terms of
  a normal type $W$. %
  $W$ is the result of inference. 
\item $\tconvsub i{\delta}{\delta'}\Delta$ denotes that $\delta$ and $\delta'$ are convertible
  local substitutions. %
  This judgment is defined by using $\tconvtrm i t{t'}T l$ pairwise. 
\end{itemize}

We give the following convertibility checking rules for types first:
\begin{mathpar}
  \inferrule
  {\ttypreds i T W l \\ \ttypreds i{T'}{W'} l \\ \tconvtypnf i W{W'} l}
  {\tconvtyp i T{T'} l}

  \inferrule
  {\lpjudge i \Gamma}
  {\tconvtypnf i \Nat{\Nat} 0}

  \inferrule
  {\lpjudge i \Gamma \\ \tyequiv[L]{l}{l'}{\Level}}
  {\tconvtypnf i {\Ty l}{\Ty{l'}}{1 + l}}

  \inferrule
  {\tconvtyp i S{S'} l \\ \tconvtyp[\Psi][\Gamma, x : S \at l] i T{T'}{l'}}
  {\tconvtypnf i{\PI l{l'}x{S}{T}}{\PI l{l'}x{S'}{T'}}{l \sqcup l'}}

  \inferrule
  {\tconvtypne iV{V'}{l}}
  {\tconvtypnf i{V}{V'}{l}}

  \inferrule
  {\tconvtypnf i{W}{W'}{l'} \\ \tyequiv{l}{l'}\Level}
  {\tconvtypnf i{W}{W'}{l}}

  \inferrule
  {\lpjudge i \Gamma \\ U : \DTyp[\Delta]{i'} l \in \Psi \\ i' \in \{\codelevel, \proglevel\} \\ i' \le i \\ \tconvsub i \delta{\delta'} \Delta}
  {\tconvtypne i{U^\delta}{U^{\delta'}}{l}}

  \inferrule
  {\tyequiv[L]{l}{l'}\Level \\ \tconvtrmnee i \nu{\nu'}{\Ty l}{1 + l}}
  {\tconvtypne i{\Elt l \nu}{\Elt{l'}{\nu'}}{l}}

  \inferrule
  {\tconvtypne iV{V'}{l'}  \\ \tyequiv[L]{l}{l'}\Level}
  {\tconvtypne iV{V'}{l}}
\end{mathpar}
For the types only available at layer $\metalevel$:
\begin{mathpar}
  \inferrule
  {\lpjudge \metalevel \Gamma \\ \tconvtyp[\Psi][\Gamma][L, \vect\ell] \metalevel{T}{T'}{l} \\ \tyequiv[L,\vect\ell]{l}{l'}{\Level}}
  {\tconvtypnf \metalevel{\UPI \ell l T}{\UPI \ell{l'}{T'}}\omega}
  
  \inferrule
  {\lpjudge \metalevel \Gamma \\ \tconvtyp[\Psi, g : \Ctx] \metalevel{T}{T'}{l} \\ \tyequiv[L]{l}{l'}{\Level}}
  {\tconvtypnf \metalevel{\CPI g l T}{\CPI g{l'}{T'}}{l}}

  \inferrule
  {\lpjudge \metalevel \Gamma \\ \tconvtyp[\Psi, U : \DTyp[\Delta] \proglevel{l_1}] \metalevel{T}{T'}{l_2} \\ \tconvctx \proglevel{\Delta}{\Delta'} \\ \tyequiv[L]{l_1}{l_3}{\Level} \\ \tyequiv[L]{l_2}{l_4}{\Level}}
  {\tconvtypnf \metalevel{\TPI U [\Delta]{l_1}{l_2} T}{\TPI U[\Delta']{l_3}{l_4}{T'}}{l_2}}

  \inferrule
  {\lpjudge \metalevel \Gamma \\ \tconvctx \proglevel{\Delta}{\Delta'} \\ \tyequiv[L]{l}{l'}{\Level}}
  {\tconvtypnf \metalevel{\CTyp[\Delta]l}{\CTyp[\Delta']{l'}}{0}}  

  \inferrule
  {\lpjudge \metalevel \Gamma \\ \tconvctx \proglevel{\Delta}{\Delta'} \\ \tconvtyp[\Psi][\Delta]\proglevel
    T{T'}l \\ \tyequiv[L]{l}{l'}{\Level}}
  {\tconvtypnf \metalevel{\CTrm[\Delta] Tl}{\CTrm[\Delta']{T'}{l'}}{0}}
\end{mathpar}

We propagate the convertibility for types pairwise to obtain the convertibility for local contexts.
\begin{mathpar}
  \inferrule
  {\judge[L] \Psi}
  {\tconvctx i \cdot\cdot}

  \inferrule
  {\judge[L] \Psi \\ g : \Ctx \in \Psi}
  {\tconvctx i{g}{g}}

  \inferrule
  {\tconvctx i \Gamma\Delta \\ \tconvtyp i T{T'} l \\ \tyequiv[L]l{l'}\Level}
  {\tconvctx i{\Gamma, x : T \at l}{\Delta, x : T' \at{l'}}}
\end{mathpar}

The convertibility of terms proceeds similarly. %
The following are checking rules that are available at both layers:
\begin{mathpar}
  \inferrule
  { T \reds W \\ \ttrmreds i t w T l \\ \ttrmreds i
    {t'}{w'}T l \\ \tconvtrmnf i w{w'}W l}
  {\tconvtrm i t{t'}T l}

  \inferrule
  {\lpjudge i \Gamma}
  {\tconvtrmnf i \Nat{\Nat}{\Ty 0}{1}}

  \inferrule
  {\lpjudge i \Gamma \\ \tyequiv[L]{l}{l'}{\Level}}
  {\tconvtrmnf i {\Ty l}{\Ty{l'}}{\Ty{1 + l}}{2 + l}}

  \inferrule
  {\tconvtrm i s{s'}{\Ty l}{1 + l} \\ \tconvtrm[\Psi][\Gamma, x : \Elt l s \at l] i t{t'}{\Ty{l'}}{1 + l'}}
  {\tconvtrmnf i{\PI l{l'}x{s}{t}}{\PI l{l'}x{s'}{t'}}{\Ty{l \sqcup l'}}{1 + {(l \sqcup l')}}}

  \inferrule
  {\lpjudge i \Gamma}
  {\tconvtrmnf i \ze \ze \Nat 0}

  \inferrule
  {\tconvtrmnf i t{t'}\Nat 0}
  {\tconvtrmnf i{\su t}{\su{t'}}\Nat 0}

  \inferrule
  {\tconvtrmnee i \nu{\nu'}{\Nat} 0}
  {\tconvtrmnf i \nu{\nu'}\Nat 0}

  \inferrule
  {\lttypwf i S l \\ \lttyping i w {\PI{l}{l'}x{S}{T}}{l \sqcup l'} \\ \lttyping i {w'}{\PI{l}{l'}x{S}{T}}{l \sqcup l'} \\ \tconvtrm[\Psi][\Gamma, x : S \at{l}] i {\APP{w}{l}{l'}x{S}{T}{x}}{\APP{w'}{l}{l'}x{S}{T}{x}}{T}{l'}}
  {\tconvtrmnf i{w}{w'}{\PI{l}{l'}x{S}{T}}{l \sqcup l'}}

  \inferrule
  {\tconvtrmnee i \nu{\nu'}{W} l}
  {\tconvtrmnf i \nu{\nu'}V l}

  \inferrule
  {\tconvtrmnf i w{w'}W{l'} \\ \tyequiv[L]{l}{l'}\Level}
  {\tconvtrmnf i w{w'}W l}
\end{mathpar}

The following rules check terms that are available only at layer $\metalevel$:
\begin{mathpar}
  \inferrule
  {\lttyping \metalevel w {\UPI \ell l T}{\omega} \\ \lttyping \metalevel {w'}{\UPI \ell l T}{\omega} \\ \tconvtrm[\Psi][\Gamma][L, \vect\ell] \metalevel{\UAPP t{\ell}}{\UAPP{t'}{\ell}}{T}{l}}
  {\tconvtrmnf \metalevel{w}{w'}{\UPI \ell l T}{\omega}}

  \inferrule
  {\lttyping \metalevel w{\CPI g l T}{l} \\ \lttyping \metalevel {w'}{\CPI g l T}{l} \\ \tconvtrm[\Psi, g : \Ctx] \metalevel{\CAPP w{g}}{\CAPP{w'}{g}}{T}{l}}
  {\tconvtrmnf \metalevel{w}{w'}{\CPI g l T}{l}}

  \inferrule
  {\lttyping \metalevel w {\TPI U [\Delta]{l}{l'} T}{l'} \\ \lttyping \metalevel {w'}{\TPI U [\Delta]{l}{l'} T}{l'} \\ \tconvtrm[\Psi, U : \DTyp[\Delta] \proglevel{l}] \metalevel{\TAPP w{U^{\id_\Delta}}}{\TAPP{w'}{U^{\id_\Delta}}}{T}{l'}}
  {\tconvtrmnf \metalevel{w}{w'}{\TPI U [\Delta]{l}{l'} T}{l'}}

  \inferrule
  {\lpjudge \metalevel \Gamma \\ \lttypwf[\Psi][\Delta] \codelevel T l \\ T = T'}
  {\tconvtrmnf \metalevel{\boxit T}{\boxit{T'}}{\CTyp[\Delta] l}{0}}

  \inferrule
  {\tconvtrmnee \metalevel \nu{\nu'}{\CTyp[\Delta] l}{0}}
  {\tconvtrmnf \metalevel \nu{\nu'}{\CTyp[\Delta] l}{0}}

  \inferrule
  {\lpjudge \metalevel \Gamma \\ \lttyping[\Psi][\Delta] \codelevel t T l \\ t = t'}
  {\tconvtrmnf \metalevel{\boxit t}{\boxit{t'}}{\CTrm[\Delta] T l}{0}}

  \inferrule
  {\tconvtrmnee \metalevel \nu{\nu'}{\CTrm[\Delta] T l}{0}}
  {\tconvtrmnf \metalevel \nu{\nu'}{\CTrm[\Delta] T l}{0}}
\end{mathpar}
Notice here convertibility of $\tbox$'ed types and terms are checked simply with
syntactic equality. %
The convertibility of neutral terms proceeds as follows. %
Similarly, we first give the checking rules that are available at both layers:
\begin{mathpar}
  \inferrule
  {\tconvtrmne i \nu{\nu'}{T}{l} \\ T \reds W}
  {\tconvtrmnee i \nu{\nu'}W l}

  \inferrule
  {\tconvtrmne i \nu{\nu'}T{l'} \\ \tyequiv[L]{l}{l'}\Level}
  {\tconvtrmne i \nu{\nu'}T l}

  \inferrule
  {\lpjudge i \Gamma \\ x : T \at l \in \Gamma}
  {\tconvtrmne i x{x}T l}

  \inferrule
  {\lpjudge i \Gamma \\ u : \DTrm[\Delta]{i'}T l \in \Psi \\ i' \in \{\varlevel, \codelevel\} \\ i' \le i \\ \tconvsub i \delta{\delta'} \Delta}
  {\tconvtrmne i{u^\delta}{u^{\delta'}}{T[\delta]}{l}}

  \inferrule
  {\tyequiv[L]l{l'}\Level \\ \tconvtyp[\Psi][\Gamma, x : \Nat \at 0] i M{M'} l \\
    \tconvtrm i {s_1}{s_3} {M[\ze/x]}l \\
    \tconvtrm[\Psi][\Gamma, x : \Nat \at 0, y : M \at l] i{s_2}{s_4}{M[\su x/x]}l \\
    \tconvtrmnee i \nu{\nu'} \Nat 0}
  {\tconvtrmne i{\ELIMN l{x.M}{s_1}{x,y. s_2}\nu}{\ELIMN{l'}{x.M'}{s_3}{x,y. s_4}{\nu'}}{M[\nu/x]}{l}}

  \inferrule
  {\tyequiv[L]{l_1}{l_3}\Level  \\ \tyequiv[L]{l_2}{l_4}\Level \\ \tconvtyp i S{S'}{l_1}
    \\
    \tconvtyp[\Psi][\Gamma, x : S \at{l_1}] i{T}{T'}{l_2} \\
    \tconvtrmnee i \nu{\nu'}{\PI {l_1}{l_2} x{S''}{T''}}{l_1 \sqcup l_2} \\ \tconvtrm i s{s'} S{l_1}}
  {\tconvtrmne i{\APP \nu {l_1}{l_2} x S T s}{\APP{\nu'}{l_3}{l_4} x {S'}{T'}{s'}}{T[s/x]}{l_2}}
\end{mathpar}
When checking applications, we simply ignore the type inferred by checking $\nu$ and
$\nu'$. %
This is fine because we already know $\nu$ and $\nu'$ are well-typed so the type
annotations must be equivalent.

Then we give the rules only available at layer $\metalevel$:
\begin{mathpar}
  \inferrule
  {\tconvtrmnee \metalevel{\nu}{\nu'}{\UPI \ell l T}{\omega} \\ |\vect\ell| = |\vect l| = |\vect
    l'| > \codelevel \\
    \forall \codelevel \le n < |\vect l| ~.~ \tyequiv[L]{\vect l(n)}{\vect l'(n)}\Level}
  {\tconvtrmne \metalevel{\UAPP \nu l}{\nu'~\$~\vect l'}{T[\vect l/\vect \ell]}{l[\vect l/\vect \ell]}}

  \inferrule
  {\tconvtrmnee \metalevel{\nu}{\nu'}{\CPI g l T}{l} \\ \tconvctx \proglevel \Delta{\Delta'}}
  {\tconvtrmne \metalevel{\CAPP \nu \Delta}{\CAPP{\nu'}{\Delta'}}{T[\Delta/g]}{l}}

  \inferrule
  {\tconvtrmnee \metalevel{\nu}{\nu'}{\TPI U[\Delta] l{l'}{T''}}{l'} \\ \tconvtyp[\Psi][\Delta] \proglevel{T}{T'}l}
  {\tconvtrmne \metalevel{\TAPP \nu{T}}{\TAPP{\nu'}{T'}}{T''[T/U]}{l'}}

  \inferrule
  {\lpjudge \metalevel \Gamma \\ \tyequiv[L]{l_1}{l_3}\Level \\ \tyequiv[L]{l_2}{l_4}\Level \\ \tconvctx \proglevel \Delta{\Delta'} \\
    \tconvtrmnee \metalevel{\nu}{\nu'}{\CTyp[\Delta]{l_2}}{0} \\
    \tconvtyp[\Psi][\Gamma,x_T : \CTyp[\Delta]{l_2} \at{0}]\metalevel{M}{M'}{l_1} \\
  \tconvtrm[\Psi, U : \DTyp[\Delta]\codelevel{l_2}]\metalevel{t_1}{t_2}{M[\boxit U/x_T]}{l_1}}
  {\tconvtrmne \metalevel{\LETBTYP{l_1}{l_2} \Delta{x_T.M}{U}{t_1}\nu}{\LETBTYP{l_3}{l_4}{\Delta'}{x_T.M'}{U}{t_2}{\nu'}}{M[t/x_T]}{l_1}}

  \inferrule
  {\lpjudge \metalevel \Gamma \\ \tyequiv[L]{l_1}{l_3}\Level \\ \tyequiv[L]{l_2}{l_4}\Level \\ \tconvctx \proglevel \Delta{\Delta'} \\
    \tconvtyp \proglevel T{T'}{l_2} \\
    \tconvtrmnee \metalevel{\nu}{\nu'}{\CTrm[\Delta]T{l_2}}{0} \\
    \tconvtyp[\Psi][\Gamma,x_T : \CTrm[\Delta]T{l_2} \at{0}]\metalevel{M}{M'}{l_1} \\
    \tconvtrm[\Psi, u : \DTrm[\Delta]\codelevel T{l_2}]\metalevel{t_1}{t_2}{M[\boxit u/x_t]}{l_1}}
  {\tconvtrmne \metalevel{\LETBTRM{l_1}{l_2} \Delta{T}{x_t.M}{u}{t_1}\nu}{\LETBTRM{l_3}{l_4}{\Delta'}{T'}{x_T.M'}{u}{t_2}{\nu'}}{M[t/x_t]}{l_1}}
\end{mathpar}
The remaining piece of the convertibility checking for neutral recursive principles. %
The recursive principles get stuck when the scrutinees are neutral or $\tbox$'ed
global variables. %
To check the convertibility of neutral recursive principles, we recursively check the
convertibility between motives, corresponding branches and the indexing universe
levels, local contexts and potentially types. %
To derive the following two conclusions:
\begin{gather*}
  \tconvtrmne \metalevel{\ELIMTYP{l_1}{l_2}Mbl\Delta
    \nu}{\ELIMTYP{l_3}{l_4}{M'}{b'}{l'}{\Delta'}{\nu'}}{M_\Typ[\nu/x_T]}{l_1} \\
  \tconvtrmne \metalevel{\ELIMTRM{l_1}{l_2}Mbl\Delta T \nu}{\ELIMTYP{l_3}{l_4}{M'}{b'}{l'}{\Delta'}{T'}{\nu'}}{M_\Trm[\nu/x_t]}{l_1}
\end{gather*}
We proceed by checking the convertibility of motives:
\begin{mathpar}
  \tyequiv[L]{l_1}{l_3}\Level

  \tyequiv[L]{l_2}{l_4}\Level
  
  \tconvtyp[\Psi, g : \Ctx][\Gamma, x_T : \CTyp[g]\ell \at 0][L,\ell] \metalevel {M_\Typ}{M_\Typ'}{l_1}

  \tconvtyp[\Psi, g : \Ctx, U_T : \DTyp[g]\proglevel\ell][\Gamma, x_t : \CTrm[g]{U_T^\id}\ell
  \at 0][L,\ell] \metalevel{M_\Trm}{M_\Trm'}{l_2}
\end{mathpar}
We do the same for all the branches as well. %
Following the previous conventions, we group all these checking into $C_A$ for
\textbf{\codelevel}onvertibility checking for all premises. %
Then what we have left is to make sure the scrutinees are convertible.
\begin{mathpar}
  \inferrule
  {C_A \\ \tyequiv[L]{l}{l'}\Level \\ \tconvctx \proglevel{\Delta}{\Delta'} \\ \tconvtrmnee \metalevel{\nu}{\nu'}{\CTyp[\Delta]{l}}{0}}
  {\tconvtrmne \metalevel{\ELIMTYP{l_1}{l_2}Mbl\Delta \nu}{\ELIMTYP{l_3}{l_4}{M'}{b'}{l'}{\Delta'}{\nu'}}{M_\Typ[\nu/x_T]}{l_1}}

  \inferrule
  {C_A \\ \tyequiv[L]{l}{l'}\Level \\ \tconvctx \proglevel{\Delta}{\Delta'} \\ \tconvtyp \proglevel
    T{T'}{1 + l} \\ \tconvtrmnee \metalevel{\nu}{\nu'}{\CTrm[\Delta]T{l}}{0}}
  {\tconvtrmne \metalevel{\ELIMTRM{l_1}{l_2}Mbl\Delta T \nu}{\ELIMTYP{l_3}{l_4}{M'}{b'}{l'}{\Delta'}{T'}{\nu'}}{M_\Trm[\nu/x_t]}{l_1}}
\end{mathpar}

If the scrutinees are $\tbox$'ed global variables, then the check is always the same,
except that the global variables are compared syntactically:
\begin{mathpar}
  \inferrule
  {C_A \\ \tyequiv[L]{l}{l'}\Level \\ \tconvctx \proglevel{\Delta}{\Delta'} \\ U :
    \DTyp[\Delta]{\codelevel} l \in \Psi  \\ \ltsubst \codelevel \delta \Delta}
  {\tconvtrmne \metalevel{\ELIMTYP{l_1}{l_2}Mbl\Delta{(\boxit{U^\delta})}}{\ELIMTYP{l_3}{l_4}{M'}{b'}{l'}{\Delta'}{(\boxit{U^\delta})}}{M_\Typ[\boxit{U^\delta}/x_T]}{l_1}}

  \inferrule
  {C_A \\ \tyequiv[L]{l}{l'}\Level \\ \tconvctx \proglevel{\Delta}{\Delta'} \\ \tconvtyp \proglevel
    T{T'}{1 + l} \\ u : \DTrm[\Delta]{i'}T l \in \Psi \\ i' \in \{\varlevel, \codelevel\} \\ \ltsubst \codelevel \delta \Delta}
  {\tconvtrmne \metalevel{\ELIMTRM{l_1}{l_2}Mbl\Delta T{(\boxit{u^\delta})}}{\ELIMTYP{l_3}{l_4}{M'}{b'}{l'}{\Delta'}{T'}{(\boxit{u^\delta})}}{M_\Trm[\boxit{u^\delta}/x_t]}{l_1}}
\end{mathpar}
Now we have finished all the convertibility rules for neutral terms.

We simply let the convertibility for terms to propagate pairwise to derive the
convertibility for local substitutions:
\begin{mathpar}
  \inferrule
  {\lpjudge i\Gamma \\ \text{$\Gamma$ ends with $\cdot$} \\ |\Gamma| = \metalevel}
  {\tconvsub i {\cdot^\metalevel}{\cdot^\metalevel}{\cdot}}

  \inferrule
  {\lpjudge i\Gamma \\ g : \Ctx \in \Psi \\ \text{$\Gamma$ ends with $g$} \\ |\Gamma| = \metalevel}
  {\tconvsub i {\cdot_g^\metalevel}{\cdot_g^\metalevel}{\cdot}}

  \inferrule
  {\lpjudge i\Gamma \\ g : \Ctx \in \Psi \\\\ \text{$\Gamma$ ends with $g$} \\ |\Gamma| = \metalevel}
  {\tconvsub i {\wk_g^\metalevel}{\wk_g^\metalevel}{g}}

  \inferrule
  {\tconvsub i {\delta}{\delta'}{\Delta} \\\\ \lttypwf[\Psi][\Delta]i T l \\ \tconvtrm i {t}{t'}{T[\delta]} l}
  {\tconvsub i {\delta, t/x}{\delta', t'/x}{\Delta, x : T \at l}}
\end{mathpar}

The convertibility algorithm is obtained by reading all the components for
convertibility rules as inputs and the neutral judgments consider types as outputs. %
If there is no corresponding rule, then two terms are not convertible; otherwise, two
terms are convertible. %
We verify some basic properties as follows:
\begin{lemma}[Soundness] Assuming $\compt i$,
  \begin{itemize}
  \item if $\tconvtyp i T{T'} l$, then $\lttypeq i T{T'} l$;
  \item if $\tconvtypnf i W{W'} l$, then $\lttypeq i W{W'} l$;
  \item if $\tconvtypne i V{V'}l$, then $\lttypeq i{V}{V'}l$;
  \item if $\tconvctx i \Gamma \Delta$, then $\lpequiv i \Gamma \Delta$;
  \item if $\tconvtrm i t{t'}T l$, then $\lttyequiv i t{t'} T l$;
  \item if $\tconvtrmnf i w{w'}W l$, then $\lttyequiv i w{w'} W l$;
  \item if $\tconvtrmne i \nu{\nu'}T l$, then $\lttyequiv i\nu{\nu'} T l$;
  \item if $\tconvtrmnee i \nu{\nu'}W l$, then $\lttyequiv i \nu{\nu'} W l$;
  \item if $\tconvsub i{\delta}{\delta'}\Delta$, then $\ltsubeq i
    {\delta}{\delta'}{\Delta}$. 
  \end{itemize}
\end{lemma}
\begin{proof}
  Mutual induction.  %
  Use $\eta$ rules for all kinds of function types. %
  Use congruence rules, presupposition and conversion rules when checking neutral
  terms. %
\end{proof}

\begin{lemma} $ $
  \begin{itemize}
  \item If $\tconvtypnf i W{W'} l$, then $\tconvtyp i W{W'} l$. 
  \item If $\tconvtrmnf i w{w'}W l$, then $\tconvtrm i w{w'} W l$. 
  \end{itemize}
\end{lemma}

Other lemmas like PER require the fundamental theorems so we postpone their proofs
until we have the semantic models. 

% \begin{lemma}[PER]
%   All convertibility judgments form PERs.
% \end{lemma}

\section{Logical Relations for \delamlang}\labeledit{sec:dt:logrel}

Previously, we have given the judgments of \delamlang, verified its syntactic
properties and given its reduction and convertibility algorithms. %
Starting this section, we establish the logical relation and prove the (weak)
normalization and convertibility properties of \delamlang. %
Following \citet{abel_decidability_2017}, we proceeds as follows:
\begin{itemize}
\item First we give a set of generic equivalence conditions for a parameterized
  discussion of the logical relation.
\item Then we give the definition of the Kripke logical relations of types and
  terms. %
  The logical relations are parameterized by layers. %
  In this step, we are only concerned about types that are available at all layers,
  i.e. those in MLTT and unrelated to meta-programming. 
\item Then we give the definition  of the Kripke logical relations of local contexts
  and local substitutions. %
\item Then we branch off two orthogonal developments.
  \begin{itemize}
  \item We give the definition of the Kripke logical relations of global contexts and
    global substitutions.
  \item We give the definition of the Kripke logical relations of types and terms,
    again. %
    But in this case, we must also give the definition for types that are related to
    meta-programming, i.e. contextual types. %
  \end{itemize}
  In fact, the definitions given by the two sub-steps above must consider each
  other. %
  Otherwise, we will not able to extend related global substitutions during the proof
  of the fundamental theorems. 
\item Next we give the semantic judgments. %
  The semantic judgments require types, terms, etc. to be stable under all universe,
  global and local substitutions. %
\item Finally, we establish the fundamental theorems for the semantic judgments. %
  Instantiating the generic equivalence gives us the proof of convertibility.
\end{itemize}
Due to layering, following \Cref{sec:cv:logrel}, the generic equivalence, logical
relation and validity judgments are all layered. %
In fact, since computation exists at both layers $\proglevel$ and $\metalevel$, the situation is very
complex. %
\citet{abel_decidability_2017} instantiate their generic equivalence twice to obtain
the decidability of convertibility checking, and we will also be doing the same. %
Due to the complication of layering, our fundamental theorems must talk about all
layers. %
The difficulties of the logical relations lie in that how we can support code running
and recursions on code at the same time and justify them in the semantics. %

\subsection{Generic Equivalence}

Similar to \Cref{sec:cv:logrel}, we first quantify four generic equivalence relations,
which will be instantiated to syntactic equivalence and convertibility later, and
their laws. %
This step provides modularity to logical relation argument: we simply instantiate this
generic equivalence to obtain different versions of the fundamental theorems. %
Due to dependent types, we define generic equivalence over $i \in \{\proglevel, \metalevel\}$:
\begin{itemize}
\item $\lttypgneeq i V{V'} l$ describes a generic type equivalence between two neutral
  types at universe level $l$ at layer $i$.
\item $\lttypgeq i {T}{T'} l$ describes a generic type equivalence between two types at
  universe level $l$ at layer $i$.
\item $\lttrmgneeq i \nu{\nu'} T l$ describes a generic type equivalence between two neutral
  terms of type $T$ at universe level $l$ at layer $i$. 
\item $\lttrmgeq i t{t'} T l$ describes a generic type equivalence between two
  terms of type $T$ at universe level $l$ at layer $i$. 
\end{itemize}

From the four generic equivalence, we induce two equivalence of local contexts and
local substitutions by using the generic equivalence pairwise:
\begin{mathpar}
  \inferrule
  {\judge[L] \Psi}
  {\ltctxgeq i \cdot\cdot}

  \inferrule
  {\judge[L] \Psi \\ g : \Ctx \in \Psi}
  {\ltctxgeq i{g}{g}}

  \inferrule
  {\ltctxgeq i \Gamma\Delta \\ \lttypgeq i T{T'} l \\ \tyequiv[L]l{l'}\Level}
  {\ltctxgeq i{\Gamma, x : T \at l}{\Delta, x : T' \at{l'}}}

  \inferrule
  {\lpjudge i\Gamma \\ \text{$\Gamma$ ends with $\cdot$} \\ |\Gamma| = \metalevel}
  {\ltsubgeq i {\cdot^\metalevel}{\cdot^\metalevel}{\cdot}}

  \inferrule
  {\lpjudge i\Gamma \\ g : \Ctx \in \Psi \\ \text{$\Gamma$ ends with $g$} \\ |\Gamma| = \metalevel}
  {\ltsubgeq i {\cdot_g^\metalevel}{\cdot_g^\metalevel}{\cdot}}

  \inferrule
  {\lpjudge i\Gamma \\ g : \Ctx \in \Psi \\\\ \text{$\Gamma$ ends with $g$} \\ |\Gamma| = \metalevel}
  {\ltsubgeq i {\wk_g^\metalevel}{\wk_g^\metalevel}{g}}

  \inferrule
  {\ltsubgeq i {\delta}{\delta'}{\Delta} \\\\ \lttypwf[\Psi][\Delta]i T l \\
    \lttrmgeq i {t}{t'}{T[\delta]} l}
  {\ltsubgeq i {\delta, t/x}{\delta', t'/x}{\Delta, x : T \at l}}
\end{mathpar}

The generic equivalence and the logical relations are invariant under all
weakenings. %
Therefore, we should make these notions clear here. %
Since there are three different contexts, we have three corresponding kinds of
weakenings. %
In particular, $\theta :: L \To L'$ is the universe weakening. %
Following previous conventions, $\gamma :: L \sep \Psi \To_g \theta$ is a global
weakening, and $\tau :: L \sep \Psi;\Gamma \To_i \Delta$ is a local weakening. %
The subscript $i$ denotes which layer the contexts $\Gamma$ and $\Delta$ live in. %
We can simultaneously weaken all three contexts at the same time. %
We simply apply $\theta$, $\gamma$ and $\tau$ in this order. %
We let $\psi$ represent this triple:
\[
  \psi := (\theta, \gamma, \tau) :: L \sep \Psi; \Gamma \To_i L' \sep \Phi;\Delta
\]
where
\begin{align*}
  \theta &:: L \To L' \\
  \gamma &:: L' \sep \Psi[\theta] \To_g \Phi \\
  \tau &:: L' \sep \theta; \Gamma[\theta][\gamma] \To_i \Delta
\end{align*}
Similarly we let
\[
  \alpha := (\theta, \gamma) :: L \sep \Psi \To L' \sep \Phi
\]
We can apply weakenings like substitutions to universe levels, types, terms, contexts
and substitutions as expected. %
The action is to shift the variables according to the specified weakenings. %
This is a standard action, despite having three separate notions, so we take it for
granted here. %
When it is clear from the context, we do not write down the weakening action at all to
avoid clutter. 

Then we give the laws of the generic equivalence. %
Since the generic equivalence at layer $\metalevel$ subsumes that at layer $\proglevel$, we first give
the laws that hold for both layers, and then incrementally add those that only hold at
layer $\metalevel$.

\begin{law}[Subsumption]$ $
  \begin{itemize}
  \item If $\lttypgneeq i V{V'} l$, then $\lttypgeq i V {V'} l$. 
  \item If $\lttypgeq i T{T'} l$, then $\lttypeq i T{T'} l$. 
  \item If $\lttrmgneeq i \nu{\nu'} T l$, then $\lttrmgeq i \nu{\nu'} T l$.
  \item If $\lttrmgeq i t{t'} T l$, then $\lttyequiv i t{t'} T l$.
  \end{itemize}
\end{law}

As a lemma, subsumption propagates to contexts and local substitutions:
\begin{lemma}[Subsumption] $ $
  \begin{itemize}
  \item If $\ltctxgeq i \Gamma \Delta$, then $\lpequiv i \Gamma \Delta$. 
  \item If $\ltsubgeq i \delta {\delta'} \Delta$, then $\ltsubeq i \delta {\delta'} \Delta$.  
  \end{itemize}
\end{lemma}

Due to subsumption, we know that components in generic equivalence are well-formed or
well-typed:
\begin{lemma}[Presupposition] $ $
  \begin{itemize}
  \item If $\lttypgeq i T{T'} l$, then $\lttypwf i T l$ and $\lttypwf i{T'} l$.
  \item If $\lttrmgeq i t{t'} T l$, then $\lttyping i t T l$ and $\lttyping i{t'}T
    l$. 
  \end{itemize}
\end{lemma}
\begin{proof}
  By subsumption and presupposition. 
\end{proof}

\begin{law}[PER]
  All four relations are PERs. 
\end{law}

\begin{law}[Type Conversion] $ $
  \begin{itemize}
  \item If $\lttrmgneeq i \nu {\nu'} T l$ and $\lttypeq i T{T'} l$, then $\lttrmgneeq i \nu {\nu'}{T'} l$.
  \item If $\lttrmgeq i t {t'} T l$ and $\lttypeq i T{T'} l$, then $\lttrmgeq i t
    {t'}{T'} l$.
  \end{itemize}
\end{law}

\begin{law}[Context Equivalence] $ $
  \begin{itemize}
  \item If $\lttypgneeq i V {V'} l$, $\gequiv\Phi\Psi$ and
    $\lpequiv[\Phi]{i}\Delta\Gamma$, then $\lttypgneeq[\Phi][\Delta] i V {V'} l$.
  \item If $\lttypgeq i T{T'} l$, $\gequiv\Phi\Psi$ and
    $\lpequiv[\Phi]{i}\Delta\Gamma$, then $\lttypgeq[\Phi][\Delta] i T{T'} l$.
  \item If $\lttrmgneeq i \nu {\nu'} T l$, $\gequiv\Phi\Psi$ and
    $\lpequiv[\Phi]{i}\Delta\Gamma$, then $\lttrmgneeq[\Phi][\Delta] i \nu {\nu'}{T} l$.
  \item If $\lttrmgeq i t {t'} T l$, $\gequiv\Phi\Psi$ and
    $\lpequiv[\Phi]{i}\Delta\Gamma$, then $\lttrmgeq[\Phi][\Delta] i t {t'}{T} l$.
  \end{itemize}
\end{law}

\begin{law}[Weakening] $ $
  \begin{itemize}
  \item If $\lttypgneeq i V {V'} l$ and
    $\psi :: L \sep \Psi; \Gamma \To_i L' \sep \Phi;\Delta$, then
    $\lttypgneeq[\Phi][\Delta][L'] i{V}{V'}{l}$.
  \item If $\lttypgeq i T{T'} l$ and
    $\psi :: L \sep \Psi; \Gamma \To_i L' \sep \Phi;\Delta$, then
    $\lttypgeq[\Phi][\Delta][L'] i{T}{T'}{l}$.
  \item If $\lttrmgneeq i \nu {\nu'} T l$ and
    $\psi :: L \sep \Psi; \Gamma \To_i L' \sep \Phi;\Delta$, then
    $\lttrmgneeq[\Phi][\Delta][L'] i{\nu}{\nu'}{T}{l}$.
  \item If $\lttrmgeq i t {t'} T l$ and
    $\psi :: L \sep \Psi; \Gamma \To_i L' \sep \Phi;\Delta$, then
    $\lttrmgeq[\Phi][\Delta][L'] i{t}{t'}{T}{l}$.
  \end{itemize}  
\end{law}

\begin{law}[Weak Head Closure] $ $
  \begin{itemize}
  \item If $\ttypreds i T W l$, $\ttypreds i{T'}{W'} l$ and $\lttypgeq i W{W'} l$, then $\lttypgeq i T{T'} l$. 
  \item If $\ttrmreds i t w T l$, $\ttrmreds i{t'}{w'}T l$ and $\lttrmgeq i w{w'} T l$, then $\lttrmgeq i t{t'} T l$.
  \end{itemize}
\end{law}

\begin{law}[Type Constructors]
  If $\lpjudge i \Gamma$, 
  \begin{itemize}
  \item if $\tyequiv[L] l{l'}\Level$, then $\lttypgeq i{\Ty l}{\Ty l'}{1 + l}$ and
    $\lttrmgeq i{\Ty l}{\Ty l'}{\Ty{1 + l}}{2 + l}$;
  \item then $\lttypgeq i{\Nat}{\Nat}0$ and $\lttrmgeq i{\Nat}{\Nat}{\Ty 0}{1}$;
  \item if $\lttypgeq i S{S'} l$ and
    $\lttypgeq[\Psi][\Gamma, x : S \at l] i T{T'}{l'}$, then \newline
    $\lttypgeq i{\PI{l}{l'}x S T}{\PI{l}{l'}x{S'}{T'}}{l \sqcup l'}$;
  \item if $\lttrmgeq i s{s'}{\Ty l}{1 + l}$ and $\lttrmgeq[\Psi][\Gamma, x : \Elt l s \at l] i
    t{t'}{\Ty{l'}}{1 + l'}$, then \newline
    $\lttrmgeq i{\PI{l}{l'}x s t}{\PI{l}{l'}x{s'}{t'}}{\Ty{l \sqcup
      l'}}{1 + {(l \sqcup l')}}$.
  \end{itemize}
\end{law}

\begin{law}[Neutral Types]  $ $
  \begin{itemize}
  \item If $\lpjudge i \Gamma$, $U : \DTyp[\Delta]{i'} l \in \Psi$, $i' \in \{\codelevel, \proglevel\}$,
    $i' \le i$ and
    $\ltsubgeq i \delta{\delta'} \Delta$, then
    $\lttypgneeq i{U^\delta}{U^{\delta'}}{l}$. 
  \item If $\tyequiv[L]{l}{l'}\Level$ and $\lttrmgneeq i \nu{\nu'}{\Ty l}{1 + l}$,
    then $\lttypgneeq i{\Elt l \nu}{\Elt{l'}{\nu'}}{l}$.
  \end{itemize}
\end{law}

\begin{law}[Congruence] $ $
  \begin{itemize}
  \item If $\lpjudge i \Gamma$, then $\lttrmgeq i \ze \ze \Nat 0$.
  \item If $\lttrmgeq i t{t'}\Nat 0$, then $\lttrmgeq i {\su t}{\su{t'}}\Nat 0$.
  \item If $\lttypwf i S l$, $\lttyping i t {\PI{l}{l'}x{S}{T}}{l \sqcup l'}$,
    $\lttyping i {t'}{\PI{l}{l'}x{S}{T}}{l \sqcup l'}$ and \newline
    $\lttrmgeq[\Psi][\Gamma, x : S \at{l}] i
    {\APP{t}{l}{l'}x{S}{T}{x}}{\APP{t'}{l}{l'}x{S}{T}{x}}{T}{l'}$, then \newline
    $\lttrmgeq i t{t'}{\PI{l}{l'}x{S}{T}}{l \sqcup l'}$.
  \end{itemize}
\end{law}

\begin{law}[Congruence for Neutrals] $ $
  \begin{itemize}
  \item If $\lpjudge i \Gamma$ and $x : T \at l \in \Gamma$, then $\lttrmgneeq i x x
    T l$. 
  \item If $\lpjudge i \Gamma$, $u : \DTrm[\Delta]{i'}T l \in \Psi$,
    $i' \in \{\varlevel, \codelevel\}$, $i' \le i$ and $\ltsubgeq i \delta{\delta'} \Delta$, then
    \newline $\lttrmgneeq i{u^\delta}{u^{\delta'}} {T[\delta]} l$.
  \item If $\tyequiv[L]l{l'}\Level$, $\lttypgeq[\Psi][\Gamma, x : \Nat \at 0] i M{M'} l$,
    $\lttrmgeq i {s_1}{s_3} {M[\ze/x]}l$, \newline
    $\lttrmgeq[\Psi][\Gamma, x : \Nat \at 0, y : M \at l] i{s_2}{s_4}{M[\su
      x/x]}l$ and
    $\lttrmgneeq i \nu{\nu'} \Nat 0$, then \newline
    $\lttrmgneeq i{\ELIMN l{x.M}{s_1}{x,y. s_2}\nu}{\ELIMN{l'}{x.M'}{s_3}{x,y. s_4}{\nu'}}{M[\nu/x]}{l}$.
  \item If $\tyequiv[L]{l_1}{l_3}\Level$, $\tyequiv[L]{l_2}{l_4}\Level$, $\lttypgeq i S{S'}{l_1}$,
    $\lttypgeq[\Psi][\Gamma, x : S \at{l_1}] i{T}{T'}{l_2}$, 
    $\lttrmgneeq i \nu{\nu'}{\PI {l_1}{l_2} x{S''}{T''}}{l_1 \sqcup l_2}$ and $\lttrmgeq i
    s{s'} S{l_1}$, then \newline
    $\lttrmgneeq i{\APP \nu {l_1}{l_2} x S T s}{\APP{\nu'}{l_3}{l_4} x
      {S'}{T'}{s'}}{T[s/x]}{l_2}$. 

  \end{itemize}
\end{law}

We derive that
\begin{lemma}[Reflexivity of Local Identity Substitutions]
  If $\lpequiv i \Gamma \Delta$, then
  $\ltsubgeq i{\id_{\Gamma}}{\id_{\Gamma}}{\Delta}$.
\end{lemma}

\begin{lemma}[Congruence of Global Variables] $ $
  \begin{itemize}
  \item If $\lpequiv i \Gamma \Delta$, $u : \DTrm[\Delta]{i'}T l \in \Psi$, $i' \in \{\varlevel, \codelevel\}$ and
    $i' \le i$, then $\lttrmgneeq i {u^{\id_{\Gamma}}}{u^{\id_{\Gamma}}} T l$.
  \item If $\lpequiv i \Gamma \Delta$, $U : \DTyp[\Delta]{i'} l \in \Psi$,
    $i' \in \{\codelevel, \proglevel\}$ and $i' \le i$, then $\lttypgneeq i{U^{\id_{\Gamma}}}{U^{\id_{\Gamma}}}{l}$.
  \end{itemize}
\end{lemma}

At this point, we give the laws that should hold for both layers. %
Next, we consider laws that only hold at layer $i = \metalevel$. %
We first give the laws for type constructors.
\begin{law}[Type Constructors] $ $
  \begin{itemize}
  \item If $\lpjudge \metalevel \Gamma$, $\lttypgeq[\Psi][\Gamma][L, \vect\ell] \metalevel{T}{T'}{l}$
    and $\tyequiv[L,\vect\ell]{l}{l'}{\Level}$, then
    $\lttypgeq \metalevel{\UPI \ell l T}{\UPI \ell{l'}{T'}}\omega$.
  \item If $\lpjudge \metalevel \Gamma$, $\lttypgeq[\Psi, g : \Ctx] \metalevel{T}{T'}{l}$ and
    $\tyequiv[L]{l}{l'}{\Level}$, then \newline
    $\lttypgeq \metalevel{\CPI g l T}{\CPI g{l'}{T'}}{l}$. 
  \item If $\lpjudge \metalevel \Gamma$, $\lttypgeq[\Psi, U : \DTyp[\Delta] \proglevel{l_1}]
    \metalevel{T}{T'}{l_2}$, $\ltctxgeq \proglevel{\Delta}{\Delta'}$, $\tyequiv[L]{l_1}{l_3}{\Level}$
    and $\tyequiv[L]{l_2}{l_4}{\Level}$, then 
    $\lttypgeq \metalevel{\TPI U [\Delta]{l_1}{l_2} T}{\TPI U[\Delta']{l_3}{l_4}{T'}}{l_2}$.
  \item If $\lpjudge \metalevel \Gamma$ and $\ltctxgeq \proglevel{\Delta}{\Delta'}$, then $\lttypgeq
    \metalevel{\CTyp[\Delta]l}{\CTyp[\Delta']l}{0}$.
  \item If $\lpjudge \metalevel \Gamma$, $\ltctxgeq \proglevel{\Delta}{\Delta'}$ and
    $\lttypgeq[\Psi][\Delta]\proglevel T{T'}l$, then \newline
    $\lttypgeq \metalevel{\CTrm[\Delta] Tl}{\CTrm[\Delta']{T'}l}{0}$.
  \end{itemize}
\end{law}

\begin{law}[Congruence] $ $
  \begin{itemize}
  \item If $\lttyping \metalevel t {\UPI \ell l T}{\omega}$, $\lttyping \metalevel {t'}{\UPI \ell l
      T}{\omega}$ and $\lttrmgeq[\Psi][\Gamma][L, \vect\ell] \metalevel{\UAPP
      t{\ell}}{\UAPP{t'}{\ell}}{T}{l}$, then
    $\tconvtrmnf \metalevel{t}{t'}{\UPI \ell l T}{\omega}$. 
  \item If $\lttyping \metalevel t{\CPI g l T}{l}$, $\lttyping \metalevel {t'}{\CPI g l T}{l}$ and \newline
    $\lttrmgeq[\Psi, g : \Ctx] \metalevel{\CAPP t{g}}{\CAPP{t'}{g}}{T}{l}$, then $\lttrmgeq
    \metalevel{t}{t'}{\CPI g l T}{l}$. 
  \item If $\lttyping \metalevel t {\TPI U [\Delta]{l}{l'} T}{l'}$,
    $\lttyping \metalevel {t'}{\TPI U [\Delta]{l}{l'} T}{l'}$, and \newline
    $\lttrmgeq[\Psi, U : \DTyp[\Delta] \proglevel{l}] \metalevel{\TAPP
      t{U^{\id_\Delta}}}{\TAPP{t'}{U^{\id_\Delta}}}{T}{l'}$, then \newline
    $\lttrmgeq \metalevel{w}{w'}{\TPI U [\Delta]{l}{l'} T}{l'}$.
  \item If $\lpjudge \metalevel \Gamma$ and $\lttypwf[\Psi][\Delta] \codelevel T l$, then $\lttrmgeq
    \metalevel{\boxit T}{\boxit{T}}{\CTyp[\Delta] l}{0}$.
  \item If $\lpjudge \metalevel \Gamma$ and $\lttyping[\Psi][\Delta] \codelevel t T l$, then
    $\lttrmgeq \metalevel{\boxit t}{\boxit{t}}{\CTrm[\Delta] T l}{0}$. 
  \end{itemize}
\end{law}

\begin{law}[Congruence for Neutrals] $ $
  \begin{itemize}
  \item If $\lttrmgneeq \metalevel{\nu}{\nu'}{\UPI \ell l T}{\omega}$, $|\vect\ell| = |\vect l| = |\vect
    l'| > \codelevel$ and
    $\forall ~ \codelevel \le n < |\vect l| ~.~ \tyequiv[L]{\vect l(n)}{\vect l'(n)}\Level$,
    then $\lttrmgneeq \metalevel{\UAPP \nu l}{\nu'~\$~\vect l'}{T[\vect l/\vect \ell]}{l[\vect
      l/\vect \ell]}$. 
  \item If $\lttrmgneeq \metalevel{\nu}{\nu'}{\CPI g l T}{l}$ and $\ltctxgeq \proglevel
    \Delta{\Delta'}$, then $\lttrmgneeq \metalevel{\CAPP \nu
      \Delta}{\CAPP{\nu'}{\Delta'}}{T[\Delta/g]}{l}$.
  \item If $\lttrmgneeq \metalevel{\nu}{\nu'}{\TPI U[\Delta] l{l'}{T''}}{l'}$ and
    $\lttypgeq[\Psi][\Delta] \proglevel{T}{T'}l$, then \newline
    $\lttrmgneeq \metalevel{\TAPP \nu{T}}{\TAPP{\nu'}{T'}}{T''[T/U]}{l'}$. 
  \item If $\lpjudge \metalevel \Gamma$, $\tyequiv[L]{l_1}{l_3}\Level$, 
    $\tyequiv[L]{l_2}{l_4}\Level$, $\ltctxgeq \proglevel \Delta{\Delta'}$, \newline
    $\lttrmgneeq \metalevel{\nu}{\nu'}{\CTyp[\Delta]{l_2}}{0}$, 
    $\lttypgeq[\Psi][\Gamma,x_T : \CTyp[\Delta]{l_2} \at{0}]\metalevel{M}{M'}{l_1}$
    and \newline
    $ \lttrmgeq[\Psi, U : \DTyp[\Delta]\codelevel{l_2}]\metalevel{t_1}{t_2}{M[\boxit{U^\id}/x_T]}{l_1}$, then \newline
    $\lttrmgneeq \metalevel{\LETBTYP{l_1}{l_2}
      \Delta{x_T.M}{U}{t_1}\nu}{\LETBTYP{l_3}{l_4}{\Delta'}{x_T.M'}{U}{t_2}{\nu'}}{M[t/x_T]}{l_1}$. 
  \item If $\lpjudge \metalevel \Gamma$, $\tyequiv[L]{l_1}{l_3}\Level$, 
    $\tyequiv[L]{l_2}{l_4}\Level$, $\ltctxgeq \proglevel \Delta{\Delta'}$,
    $\lttypgeq \proglevel T{T'}{l_2}$,
    $\lttrmgneeq \metalevel{\nu}{\nu'}{\CTrm[\Delta]T{l_2}}{0}$, 
    $\lttypgeq[\Psi][\Gamma,x_T : \CTrm[\Delta]T{l_2} \at{0}]\metalevel{M}{M'}{l_1}$
    and \newline
    $ \lttrmgeq[\Psi, u : \DTrm[\Delta]\codelevel T{l_2}]\metalevel{t_1}{t_2}{M[\boxit{u^\id}/x_t]}{l_1}$, then \newline
    $\lttrmgneeq \metalevel{\LETBTRM{l_1}{l_2}
      \Delta{T}{x_t.M}{u}{t_1}\nu}{\LETBTRM{l_3}{l_4}{\Delta'}{T'}{x_T.M'}{u}{t_2}{\nu'}}{M[t/x_t]}{l_1}$. 
  \end{itemize}
\end{law}

Then we consider the law for neutral forms for recursive principles for code. %
The law follows a similar line to the equivalence judgments and the convertibility
checking: the evaluation is blocked when the scrutinee is neutral or is a $\tbox$'ed
global variable. 
\begin{law}[Neutral Recursion on Code] $ $
  \begin{itemize}
  \item If all motives and branches are related by corresponding generic equivalence,
    moreover, $\tyequiv[L]{l}{l'}\Level$, $\ltctxgeq \proglevel{\Delta}{\Delta'}$ and
    $\lttrmgneeq \metalevel{\nu}{\nu'}{\CTyp[\Delta]{l}}{0}$, then \newline
    $\lttrmgneeq \metalevel{\ELIMTYP{l_1}{l_2}Mbl\Delta
      \nu}{\ELIMTYP{l_3}{l_4}{M'}{b'}{l'}{\Delta'}{\nu'}}{M_\Typ[l/\ell,\Delta/g,\nu/x_T]}{l_1}$. 
  \item If all motives and branches are related by corresponding generic equivalence,
    moreover, $\tyequiv[L]{l}{l'}\Level$, $\ltctxgeq \proglevel{\Delta}{\Delta'}$,
    $U : \DTyp[\Delta]{\codelevel} l \in \Psi$ and $\ltsubst \codelevel \delta \Delta$, then \newline
    $\lttrmgneeq \metalevel{\ELIMTYP{l_1}{l_2}Mbl\Delta{(\boxit{U^\delta})}}{\ELIMTYP{l_3}{l_4}{M'}{b'}{l'}{\Delta'}{(\boxit{U^\delta})}}{M_\Typ[l/\ell,\Delta/g,\boxit{U^\delta}/x_T]}{l_1}$.
  \item If all motives and branches are related by corresponding generic equivalence,
    moreover, $\tyequiv[L]{l}{l'}\Level$, $\ltctxgeq \proglevel{\Delta}{\Delta'}$,
    $\tconvtyp \proglevel T{T'}{l}$ and
    $\tconvtrmnee \metalevel{\nu}{\nu'}{\CTrm[\Delta]T{l}}{0}$, then \newline
    $\lttrmgneeq \metalevel{\ELIMTRM{l_1}{l_2}Mbl\Delta T \nu}{\ELIMTYP{l_3}{l_4}{M'}{b'}{l'}{\Delta'}{T'}{\nu'}}{M_\Trm[l/\ell,\Delta/g,T/U_T,\nu/x_t]}{l_2}$.
  \item If all motives and branches are related by corresponding generic equivalence,
    moreover, $\tyequiv[L]{l}{l'}\Level$, $\ltctxgeq \proglevel{\Delta}{\Delta'}$,
    $\tconvtyp \proglevel T{T'}{l}$, $u : \DTrm[\Delta]{i'}T l \in \Psi$,
    $ i' \in \{\varlevel, \codelevel\}$ and $\ltsubst \codelevel \delta \Delta$, then \newline
    $\lttrmgneeq \metalevel{\ELIMTRM{l_1}{l_2}Mbl\Delta T
      {(\boxit{u^\delta})}}{\ELIMTYP{l_3}{l_4}{M'}{b'}{l'}{\Delta'}{T'}{(\boxit{u^\delta})}}{M_\Trm[l/\ell,\Delta/g,T/U_T,\boxit{u^\delta}/x_t]}{l_2}$.
  \end{itemize}
\end{law}
We conclude all the laws here fore the generic equivalence.

This generic equivalence will be instantiated twice times: the syntactic equivalence
judgments, and the convertibility checking judgments. %
\delamlang is way more complex than \citet{abel_decidability_2017}'s work because
\delamlang has computations at two layers, $\proglevel$ and $\metalevel$. %
Therefore, we must derive the necessary properties from the fundamental theorems to be
proved shortly at each layer. %

\subsection{Kripke Logical Relations for MLTT}\labeledit{sec:dt:klogrel}

The Kripke logical relations are parameterized by the generic equivalence. %
It is additionally indexed by another layering index $j \in \{\proglevel, \metalevel\}$, which
quantifies the types described by the relations. %
When $j = \proglevel$, we consider types from MLTT. %
When $j = \metalevel$, we consider all possible types. %
The reason for this distinction is to handle the lifting property from layer $\codelevel$
(which has the terms as $\proglevel$) or $\proglevel$ to
$\metalevel$, where terms from MLTT are brought to \delamlang. %
On the semantic side, we need to make sure that terms from MLTT can interact with
``native'' terms in \delamlang coherent. %
We further restrict $j = \proglevel$ when $i = \proglevel$. 

The Kripke logical relations are defined by
\begin{enumerate}
\item recursion on $j$, which effectively means the logical relations are
  2-layered; also note when $i = \proglevel$, $j = \proglevel$ is determined;
\item a transfinite well-founded recursion on the universe levels, and
\item induction-recursion on related types and terms. 
\end{enumerate}
In particular, the recursion on $j$ is necessary, as the relations when $j = \metalevel$ depend
on the validity judgments of $j = \proglevel$. %
When we do a recursion on universe levels, we must mind the well-foundedness of
universe levels. %
As we have discussed in \Cref{sec:dt:ulevel}, we are sure that all universe
levels must find a finite number of steps to descend to $0$. %
The only problem is $\omega$, which is not finite. %
Thus we must include one large cardinal to handle this level, hence the transfinite
recursion. %
Luckily, we do not have to think about it most of the time as we cannot really use
$\omega$ to do anything special at all. %
Note that our relations do not exactly follow \citet{abel_decidability_2017} tightly,
where two relations are defined for types and terms respectively. %
In our case, we provide simpler inductive-recursive definitions, where only one
relation is defined for types and for terms respectively. %
This style is more akin to the PER models in untyped domains by
\citet{abel_normalization_2013,hu_jang_pientka_2023}, except that our logical
relations are Kripke. %
We follow a proof schema that combines that of \citet{abel_decidability_2017} and that
of \citet{abel_normalization_2013,hu_jang_pientka_2023}. %
We define the following judgments:
\begin{itemize}
\item $\D :: \ldtypeq i j T{T'} l$ denotes that two types $T$ and $T'$ are related. %
  This relation is defined inductively. %
  We use $\D$ to mark give a name to the derivation as we will do recursion on it. %
\item $\ldtyequiv i j t{t'} \D$ denotes that two terms $t$ and $T'$ related by $\D$. %
  This relation is defined by a recursion on $\D$. 
\item $\E :: \ldctxeq i j \Gamma \Delta$ denotes that two contexts are related. %
  It is a generalization of $\D$. 
\item $\ldsubeq i j \delta{\delta'} \E$ denotes that two local substitutions
  $\delta$ and $\delta'$ are related. %
  It is a generalization of $\ldtyequiv i j t{t'} \D$ by doing recursion on $\E$. 
\end{itemize}
% The index $i$ is fixed for each instance of generic equivalence.
For convenience, we define the following:
\begin{align*}
  \ldtypwf i j T l &:= \ldtypeq i j T T l \\
  \ldtyequivt i j t{t'} T l & := \text{for some $T'$, } \D :: \ldtypeq i j T{T'}
l \tand \ldtyequiv i j t{t'}{\D} \\
  \ldtyping i j t T l & := \ldtyequivt i j t{t} T l \\
  \ldctxwf i j \Gamma &:=  \ldctxeq i j \Gamma \Gamma \\
  \ldsubeq i j \delta{\delta'} \Delta &:= \text{for some $\Delta'$, }\E :: \ldctxeq i
                                        j \Delta{\Delta'} \tand \ldsubeq i j \delta{\delta'} \E \\
  \ldsubst i j \delta \Delta &:= \ldsubeq i j \delta{\delta} \Delta
\end{align*}

Now we proceed to define the relations. %
We begin with the natural numbers. %
\begin{mathpar}
  \D :: \inferrule
  {\ttypreds i T{\Nat}0 \\ \ttypreds i {T'}\Nat0}
  {\ldtypeq i j T{T'}0}
\end{mathpar}
Then $\ldtyequiv i j t{t'} \D$ is defined by $\ldtyeqnat i t{t'}$, which we
define as follows:
\begin{mathpar}
  \inferrule
  {\ttrmreds i t w \Nat 0 \\ \ttrmreds i{t'}{w'}\Nat 0 \\ \lttrmgeq i w{w'}\Nat 0
    \\ \ldtynfeqnat i w{w'}}
  {\ldtyeqnat i t{t'}}

  \inferrule
  { }
  {\ldtynfeqnat i \ze \ze}

  \inferrule
  {\ldtyeqnat i t{t'}}
  {\ldtynfeqnat i{\su t}{\su{t'}}}

  \inferrule
  {\lttrmgneeq i \nu{\nu'}\Nat 0}
  {\ldtynfeqnat i{\nu}{\nu'}}  
\end{mathpar}

Then we consider universes.
\begin{mathpar}
  \D :: \inferrule
  {\ttypreds i T{\Ty{l_1}}{1 + l_1} \\ \ttypreds i {T'}{\Ty{l_2}}{1 + l_2} \\\\
  \tyequiv[L]{l_1}{l}\Level \\ \tyequiv[L]{l_2}{l}\Level}
  {\ldtypeq i j{T}{T'}{1 + l}}
\end{mathpar}
Then $\ldtyequiv i j t{t'} \D$ is defined by
\begin{itemize}
\item $\ttrmreds i t w{\Ty l}{1 + l}$,
\item $\ttrmreds i {t'}{w'}{\Ty l}{1 + l}$,
\item $\lttrmgeq i w{w'}{\Ty l}{1 + l}$, which means that $t$ and $t'$ are equivalent
  types at level $l$,
\item $\ldtypeq i j{\Elt l w}{\Elt l{w'}}{l}$, which means that the corresponding
  types of $w$ and $w'$ are related. 
\end{itemize}
The last condition requires the well-founded recursion on the universe levels in order
to refer back to the relation for types. %
Notice that the universe level decreases by one so this definition is valid. %

Then we define the relation for $\Pi$ types.
\begin{mathpar}
  \D :: \inferrule
  {\ttypreds i T{\PI l{l'} x{S_1}{T_1}}{l \sqcup l'} \\ \ttypreds i {T'}{\PI l{l'}
      x{S_2}{T_2}}{l \sqcup l'} \\
    \lttypwf i{S_1} l \\ \lttypwf i{S_2} l \\ \lttypwf[\Psi][\Gamma, x : S_1 \at l]
    i{T_1}{l'} \\
    \lttypwf[\Psi][\Gamma, x : S_2 \at l] i{T_2}{l'} \\
    \lttypgeq i{\PI l{l'} x{S_1}{T_1}}{\PI l{l'} x{S_2}{T_2}}{l \sqcup l'} \\
    \E :: (\forall~ \psi :: L' \sep \Phi;\Delta \To_i L \sep \Psi; \Gamma ~.~
    \ldtypeq[\Phi][\Delta][L'] i j{S_1}{S_2}{l}) \\
    \F :: (\forall~ \psi :: L' \sep \Phi;\Delta \To_i L \sep \Psi; \Gamma \tand \ldtyequiv[\Phi][\Delta][L']
    i j {s}{s'}{\E(\psi)} ~.~
    \ldtypeq[\Phi][\Delta][L'] i j{T_1[s/x]}{T_2[s'/x]}{l'})}
  {\ldtypeq i j{T}{T'}{l \sqcup l'}}
\end{mathpar}
This case is particularly complex. %
Let us digest the premises one by one. %
First, we require that both types $T$ and $T'$ to reduce to some $\Pi$ types. %
The typing judgments require the components of the $\Pi$ types are well-formed. %
Further, the $\Pi$ types themselves are equivalent. %
Then $\E$ requires that $S_1$ and $S_2$ are related under any weakening. %
This makes the relation Kripke. %
$\F$ is similar but require $T_1$ and $T_2$ remain related given any two related terms
$s$ and $s'$. %
Furthermore, in reality, we should put down sufficient premises for equivalent
universes to allow syntactically different universes to appear in both $T$ and $T'$,
but due to the size of this rule, we only use $l$ and $l'$ here. %
We apply the same principle in other rules in the rest of this technical report.

Then $\ldtyequiv i j t{t'} \D$ is defined by
\begin{itemize}
\item $\ttrmreds i t w{\PI l{l'} x{S_1}{T_1}}{l \sqcup l'}$, 
\item $\ttrmreds i {t'}{w'}{\PI l{l'} x{S_2}{T_2}}{l \sqcup l'}$,
\item $\lttrmgeq i w{w'}{\PI l{l'} x{S_1}{T_1}}{l \sqcup l'}$,
\item $\lttrmgeq i w{w'}{\PI l{l'} x{S_2}{T_2}}{l \sqcup l'}$, which is duplicated to
  make symmetry a simpler property,
\item
  for all $\forall\ \psi :: L' \sep \Phi;\Delta \To_i L \sep \Psi; \Gamma$ and $\Ac ::
  \ldtyequiv[\Phi][\Delta][L'] i j {s}{s'}{\E(\psi)}$, moreover, we add
  equivalence assumptions to remove the effect of type annotations,
  $\lttypgeq[\Phi][\Delta][L'] i{S_1}{S_1'}l$, $\lttypgeq[\Phi][\Delta][L'] i{S_2}{S_2'}l$, \break
  $\lttypgeq[\Phi][\Delta, x : S_1 \at l][L'] i{T_1}{T_1'}{l'}$,
  $\lttypgeq[\Phi][\Delta, x : S_2 \at l][L'] i{T_2}{T_2'}{l'}$, then \break
  $\ldtyequiv[\Phi][\Delta][L'] i j{\APP w l{l'} x{S_1'}{T_1'} s}{\APP{w'}l{l'}
    x{S_2'}{T_2'}{s'}}{\F(\psi, \Ac)}$.
\end{itemize}

The next case is neutral types.
\begin{mathpar}
  \D :: \inferrule
  {\ttypreds i T{V}l \\ \ttypreds i{T'}{V'}l \\
    \lttypgneeq i V{V'} l}
  {\ldtypeq i j{T}{T'}{l}}
\end{mathpar}
Then $\ldtyequiv i j t{t'} \D$ is defined by
\begin{itemize}
\item $\ttrmreds i t \nu{V}{l}$,
\item $\ttrmreds i {t'}{\nu'}{V'}{l}$,
\item $\lttrmgneeq i \nu{\nu'}{V}{l}$,
\item $\lttrmgneeq i \nu{\nu'}{V'}{l}$.
\end{itemize}
Again, we duplicate $\nu \sim \nu'$ to make symmetry easy.

The last case possible for both layers is equivalence of universe levels. %
This case is introduced for bureaucratic purposes and often ignored.
\begin{mathpar}
  \D :: \inferrule
  {\E :: \ldtypeq i j{T}{T'}{l'} \\ \tyequiv[L]l{l'}\Level}
  {\ldtypeq i j{T}{T'}{l}}
\end{mathpar}
Then $\ldtyequiv i j t{t'} \D$ is defined by $\ldtyequiv i j t{t'} \E$. 

Now we have given all logical relations for types and terms that are available at both layers. %
In order to give the logical relations for layer $\metalevel$, we need to first give the
logical relations for local contexts and local substitutions, as the semantics of
contextual types depend on them. %
The complexity of our logical relations primarily comes from the fact that we must be
able to handle computation at two layers ($\proglevel$ and $\metalevel$) and different lifting
behaviors. %

We proceed by defining related local contexts inductively and then the corresponding
recursive case for local substitutions. %
\begin{mathpar}
  \E :: \inferrule
  {\judge[L] \Psi}
  {\ldctxeq i j \cdot\cdot}
\end{mathpar}
Then $\ldsubeq i j \delta{\delta'} \E$ is defined by $\lpjudge i \Gamma$ and then checking $\Gamma$:
\begin{itemize}
\item if $\Gamma$ ends with $\cdot$, then $\delta = \delta' = \cdot^{|\Gamma|}$;
\item if $\Gamma$ ends with $g$, then $g : \Ctx \in \Psi$ and $\delta = \delta' = \cdot^{|\Gamma|}_g$;
\end{itemize}

\begin{mathpar}
  \E :: \inferrule
  {\judge[L] \Psi \\ g : \Ctx \in \Psi}
  {\ldctxeq i j {g}{g}}
\end{mathpar}
Then $\ldsubeq i j \delta{\delta'} \E$ is defined as $\lpjudge i \Gamma$ and
\begin{itemize}
\item $\Gamma$ also ends with $g$,
\item $\delta = \delta' = \wk^{|\Gamma|}_g$.
\end{itemize}

\begin{mathpar}
  \E :: \inferrule
  {\F :: \forall~\alpha :: L' \sep \Phi \To L \sep \Psi ~.~ \ldctxeq[\Phi][L'] i j
    \Delta{\Delta'} \\ \D :: \ldStypeq \F i j T{T'} l \\ \tyequiv[L]l{l'}\Level}
  {\ldctxeq i j{\Delta, x : T \at l}{\Delta', x : T' \at{l'}}}
\end{mathpar}
Then $\ldsubeq i j \delta{\delta'} \E$ is defined as
\begin{itemize}
\item $\delta = \delta_1, t/x$,
\item $\delta' = \delta_1', t'/x$,
\item $\Cc :: \forall~\alpha :: L' \sep \Phi \To L \sep \Psi ~.~ \ldsubeq[\Phi][\Gamma][L'] i j {\delta_1}{\delta_1'}{\F(\alpha)}$,
\item $\forall~\alpha :: L' \sep \Phi \To L \sep \Psi ~.~ \ldtyequiv[\Phi][\Gamma][L'] i j t{t'}{\D(\alpha,
    \Cc(\alpha))}$.
\end{itemize}
where we let $\D :: \ldStypeq \F i j T{T'} l$ to be
\[
  \forall~\alpha :: L' \sep \Phi \To L \sep \Psi \tand \ldsubeq[\Phi][\Gamma][L'] i j \delta{\delta'}{\F(\alpha)} ~.~ \ldtypeq[\Phi][\Gamma][L'] i j {T[\delta]}{T'[\delta']}{l}
\]
The judgment $\ldStypeq \F i j T{T'} l$ requires the stability under local substitutions
of the relation between $T$ and $T'$. %
We apply the same principle to related terms and derive the judgment
$\ldStyequiv \F i j t{t'}{\D}$, which is given by
\[
  \forall~\alpha :: L' \sep \Phi \To L \sep \Psi \tand \Cc :: \ldsubeq[\Phi][\Gamma][L'] i j
  \delta{\delta'}{\F(\alpha)} ~.~ \ldtyequiv[\Phi][\Gamma][L'] i j {t[\delta]}{t'[\delta']}{\D(\alpha,
    \Cc)}
\]

\subsection{Properties for Logical Relations When $j = \proglevel$}\labeledit{sec:dt:logrel-prop}

In this section, we pause our progress of defining the logical relations when $j =
\metalevel$. %
Following previous lines of work, when we give the Kripke logical relations to
contextual types, we will have to refer to the validity judgments of types and terms
in the logical relations. %
Therefore, in this section, we first work out a list of properties of the logical
relations when $j = \proglevel$, and then in the next section, we define semantic judgments for
global contexts and global substitutions, and then validity judgments. %
Once we have the validity judgments when $j = \proglevel$, we can then finish writing down the
logical relations for contextual types. %
Without further ado, let us start proving some lemmas. %
The list of lemmas mainly follows \citet{abel_normalization_2013,hu_jang_pientka_2023}
though we also take \citet{abel_decidability_2017} into consideration. 

\begin{lemma}[Weakening] $ $
  \begin{itemize}
  \item If $\D :: \ldtypeq i \proglevel{T}{T'}{l}$ and
    $\psi :: L' \sep \Phi;\Delta \To_i L \sep \Psi; \Gamma$, then
    $\E :: \ldtypeq[\Phi][\Delta][L'] i \proglevel{T}{T'}{l}$.
  \item If $\ldtyequiv i \proglevel{t}{t'}\D$ and
    $\psi :: L' \sep \Phi;\Delta \To_i L \sep \Psi; \Gamma$, then
    $\ldtyequiv[\Phi][\Delta][L'] i \proglevel{t'}{t}\E$.
  \end{itemize}
\end{lemma}
\begin{proof}
  Induction on $\D$. %
  Note that typing judgments are invariant under weakenings. %
  Also use the weakening law for generic equivalence. %
\end{proof}

\begin{lemma}[Escape] $ $
  \begin{itemize}
  \item If $\D :: \ldtypeq i \proglevel{T}{T'}{l}$, then $\lttypgeq i {T}{T'} l$.
  \item If $\ldtyequiv i \proglevel{t}{t'}\D$, then $\lttrmgeq i t{t'} T l$ and $\lttrmgeq i t{t'}{T'} l$.
  \end{itemize}  
\end{lemma}
\begin{proof}
  Induction on $\D$. %
  Use the subsumption law of the generic equivalence. %
\end{proof}

\begin{lemma}[Reflexivity of Neutral]
  If $\D :: \ldtypeq i \proglevel{T}{T'}l$, $\lttrmgneeq i \nu{\nu'} T l$ and $\lttrmgneeq i \nu{\nu'}{T'} l$, then $\ldtyequiv i \proglevel{\nu}{\nu'}\D$.
\end{lemma}
\begin{proof}
  Induction on $\D$. %
  Use the conversion and subsumption laws of generic equivalence and then the
  soundness of multi-step reductions and the escape lemma to obtain the equivalence
  between $\nu$ and $\nu'$ as normal forms. %
  In the function case, we use the congruence for neutral law of generic equivalence
  to show that the results of applying two neutral function are related.
\end{proof}

\begin{lemma}[Weak Head Expansion] $ $
  \begin{itemize}
  \item If $\D :: \ldtypeq i \proglevel{T}{T'}{l}$, $\ttypreds i{T_1} T l$ and $\ttypreds
    i{T_1'}{T'}l$, then \newline $\ldtypeq i \proglevel {T_1}{T_1'} l$.
  \item If $\ldtyequiv i \proglevel{t}{t'}\D$, $\ttrmreds i{t_1}{t}{T} l$ and $\ttrmreds
    i{t_1'}{t'}{T'}l$, then $\ldtyequiv i \proglevel {t_1}{t_1'} \D$.
  \end{itemize}  
\end{lemma}
\begin{proof}
  Induction on $\D$. Use transitivity of reductions. 
\end{proof}

\begin{lemma}[Symmetry] $ $
  \begin{itemize}
  \item If $\D :: \ldtypeq i \proglevel{T}{T'}{l}$, then $\E :: \ldtypeq i \proglevel{T'}{T}{l}$.
  \item If $\ldtyequiv i \proglevel{t}{t'}\D$, then $\ldtyequiv i \proglevel{t'}{t}\E$. 
  \end{itemize}
\end{lemma}
\begin{proof}
  We do induction on $\D$. %
  Since our definition is designed with symmetry in mind, symmetry is rather immediate.
\end{proof}

\begin{lemma}[Right Irrelevance]
  If $\D :: \ldtypeq i \proglevel{T}{T'}{l}$, $\E :: \ldtypeq i \proglevel{T}{T''}{l}$ and
  $\ldtyequiv i \proglevel{t}{t'}\D$, then $\ldtyequiv i \proglevel{t}{t'}\E$.
\end{lemma}
\begin{proof}
  We do induction on $\D$ and then invert $\E$. %
  We consider the function case. %
  In this case, we have premises
  \begin{align*}
    \D_1 &:: (\forall~ \psi :: L' \sep \Phi;\Delta \To_i L \sep \Psi; \Gamma ~.~
    \ldtypeq[\Phi][\Delta][L'] i \proglevel{S_1}{S_2}{l_1}) \\
    \D_2 &:: (\forall~ \psi :: L' \sep \Phi;\Delta \To_i L \sep \Psi; \Gamma ~.~ \forall~ \ldtyequiv[\Phi][\Delta][L']
    i \proglevel {s}{s'}{\D_1(\psi)} ~.~
    \ldtypeq[\Phi][\Delta][L'] i \proglevel{T_1[s/x]}{T_2[s'/x]}{l_2}) \\
    \E_1 &:: (\forall~ \psi :: L' \sep \Phi;\Delta \To_i L \sep \Psi; \Gamma ~.~
    \ldtypeq[\Phi][\Delta][L'] i \proglevel{S_1}{S_2'}{l_1}) \\
    \E_2 &:: (\forall~ \psi :: L' \sep \Phi;\Delta \To_i L \sep \Psi; \Gamma ~.~ \forall~ \ldtyequiv[\Phi][\Delta][L']
    i \proglevel {s}{s'}{\E_1(\psi)} ~.~
    \ldtypeq[\Phi][\Delta][L'] i \proglevel{T_1[s/x]}{T_2'[s'/x]}{l_2})
  \end{align*}
  from $\D$ and $\E$. %
  Here $\D_1$ and $\D_2$ are the premises of $\D$, and likewise for $\E$. %
  By determinacy, we know that
  \[
    T \reds \PI{l_1}{l_2} x{S_1}{T_1}
  \]
  must be unique.  %
  Moreover, we also know
  \begin{gather*}
    T' \reds \PI{l_1}{l_2} x{S_2}{T_2} \\
    T'' \reds \PI{l_1}{l_2} x{S_2'}{T_2'}
  \end{gather*}
  The most difficult part is to show that given
  \begin{itemize}
  \item $\psi :: L' \sep \Phi;\Delta \To_i L \sep \Psi; \Gamma$,    
  \item $\Ac :: \ldtyequiv[\Phi][\Delta][L'] i \proglevel {s}{s'}{\E_1(\psi)}$
  \item $\lttypgeq[\Phi][\Delta][L'] i{S_1}{S_3}l$,
  \item $\lttypgeq[\Phi][\Delta][L'] i{S_2}{S_4}l$, 
  \item $\lttypgeq[\Phi][\Delta, x : S_1 \at{l_1}][L'] i{T_1}{T_3}{l_2}$, and
  \item $\lttypgeq[\Phi][\Delta, x : S_2' \at{l_1}][L'] i{T_2'}{T_4}{l_2}$,
  \end{itemize}
  then
  \[
    \ldtyequiv[\Phi][\Delta][L'] i \proglevel{\APP w{l_1}{l_2} x{S_3}{T_3} s}{\APP{w'}{l_1}{l_2}
      x{S_4}{T_4}{s'}}{\E_2(\psi, \Ac)}
  \]

  From $\ldtyequiv i \proglevel{t}{t'}\D$, together with
  \begin{itemize}
  \item $\Bc :: \ldtyequiv[\Phi][\Delta][L'] i \proglevel {s}{s'}{\D_1(\psi)}$ by IH,
  \item $\lttypgeq[\Phi][\Delta, x : S_2 \at{l_1}][L'] i{T_2}{T_4}{l_2}$, where we
    know $\lttypeq i{S_2}{S_2'}{l_1}$ by escape, the subsumption law of generic
    equivalence, transitivity and local context equivalence of syntactic and generic
    equivalence. %
  \end{itemize}
  we have
  \[
    \ldtyequiv[\Phi][\Delta][L'] i \proglevel{\APP w{l_1}{l_2} x{S_3}{T_3} s}{\APP{w'}{l_1}{l_2}
      x{S_4}{T_4}{s'}}{\D_2(\psi, \Bc)}
  \]
  By another IH, we have the goal.
\end{proof}

\begin{lemma}[Left Irrelevance]
  If $\D :: \ldtypeq i \proglevel{T'}{T}{l}$, $\E :: \ldtypeq i \proglevel{T''}{T}{l}$ and
  $\ldtyequiv i \proglevel{t}{t'}\D$, then $\ldtyequiv i \proglevel{t}{t'}\E$.
\end{lemma}
\begin{proof}
  Immediate by symmetry and right irrelevance. 
\end{proof}

The left and right irrelevance lemmas are called the irrelevance lemma. %
It says that the exact relation between types is not important as long as their normal
forms are related. %

\begin{lemma}[Reflexivity and Transitivity] $ $
  \begin{itemize}
  \item If $\D_1 :: \ldtypeq i \proglevel{T_1}{T_2}{l}$ and $\D_2 :: \ldtypeq i
    \proglevel{T_2}{T_3}{l}$, then $\D_3 :: \ldtypeq i \proglevel{T_1}{T_3}{l}$.
  \item If $\E :: \ldtypeq i \proglevel{T_1}{T_1}{l}$, $\ldtyequiv i \proglevel{t_1}{t_2}{\D_1}$ and
    $\ldtyequiv i \proglevel{t_2}{t_3}{\D_2}$, then $\ldtyequiv i \proglevel{t_1}{t_3}{\D_3}$.
  \item $\F :: \ldtypeq i \proglevel{T_1}{T_1}{l}$.
  \item If $\ldtyequiv i \proglevel{t_1}{t_2}{\D_1}$, then $\ldtyequiv i \proglevel{t_1}{t_1}{\F}$.
  \end{itemize}
\end{lemma}
\begin{proof}
  We do induction on $\D_1$ and then invert $\D_2$. %
  The function case is the most complex one. %
  We have the following premises:
  \begin{align*}
    \Ac_1 &:: (\forall~ \psi :: L' \sep \Phi;\Delta \To_i L \sep \Psi; \Gamma ~.~
    \ldtypeq[\Phi][\Delta][L'] i \proglevel{S_1}{S_2}{l_1}) \\
    \Ac_2 &:: (\forall~ \psi :: L' \sep \Phi;\Delta \To_i L \sep \Psi; \Gamma ~.~ \forall~ \ldtyequiv[\Phi][\Delta][L']
    i \proglevel {s}{s'}{\Ac_1(\psi)} ~.~
    \ldtypeq[\Phi][\Delta][L'] i \proglevel{S_1'[s/x]}{S_2'[s'/x]}{l_2}) \\
    \Bc_1 &:: (\forall~ \psi :: L' \sep \Phi;\Delta \To_i L \sep \Psi; \Gamma ~.~
    \ldtypeq[\Phi][\Delta][L'] i \proglevel{S_2}{S_3}{l_1}) \\
    \Bc_2 &:: (\forall~ \psi :: L' \sep \Phi;\Delta \To_i L \sep \Psi; \Gamma ~.~ \forall~ \ldtyequiv[\Phi][\Delta][L']
    i \proglevel {s}{s'}{\E_1(\psi)} ~.~
    \ldtypeq[\Phi][\Delta][L'] i \proglevel{S_2'[s/x]}{S_3'[s'/x]}{l_2})
  \end{align*}
  from $\D$ and $\E$. %
  Here $\Ac_1$ and $\Ac_2$ are the premises of $\D_1$, and $\Bc_1$ and $\Bc_2$ are from $\D_2$. %
  By determinacy, we know that
  \begin{gather*}
    T_1 \reds \PI{l_1}{l_2} x{S_1}{S_1'} \\
    T_2 \reds \PI{l_1}{l_2} x{S_2}{S_2'} \\
    T_3 \reds \PI{l_1}{l_2} x{S_3}{S_3'}
  \end{gather*}

  Now, we should construct the transitivity for types and terms at the same time to
  understand how this proof is going to check out. %
  First we let 
  \[
    \Cc_1 :: (\forall~ \psi :: L' \sep \Phi;\Delta \To_i L \sep \Psi; \Gamma ~.~
    \ldtypeq[\Phi][\Delta][L'] i \proglevel{S_1}{S_3}{l_1})
  \]
  be the result of IH on $\Ac_1$ and $\Bc_1$. %
  Then our goal is to show that 
  
  The most difficult part is to show that given
  \begin{itemize}
  \item $\psi :: L' \sep \Phi;\Delta \To_i L \sep \Psi; \Gamma$,    
  \item $\F :: \ldtyequiv[\Phi][\Delta][L'] i \proglevel {s}{s'}{\Cc_1(\psi)}$
  \item $\lttypgeq[\Phi][\Delta][L'] i{S_1}{S_4}l$,
  \item $\lttypgeq[\Phi][\Delta][L'] i{S_3}{S_5}l$, 
  \item $\lttypgeq[\Phi][\Delta, x : S_1 \at{l_1}][L'] i{S_1'}{S_4'}{l_2}$, and
  \item $\lttypgeq[\Phi][\Delta, x : S_3 \at{l_1}][L'] i{S_3'}{S_5'}{l_2}$,
  \end{itemize}
  then
  \begin{itemize}
  \item
    \[
      \Cc_2 :: \ldtypeq[\Phi][\Delta][L'] i \proglevel{S_1'[s/x]}{S_3'[s'/x]}{l_2}
    \]
  \item 
    \[
      \ldtyequiv[\Phi][\Delta][L'] i \proglevel{\APP w_1{l_1}{l_2} x{S_4}{S_4'} s}{\APP{w_3}{l_1}{l_2}
        x{S_5}{S_5'}{s'}}{\Cc_2}
    \]
    where $t_k \reds w_k$ for $k \in \{1, 2, 3\}$. 
  \end{itemize}
  Our plan is the following:
  \begin{enumerate}
  \item We first relate $S_1'[s/x]$ and $S_2'[s/x]$, resp.
    $\APP {w_1}{l_1}{l_2} x{S_4}{S_4'} s$ and $\APP {w_2}{l_1}{l_2} x{S_4}{S_4'} s$, which
    requires
    \[
      \ldtyequiv[\Phi][\Delta][L'] i \proglevel {s}{s}{\Ac_1(\psi)}
    \]
    which is derived from reflexivity and irrelevance;
  \item and then relate $S_2'[s/x]$ and $S_3'[s/x]$, resp.
    $\APP {w_2}{l_1}{l_2} x{S_4}{S_4'} s$ and $\APP {w_2}{l_1}{l_2} x{S_5}{S_5'}{s'}$,
    which requires
    \[
      \ldtyequiv[\Phi][\Delta][L'] i \proglevel {s}{s'}{\Bc_1(\psi)}
    \]
    which is immediate from irrelevance.
  \end{enumerate}
  The final missing piece is reflexivity. %
  For types, it is just a result from transitivity and symmetry. %
  For terms, it is reflexivity and transitivity of related types and then
  irrelevance. %
  
  This concludes the function case. 
\end{proof}

Next, we work on the properties for related local contexts and local substitutions. %
We first consider the built property of weakening. %
Weakening can be seen as a case of the Yoneda lemma, where we only depend on
composition of weakenings.
\begin{lemma}[Weakening] $ $
  \begin{itemize}
  \item If $\D :: \ldctxeq i \proglevel \Delta{\Delta'}$ and  $\alpha :: L' \sep \Phi \To L \sep
    \Psi$, then $\E :: \ldctxeq[\Phi][L'] i \proglevel \Delta{\Delta'}$.
  \item If $\ldsubeq i \proglevel{\delta}{\delta'}\D$, $\alpha :: L' \sep \Phi \To L \sep \Psi$
    and $\tau :: L' \sep \Phi; \Gamma' \To_i \Gamma$, then
    $\ldsubeq[\Phi][\Gamma'][L'] i \proglevel{\delta}{\delta'}\E$.
  \end{itemize}
\end{lemma}
\begin{proof}
  Induction on $\D$. %
  In the two base cases, we know $\judge[L']\Phi$ by weakening. %
  In the step case, we simply store the given weakenings in a composition. %
\end{proof}
Notice that the second statement is a little bit strengthened. %
It is possible due to local weakening of related types and terms. %

We should first prove the reflexivity between identity local substitutions before
proving the the escape lemma.
\begin{lemma}[Reflexive Local Weakenings]
  If $\lpequiv i{\Delta,\Gamma}{\Delta',\Gamma}$ and $\D :: \ldctxeq i \proglevel \Delta{\Delta'}$, then \break
  $\ldsubeq[\Psi][\Delta,\Gamma] i \proglevel{\wk^{|\Gamma|}_{\Delta}}{\wk^{|\Gamma|}_{\Delta}}\D$. 
\end{lemma}
\begin{proof}
  We do induction on $\D$ and consider the step case
  \begin{mathpar}
    \D :: \inferrule
    {\E :: \forall \alpha :: L' \sep \Phi \To L \sep \Psi ~.~ \ldctxeq[\Phi][L'] i \proglevel {\Delta}{\Delta'} \\
      \F :: \ldStypeq \E i \proglevel T{T'} l \\ \tyequiv[L]l{l'}\Level}
    {\ldctxeq i \proglevel{\Delta, x : T \at l}{\Delta', x : T' \at{l'}}}
  \end{mathpar}
  \begin{align*}
    \Bc :: & \forall \alpha :: L' \sep \Phi \To L \sep \Psi ~.~ \ldsubeq[\Phi][\Delta, x : T \at l,\Gamma][L'] i \proglevel{\wk^{1 + |\Gamma|}_{\Delta}}{\wk^{1 + |\Gamma|}_{\Delta}}{\E(\alpha)}
      \byIH \\
    \alpha :: & L' \sep \Phi \To L \sep \Psi
                \tag{by assumption} \\
    & \ldtypeq[\Phi][\Delta, x : T \at l,\Gamma][L'] i \proglevel {T[\wk^{1 + |\Gamma|}_{\Delta}]}{T'[\wk^{1 + |\Gamma|}_{\Delta}]}{l}
      \tag{as $\F(\alpha, \Bc(\alpha))$} \\
    \Ac :: & \ldtypeq[\Phi][\Delta, x : T \at l,\Gamma][L'] i \proglevel {T}{T'}{l} \\
    & \lpjudge[\Phi][L'] i{\Gamma, x : T \at l,\Delta}
      \tag{by weakening and presupposition} \\
    & \lttrmgneeq[\Phi][\Delta, x : T \at l,\Gamma][L'] i x x T l
      \tag{by congruence for neutrals of generic equivalence} \\
    & \ldtyequiv[\Phi][\Delta, x : T \at l,\Gamma][L'] i \proglevel{x}{x}\Ac
      \tag{by reflexivity of neutral} \\
    & \ldsubeq[\Psi][\Gamma, x : T \at l,\Delta] i \proglevel{\wk^{1 + |\Gamma|}_{\Delta},x/x}{\wk^{1 + |\Gamma|}_{\Delta},x/x}\D
  \end{align*}
  Thus we conclude the goal. 
\end{proof}

\begin{corollary}[Reflexive Local Identity Substitutions]\labeledit{lem:dt:refl-lid}
  If $\lpequiv i \Gamma \Delta$ and $\D :: \ldctxeq i \proglevel \Gamma \Delta$, then \break
  $\ldsubeq i \proglevel{\id}{\id}\D$. 
\end{corollary}
\begin{proof}
  This is a specialization of the previous lemma. 
\end{proof}

Knowing all local identity substitutions are reflexively related, we can then prove
the escape lemma.
\begin{lemma}[Escape] $ $
  \begin{itemize}
  \item If $\D :: \ldctxeq i \proglevel{\Delta_1}{\Delta_2}$, then $\ltctxgeq i{\Delta_1}{\Delta_2}$. 
  \item If $\ldsubeq i \proglevel{\delta}{\delta'}\D$, then $\ltsubgeq i
    {\delta}{\delta'}{\Delta_1}$ and $\ltsubgeq i {\delta}{\delta'}{\Delta_2}$.
  \end{itemize}
\end{lemma}
\begin{proof}
  Induction on $\D$. %
  We consider the step case.
  \begin{mathpar}
    \D :: \inferrule
    {\E :: \forall \alpha :: L' \sep \Phi \To L \sep \Psi ~.~ \ldctxeq[\Phi][L'] i
      \proglevel{\Delta_1}{\Delta_2} \\ \F :: \ldStypeq \E i \proglevel T{T'} l \\ \tyequiv[L]l{l'}\Level}
    {\ldctxeq i \proglevel{\Delta_1, x : T \at l}{\Delta_2, x : T' \at{l'}}}
  \end{mathpar}
  \begin{align*}
    & \ltctxgeq i{\Delta_1}{\Delta_2}
      \byIH \\
    & \ldsubeq[\Psi][\Delta_1] i \proglevel{\id}{\id}{\E(\id)}
      \tag{by \Cref{lem:dt:refl-lid}} \\
    & \ldtypeq[\Psi][\Delta_1] i \proglevel {T}{T'}{l}
      \tag{as $\F(\id, \E(\id))$} \\
    & \lttypgeq[\Psi][\Delta_1] i {T}{T'} l
      \tag{by escape}
  \end{align*}
  Hence we conclude the first statement.
  
  In the second statement, we have $\delta = \delta_1, t/x$ and $\delta' = \delta_1',
  t'/x$, then
  \begin{align*}
    & \ldsubeq i \proglevel {\delta_1}{\delta_1'}{\E(\id)}
      \tag{by assumption} \\
    & \ltsubgeq i {\delta_1}{\delta_1'}{\Delta_1} \tand \ltsubgeq i {\delta_1}{\delta_1'}{\Delta_2}
      \byIH \\
    & \ldtyequiv i \proglevel t{t'}{\F(\id, \E(\id))}
      \tag{by assumption} \\
    & \lttrmgeq i t{t'}{T[\delta_1]}l
      \tag{by escape of related terms}
  \end{align*}
  Therefore we conclude the second statement as well.
\end{proof}

\begin{lemma}[Symmetry] $ $
  \begin{itemize}
  \item If $\D :: \ldctxeq i \proglevel{\Delta_1}{\Delta_2}$, then $\E :: \ldctxeq i \proglevel{\Delta_2}{\Delta_1}$.
  \item If $\ldsubeq i \proglevel{\delta}{\delta'}\D$, then $\ldsubeq i \proglevel{\delta'}{\delta}\E$.
  \end{itemize}  
\end{lemma}
\begin{proof}
  Induction on $\D$. %
  Use the symmetry of related types to obtain the goal.
\end{proof}

\begin{lemma}[Right Irrelevance]
  If $\D :: \ldctxeq i \proglevel{\Delta_1}{\Delta_2}$, $\E :: \ldctxeq i
  \proglevel{\Delta_1}{\Delta_3}$ and $\ldsubeq i \proglevel{\delta}{\delta'}\D$, then
  $\ldsubeq i \proglevel{\delta}{\delta'}\E$.
\end{lemma}
\begin{proof}
  Do induction on $\D$ and then invert $\E$. %
  Use the right irrelevance of related terms. %
\end{proof}

\begin{lemma}[Left Irrelevance]
  If $\D :: \ldctxeq i \proglevel{\Delta_1}{\Delta_2}$, $\E :: \ldctxeq i
  \proglevel{\Delta_3}{\Delta_1}$ and $\ldsubeq i \proglevel{\delta}{\delta'}\D$, then
  $\ldsubeq i \proglevel{\delta}{\delta'}\E$.
\end{lemma}
\begin{proof}
   A direct consequence of right irrelevance and symmetry. 
\end{proof}

\begin{lemma}[Reflexivity and Transitivity] $ $
  \begin{itemize}
  \item If $\D_1 :: \ldctxeq i \proglevel{\Delta_1}{\Delta_2}$ and
    $\D_2 :: \ldctxeq i \proglevel{\Delta_2}{\Delta_3}$, then
    $\D_3 :: \ldctxeq i \proglevel{\Delta_1}{\Delta_3}$.
  \item If $\ldsubeq i \proglevel{\delta_1}{\delta_2}{\D_1}$ and
    $\ldsubeq i \proglevel{\delta_2}{\delta_3}{\D_2}$, then
    $\ldsubeq i \proglevel{\delta_1}{\delta_3}{\D_3}$.
  \item $\E :: \ldctxeq i \proglevel{\Delta_1}{\Delta_1}$.
  \item If $\ldsubeq i \proglevel{\delta_1}{\delta_2}{\D_1}$, then $\ldsubeq i \proglevel{\delta_1}{\delta_1}{\E}$.
  \end{itemize}  
\end{lemma}
\begin{proof}
  Do induction on $\D_1$ and then invert $\D_2$. %
  The proof proceeds very similarly to the relations for types and terms. %
  We only consider the step case. %
  In this case, we have the following premises:
  \begin{itemize}
  \item $\F_1 :: \forall~\alpha :: L' \sep \Phi \To L \sep \Psi ~.~
    \ldctxeq[\Phi][L'] i \proglevel {\Delta_1'}{\Delta_2'}$,
  \item $\F_2 :: \forall~\alpha :: L' \sep \Phi \To L \sep \Psi ~.~
    \ldctxeq[\Phi][L'] i \proglevel {\Delta_2'}{\Delta_3'}$,
  \item $\Ac_1 :: \ldStypeq{\F_1} i \proglevel {T_1}{T_2} l$,
  \item $\Ac_2 :: \ldStypeq{\F_2} i \proglevel{T_2}{T_3} l$,
  \item If $\ldStypeq{\F_1} i \proglevel T{T'} l$ and $\ldStypeq{\F_2} i \proglevel{T'}{T''} l$, then
    $\ldStypeq {\F_3} i \proglevel T{T''} l$.
  \item
    $\forall~\alpha :: L' \sep \Phi \To L \sep \Psi ~.~ \ldtyequiv[\Phi][\Gamma][L']
    i \proglevel {t_1}{t_2}{\Ac_1(\alpha, \F_1(\alpha))}$,
  \item
    $\forall~\alpha :: L' \sep \Phi \To L \sep \Psi ~.~ \ldtyequiv[\Phi][\Gamma][L']
    i \proglevel {t_2}{t_3}{\Ac_2(\alpha, \F_2(\alpha))}$,
  \end{itemize}
  By IH, it is easy to show
  \[
    \F_3 :: \forall~\alpha :: L' \sep \Phi \To L \sep \Psi ~.~
    \ldctxeq[\Phi][L'] i \proglevel {\Delta_1'}{\Delta_3'}
  \]
  The difficult goals are
  \begin{itemize}
  \item $\Ac_3 :: \ldStypeq{\F_3} i \proglevel{T_1}{T_3} l$, and
  \item
    $\forall~\alpha :: L' \sep \Phi \To L \sep \Psi ~.~ \ldtyequiv[\Phi][\Gamma][L']
    i \proglevel {t_1}{t_3}{\Ac_3(\alpha, \F_3(\alpha))}$. 
  \end{itemize}
  We first assume $\alpha :: L' \sep \Phi \To L \sep \Psi$. %
  To prove the first statement, we further assume $\ldsubeq[\Phi][\Gamma][L'] i j
  \delta{\delta'}{\F_3(\alpha)}$. %
  Then we do a similar reasoning to the function case for transitivity of related
  types. %
  \begin{enumerate}
  \item We first relate $T_1[\delta]$ and $T_2[\delta]$ by using reflexivity.
  \item Then we related $T_2[\delta]$ and $T_3[\delta']$.
  \item Then we apply transitivity of related types. 
  \end{enumerate}

  To prove the second statement, we can simply use irrelevance so that transitivity
  can eventually apply on $\Ac_3(\alpha, \F_3(\alpha))$. 
\end{proof}

\begin{lemma}[Transitivity]
  Given
  \begin{itemize}
  \item 
    $\F_1 :: \forall~\alpha :: L' \sep \Phi \To L \sep \Psi ~.~ \ldctxeq[\Phi][L'] i \proglevel
    {\Delta_1}{\Delta_2}$,
  \item $\F_2 :: \forall~\alpha :: L' \sep \Phi \To L \sep \Psi ~.~ \ldctxeq[\Phi][L'] i \proglevel
    {\Delta_2}{\Delta_3}$, and
  \item $\F_3 :: \forall~\alpha :: L' \sep \Phi \To L \sep \Psi ~.~ \ldctxeq[\Phi][L'] i \proglevel
    {\Delta_1}{\Delta_3}$,
  \end{itemize}
  then
  \begin{itemize}
  \item if $\D_1 :e: \ldStypeq{\F_1} i \proglevel T{T'} l$ and $\D_2 :: \ldStypeq{\F_2} i
    \proglevel{T'}{T''} l$, then 
    $\D_3 :: \ldStypeq{\F_3} i \proglevel T{T'} l$;
  \item if $\ldStyequiv{\F_1} i \proglevel t{t'}{\D_1}$ and $\ldStyequiv{\F_2} i \proglevel
    {t'}{t''}{\D_3}$, then $\ldStyequiv{\F_3} i \proglevel t{t''}{\D_3}$.
  \end{itemize}
\end{lemma}
\begin{proof}
  The first statement is what we have proved in transitivity previously. %
  The second statement uses the transitivity of related terms. %
\end{proof}

\subsection{Semantic Well-formedness of Global Contexts and Related Global Substitutions}\labeledit{sec:dt:semglob}

Following previous lines of work, when we give the Kripke logical relations to
contextual types, we will have to refer to related types and terms in the logical
relations under some invariants. %
Therefore, as early as it might seem, we must consider the semantics for global
contexts and global substitutions to see what it needs to learn how exactly the
semantics of contextual types can be defined. %
We then describe the semantics of global contexts and global substitutions, similar to
that of local contexts and local substitutions, as inductive-recursive definitions. %
Types and terms in a global substitution consist of two components:
\begin{enumerate}
\item the use of logical relations showing related terms having related computation,
\item together with the maintenance of syntactic structures, if the type or the term
  comes from layer $\varlevel$ or $\codelevel$.
\end{enumerate}
To handle the first component, we define several auxiliary definitions. %
These definitions essentially state that logical relations should remain stable under
weakenings and local substitutions. %
We give the definitions as follows. %
We define $\ldStypeqge \proglevel \proglevel T{T'} l$ to be given
\begin{itemize}
\item $\alpha :: L' \sep \Phi \To L \sep \Psi$,
\item $k \ge \proglevel$, and
\item $\ldsubeq[\Phi][\Delta][L'] k \proglevel \delta{\delta'}{\Gamma}$,
\end{itemize}
it holds that
\[
  \ldtypeq[\Phi][\Delta][L'] k \proglevel{T[\delta]}{T'[\delta']} l
\]
Similarly, we define $\ldStyequivge \proglevel \proglevel t{t'} T l$ to be given
\begin{itemize}
\item $\alpha :: L' \sep \Phi \To L \sep \Psi$, 
\item $k \ge \proglevel$, and
\item $\ldsubeq[\Phi][\Delta][L'] k \proglevel \delta{\delta'}{\Gamma}$,
\end{itemize}
we have
\[
  \ldtyequivt[\Phi][\Delta][L'] k \proglevel{t[\delta]}{t'[\delta']}{T[\delta]} l
\] %
We also need a similar relation for local contexts and local substitutions. %
$\ldSctxeqge \proglevel \proglevel \Delta{\Delta'}$ is defined as
\[
  \forall ~ \alpha :: L' \sep \Phi \To L \sep \Psi \tand k \ge \proglevel ~.~ \ldSctxeq[\Phi][L'] k j{\Delta}{\Delta}
\]
We define $\ldSsubeqge \proglevel \proglevel \delta{\delta'}\Delta$ as 
\begin{itemize}
\item $\alpha :: L' \sep \Phi \To L \sep \Psi$, 
\item $k \ge \proglevel$, and
\item $\ldsubeq[\Phi][\Delta'][L'] k \proglevel{\delta_1}{\delta_1'}{\Gamma}$,
\end{itemize}
we have
\[
  \ldsubeq[\Phi][\Delta'][L'] k \proglevel{\delta \circ \delta_1}{\delta' \circ \delta_1}{\Delta}
\] %

We define their non-relational version:
\begin{align*}
  \ldStypwfge \proglevel \proglevel T l &:= \ldStypeqge \proglevel \proglevel T{T} l \\
  \ldStypingge \proglevel \proglevel t T l &:= \ldStyequivge \proglevel \proglevel t{t} T l \\
  \ldSctxwfge \proglevel \proglevel \Delta &:= \ldSctxeqge \proglevel \proglevel \Delta{\Delta} \\
  \ldSsubstge \proglevel \proglevel \delta \Delta &:= \ldSsubeqge \proglevel \proglevel \delta\delta \Delta
\end{align*}
These would have been the semantic judgments for the fundamental theorems for MLTT. %
However, they are not enough if we want to establish the fundamental theorems for the
whole \delamlang. %
That is because the semantic judgments for \delamlang require stability under all
sorts of substitutions, including universe, global and local substitutions. %
Thus, in this section, we will define when a global context and global
substitutions are semantically well-formed in order to state the semantic judgments. %

These definitions are closed under weakenings and local substitutions by design:
\begin{lemma}[Weakenings] $ $
  \begin{itemize}
  \item If $\ldStypeqge \proglevel \proglevel T{T'} l$ and $\alpha :: L' \sep \Phi \To L \sep \Psi$,
    then $\ldStypeqge[\Phi][\Gamma][L'] \proglevel \proglevel T{T'} l$.
  \item If $\ldStyequivge \proglevel \proglevel t{t'} T l$ and $\alpha :: L' \sep \Phi \To L \sep \Psi$,
    then $\ldStyequivge[\Phi][\Gamma][L'] \proglevel \proglevel t{t'} T l$.
  \item If $\ldSctxeqge \proglevel \proglevel \Delta{\Delta'}$ and $\alpha :: L' \sep \Phi \To L \sep
    \Psi$, then $\ldSctxeqge[\Phi][L'] \proglevel \proglevel \Delta{\Delta'}$.
  \item If $\ldSsubeqge \proglevel \proglevel \delta{\delta'}\Delta$ and $\alpha :: L' \sep \Phi \To L
    \sep \Psi$, then $\ldSsubeqge[\Phi][\Gamma][L'] \proglevel \proglevel \delta{\delta'}\Delta$.
  \end{itemize}
\end{lemma}

\begin{lemma}[Local Substitutions]\labeledit{lem:dt:lsubst-ge} $ $
  \begin{itemize}
  \item If $\ldStypeqge \proglevel \proglevel T{T'} l$ and $\ldSsubeqge[\Psi][\Delta] \proglevel \proglevel \delta{\delta'}\Gamma$,
    then $\ldStypeqge[\Psi][\Delta] \proglevel \proglevel {T[\delta]}{T'[\delta']} l$.
  \item If $\ldStyequivge \proglevel \proglevel t{t'} T l$ and $\ldSsubeqge[\Psi][\Delta] \proglevel \proglevel \delta{\delta'}\Gamma$,
    then $\ldStyequivge[\Psi][\Delta] \proglevel \proglevel {t[\delta]}{t'[\delta']}{T[\delta]} l$.
  \item If $\ldSsubeqge \proglevel \proglevel \delta{\delta'}\Delta$ and
    $\ldSsubeqge[\Psi][\Delta'] \proglevel \proglevel{\delta_1}{\delta_1'}\Gamma$, then
    $\ldSsubeqge[\Phi][\Delta'] \proglevel \proglevel{\delta \circ \delta_1}{\delta' \circ
      \delta_1'}\Delta$.
  \end{itemize}
\end{lemma}

\begin{lemma}[Local Weakenings]
  If $\lpequiv \proglevel{\Delta,\Gamma}{\Delta',\Gamma}$ and $\ldSctxeqge \proglevel
  \proglevel{\Delta}{\Delta'}$, then \break
  $\ldSsubeqge[\Psi][\Delta,\Gamma] i j{\wk^{|\Gamma|}_{\Delta}}{\wk^{|\Gamma|}_{\Delta}}\Delta$.
\end{lemma}
\begin{proof}
  Notice that a local weakening only shorten a well-formed local substitution.
\end{proof}

\begin{lemma}[Semantic Conversion]
  If $\ldStyequivge \proglevel \proglevel t{t'} T l$ and $\ldStypeqge \proglevel \proglevel T{T'} l$, $\ldStyequivge \proglevel \proglevel t{t'}{T'} l$.
\end{lemma}

\begin{lemma}[PER]
  $\ldStypeqge \proglevel \proglevel T{T'} l$, $\ldStyequivge \proglevel \proglevel t{t'} T l$,
  $\ldSctxeqge \proglevel \proglevel \Delta{\Delta'}$ and $\ldSsubeqge i j \delta{\delta'}\Delta$ are
  PERs. 
\end{lemma}

Next, we handle the second component to record the syntactic information of types and terms. %
Following \citet{hu2024layered,10.1007/978-3-031-57262-3_3}, we use an inductive
definition to keep track of the syntactic structure, in which semantic information is
also maintained. %
In fact, we need three mutually inductive judgments very similar to typing judgments. %
\begin{itemize}
\item $\ldtypwf \codelevel \proglevel T l$ stores the syntactic information and the semantic
  information of $T$ and all its sub-structures.
\item $\ldtyping i \proglevel t T l$ stores the syntactic information and the semantic
  information of $t$ and all its sub-structures.
\item $\ldsubst i \proglevel \delta{\Delta}$ stores the syntactic information and the semantic
  information of all terms in $\delta$. 
\end{itemize}
We restrict the parameter layer $i \in \{\varlevel, \codelevel\}$. %
These definitions will make use of $\ldStypwfge \proglevel \proglevel T l$, $\ldStypingge \proglevel \proglevel t T l$ and
$\ldSsubstge \proglevel \proglevel \delta \Delta$. %
% This is fine as we should consider definitions above as shorthands for large universal
% quantifications. %
% In other words, $\ldStypwfge \codelevel \proglevel T l$ etc. exists but does not cause cyclic
% dependencies. %

Their definitions are designed to imply semantic information:
\begin{lemma}\labeledit{lem:dt:ldc-pm} $ $
  \begin{itemize}
  \item If $\ldtypwf \codelevel \proglevel T l$, then $\ldStypwfge \proglevel \proglevel T l$.
  \item If $\ldtyping i \proglevel t T l$, then $\ldStypingge \proglevel \proglevel t T l$.
  \item If $\ldsubst i \proglevel \delta \Delta$, then $\ldSsubstge \proglevel \proglevel \delta \Delta$.
  \end{itemize}
\end{lemma}
Further we let
\begin{align*}
  \ldtypeq \codelevel \proglevel T{T'} l &:= \ldtypwf \codelevel \proglevel T l \tand T = T' \\
  \ldtyequivt i \proglevel t{t'} T l & := \ldtyping i \proglevel t T l \tand t = t' \\
  \ldsubeq i \proglevel \delta{\delta'}{\Delta} &:= \ldsubst i \proglevel \delta{\Delta} \tand \delta = \delta'
\end{align*}

We first consider the judgment for types:
\begin{mathpar}
  \inferrule
  {\tyequiv[L]l 0\Level \\ \ldStypwfge \proglevel \proglevel \Nat l}
  {\ldtypwf \codelevel \proglevel \Nat l}

  \inferrule
  {\tyequiv[L]l{1 + l'}\Level \\ \ldStypwfge \proglevel \proglevel{\Ty{l'}}{l}}
  {\ldtypwf \codelevel \proglevel{\Ty{l'}}l}

  \inferrule
  {\tyequiv[L]l{l_1 \sqcup l_2}\Level \\ \typing[L]{l_1}\Level \\ \typing[L]{l_2}\Level \\
    \ldtypwf \codelevel \proglevel S{l_1} \\ \ldtypwf[\Psi][\Gamma, x : S \at{l_1}] \codelevel \proglevel T{l_2} \\
    \ldStypwfge \proglevel \proglevel{{\PI{l_1}{l_2} x S T}} l}
  {\ldtypwf \codelevel \proglevel{\PI{l_1}{l_2} x S T}{l}}

  \inferrule
  {\tyequiv[L]l{l'}\Level \\  U : \DTyp[\Delta]{\codelevel}{l'} \in \Psi \\ \ldsubst \codelevel \proglevel \delta{\Delta} \\ \ldStypwfge \proglevel \proglevel{U^\delta}{l}}
  {\ldtypwf \codelevel \proglevel{U^\delta}{l}}
  
  \inferrule
  {\tyequiv[L]{l'}l\Level \\
    \ldtyping \codelevel \proglevel t{\Ty{l'}}{1 + l} \\
    \ldStypwfge \proglevel \proglevel{\Elt{l'} t}{l}}
  {\ldtypwf \codelevel \proglevel{\Elt{l'} t}{l}}
\end{mathpar}
The judgment not only keeps track of the syntactic structure of types but also the
semantic information for both layers $\proglevel$ and $\metalevel$. %
The semantic information for both layers is critical to enable code running, when we
refer to the code of types at layer $\proglevel$ and $\metalevel$, resp.. % 

The judgment for local substitutions is simple, which simply generalizes that for
terms:
\begin{mathpar}
  \inferrule
  {\ldSsubstge \proglevel \proglevel{\cdot^k}{\cdot}}
  {\ldsubst i \proglevel {\cdot^k}{\cdot}}

  \inferrule
  {\ldSsubstge \proglevel \proglevel{\cdot_g^k}{\cdot}}
  {\ldsubst i \proglevel {\cdot_g^k}{\cdot}}

  \inferrule
  {\ldSsubstge \proglevel \proglevel{\wk_g^k}{g}}
  {\ldsubst i \proglevel {\wk_g^k}{g}}

  \inferrule
  {\ldsubst i \proglevel {\delta}{\Delta} \\ \ldtyping i \proglevel t{T[\delta]}{l} \\
    \ldSsubstge \proglevel \proglevel{\delta, t/x}{\Delta, x : T \at l}}
  {\ldsubst i \proglevel {\delta, t/x}{\Delta, x : T \at l}}
\end{mathpar}

We also keep track of both syntactic and semantic information of terms. %
We give the definition as follows (the parameter $i$ is restricted to be in $\{\varlevel, \codelevel\}$
when used only for the rules below):
\begin{mathpar}
  \inferrule
  {x : T \at{l'} \in \Gamma \\
    \tyequiv[L]l{l'}\Level \\
    \ldStypeqge \proglevel \proglevel{T'}{T}l \\
    \ldStypingge \proglevel \proglevel x{T'}l}
  {\ldtyping i \proglevel x{T'}{l}}

  \inferrule
  {u : \DTrm[\Delta]{i'}{T}{l'} \in \Psi \\ i' \in \{\varlevel, \codelevel\} \\ i' \le i \\
    \ldsubst i \proglevel \delta{\Delta} \\
    \tyequiv[L]l{l'}\Level \\
    \ldStypeqge \proglevel \proglevel{T'}{T[\delta]}l \\
    \ldStypingge \proglevel \proglevel{u^\delta}{T'}l}
  {\ldtyping i \proglevel{u^\delta}{T'}l}

  \inferrule
  {\tyequiv[L]l{2+l'}\Level \\
    \ldStypeqge \proglevel \proglevel{T}{\Ty{1 + l'}}l \\
    \ldStypingge \proglevel \proglevel{\Ty{l'}}T l}
  {\ldtyping \codelevel \proglevel{\Ty{l'}}{T}{l}}
  
  \inferrule
  {\tyequiv[L]l{1}\Level \\
    \ldStypeqge \proglevel \proglevel{T}{\Ty 0}l \\
    \ldStypingge \proglevel \proglevel{\Nat}T l}
  {\ldtyping \codelevel \proglevel \Nat T{l}}

  \inferrule
  {\tyequiv[L]l{0}\Level \\
    \ldStypeqge \proglevel \proglevel{T}{\Nat}l \\
    \ldStypingge \proglevel \proglevel{\ze}T l} 
  {\ldtyping \codelevel \proglevel \ze T l}

  \inferrule
  {\ldtyping \codelevel \proglevel t \Nat 0 \\
   \tyequiv[L]l{0}\Level \\
   \ldStypeqge \proglevel \proglevel{T}{\Nat}l \\
   \ldStypingge \proglevel \proglevel{\su t}T l}
  {\ldtyping \codelevel \proglevel {\su t}T l}

  \inferrule
  {\ldtypwf[\Psi][\Gamma, x : \Nat \at 0] \codelevel \proglevel M l \\
    \ldtyping \codelevel \proglevel s {M[\ze/x]}l \\
    \ldtyping[\Psi][\Gamma, x : \Nat \at 0, y : M \at l] \codelevel \proglevel {s'}{M[\su x/x]}l \\
    \ldtyping \codelevel \proglevel t \Nat \ze \\
    \tyequiv[L]{l'}l\Level \\
    \ldStypeqge \proglevel \proglevel{T}{M[t/x]}{l'} \\
    \ldStypingge \proglevel \proglevel{\ELIMN l{x.M}s{x,y. s'}t}T{l'}}
  {\ldtyping \codelevel \proglevel{\ELIMN l{x.M}s{x,y. s'}t}{T}{l'}}

  \inferrule
  {\typing[L]{l_1}\Level  \\ \typing[L]{l_2}\Level \\
    \ldtyping \codelevel \proglevel s{\Ty{l_1}}{1 + l_1} \\ \ldtyping[\Psi][\Gamma, x : \Elt{l_1} S \at {l_1}] \codelevel \proglevel t{\Ty{l_2}}{1 + l_2} \\
    \tyequiv[L]l{\su{(l_1 \sqcup l_2)}}\Level \\
    \ldStypeqge \proglevel \proglevel{T}{\Ty{l_1 \sqcup l_2}}l \\
    \ldStypingge \proglevel \proglevel{\PI{l_1}{l_2} x s t}T l}
  {\ldtyping \codelevel \proglevel{\PI{l_1}{l_2} x s t}{T}{l}}

  \inferrule
  {\typing[L]{l_1}\Level  \\ \typing[L]{l_2}\Level \\ \ldtypwf \codelevel \proglevel S{l_1} \\
    \ldtyping[\Psi][\Gamma, x : S \at{l_1}] \codelevel \proglevel t{T}{l_2} \\
    \tyequiv[L]l{l_1 \sqcup l_2}\Level \\
    \ldStypeqge \proglevel \proglevel{T'}{\PI{l_1}{l_2} x S T}l \\
    \ldStypingge \proglevel \proglevel{\LAM {l_1}{l_2} x S t}{T'}l}
  {\ldtyping \codelevel \proglevel{\LAM {l_1}{l_2} x S t}{T'}{l}}

  \inferrule
  {\typing[L]{l_1}\Level  \\ \ldtypwf \codelevel \proglevel S{l_1} \\
    \ldtypwf[\Psi][\Gamma, x : S \at{l_1}] \codelevel \proglevel {T}{l_2} \\
    \ldtyping \codelevel \proglevel t{\PI {l_1}{l_2} x S T}{l_1 \sqcup l_2} \\ \ldtyping \codelevel \proglevel s S{l_1} \\
    \tyequiv[L]l{l_2}\Level \\
    \ldStypeqge \proglevel \proglevel{T'}{T[s/x]}l \\
    \ldStypingge \proglevel \proglevel{\APP t{l_1}{l_2} x S T s}{T'}l}
  {\ldtyping \codelevel \proglevel{\APP t{l_1}{l_2} x S T s}{T'}{l}}
\end{mathpar}
These rules strictly decrease on the structures of types, terms and local
substitutions. %
The purpose of these rules are to remember all computational and syntactic information
of all sub-structures. %
Notice that the annotated types do not necessarily match up with the precise types
which terms might have. %
For example, $\ze$ is quantified to have seemingly some arbitrary type $T$. %
This, however, is not an accurate understanding. %
In fact, $T$ is quantified by $\ldStypingge \proglevel \proglevel{\ze}T l$, so if we provide a proof of
$\ldStypeqge \proglevel \proglevel{T}{T'}l$, then we obtain $\ldStypingge \proglevel \proglevel{\ze}{T'} l$ by
irrelevance. %
More specifically,
\begin{lemma}[Conversion]
  If $\ldtyping \codelevel \proglevel t T l$ and $\ldStypeqge \proglevel \proglevel{T}{T'}l$, then $\ldtyping \codelevel \proglevel t{T'} l$.
\end{lemma}
The rules are designed in this way so that there is a semantic structure which the
recursors can recurse on. %
Moreover, these rules are closed under weakenings and local substitutions, which is particularly
important for use of global variables.
\begin{lemma}[Weakenings] $ $
  \begin{itemize}
  \item If $\ldtypwf \codelevel \proglevel T l$ and $\alpha :: L' \sep \Phi \To L \sep \Psi$,
    then $\ldtypwf[\Phi][\Gamma][L'] \codelevel \proglevel T l$.
  \item If $i \in \{\varlevel, \codelevel\}$, $\ldtyping i \proglevel t T l$ and $\alpha :: L' \sep \Phi \To L \sep \Psi$,
    then $\ldtyping[\Phi][\Gamma][L'] i \proglevel t T l$.
  \item If $i \in \{\varlevel, \codelevel\}$, $\ldsubst i \proglevel \delta\Delta$ and
    $\alpha :: L' \sep \Phi \To L \sep \Psi$, then
    $\ldsubst[\Phi][\Gamma][L'] i \proglevel \delta\Delta$.
  \end{itemize}
\end{lemma}

\begin{lemma}[Local Substitutions]\labeledit{lem:dt:lsubst-cp} $ $
  \begin{itemize}
  \item If $\ldtypwf \codelevel \proglevel T l$ and $\ldsubst[\Psi][\Delta] \codelevel \proglevel \delta\Gamma$,
    then $\ldtypwf[\Psi][\Delta] \codelevel \proglevel {T[\delta]} l$.
  \item If $i \in \{\varlevel, \codelevel\}$, $\ldtyping i \proglevel t T l$ and $\ldsubst[\Psi][\Delta] i \proglevel \delta\Gamma$,
    then $\ldtyping[\Psi][\Delta] i \proglevel {t[\delta]}{T[\delta]} l$.
  \item If $i \in \{\varlevel, \codelevel\}$, $\ldsubst i \proglevel \delta\Delta$ and
    $\ldsubst[\Psi][\Delta'] i \proglevel{\delta'}\Gamma$, then
    $\ldsubst[\Phi][\Delta'] i \proglevel{\delta \circ \delta'}\Delta$.
  \end{itemize}
\end{lemma}
\begin{proof}
  We proceed by mutual induction. %
  The syntax is already closed under local substitutions. %
  The semantics is also closed under local substitutions by \Cref{lem:dt:lsubst-ge}. 
\end{proof}

We can lift from layer $\varlevel$ to layer $\codelevel$.
\begin{lemma}[Lifting] $ $
  \begin{itemize}
  \item If $\ldtyping \varlevel \proglevel t T l$, then $\ldtyping \codelevel \proglevel t T l$.
  \item If $\ldsubst \varlevel \proglevel \delta\Delta$, then $\ldsubst \codelevel \proglevel \delta\Delta$.
  \end{itemize}
\end{lemma}

\begin{lemma}[Local Weakenings]
  If $\lpequiv \proglevel{\Delta,\Gamma}{\Delta',\Gamma}$ and $\ldSctxeqge \proglevel
  \proglevel{\Delta}{\Delta'}$, then \break
  $\ldsubst[\Psi][\Delta,\Gamma] \codelevel \proglevel{\wk^{|\Gamma|}_{\Delta}}\Delta$.
\end{lemma}

In the semantic rules for global contexts and global substitutions, we use the
auxiliary definitions above. %
Similar to the counterparts for local contexts and local substitutions, the rules are
defined in an induction-recursion. %
The key idea is to make sure the relation is invariant under universe weakenings. %
We also impose invariance under global substitutions for related local contexts, types
and terms. %
We are going to define the following definitions:
\begin{itemize}
\item $\D :: \ldSgctxeq \Psi \Phi$ denotes the semantic related global contexts.
\item $\ldSgsubeq \sigma{\sigma'} \D$ denotes the semantic relation between the global
  substitutions $\sigma$ and $\sigma'$. %
  The two definitions above are defined inductive-recursively. %
  
\item $\F :: \ltSctxwf \E \Gamma$ given $\E :: \forall~\theta :: L' \To L ~.~ \ldSgctx
  \Phi$ denotes $\Gamma$ is semantically well-formedness and is stable under universe
  weakening and related global substitutions. 
\item $\ltStypwf \F \proglevel \proglevel T l$ denotes that $T$ is semantically well-formedness and is
  stable under universe weakening and related global substitutions.
\end{itemize}

We now move on to the actual definitions.
\begin{mathpar}
  \D :: \inferrule
  { }
  {\ldSgctxeq \cdot \cdot}
\end{mathpar}
$\ldSgsubeq \sigma{\sigma'} \D$ is defined as $\judge[L]\Psi$ and $\sigma = \sigma' = \cdot$.

\begin{mathpar}
  \D :: \inferrule
  {\E :: \forall~\theta :: L' \To L ~.~ \ldSgctxeq[L']\Phi{\Phi'}}
  {\ldSgctxeq{\Phi, g : \Ctx}{\Phi', g : \Ctx}}
\end{mathpar}
$\ldSgsubeq \sigma{\sigma'} \D$ is defined as
\begin{itemize}
\item $\sigma = \sigma_1, \Gamma /g$ and $\sigma' = \sigma_1', \Gamma' /g$, 
\item $\forall~\theta :: L' \To L ~.~
  \ldSgsubeq[\Psi][L']{\sigma_1}{\sigma_1'}{\E(\theta)}$, and
\item $\ldSctxeqge \proglevel \proglevel \Gamma{\Gamma'}$.
\end{itemize}

\begin{mathpar}
  \D :: \inferrule
  {\E :: \forall~\theta :: L' \To L ~.~ \ldSgctxeq[L']\Phi{\Phi'} \\
    \ltSctxeq \E \Gamma{\Gamma'} \\ i \in \{\codelevel, \proglevel\}}
  {\ldSgctxeq{\Phi, U : \DTyp i l}{\Phi', U : \DTyp[\Gamma'] i l}}
\end{mathpar}
$\ldSgsubeq \sigma{\sigma'} \D$ is defined as
\begin{itemize}
\item $\sigma = \sigma_1, T/U$ and $\sigma' = \sigma_1', T'/U$,
\item $\forall~\theta :: L' \To L ~.~
  \ldSgsubeq[\Psi][L']{\sigma_1}{\sigma_1'}{\E(\theta)}$,
\item We analyze $i$.
  \begin{itemize}
  \item If $i = \codelevel$, then $T = T'$, $\ldtypwf \codelevel \proglevel T l$ and
    $\ldtypwf[\Psi][\Gamma'] \codelevel \proglevel T l$.
  \item Otherwise, $i = \proglevel$, then $\ldStypeqge \proglevel \proglevel T{T'} l$ and $\ldStypeqge[\Psi][\Gamma'] \proglevel \proglevel T{T'} l$.
  \end{itemize}
\end{itemize}
where $\ltSctxeq \E \Gamma \Delta$ is defined as
\[
  \forall~~\theta :: L' \To L \tand \ldSgsubeq[\Psi][L'] \sigma{\sigma'}{\E(\theta)} \tand k \in
  \{\proglevel, \metalevel\}~.~
  \ldctxeq[\Psi][L'] k \proglevel{\Gamma[\sigma]}{\Delta[\sigma']}
\]

\begin{mathpar}
  \D :: \inferrule
  {\E :: \forall~\theta :: L' \To L ~.~ \ldSgctxeq[L']\Phi{\Phi'} \\
    \F :: \ltSctxeq \E \Gamma{\Gamma'} \\
    \Ac :: \ltStypeq \F \proglevel \proglevel T{T'} l \\
    i \in \{\varlevel, \codelevel\}}
  {\ldSgctxeq{\Phi, u : \DTrm i T l}{\Phi', u : \DTrm[\Gamma'] i{T'} l}}
\end{mathpar}
$\ldSgsubeq \sigma{\sigma'} \D$ is defined as
\begin{itemize}
\item $\sigma = \sigma_1, t/u$ and $\sigma' = \sigma_1', t'/u$,
\item $\forall~\theta :: L' \To L ~.~
  \ldSgsubeq[\Psi][L']{\sigma_1}{\sigma_1'}{\E(\theta)}$,
\item Since we know $i \in \{\varlevel, \codelevel\}$, we have $t = t'$,
  $\ldtyping[\Psi][\Gamma[\sigma_1]] i \proglevel t{T[\sigma_1]} l$ and \break
  $\ldtyping[\Psi][\Gamma'[\sigma_1']] i \proglevel t{T'[\sigma_1']} l$. %
  In an extension where $i = \proglevel$ is possible, then we use \break
  $\ldStyequivge[\Psi][\Gamma[\sigma_1]] \proglevel \proglevel t{t'}{T[\sigma_1]} l$ and 
  $\ldStyequivge[\Psi][\Gamma'[\sigma_1']] \proglevel \proglevel t{t'}{T'[\sigma_1']} l$.
\end{itemize}
where $\ltStypeq \F \proglevel \proglevel T{T'} l$ is defined as given
\begin{itemize}
\item $\theta :: L' \To L$, 
\item $\Bc :: \ldSgsubeq[\Psi][L'] \sigma{\sigma'}{\E(\theta)}$,
\item $k \in \{\proglevel, \metalevel\}$, and
\item $\ldsubeq[\Psi][\Delta][L'] k \proglevel \delta{\delta'}{\F(\theta, \Bc, k)}$,
\end{itemize}
we have
\[
  \ldtypeq[\Psi][\Delta][L'] k \proglevel{T[\sigma][\delta]}{T'[\sigma'][\delta']} l
\]

Now we consider the properties of the new definitions.
\begin{lemma}[Weakening] $ $
  \begin{itemize}
  \item If $\D :: \ldSgctxeq \Phi{\Phi'}$ and $\theta :: L' \To L$, then
    $\D' :: \ldSgctxeq[L'] \Phi{\Phi'}$.
  \item If $\ldSgsubeq \sigma{\sigma'} \D$ and $\theta :: L' \To L$, then
    $\ldSgsubeq[\Psi][L'] \sigma{\sigma'}{\D'}$. 
  \end{itemize}
\end{lemma}
\begin{proof}
  Immediate by induction on $\D$. We simply use the composition of universe weakenings. 
\end{proof}

\begin{lemma}[Reflexive Global Weakening]
  If $\judge[L]{\Phi,\Psi}$ and $\D :: \ldSgctxeq \Phi{\Phi'}$, then \break
  $\ldSgsubeq[\Phi,\Psi]{\wk_{\Phi}^{|\Psi|}}{\wk_{\Phi}^{|\Psi|}}{\D}$. 
\end{lemma}
\begin{proof}
  We do induction on $\D$.
  \begin{itemize}[label=Case]
  \item
    \begin{mathpar}
      \D :: \inferrule
      {\E :: \forall~\theta :: L' \To L ~.~ \ldSgctxeq[L']{\Phi_1}{\Phi_1'}}
      {\ldSgctxeq{\Phi_1, g : \Ctx}{\Phi_1', g : \Ctx}}
    \end{mathpar}
    We have given $\theta :: L' \To L$,
    \begin{align*}
      & \ldSgsubeq[\Phi_1,g:\Ctx,\Psi]{\wk_{\Phi_1}^{1 + |\Psi|}}{\wk_{\Phi_1}^{1 +
        |\Psi|}}{\E(\theta)}
        \byIH \\
      & \alpha :: L'' \sep \Psi' \To L' \sep \Phi_1,g:\Ctx,\Psi \\
      & k \in \{\proglevel, \metalevel\}
        \tag{by assumption} \\
      & \ldctxeq[\Psi'][L''] k \proglevel{g}{g}
        \tag{by definition}
    \end{align*}
    In the last line, we know $g$ exits in $\Psi'$ because $\alpha$ is a weakening. %
    Hence we have this case.
  \item
    \begin{mathpar}
      \D :: \inferrule
      {\E :: \forall~\theta :: L' \To L ~.~ \ldSgctxeq[L']{\Phi_1}{\Phi_1'} \\
        \ltSctxeq \E \Gamma{\Gamma'} \\ i \in \{\codelevel, \proglevel\}}
      {\ldSgctxeq{\Phi_1, U : \DTyp i l}{\Phi_1', U : \DTyp[\Gamma'] i l}}
    \end{mathpar}
    \begin{itemize}[label=Subcase]
    \item If $i = \proglevel$, then given $\theta :: L' \To L$ and $k \in \{\proglevel, \metalevel\}$, then we
      further assume \break
      $\alpha :: L'' \sep \Psi' \To L' \sep \Phi_1,U : \DTyp i l,\Psi$ and
      $\ldsubeq[\Psi'][\Delta][L''] k \proglevel \delta{\delta'}{\Gamma}$, we should prove
      \[
        \ldtypeq[\Psi'][\Delta][L''] k \proglevel{U^\delta}{U^{\delta'}} l
      \]
      There is a symmetric proof for $\Gamma'$ which we omit here. %
      We proceed as follows:
      \begin{align*}
        & \ltsubgeq[\Psi'][\Delta][L''] k \delta{\delta'}{\Gamma}
          \tag{by escape, \Cref{lem:dt:ldc-pm}} \\
        & \lttypgneeq[\Psi'][\Delta][L''] k{U^\delta}{U^{\delta'}} l
          \tag{by the neutral types law} \\
        & \ldtypeq[\Psi'][\Delta][L''] k \proglevel{U^\delta}{U^{\delta'}} l
          \tag{by the neutral type case}
      \end{align*}
      
    \item If $i = \codelevel$, then we know $\ldsubst[\Psi] \codelevel \proglevel{\id}\Gamma$. %
      Combining the previous case, we have the goal.
    \end{itemize}
  \item
    \begin{mathpar}
      \D :: \inferrule
      {\E :: \forall~\theta :: L' \To L ~.~ \ldSgctxeq[L']{\Phi_1}{\Phi_1'} \\
        \F :: \ltSctxeq \E \Gamma{\Gamma'} \\
        \Ac :: \ltStypeq \F \proglevel \proglevel T{T'} l \\
        i \in \{\varlevel, \codelevel\}}
      {\ldSgctxeq{\Phi_1, u : \DTrm i T l}{\Phi_1', u : \DTrm[\Gamma'] i T' l}}
    \end{mathpar}
    We first have
    \[
      \Bc :: \ldSgsubeq[\Phi_1,u : \DTrm i T l,\Psi]{\wk_{\Phi_1}^{1 + |\Psi|}}{\wk_{\Phi_1}^{1 +
          |\Psi|}}{\E(\theta)}
    \]
    by IH.

    Then without loss of generality, we should prove
    \[
      \ldtyping[\Phi_1,u : \DTrm i T l,\Psi] \codelevel \proglevel{u^\id}T l
    \]
    which in turn requires
    \[
      \ldStypingge[\Phi_1,u : \DTrm i T l,\Psi] \proglevel \proglevel{u^\id}T l
    \]
    Given $\theta :: L' \To L$ and $k \ge i$,  $\alpha :: L'' \sep \Psi' \To L'
    \sep \Phi_1,g:\Ctx,\Psi$ and
    $\Cc :: \ldsubeq[\Psi'][\Delta][L''] k \proglevel \delta{\delta'}{\Gamma}$, we should
    prove
    \[
      \ldtyequivt[\Psi'][\Delta][L''] k \proglevel{u^\delta}{u^{\delta'}}{T[\delta]} l
    \]
    We proceed as follows:
    \begin{align*}
      & \ltsubgeq[\Psi'][\Delta][L''] k \delta{\delta'}{\Gamma}
        \tag{by escape, \Cref{lem:dt:ldc-pm}} \\
      & \ldtypeq[\Psi'][\Delta][L''] k \proglevel{T[\wk_{\Phi_1}^{1 +
        |\Psi|}][\delta]}{T[\wk_{\Phi_1}^{1 + |\Psi|}][\delta']}l
        \tag{as $\Ac(\theta, \Bc, k, \Cc)$} \\
      \Cc' ::~& \ldtypeq[\Psi'][\Delta][L''] k \proglevel{T[\delta]}{T[\delta']}l
                \tag{omit weakening}\\
      & \lttrmgneeq[\Psi'][\Delta][L''] k{u^\delta}{u^{\delta'}}{T[\delta]} l
        \tag{by the congruence for neutrals law} \\
      & \ldtyequiv[\Psi'][\Delta][L''] k \proglevel{u^\delta}{u^{\delta'}}{\Cc'}
        \tag{by reflexivity of neutral}
    \end{align*}
  \end{itemize}
\end{proof}

\begin{corollary}[Reflexive Global Identity]
  If $\judge[L]{\Phi}$ and $\D :: \ldSgctx \Phi$, then
  $\ldSgsubeq[\Phi]{\id_{\Phi}}{\id_{\Phi}}{\D}$. 
\end{corollary}

\begin{lemma}[Escape] $ $
  \begin{itemize}
  \item If $\D :: \ldSgctxeq \Phi{\Phi'}$, then $\judge[L]\Phi$ and $\judge[L]\Phi'$.
  \item If $\ldSgsubeq{\sigma}{\sigma'}{\D}$, then $\ptyequiv{\sigma}{\sigma'}\Phi$
    and $\ptyequiv{\sigma}{\sigma'}{\Phi'}$. 
  \end{itemize}
\end{lemma}
\begin{proof}
  We do induction on $\D$. We consider one case:
  \begin{mathpar}
    \D :: \inferrule
    {\E :: \forall~\theta :: L' \To L ~.~ \ldSgctxeq[L']{\Phi_1}{\Phi_1'} \\
      \F :: \ltSctxeq \E \Gamma{\Gamma'} \\
      \Ac :: \ltStypeq \F \proglevel \proglevel T{T'} l \\
      i \in \{\varlevel, \codelevel\}}
    {\ldSgctxeq{\Phi_1, u : \DTrm i T l}{\Phi_1', u : \DTrm[\Gamma'] i{T'} l}}
  \end{mathpar}
  We proceed as follows:
  \begin{align*}
    & \judge[L]{\Phi_1} \tand \judge[L]{\Phi_1'}
      \byIH \\
    & \ldctxeq \proglevel \proglevel \Gamma{\Gamma'}
      \tag{by $\F$} \\
    & \lpequiv \proglevel \Gamma{\Gamma'}
      \tag{applying escape to the previous line} \\
    & \ldtypwf \proglevel \proglevel T l
      \tag{by $\Ac$, using reflexive global and local identities} \\
    & \lttypwf \proglevel T l \tand \lttypwf \proglevel{T'} l
      \tag{by escape} \\
    & \judge[L]{\Phi_1, u : \DTrm i T l} \tand \judge[L]{\Phi_1', u : \DTrm[\Gamma'] i{T'} l}
      \tag{by definition}
  \end{align*}
  Then we consider the global substitutions. %
  We know $\sigma = \sigma_1, t/u$ and $\sigma' = \sigma_1', t'/u$ and $t = t'$. %
  Then by IH, we have
  \[
    \ptyequiv{\sigma_1}{\sigma_1'}{\Phi_1} \tand \ptyequiv{\sigma_1}{\sigma_1'}{\Phi_1'}
  \]
  Since we know $i \in \{\varlevel,\codelevel\}$, by passing in a local identity substitution, we have
  \[
    \ldtyequivt[\Psi][\Gamma[\sigma_1]] \proglevel \proglevel t{t'}{T[\sigma_1]} l
  \]
  We then have $\lttyequiv[\Psi][\Gamma[\sigma_1]]i t{t'}{T[\sigma_1]} l$ by escape. %
  We do it similarly for
  $\lttyequiv[\Psi][\Gamma'[\sigma_1]] \proglevel t{t'}{T'[\sigma_1]} l$.  By presupposition,
  unlifting and analysis using $i$, we can conclude the desired goal.
\end{proof}

\begin{lemma}[Symmetry] $ $
  \begin{itemize}
  \item If $\D :: \ldSgctxeq \Phi{\Phi'}$, then $\E :: \ldSgctxeq{\Phi'}{\Phi}$.
  \item If $\ldSgsubeq{\sigma}{\sigma'}{\D}$, then
  $\ldSgsubeq{\sigma'}{\sigma}{\E}$.
  % \item If $\ldStypeqge \proglevel \proglevel T{T'} l$,
  %   then $\ldStypeqge \proglevel \proglevel {T'}{T} l$.
  % \item If $\ldStyequivge \proglevel \proglevel t{t'} T l$,
  %   then $\ldStyequivge \proglevel \proglevel t{t'} T l$.
  % \item If $\ldSctxeqge \proglevel \proglevel {\Delta}{\Delta'}$,
  %   then $\ldSctxeqge \proglevel \proglevel {\Delta'}{\Delta}$.
  % \item If $\ldSsubeqge \proglevel \proglevel \delta{\delta'}\Delta$,
  %   then $\ldSsubeqge \proglevel \proglevel {\delta'}\delta\Delta$.
  \end{itemize}
\end{lemma}
\begin{proof}
  We do induction on $\D$. %
  Apply symmetry for previous definitions. %
\end{proof}

\begin{lemma}[Right Irrelevance]
  If $\D :: \ldSgctxeq \Phi{\Phi'}$, $\E :: \ldSgctxeq \Phi{\Phi''}$ and
  $\ldSgsubeq{\sigma}{\sigma'}{\D}$, then $\ldSgsubeq{\sigma}{\sigma'}{\E}$.
\end{lemma}
\begin{proof}
  We do induction on $\D$ and then invert $\E$. %
  % This lemma is also immediate because $\D$ and $\E$ generates precisely the same set
  % of premises in each case. %
  % $\ldSgsubeq{\sigma}{\sigma'}{\D}$ only depends on $\D$ in the recursive relation
  % which is rather immediate to swap to $\E$ by IH.
  This lemma is proved similar to previous irrelevance lemma. %
  IHs are sufficient to discharge proof obligations. 
\end{proof}

\begin{lemma}[Left Irrelevance]
  If $\D :: \ldSgctxeq {\Phi'}{\Phi}$, $\E :: \ldSgctxeq {\Phi''}\Phi$ and
  $\ldSgsubeq{\sigma}{\sigma'}{\D}$, then $\ldSgsubeq{\sigma}{\sigma'}{\E}$.
\end{lemma}
\begin{proof}
  Apply symmetry and right irrelevance. 
\end{proof}

\begin{lemma}[Transitivity] $ $
  \begin{itemize}
  \item If $\D_1 :: \ldSgctxeq{\Phi_1}{\Phi_2}$ and $\D_2 ::
    \ldSgctxeq{\Phi_2}{\Phi_3}$, then $\D_3 :: \ldSgctxeq{\Phi_1}{\Phi_3}$.
  \item If $\E :: \ldSgctxeq{\Phi_1}{\Phi_1}$, $\ldSgsubeq{\sigma_1}{\sigma_2}{\D_1}$ and
    $\ldSgsubeq{\sigma_2}{\sigma_3}{\D_2}$, then $\ldSgsubeq{\sigma_1}{\sigma_3}{\D_3}$.
  % \item If $i \ge \codelevel$, $\ldStypeqge i \proglevel T{T'} l$ and $\ldStypeqge i \proglevel {T'}{T''} l$,
  %   then $\ldStypeqge i \proglevel {T}{T''} l$.
  % \item If $\ldStyequivge i \proglevel t{t'} T l$ and $\ldStyequivge i \proglevel {t'}{t''} T l$,
  %   then $\ldStyequivge i \proglevel t{t''} T l$.
  % \item If $i \ge \proglevel$, $\ldSctxeqge i \proglevel {\Delta}{\Delta'}$ and $\ldSctxeqge i \proglevel {\Delta'}{\Delta''}$,
  %   then $\ldSctxeqge i \proglevel {\Delta}{\Delta''}$.
  % \item If $\ldSsubeqge \proglevel \proglevel \delta{\delta'}\Delta$ and $\ldSsubeqge \proglevel \proglevel {\delta'}{\delta''}\Delta$,
  %   then $\ldSsubeqge \proglevel \proglevel {\delta}{\delta''}\Delta$.
  \end{itemize}
\end{lemma}
\begin{proof}
  We do induction on $\D_1$ and then invert
  $\D_2$. %
  % The other three statements are proved during the process of proving the first two
  % statements. %
  The lemma is immediate after use of reflexivity and transitivity of previous
  definitions.  %
\end{proof}

\subsection{Logical Relations When $j = \metalevel$}

We are ready for defining the rest of the logical relations. %
The only cases left are those when $j = \metalevel$. %
In this case, we also know that $i = \metalevel$. %
We first begin with the contextual types for types.
\begin{mathpar}
  \D :: \inferrule
  {\ttypreds \metalevel T{\CTyp[\Delta] l}{0} \\ \ttypreds \metalevel{T'}{\CTyp[\Delta'] l}{0}
    \\
    \lpjudge \metalevel \Gamma \\
    \lpjudge \proglevel \Delta \\
    \lpjudge \proglevel {\Delta'} \\
    \lttypgeq \metalevel {\CTyp[\Delta] l}{\CTyp[\Delta'] l}{0} \\
    \ldSctxeqge \proglevel \proglevel \Delta{\Delta'}}
  {\ldtypeq \metalevel \metalevel{T}{T'}{0}}
\end{mathpar}
Then $\ldtyequiv \metalevel \metalevel t{t'} \D$ is defined by
\begin{itemize}
\item $\ttrmreds \metalevel t w {\CTyp[\Delta] l}{0}$,
\item $\ttrmreds \metalevel{t'}{w'}{\CTyp[\Delta'] l}{0}$,
\item $\lttrmgeq \metalevel w{w'}{\CTyp[\Delta] l}{0}$,
\item $\lttrmgeq \metalevel w{w'}{\CTyp[\Delta'] l}{0}$,
\item then we have two auxiliary definitions $\ldtynfeqm w{w'}{\CTyp[\Delta] l}$, and
\item $\ldtynfeqm w{w'}{\CTyp[\Delta'] l}$.
\end{itemize}
\begin{mathpar}
  \inferrule
  {\ldtypwf[\Psi][\Delta] \codelevel \proglevel {T_1} l}
  {\ldtynfeqm {\boxit{T_1}}{\boxit{T_1}}{\CTyp[\Delta] l}}

  \inferrule
  {\lttrmgneeq \metalevel \nu{\nu'}{\CTyp[\Delta] l}{0}}
  {\ldtynfeqm \nu{\nu'}{\CTyp[\Delta] l}}
\end{mathpar}
Notice how the relation relies on previous relations where $j = \proglevel$. %
In particular, the local contexts $\Delta$ and $\Delta'$ are related when for layers
$\ge \proglevel$, so that they are contexts for boxed types.  When we relate $t$ and $t'$, we
follow the routine: we first reduce them to weak head normal forms, and say that the
normal forms are generically equivalent under two types to make symmetry an easy
proof. %
Then the weak head normal forms are related by the auxiliary relation
$\ldtynfeqm w{w'}{\CTyp[\Delta] l}$. %
This relation treats a contextual type as a sum type. %
There are two possibilities: a normal form of a contextual type is either neutral,
where there is nothing else to do, or contains the code of a type. %
Due to the static code lemma, we know that these types must be syntactically equal. %
In this case, if the type is $T_1$, we require
$\ldtypwf[\Psi][\Delta] \codelevel \proglevel {T_1} l$. This predicate is what we just defined: it
keeps track of both the syntactic and semantic information about $T_1$. %
In this way, we are sure that $T_1$ semantically works fine whenever its syntactic
information and/or its semantic information is needed. %
The same principle is applied to contextual types for terms.
\begin{mathpar}
  \D :: \inferrule
  {\ttypreds \metalevel T{\CTrm[\Delta]{T_1} l}{0} \\ \ttypreds \metalevel{T'}{\CTrm[\Delta']{T_1'} l}{0}
    \\
    \lpjudge \metalevel \Gamma \\
    \lttypwf[\Psi][\Delta] \proglevel{T_1} l \\
    \lttypwf[\Psi][\Delta'] \proglevel{T_1'} l \\
    \lttypgeq \metalevel {\CTrm[\Delta]{T_1} l}{\CTrm[\Delta']{T_1'} l}{0} \\
    \ldSctxeqge \proglevel \proglevel \Delta{\Delta'} \\
    \ldStypeqge[\Psi][\Delta] \proglevel \proglevel{T_1}{T_1'} l \\
    \ldStypeqge[\Psi][\Delta'] \proglevel \proglevel{T_1'}{T_1} l}
  {\ldtypeq \metalevel \metalevel{T}{T'}{0}}
\end{mathpar}
Then $\ldtyequiv \metalevel \metalevel t{t'} \D$ is defined by
\begin{itemize}
\item $\ttrmreds \metalevel t w {\CTrm[\Delta]{T_1}l}{0}$,
\item $\ttrmreds \metalevel{t'}{w'}{\CTrm[\Delta']{T_1'} l}{0}$,
\item $\lttrmgeq \metalevel w{w'}{\CTrm[\Delta]{T_1} l}{0}$,
\item $\lttrmgeq \metalevel w{w'}{\CTrm[\Delta']{T_1'} l}{0}$,
\item then we have two auxiliary definitions $\ldtynfeqm w{w'}{\CTrm[\Delta]{T_1} l}$,
  and $\ldtynfeqm w{w'}{\CTrm[\Delta']{T_1'} l}$.
\end{itemize}
\begin{mathpar}
  \inferrule
  {\ldtyping[\Psi][\Delta] \codelevel \proglevel {t_1}{T_1} l}
  {\ldtynfeqm {\boxit{t_1}}{\boxit{t_1}}{\CTrm[\Delta]{T_1} l}}

  \inferrule
  {\lttrmgneeq \metalevel \nu{\nu'}{\CTrm[\Delta]{T_1} l} 0}
  {\ldtynfeqm \nu{\nu'}{\CTrm[\Delta]{T_1} l}}
\end{mathpar}

Then we take a look at the two kinds of meta-functions. %
Let us first consider the meta-functions for local contexts. %
\begin{mathpar}
  \D :: \inferrule
  {\ttypreds \metalevel T{\CPI g l{T_1}}{l} \\ \ttypreds \metalevel{T'}{\CPI g l{T_1'}}{l} \\
    \lpjudge \metalevel \Gamma \\
    \lttypwf[\Psi, g : \Ctx] \metalevel{T_1} l \\
    \lttypwf[\Psi, g : \Ctx] \metalevel{T_1'} l \\
    \lttypgeq \metalevel {\CPI g l{T_1}}{\CPI g l{T_1'}}{l} \\
    \E :: \forall~ \psi :: L' \sep \Phi;\Delta'' \To_i L \sep \Psi; \Gamma \tand
    \ldSctxeqge[\Phi][L'] \proglevel \proglevel \Delta{\Delta'} ~.~
    \ldtypeq[\Phi][\Delta''][L'] \metalevel \metalevel{T_1[\Delta/g]}{T_1'[\Delta'/g]}{l}}
  {\ldtypeq \metalevel \metalevel{T}{T'}l}
\end{mathpar}
Then $\ldtyequiv \metalevel \metalevel t{t'} \D$ is defined by
\begin{itemize}
\item $\ttrmreds \metalevel t w{\CPI g l{T_1}}{l}$,
\item $\ttrmreds \metalevel{t'}{w'}{\CPI g l{T_1'}}{l}$,
\item $\lttrmgeq \metalevel w{w'}{\CPI g l{T_1}}{l}$,
\item $\lttrmgeq \metalevel w{w'}{\CPI g l{T_1'}}{l}$, and
\item given
  \begin{itemize}
  \item $\psi :: L' \sep \Phi;\Delta'' \To_i L \sep \Psi; \Gamma$, and
  \item $\Ac :: \ldSctxeqge[\Phi][L'] \proglevel \proglevel \Delta{\Delta'}$,
  \end{itemize}
  then 
  \[
    \ldtyequiv[\Phi][\Delta''][L'] \metalevel \metalevel{\CAPP{w}{\Delta}}{\CAPP{w'}{\Delta'}}{\E(\psi,
      \Ac)}
  \]
\end{itemize}

We apply a similar principle to the meta-functions for types.
\begin{mathpar}
  \D :: \inferrule
  {\ttypreds \metalevel T{\TPI U[\Delta]{l}{l'}{T_1}}{l'} \\
    \ttypreds \metalevel{T'}{\TPI U[\Delta']{l}{l'}{T_1'}}{l'} \\
    \lpjudge \metalevel \Gamma \\
    \lttypwf[\Psi, U : \DTyp[\Delta] \proglevel l] \metalevel{T_1}{l'} \\
    \lttypwf[\Psi, U : \DTyp[\Delta'] \proglevel l] \metalevel{T_1}{l'} \\
    \lttypgeq \metalevel {\TPI U[\Delta]{l}{l'}{T_1}}{\TPI
      U[\Delta']{l}{l'}{T_1'}}{l'} \\
    \ldSctxeqge[\Phi][L'] \proglevel \proglevel \Delta{\Delta'} \\
    \E_1 :: \forall~ \psi :: L' \sep \Phi;\Delta'' \To_i L \sep \Psi; \Gamma \tand
    \ldStypeqge[\Phi][\Delta][L'] \proglevel \proglevel {T_2}{T_2'} l ~.~
    \ldtypeq[\Phi][\Delta''][L'] \metalevel \metalevel{T_1[T_2/U]}{T_1'[T_2'/U]}{l'} \\
    \E_2 :: \forall~ \psi :: L' \sep \Phi;\Delta'' \To_i L \sep \Psi; \Gamma \tand
    \ldStypeqge[\Phi][\Delta'][L'] \proglevel \proglevel {T_2}{T_2'} l ~.~
    \ldtypeq[\Phi][\Delta''][L'] \metalevel \metalevel{T_1'[T_2/U]}{T_1[T_2'/U]}{l'}}
  {\ldtypeq \metalevel \metalevel{T}{T'}{l'}}
\end{mathpar}
Notice the symmetry between $\E_1$ and $\E_2$. %
Then $\ldtyequiv \metalevel \metalevel t{t'} \D$ is defined by
\begin{itemize}
\item $\ttrmreds \metalevel t w{\TPI U[\Delta]{l}{l'}{T_1}}{l'}$,
\item $\ttrmreds \metalevel{t'}{w'}{\TPI U[\Delta']{l}{l'}{T_1'}}{l'}$,
\item $\lttrmgeq \metalevel w{w'}{\TPI U[\Delta]{l}{l'}{T_1}}{l'}$,
\item $\lttrmgeq \metalevel w{w'}{\TPI U[\Delta']{l}{l'}{T_1'}}{l'}$, and
\item given
  \begin{itemize}
  \item $\psi :: L' \sep \Phi;\Delta'' \To_i L \sep \Psi; \Gamma$, and
  \item $\Ac :: \ldStypeqge[\Phi][\Delta][L'] \proglevel \proglevel {T_2}{T_2'} l$,
  \end{itemize}
  then 
  \[
    \ldtyequiv[\Phi][\Delta''][L'] \metalevel \metalevel{\TAPP{w}{T_2}}{\TAPP{w'}{T_2'}}{\E_1(\psi,
      \Ac)}
  \]
\item and last symmetrically given
  \begin{itemize}
  \item $\psi :: L' \sep \Phi;\Delta'' \To_i L \sep \Psi; \Gamma$, and
  \item $\Ac :: \ldStypeqge[\Phi][\Delta'][L'] \proglevel \proglevel {T_2}{T_2'} l$,
  \end{itemize}
  then 
  \[
    \ldtyequiv[\Phi][\Delta''][L'] \metalevel \metalevel{\TAPP{w'}{T_2}}{\TAPP{w}{T_2'}}{\E_2(\psi,
      \Ac)}
  \]
\end{itemize}

The last type is the universe-polymorphic function types. %
They are special because they parameterize over universe levels. %
Therefore, a universe-polymorphic function cannot live in any finite universe as it
can be instantiated to any small universe. %
Our semantics must incorporate this fact and introduce the $\omega$ level, which
further requires a transfinite recursion on the universe levels. %
\begin{mathpar}
  \D :: \inferrule
  {\ttypreds \metalevel T{\UPI \ell l{T_1}}{\omega} \\
    \ttypreds \metalevel{T'}{\UPI \ell l{T_1'}}{\omega} \\
    \lpjudge \metalevel \Gamma \\
    \lttypwf[\Psi][\Gamma][L, \vect \ell] \metalevel{T_1} l \\
    \lttypwf[\Psi][\Gamma][L, \vect \ell] \metalevel{T_1'} l \\
    \lttypgeq \metalevel{\UPI \ell l{T_1}}{\UPI \ell l{T_1'}}{\omega} \\
    |\vect\ell| > 0 \\
    \begin{minipage}{\textwidth}
      $\E :: \forall~ \psi :: L' \sep \Phi;\Delta \To_i L \sep \Psi; \Gamma \tand
      |\vect\ell| = |\vect l| = |\vect l'| \tand $ \newline
      \phantom{10000} $(\forall 0 \le n < |\vect l| ~.~
      \tyequiv[L']{\vect l(n)}{\vect l'(n)}\Level) ~.~ \ldtypeq[\Phi][\Delta][L'] \metalevel
      \metalevel{T_1[\vect l/\vect \ell]}{T_1'[\vect l'/\vect \ell]}{l[\vect l/\vect \ell]}$
    \end{minipage}}
  {\ldtypeq \metalevel \metalevel{T}{T'}\omega}
\end{mathpar}
Then $\ldtyequiv \metalevel \metalevel t{t'} \D$ is defined by
\begin{itemize}
\item $\ttrmreds \metalevel t w{\UPI \ell l{T_1}}{\omega}$,
\item $\ttrmreds \metalevel{t'}{w'}{\UPI \ell l{T_1'}}{\omega}$,
\item $\lttrmgeq \metalevel w{w'}{\UPI \ell l{T_1}}{\omega}$,
\item $\lttrmgeq \metalevel w{w'}{\UPI \ell l{T_1'}}{\omega}$, and
\item given
  \begin{itemize}
  \item $\psi :: L' \sep \Phi;\Delta \To_i L \sep \Psi; \Gamma$, 
  \item $\Ac :: |\vect\ell| = |\vect l| = |\vect l'|$, and
  \item $\Bc :: \forall 0 \le n < |\vect l| ~.~ \tyequiv[L]{\vect l(n)}{\vect l'(n)}\Level$,
  \end{itemize}
  then 
  \[
    \ldtyequiv[\Phi][\Delta][L'] \metalevel \metalevel{\UAPP{w}{l}}{w~\$~\vect l'}{\E(\psi,
      \Ac,\Bc)}
  \]
\end{itemize}
Notice how the premise $\E$ changes its universe levels based on different
universes, where exactly transfinite recursion becomes necessary. %

Now we have finished defining all definitions, including the Kripke logical relations
for types and terms. %
The logical relations for local contexts and local substitutions have been given
generally in \Cref{sec:dt:klogrel}. %
Then we will examine the properties of the logical relations when $i = j = \metalevel$. %
Then we are ready for giving the definitions for the semantic judgments as well as
establish the fundamental theorems. %
By instantiating the fundamental theorems with syntactic equivalence, we are able to
obtain a few consequence lemmas, which will be subsequently used in our second
instantiation. %
In the second instantiation of the fundamental theorems, we use the convertibility
checking judgments, from which we derive our final desired the decidability theorem
for convertibility checking. %

\subsection{Properties for Logical Relations When $i = j = \metalevel$}

Now we move on to consider the properties of the logical relations at the final
layers, i.e. when $i = j = \metalevel$. %
In this case, we are considering the relations of all possible types among all
possible terms. %
We first begin with the regular properties, when we will state an important property,
layering restriction, which states how the logical relations are transferred between
$i \in \{\proglevel, \metalevel\}$ when $j = \metalevel$. %
Let us first begin with the simple ones. %
These properties follow \Cref{sec:dt:logrel-prop} quite closely as for the types that
exist at both layers $\proglevel$ and $\metalevel$ the lemma proceeds in a similar way. %
The only difference comes in for types only available at layer $\metalevel$. 
\begin{lemma}[Weakening] $ $
  \begin{itemize}
  \item If $\D :: \ldtypeq \metalevel \metalevel{T}{T'}{l}$ and
    $\psi :: L' \sep \Phi;\Delta \To_m L \sep \Psi; \Gamma$, then
    $\E :: \ldtypeq[\Phi][\Delta][L'] \metalevel \metalevel{T}{T'}{l}$.
  \item If $\ldtyequiv \metalevel \metalevel{t}{t'}\D$ and
    $\psi :: L' \sep \Phi;\Delta \To_m L \sep \Psi; \Gamma$, then
    $\ldtyequiv[\Phi][\Delta][L'] \metalevel \metalevel{t'}{t}\E$.
  \end{itemize}
\end{lemma}
\begin{proof}
  Induction on $\D$. %
  For the cases that overlap with $j = \proglevel$, they can be ported directly to this
  lemma. %
  Therefore, the only cases we need to consider are the new rules. %
  The meta-functions are rather routine as weakenings are built in their
  definitions. %
  Let us consider contextual types.
  \begin{mathpar}
    \D :: \inferrule
    {\ttypreds \metalevel T{\CTrm[\Delta]{T_1} l}{0} \\ \ttypreds \metalevel{T'}{\CTrm[\Delta']{T_1'} l}{0}
      \\
      \lttypgeq \metalevel {\CTrm[\Delta]{T_1} l}{\CTrm[\Delta']{T_1'} l}{0} \\
      \ldSctxeqge \proglevel \proglevel \Delta{\Delta'} \\
      \ldStypwfge[\Psi][\Delta] \proglevel \proglevel T l \\
      \ldStypwfge[\Psi][\Delta'] \proglevel \proglevel{T'} 0}
    {\ldtypeq \metalevel \metalevel{T}{T'}{l}}
  \end{mathpar}
  This case is almost immediate. %
  For judgments like $\ldSctxeqge \proglevel \proglevel \Delta{\Delta'}$ and
  $\ldStypwfge[\Psi][\Delta] \proglevel \proglevel T l$, we know we can extract from
  $\psi :: L' \sep \Phi;\Delta \To_m L \sep \Psi; \Gamma$
  $\alpha :: L' \sep \Phi \To L \sep \Psi$. %
  In particular, the local context in $\ldStypwfge[\Psi][\Delta] \proglevel \proglevel T l$ is not
  impacted by the weakening. %
  For $\ldtyequiv \metalevel \metalevel t{t'} \D$, we shall also apply the weakening lemma to $\alpha$
  to weaken $\ldtyping[\Psi][\Delta] \codelevel \proglevel {t_1}T l$.
\end{proof}

\begin{lemma}[Escape] $ $
  \begin{itemize}
  \item If $\D :: \ldtypeq \metalevel \metalevel{T}{T'}{l}$, then $\lttypgeq \metalevel {T}{T'} l$.
  \item If $\ldtyequiv \metalevel \metalevel{t}{t'}\D$, then $\lttrmgeq \metalevel t{t'} T l$ and $\lttrmgeq \metalevel t{t'}{T'} l$.
  \end{itemize}  
\end{lemma}
\begin{proof}
  Case analyze $\D$. %
  The lemma holds by construction.
\end{proof}

\begin{lemma}[Reflexivity of Neutral]
  If $\D :: \ldtypeq \metalevel \metalevel{T}{T'}l$, $\lttrmgneeq \metalevel \nu{\nu'} T l$ and $\lttrmgneeq \metalevel
  \nu{\nu'}{T'} l$, then $\ldtyequiv \metalevel \metalevel{\nu}{\nu'}\D$.
\end{lemma}
\begin{proof}
  Induction on $\D$. %
  We proceed similarly to the counterpart when $j = \proglevel$. 
\end{proof}

\begin{lemma}[Weak Head Expansion] $ $
  \begin{itemize}
  \item If $\D :: \ldtypeq \metalevel \metalevel{T}{T'}{l}$, $\ttypreds \metalevel{T_1} T l$ and $\ttypreds
    \metalevel{T_1'}{T'}l$, then\break $\ldtypeq \metalevel \metalevel{T_1}{T_1'} l$.
  \item If $\ldtyequiv \metalevel \metalevel{t}{t'}\D$, $\ttrmreds \metalevel{t_1}{t}{T} l$ and $\ttrmreds
    \metalevel{t_1'}{t'}{T'}l$, then $\ldtyequiv \metalevel \metalevel{t_1}{t_1'} \D$.
  \end{itemize}  
\end{lemma}
\begin{proof}
  Induction on $\D$. %
  Use transitivity of multi-step reductions.
\end{proof}

\begin{lemma}[Symmetry] $ $
  \begin{itemize}
  \item If $\D :: \ldtypeq \metalevel \metalevel{T}{T'}{l}$, then $\E :: \ldtypeq \metalevel \metalevel{T'}{T}{l}$.
  \item If $\ldtyequiv \metalevel \metalevel{t}{t'}\D$, then $\ldtyequiv \metalevel \metalevel{t'}{t}\E$. 
  \end{itemize}
\end{lemma}
\begin{proof}
  Induction on $\D$. %
  Symmetry holds by design. %
  The verbose case is the meta-functions for types. %
  Effectively, $\E_1$ and $\E_2$ are duplicated to make sure that symmetry can be easily proved.
\end{proof}

\begin{lemma}[Right Irrelevance]
  If $\D :: \ldtypeq \metalevel \metalevel{T}{T'}{l}$, $\E :: \ldtypeq \metalevel \metalevel{T}{T''}{l}$ and
  $\ldtyequiv \metalevel \metalevel{t}{t'}\D$, then $\ldtyequiv \metalevel \metalevel{t}{t'}\E$.
\end{lemma}
\begin{proof}
  Induction on $\D$. %
  Again, the overlapping cases for $j \in \{\proglevel, \metalevel\}$ can still be ported immediately,
  with only layers changed. %
  For contextual types, it is obvious as the logical relations do not depend on the
  premises at all. %
  The remaining cases are meta-functions and universe-polymorphic functions. %
  These cases are simpler than that of dependent functions because they only take
  simple IH to go through. 
\end{proof}

\begin{lemma}[Left Irrelevance]
  If $\D :: \ldtypeq \metalevel \metalevel{T'}{T}{l}$, $\E :: \ldtypeq \metalevel \metalevel{T''}{T}{l}$ and
  $\ldtyequiv \metalevel \metalevel{t}{t'}\D$, then $\ldtyequiv \metalevel \metalevel{t}{t'}\E$.
\end{lemma}
\begin{proof}
  Immediate by symmetry and right irrelevance. 
\end{proof}

\begin{lemma}[Reflexivity and Transitivity] $ $
  \begin{itemize}
  \item If $\D_1 :: \ldtypeq \metalevel \metalevel{T_1}{T_2}{l}$ and $\D_2 :: \ldtypeq \metalevel
    \metalevel{T_2}{T_3}{l}$, then $\D_3 :: \ldtypeq \metalevel \metalevel{T_1}{T_3}{l}$.
  \item If $\E :: \ldtypeq \metalevel \metalevel{T_1}{T_1}{l}$, $\ldtyequiv \metalevel \metalevel{t_1}{t_2}{\D_1}$ and
    $\ldtyequiv \metalevel \metalevel{t_2}{t_3}{\D_2}$, then $\ldtyequiv \metalevel \metalevel{t_1}{t_3}{\D_3}$.
  \item $\F :: \ldtypeq \metalevel \metalevel{T_1}{T_1}{l}$.
  \item If $\ldtyequiv \metalevel \metalevel{t_1}{t_2}{\D_1}$, then $\ldtyequiv \metalevel \metalevel{t_1}{t_1}{\F}$.
  \end{itemize}
\end{lemma}
\begin{proof}
  We do induction on $\D_1$ and then invert $\D_2$. %
  The cases available when $j = \proglevel$ still can be ported to this lemma. %
  In fact, all remaining cases are simpler than the case for dependent functions. %
  This is because for dependent functions, we must consider reflexivity of related
  input arguments, where for all new cases for meta-programming, we only need
  transitivity from layer $\proglevel$, which has been an established fact at this point. %
\end{proof}

At last, we must prove one very important lemma which connects the logical relations
when $j$ takes different values. %
We need this lemma to explain the fact that a variable at a lower layer suddenly can
be substituted by a term that is not well-typed at its original layer. %
This phenomenon typically occurs when we run code from MLTT at layer $\metalevel$. %
In this case, a variable originally only expected to be substituted by a term from
MLTT must also be able to handle a term only available at layer $\metalevel$, which might
contain, for example, a recursion principle for code. %
A similar lemma occurs in \Cref{sec:cv,sec:cv:logrel}, and a dependently typed version
must also be proved here.
\begin{lemma}[Layering Restriction] $ $
  \begin{itemize}
  \item If $\D :: \ldtypeq \metalevel \proglevel T{T'}l$, then $\E :: \ldtypeq \metalevel \metalevel T{T'}l$.
  \item The following two relations are equivalent:
    \[
      \ldtyequiv \metalevel \proglevel t{t'}\D \tand \ldtyequiv \metalevel \metalevel t{t'}\E
    \]
  \end{itemize}
\end{lemma}
The idea of this lemma is that, if we know a type is coming from MLTT only, then it is
possible to regard its term as a term at both layers $\proglevel$ and $\metalevel$. %
The direction going from $\proglevel$ to $\metalevel$ should be intuitive; it resembles the lifting
lemma on the syntactic side. %
The backward direction, however, might appear counter-intuitive. %
Yet, if we consider the example we discussed above, then we should consider this lemma
describing a process of bringing a term from $\metalevel$ back to $\proglevel$, performing the
substitution, and then finally lifting the result back to $\metalevel$. 
\begin{proof}
  We do induction on $\D$. %
  Since we know $T$ and $T'$ are related when $j = \proglevel$, we do not need to consider the
  cases for meta-programming. %
  The most complex case is the function case. %
  In this case, we have premises
  \begin{align*}
    \D_1 &:: (\forall~ \psi :: L' \sep \Phi;\Delta \To_m L \sep \Psi; \Gamma ~.~
           \ldtypeq[\Phi][\Delta][L'] \metalevel \proglevel{S_1}{S_2}{l_1}) \\
    \D_2 &:: (\forall~ \psi :: L' \sep \Phi;\Delta \To_m L \sep \Psi; \Gamma \tand \ldtyequiv[\Phi][\Delta][L']
           \metalevel \proglevel {s}{s'}{\D_1(\psi)} ~.~
           \ldtypeq[\Phi][\Delta][L'] \metalevel \proglevel{T_1[s/x]}{T_2[s'/x]}{l_2})
  \end{align*}
  By determinacy, we know that
  \begin{align*}
    T \reds \PI{l_1}{l_2} x{S_1}{T_1} \\
    T' \reds \PI{l_1}{l_2} x{S_2}{T_2}
  \end{align*}
  must be unique.  %
  First we show
  \[
    \E_1 :: (\forall~ \psi :: L' \sep \Phi;\Delta \To_i L \sep \Psi; \Gamma ~.~
    \ldtypeq[\Phi][\Delta][L'] \metalevel \metalevel{S_1}{S_2'}{l_1})
  \]
  by a simple IH.

  Next we must show that given
  \begin{itemize}
  \item $\psi :: L' \sep \Phi;\Delta \To_m L \sep \Psi; \Gamma$,    
  \item $\Ac :: \ldtyequiv[\Phi][\Delta][L'] \metalevel \metalevel {s}{s'}{\E_1(\psi)}$
  \end{itemize}
  then
  \[
    \E_2 :: \ldtypeq[\Phi][\Delta][L'] \metalevel \metalevel{T_1[s/x]}{T_2'[s'/x]}{l_2})
  \]
  holds. %
  Notice that the goal is almost applicable for $\D_2$ except that $\Ac$ does not
  satisfy the required premise $\ldtyequiv[\Phi][\Delta][L'] \metalevel \proglevel
  {s}{s'}{\D_1(\psi)}$. %
  But this is fine as we apply IH to use the iff in the second statement to obtain the
  required premise. %

  For the second statement, we must establish an iff relation. %
  That boils down to a symmetric proof. %
  Then in this case, we are given
  \begin{itemize}
  \item $\psi :: L' \sep \Phi;\Delta \To_m L \sep \Psi; \Gamma$,    
  \item $\lttypgeq[\Phi][\Delta][L'] i{S_1}{S_3}l$,
  \item $\lttypgeq[\Phi][\Delta][L'] i{S_2}{S_4}l$, 
  \item $\lttypgeq[\Phi][\Delta, x : S_1 \at{l_1}][L'] i{T_1}{T_3}{l_2}$, and
  \item $\lttypgeq[\Phi][\Delta, x : S_2 \at{l_1}][L'] i{T_2}{T_4}{l_2}$,
  \end{itemize}
  and then we have to show the following equivalence:
  \[
    \Bc :: \ldtyequiv[\Phi][\Delta][L'] \metalevel \proglevel {s}{s'}{\D_1(\psi)}
  \]
  implying
  \[
    \ldtyequiv[\Phi][\Delta][L'] \metalevel \proglevel{\APP w{l_1}{l_2} x{S_3}{T_3} s}{\APP{w'}{l_1}{l_2}
      x{S_4}{T_4}{s'}}{\D_2(\psi, \Bc)}
  \]
  is equivalent to
  \[
    \ldtyequiv[\Phi][\Delta][L'] \metalevel \metalevel {s}{s'}{\E_1(\psi)}
  \]
  implying
  \[
    \ldtyequiv[\Phi][\Delta][L'] \metalevel \metalevel{\APP w{l_1}{l_2} x{S_3}{T_3} s}{\APP{w'}{l_1}{l_2}
      x{S_4}{T_4}{s'}}{\E_2}
  \]

  let us consider the inverse direction. %
  In this case, we assume $\ldtyequiv[\Phi][\Delta][L'] \metalevel \metalevel {s}{s'}{\E_1(\psi)}$. %
  By IH, we have $\ldtyequiv[\Phi][\Delta][L'] \metalevel \proglevel {s}{s'}{\D_1(\psi)}$. %
  From this, we further obtain
  \[
    \ldtyequiv[\Phi][\Delta][L'] \metalevel \proglevel{\APP w{l_1}{l_2} x{S_3}{T_3} s}{\APP{w'}{l_1}{l_2}
      x{S_4}{T_4}{s'}}{\D_2(\psi, \Bc)}
  \] %
  Another IH is sufficient to establish the goal. 
\end{proof}
This lemma, when combined with premises like $\ldStypwfge[\Psi][\Delta] \proglevel \proglevel T l$, is
the bridge to reveal the true complication of supporting lifting in \delamlang. %

We also need a counterpart for local contexts and local substitutions.
\begin{lemma}[Layering Restriction] $ $
  \begin{itemize}
  \item If $\D :: \ldctxeq \metalevel \proglevel{\Delta}{\Delta'}$, then $\E :: \ldctxeq \metalevel \metalevel{\Delta}{\Delta'}$.
  \item The following two relations are equivalent:
    \[
      \ldsubeq \metalevel \proglevel{\delta}{\delta'}\D \tand \ldsubeq \metalevel \metalevel {\delta}{\delta'}\E
    \]
  \end{itemize}
\end{lemma}
\begin{proof}
  Induction on $\D$. %
  The step case is very similar to function case above. %
\end{proof}

\subsection{Semantic Judgments and Fundamental Theorems}

After establishing all logical relations and their properties, we are ready for giving
the definitions for semantic judgments and then moving on to proving the fundamental
theorems. %
The semantic judgments intuitively should say that logical relations are stable under
all substitutions. %
More concretely, we have
\begin{align*}
  \lsgctx \Psi &:= \forall ~ \tyequiv[L']{\phi}{\phi'}L ~.~
                 \ldSgctxeq[L']{\Psi[\phi]}{\Psi[\phi']}
\end{align*}
$\lsgsubeq \sigma{\sigma'} \Phi$ is defined as $\lsgctx \Psi$, $\lsgctx \Phi$ and
given
\begin{itemize}
\item $\tyequiv[L']{\phi}{\phi'}L$, and
\item $\ldSgsubeq[\Psi'][L']{\sigma_1}{\sigma_1'}{\Psi[\phi]}$, 
\end{itemize}
then
\[
  \ldSgsubeq[\Psi'][L']{\sigma[\phi] \circ \sigma_1}{\sigma_1[\phi] \circ \sigma_1'}{\Phi[\phi]}
\]
We also define
\[
  \lsgsubst \sigma \Phi := \lsgsubeq \sigma\sigma \Phi
\]

$\lsctxeq i \Gamma{\Delta}$ where $i \in \{\proglevel, \metalevel \}$ is defined as $\lsgctx \Psi$ and given
\begin{itemize}
\item $\tyequiv[L']{\phi}{\phi'}L$, 
\item $\ldSgsubeq[\Phi][L'] \sigma{\sigma'}{\Psi[\phi]}$, and 
\item $k \ge i$,
\end{itemize}
then
\[
  \ldSctxeq[\Phi][L'] k{\typeof i} {\Gamma[\phi][\sigma]}{\Delta[\phi'][\sigma']}
\]
We also define
\[
  \lsctxwf i \Gamma := \lsctxeq i \Gamma \Gamma
\]

$\lstypeq i T{T'}l$ is defined as $\lsgctx \Psi$ and $\lsctxwf{\typeof i} \Gamma$ and given
\begin{itemize}
\item $\tyequiv[L']{\phi}{\phi'}L$, 
\item $\ldSgsubeq[\Phi][L'] \sigma{\sigma'}{\Psi[\phi]}$, 
\item $k \ge i$,
\item $\ldsubeq[\Phi][\Delta][L'] k{\typeof i} \delta{\delta'}{\Gamma[\phi][\sigma]}$, 
\end{itemize}
then
\[
  \ldtypeq[\Phi][\Delta][L']k{\typeof i}{T[\phi][\sigma][\delta]}{T'[\phi'][\sigma'][\delta']}{l[\phi]}
\]
We also define
\[
  \lstypwf i T l := \lstypeq i T{T}l
\]

$\lstyequiv i t{t'}T l$ is defined as $\lsgctx \Psi$ and $\lsctxwf{\typeof i} \Gamma$ and
$\lstypwf{\typeof i} T l$ given
\begin{itemize}
\item $\tyequiv[L']{\phi}{\phi'}L$, 
\item $\ldSgsubeq[\Phi][L'] \sigma{\sigma'}{\Psi[\phi]}$, 
\item $k \ge i$,
\item $\ldsubeq[\Phi][\Delta][L'] k{\typeof i} \delta{\delta'}{\Gamma[\phi][\sigma]}$, 
\end{itemize}
then
\[
  \ldtyequivt[\Phi][\Delta][L']k{\typeof i}{t[\phi][\sigma][\delta]}{t'[\phi'][\sigma'][\delta']}{T[\phi][\sigma][\delta]}{l[\phi]}
\]
We also define
\[
  \lstyping i t T l := \lstyequiv i t{t}T l
\]

$\lssubeq i \delta{\delta'}{\Delta}$ is defined as $\lsgctx \Psi$ and
$\lsctxwf{\typeof i} \Gamma$ and $\lsctxwf{\typeof i} \Delta$ given
\begin{itemize}
\item $\tyequiv[L']{\phi}{\phi'}L$, 
\item $\ldSgsubeq[\Phi][L'] \sigma{\sigma'}{\Psi[\phi]}$, 
\item $k \ge i$,
\item $\ldsubeq[\Phi][\Delta'][L'] k{\typeof i} {\delta_1}{\delta_1'}{\Gamma[\phi][\sigma]}$, 
\end{itemize}
then
\[
  \ldsubeq[\Phi][\Delta'][L'] k{\typeof i}{\delta[\phi][\sigma] \circ \delta_1}{\delta'[\phi'][\sigma'] \circ \delta_1'}{\Delta[\phi][\sigma]}
\]
We also define
\[
  \lssubst i \delta{\Delta} := \lssubeq i \delta{\delta}{\Delta}
\]

Notice that for definitions above, when $k \in \{\varlevel, \codelevel \}$ is possible, we can use the
local substitution lemma and ignore the local substitutions completely. %
This will be a frequent pattern in the proofs of the semantic rules. 

Summarizing the semantic judgments, we shall arriving at our statement of the
fundamental theorems:
\begin{theorem}[Fundamental] $ $
  \begin{itemize}
  \item If $\judge[L]\Psi$, then $\lsgctx \Psi$. 
  \item If $\lpjudge i \Gamma$ and $i \in \{\proglevel, \metalevel\}$, then $\lsctxwf i \Gamma$. 
  \item If $\lpequiv i \Gamma\Delta$ and $i \in \{\proglevel, \metalevel\}$, then $\lsctxeq i \Gamma{\Delta}$. 
  \item If $\lttypwf i T l$, then $\lstypwf i T l$. 
  \item If $\lttypeq i T{T'} l$, then $\lstypeq i T{T'} l$. 
  \item If $\lttyping i t T l$, then $\lstyping i t T l$.
  \item If $\lttyequiv i t{t'} T l$, then $\lstyequiv i t{t'} T l$.
  \item If $\ltsubst i \delta \Delta$, then $\lssubst i \delta{\Delta}$. 
  \item If $\ltsubeq i\delta{\delta'}\Delta$, then $\lssubeq i \delta{\delta'}{\Delta}$.
  \end{itemize}
\end{theorem}

The following lemma is obvious. 
\begin{lemma}[PER]
  All semantic judgments are PER.
\end{lemma}

\begin{lemma}[Reduction Expansion] $ $
  \begin{itemize}
  \item If $\lstypeq i T{T'} l$, $i \in \{\proglevel, \metalevel\}$ and $\ttypred i {T_1}{T'} l$, then $\lstypeq i{T}{T_1} l$.
  \item If $\lstyequiv i t{t'} T l$, $i \in \{\proglevel, \metalevel\}$ and $\ttrmred i{t_1}{t'}T l$, then
    $\lstyequiv i t{t_1} T l$.
  \end{itemize}
\end{lemma}
\begin{proof}
  It is easy to easy due to the weak head expansion lemma for the logical relations
  and the stability of reduction under all substitutions. 
\end{proof}

The escape lemma recovers Kripke logical relations by passing in corresponding
identity substitutions:
\begin{lemma}[Escape] $ $
  \begin{itemize}
  \item If $\lsgctx \Psi$, then $\ldSgctxeq{\Psi}{\Psi}$. 
  \item If $\lsgsubeq \sigma{\sigma'} \Phi$, then
    $\ldSgsubeq{\sigma}{\sigma'}{\Phi}$. 
  \item If $\lsctxeq i \Gamma\Delta$, $i \in \{\proglevel, \metalevel\}$ and $k \ge i$, then
    $\ldSctxeq k{\typeof i} {\Gamma}{\Delta}$.
  \item If $\lstypeq i T{T'} l$ and $k \ge i$, then $\ldtypeq k{\typeof i}{T}{T'}{l}$.
  \item If $\lstyequiv i t{t'} T l$ and $k \ge i$, then
    $\ldtyequivt k{\typeof i}{t}{t'}{T}{l}$.
  \item If $\lssubeq i \delta{\delta'}{\Delta}$ and $k \ge i$, then
    $\ldsubeq k{\typeof i} \delta{\delta'}{\Delta}$.
  \end{itemize}
\end{lemma}
By chaining escape lemmas, we can obtain that semantic equivalence judgments generic
and syntactic equivalences.

The semantic judgments are stable under substitutions.
\begin{lemma}[Universe Substitutions] $ $
  \begin{itemize}
  \item If $\lsgctx \Psi$ and $\typing[L']{\phi}L$, then $\lsgctx[L']{\Psi[\phi]}$.
  \item If $\lsgsubeq \sigma{\sigma'} \Phi$ and $\tyequiv[L']{\phi}{\phi'}L$, then
    $\lsgsubeq[\Psi[\phi]][L']{\sigma[\phi]}{\sigma'[\phi']}{\Phi[\phi]}$.
  \item If $\lsctxeq i \Gamma\Delta$, $i \in \{\proglevel, \metalevel\}$ and $\tyequiv[L']{\phi}{\phi'}L$, then
    $\lsctxeq[\Psi[\phi]][L'] i{\Gamma[\phi]}{\Delta[\phi']}$.
  \item If $\lstypeq i T{T'} l$ and $\tyequiv[L']{\phi}{\phi'}L$, then
    $\lstypeq[\Psi[\phi]][\Gamma[\phi]][L'] i{T[\phi]}{T'[\phi']}{l[\phi]}$.
  \item If $\lstyequiv i t{t'} T l$ and $\tyequiv[L']{\phi}{\phi'}L$,
    then
    $\lstyequiv[\Psi[\phi]][\Gamma[\phi]][L'] i{t[\phi]}{t'[\phi']}{T[\phi]}{l[\phi]}$.
  \item If $\lssubeq i \delta{\delta'}{\Delta}$ and $\tyequiv[L']{\phi}{\phi'}L$, then
    $\lssubeq[\Psi[\phi]][\Gamma[\phi]][L'] i{\delta[\phi]}{\delta'[\phi']}{\Delta[\phi]}$.
  \end{itemize}  
\end{lemma}
\begin{proof}
  Use composition of universe substitutions.
\end{proof}

\begin{lemma}[Global Substitutions] $ $
  \begin{itemize}
  \item If $\lsgsubeq \sigma{\sigma'} \Phi$ and $\lsgsubeq[\Psi']{\sigma_1}{\sigma_1'}{\Psi}$, then
    $\lsgsubeq[\Psi']{\sigma \circ \sigma_1}{\sigma' \circ \sigma_1'}{\Phi}$.
  \item If $\lsctxeq i \Gamma\Delta$, $i \in \{\proglevel, \metalevel\}$ and $\lsgsubeq[\Phi] \sigma{\sigma'} \Psi$, then
    $\lsctxeq[\Phi] i{\Gamma[\sigma]}{\Delta[\sigma']}$.
  \item If $\lstypeq i T{T'} l$ and $\lsgsubeq[\Phi] \sigma{\sigma'} \Psi$, then
    $\lstypeq[\Phi][\Gamma[\sigma]] i{T[\sigma]}{T'[\sigma']}{l}$.
  \item If $\lstyequiv i t{t'} T l$ and $\lsgsubeq[\Phi] \sigma{\sigma'} \Psi$,
    then
    $\lstyequiv[\Phi][\Gamma[\sigma]] i{t[\sigma]}{t'[\sigma']}{T[\sigma]}{l}$.
  \item If $\lssubeq i \delta{\delta'}{\Delta}$ and $\lsgsubeq[\Phi] \sigma{\sigma'} \Psi$, then
    $\lssubeq[\Phi][\Gamma[\sigma]] i{\delta[\sigma]}{\delta'[\sigma']}{\Delta[\sigma]}$.
  \end{itemize}  
\end{lemma}
\begin{proof}
  The principle is also to use the composition of global substitutions. %
  We also make use of the commutativity of substitutions, e.g.
  \[
    t[\phi][\sigma[\phi]] = t[\sigma][\phi]
  \]
  This equation allows us to swap $\sigma$ forwards, which will be frequently used in
  this proof.
\end{proof}

\begin{lemma}[Local Substitutions] $ $
  \begin{itemize}
  \item If $\lstypeq i T{T'} l$ and $\lssubeq[\Psi][\Delta] i \delta{\delta'}{\Gamma}$, then
    $\lstypeq[\Psi][\Delta] i{T[\delta]}{T'[\delta']}{l}$.
  \item If $\lstyequiv i t{t'} T l$ and
    $\lssubeq[\Psi][\Delta] i \delta{\delta'}{\Gamma}$, then
    $\lstyequiv[\Psi][\Delta] i{t[\delta]}{t'[\delta']}{T[\delta]}{l}$.
  \item If $\lssubeq i \delta{\delta'}{\Delta}$ and $\lssubeq[\Psi][\Gamma'] i{\delta_1}{\delta_1'}{\Gamma}$, then
    $\lssubeq[\Psi][\Gamma'] i{\delta \circ \delta_1}{\delta' \circ \delta_1'}{\Delta}$.
  \end{itemize}  
\end{lemma}
\begin{proof}
  We apply the similar technique here. %
  We make use equations similar to below:
  \begin{gather*}
    t[\phi][\delta[\phi]] = t[\delta][\phi] \\
    t[\sigma][\delta[\sigma]] = t[\delta][\sigma]
  \end{gather*}
  These equations will swap $\delta$ forwards. 
\end{proof}

The theorem proceeds by doing induction on the derivations. %
The proof though is rather verbose due to how the semantic judgments are defined. %
One pattern that is worth mentioning is that the proof should work ``backwards'' from
the layers. %
That is, we should work out the proofs from layer $\metalevel$, and then $\proglevel$ and then $\codelevel$ and
finally $\varlevel$. %
This pattern makes sense if we consider what information layers contain. %
The semantics of a term at layer $\metalevel$ only contains its computational contents at layer
$\metalevel$. %
However, for a term at layer $\proglevel$, due to lifting, its semantics must explain how this
term computes at both layers $\proglevel$ and $\metalevel$. %
For a term at layer $\codelevel$, in addition to its collective information as a term at layer
$\proglevel$, it should also has all information about its sub-structures. %
At last, if a term is at layer $\varlevel$, then we know it must be well-formed at layer $\codelevel$
but also it represents a variable. %
Thus, due to lifting, information contained at each layer strictly increases as the
layer decreases.  To build up information at a smaller layer, we should prove the
fundamental theorems from a higher layer. %
In the next section, we will start proving the fundamental theorems and make sure that
all syntactically well-formed types and terms at all layers are semantically
well-formed.

\subsection{Proving Fundamental Theorems}

To demonstrate the idea described at the end of the previous subsection, let us first
consider the simplest case. %
We often proceed by first proving the semantic rule for types and then go on and prove
the rules for terms. %
\begin{lemma}
  \begin{mathpar}
    \inferrule
    {\lsctxwf {\typeof i} \Gamma}
    {\lstypeq i\Nat \Nat \ze}
  \end{mathpar}
\end{lemma}
\begin{proof}
  From $\lsctxwf {\typeof i} \Gamma$, we also know $\lsgctx \Psi$. %
  Now assume $\tyequiv[L']{\phi}{\phi'}L$ and
  $\ldSgsubeq[\Phi][L'] \sigma{\sigma'}{\Psi[\phi]}$, we have to consider all
  $i \ge \codelevel$. %
  \begin{itemize}
  \item[Case $i = \metalevel$] Then assuming
    $\ldsubeq[\Phi][\Delta][L'] \metalevel \metalevel \delta{\delta'}{\Gamma[\phi][\sigma]}$, we must
    show
    \[
      \ldtypeq[\Phi][\Delta][L'] \metalevel \metalevel \Nat \Nat \ze
    \]
    This holds by definition.
    
  \item[Case $i = \proglevel$] Then assuming some $k \ge \proglevel$ and
    $\ldsubeq[\Phi][\Delta][L'] k \proglevel \delta{\delta'}{\Gamma[\phi][\sigma]}$, we must
    show
    \[
      \ldtypeq[\Phi][\Delta][L'] k \proglevel \Nat \Nat \ze
    \]
    Again, this also holds by definition as we see that $k \in \{\proglevel, \metalevel\}$.
    
  \item[Case $i = \codelevel$] This is the last case. %
    Assuming some $k \ge \codelevel$ and
    $\ldsubeq[\Phi][\Delta][L'] k \proglevel \delta{\delta'}{\Gamma[\phi][\sigma]}$, we must
    show
    \[
      \ldtypeq[\Phi][\Delta][L'] k \proglevel \Nat \Nat \ze
    \]
    In the previous case of $i = \proglevel$, we have given the proof for $k \in \{\proglevel, \metalevel\}$, so
    essentially we only have one case $k = \codelevel$ left. %
    In this case, we apply the local substitution lemma so we do not have to introduce
    any local substitutions at all. %
    Looking up the rules in \Cref{sec:dt:semglob}, we need to show
    \[
      \ldStypwfge[\Phi][\Gamma[\phi][\sigma]][L'] \proglevel \proglevel \Nat \ze
    \]
    This is again have been given by the case of $i = \proglevel$, modulo converting weakenings
    to substitutions.  %
    Hence we conclude the proof.
  \end{itemize}
\end{proof}

\begin{lemma}
  \begin{mathpar}
    \inferrule
    {\lsctxwf {\typeof i} \Gamma \\ \tyequiv[L]{l}{l'}\Level}
    {\lstypeq i{\Ty l}{\Ty{l'}}{1 + l}}
  \end{mathpar}
\end{lemma}
\begin{proof}
  Similar to the previous lemma. 
\end{proof}

\begin{lemma}
  \begin{mathpar}
    \inferrule
    {\lsctxwf {\typeof i} \Gamma}
    {\lstyequiv i \ze\ze \Nat \ze}
  \end{mathpar}
\end{lemma}
\begin{proof}
  From the previous lemma, we obtain $\lsgctx \Psi$, $\lsctxwf {\typeof i} \Gamma$ and
  $\lstypwf i \Nat \ze$. %
  Simulating the previous lemma, assuming $\tyequiv[L']{\phi}{\phi'}L$ and
  $\ldSgsubeq[\Phi][L'] \sigma{\sigma'}{\Psi[\phi]}$, we have to consider all
  $i \ge \codelevel$. %
  \begin{itemize}
  \item[Case $i = \metalevel$] Then assuming
    $\ldsubeq[\Phi][\Delta][L'] \metalevel \metalevel \delta{\delta'}{\Gamma[\phi][\sigma]}$, we must
    show
    \[
      \ldtyequivt[\Phi][\Delta][L'] \metalevel \metalevel \ze \ze \Nat \ze
    \]
    This is the same as showing
    \[
      \D :: \ldtypeq[\Phi][\Delta][L'] \metalevel \metalevel \Nat \Nat \ze \tand
      \ldtyequiv[\Phi][\Delta][L'] \metalevel \metalevel \ze \ze \D
    \]
    This is immediate by the congruence law of the generic equivalence and the
    definition of the logical relations.
    
  \item[Case $i = \proglevel$] Then assuming some $k \ge \proglevel$ and
    $\ldsubeq[\Phi][\Delta][L'] k \proglevel \delta{\delta'}{\Gamma[\phi][\sigma]}$, we must
    show
    \[
      \D :: \ldtypeq[\Phi][\Delta][L'] k \proglevel \Nat \Nat \ze \tand
      \ldtyequiv[\Phi][\Delta][L'] k \proglevel \ze \ze \D
    \]
    Following a similar proof to the case above, we also establish this case knowing
    $k \in \{\proglevel, \metalevel\}$.
    
  \item[Case $i = \codelevel$] Assuming some $k \ge \codelevel$ and
    $\ldsubeq[\Phi][\Delta][L'] k \proglevel \delta{\delta'}{\Gamma[\phi][\sigma]}$, we must
    show
    \[
      \ldtyequivt[\Phi][\Delta][L'] k \proglevel \ze \ze \Nat \ze
    \]
    The only additional analysis is when $k = \codelevel$ as other values for $k$ have been
    considered in the previous case. %
    In this case, we apply local substitution lemma so we only need to prove
    \[
      \ldStypingge[\Phi][\Gamma[\phi][\sigma]][L'] \proglevel \proglevel \ze \Nat \ze
    \]
    This clearly has been proven in the previous case. 
  \end{itemize}
\end{proof}

\begin{lemma}
  \begin{mathpar}
    \inferrule
    {\lstyequiv i t{t'}\Nat\ze}
    {\lstyequiv i{\su t}{\su{t'}}\Nat\ze}
  \end{mathpar}
\end{lemma}
\begin{proof}
  From $\lstyequiv i t{t'}\Nat\ze$, we also know $\lsgctx \Psi$,
  $\lsctxwf {\typeof i} \Gamma$ and \break $\lstypeq i \Nat\Nat\ze$. %
  Now assume $\tyequiv[L']{\phi}{\phi'}L$ and
  $\ldSgsubeq[\Phi][L'] \sigma{\sigma'}{\Psi[\phi]}$, we have to consider all
  $i \ge \codelevel$. %
  \begin{itemize}
  \item[Case $i = \metalevel$] Then assuming
    $\ldsubeq[\Phi][\Delta][L'] \metalevel \metalevel \delta{\delta'}{\Gamma[\phi][\sigma]}$, we must
    show
    \[
      \D :: \ldtypeq[\Phi][\Delta][L'] \metalevel \metalevel \Nat \Nat \ze \tand \ldtyequiv[\Phi][\Delta][L'] \metalevel \metalevel{\su t[\phi][\sigma][\delta]}{\su{t'}[\phi'][\sigma'][\delta']} \D
    \]
    From $\lstyequiv i t{t'}\Nat\ze$, we obtain such $\D$ and also
    \[
      \ldtyequiv[\Phi][\Delta][L'] \metalevel \metalevel{t[\phi][\sigma][\delta]}{t'[\phi'][\sigma'][\delta']} \D
    \]
    From here we obtain $\ldtynfeqnat i{\su t}{\su{t'}}$ which leads to our desired
    goal. 
  \item[Case $i = \proglevel$] Then assuming some $k \ge \proglevel$ and
    $\ldsubeq[\Phi][\Delta][L'] k \proglevel \delta{\delta'}{\Gamma[\phi][\sigma]}$, we must
    show
    \[
      \D :: \ldtypeq[\Phi][\Delta][L'] k \proglevel \Nat \Nat \ze \tand
      \ldtyequiv[\Phi][\Delta][L'] k \proglevel{\su t[\phi][\sigma][\delta]}{\su{t'}[\phi'][\sigma'][\delta']}\D
    \]
    We follow the previous case.
  \item [Case $i = \codelevel$] In this case, we only consider the most interesting case of
    $k = \codelevel$. %
    In this case, we apply local substitution lemma, so we know
    $\ldStypingge[\Phi][\Gamma[\phi][\sigma]][L'] \proglevel \proglevel{t[\phi][\sigma]}\Nat\ze$ and we
    must prove
    \[
      \ldStypingge[\Phi][\Gamma[\phi][\sigma]][L'] \proglevel \proglevel{\su t[\phi][\sigma]}\Nat\ze
    \]
    This clearly has been given by the previous case. %
    As we have seen in the last few proofs with $k = \codelevel$, it is a common pattern that
    we use the local substitution lemma to get rid of the local substitution lemma. %
    Then what is left for the proof obligation is given by $i = \proglevel$ modulo converting
    weakenings to substitutions. %
    Essentially, the semantics of layer $\codelevel$ simply remembers the derivation given by
    the semantic rules. %
    For this reason, we will keep cases of $i = \codelevel$ short.
  \end{itemize}
\end{proof}

The semantic rules are pretty sensitive to the orders in which they are proved. %
To handle $\Pi$ types, it is more convenient if we have the rules for contexts
ready.
\begin{lemma}
  \begin{mathpar}
  \inferrule
  {\lsgctx \Psi}
  {\lsctxeq i \cdot\cdot}
  \end{mathpar}
\end{lemma}
\begin{proof}
  Immediate.
\end{proof}

\begin{lemma}
  \begin{mathpar}
  \inferrule
  {\lsgctx \Psi \\ g : \Ctx \in \Psi}
  {\lsctxeq i{g}{g}}
  \end{mathpar}
\end{lemma}
\begin{proof}
  Now assuming $\tyequiv[L']{\phi}{\phi'}L$ and
  $\ldSgsubeq[\Phi][L'] \sigma{\sigma'}{\Psi[\phi]}$, we have to consider all
  $i \in \{\proglevel, \metalevel\}$. %
  We only consider $i = \metalevel$ here as the proof for $i = \proglevel$ is very similar.
  We know $g : \Ctx \in \Psi[\phi]$ as well. %
  Then we have the following after lookup
  \[
    \ldSctxeqge[\Phi][L'] \proglevel \proglevel{\sigma(g)}{\sigma'(g)}
  \]
  We are very close to our goal
  \[
    \ldSctxeq[\Phi][L'] \metalevel \metalevel{\sigma(g)}{\sigma'(g)}
  \]
  First we obtain $\ldSctxeq[\Phi][L'] \metalevel \proglevel{\sigma(g)}{\sigma'(g)}$. %
  Then by layering restriction, we have the goal by lifting $\proglevel$ to $\metalevel$.
\end{proof}

\begin{lemma}
  \begin{mathpar}
    \inferrule
    {\lsctxeq i \Gamma\Delta \\ \lstypeq i T{T'} l \\ \tyequiv[L]l{l'}\Level}
    {\lsctxeq i{\Gamma, x : T \at l}{\Delta, x : T' \at{l'}}}
  \end{mathpar}
\end{lemma}
\begin{proof}
  Now assuming $\tyequiv[L']{\phi}{\phi'}L$ and
  $\ldSgsubeq[\Phi][L'] \sigma{\sigma'}{\Psi[\phi]}$, we have to consider all
  $i \in \{\proglevel, \metalevel\}$. %
  We only consider $i = \metalevel$ here as the proof for $i = \proglevel$ is very similar. %
  We should prove
  \[
    \ldSctxeq[\Phi][L'] \metalevel \metalevel{\Gamma[\phi][\sigma], x : T[\phi][\sigma] \at{l[\phi]}}{\Delta[\phi'][\sigma'], x : T'[\phi'][\sigma'] \at{l'[\phi']}}
  \]
  We first obtain 
 \[
    \ldSctxeq[\Phi][L'] \metalevel \metalevel{\Gamma[\phi][\sigma]}{\Delta[\phi'][\sigma']}
  \]
  To obtain the goal, we must show $T[\phi][\sigma] \approx T'[\phi'][\sigma']$ is
  stable under local substitutions. %
  This is immediate by the semantic judgment, after converting weakenings into
  universe and global substitutions.
\end{proof}

\begin{lemma}
  \begin{mathpar}
  \inferrule
  {\lsctxwf{\typeof i}\Gamma \\ \text{$\Gamma$ ends with $\cdot$} \\ |\Gamma| = k'}
  {\lssubeq i {\cdot^{k'}}{\cdot^{k'}}{\cdot}}
  \end{mathpar}
\end{lemma}
\begin{proof}
  From $\lsctxwf {\typeof i} \Gamma$, we also know $\lsgctx \Psi$. %
  Assuming
  $\tyequiv[L']{\phi}{\phi'}L$, $\ldSgsubeq[\Phi][L'] \sigma{\sigma'}{\Psi[\phi]}$. %
  We should consider all possible $i$.
  \begin{itemize}
  \item[Case $i = \metalevel$] Then assuming
    $\ldsubeq[\Phi][\Delta][L'] \metalevel \metalevel \delta{\delta'}{\Gamma[\phi][\sigma]}$, we must
    show
    \[
      \D :: \ldctxeq[\Phi][L'] \metalevel \metalevel \cdot \cdot
      \tand
      \ldsubeq[\Phi][\Delta][L'] \metalevel \metalevel{\cdot^{k'} \circ \delta}{\cdot^{k'} \circ \delta'}\D
    \]
    $\D$ is immediate. %
    Now we should consider the composition. %
    We know
    \[
      {\cdot^{k'} \circ \delta} = {\cdot^{k'} \circ \delta'} = \cdot_{\widecheck{\delta}}^{\widehat\delta}
    \]
    We then have the goal by definition.
  \item[Case $i = \proglevel$] Similar.
  \item[Case $i = \codelevel$] By the local substitution lemma and the rule in
    \Cref{sec:dt:semglob}, we conclude this case by repeating the previous case. 
  \item[Case $i = \varlevel$] Similar.
  \end{itemize}
\end{proof}

\begin{lemma}
  \begin{mathpar}
  \inferrule
  {\lsctxwf{\typeof i}\Gamma \\ g : \Ctx \in \Psi \\ \text{$\Gamma$ ends with $g$} \\ |\Gamma| = k'}
  {\lssubeq i {\cdot_g^{k'}}{\cdot_g^{k'}}{\cdot}}
  \end{mathpar}
\end{lemma}
\begin{proof}
  Similar to the previous lemma. %
  We will need to do a case analysis on the result of lookup of $g$, but otherwise the
  result is straightforward.
\end{proof}

\begin{lemma}
  \begin{mathpar}
  \inferrule
  {\lsctxwf{\typeof i}\Gamma \\ g : \Ctx \in \Psi \\ \text{$\Gamma$ ends with $g$} \\ |\Gamma| = k'}
  {\lssubeq i {\wk_g^{k'}}{\wk_g^{k'}}{g}}
  \end{mathpar}
\end{lemma}
\begin{proof}
  From $\lsctxwf {\typeof i} \Gamma$, we also know $\lsgctx \Psi$. %
  The semantic well-formedness for $g$ is established by a previous lemma.
  
  Now we assume $\tyequiv[L']{\phi}{\phi'}L$,
  $\ldSgsubeq[\Phi][L'] \sigma{\sigma'}{\Psi[\phi]}$. %
  We should consider all possible $i$.
  \begin{itemize}
  \item[Case $i = \metalevel$] Then assuming
    $\ldsubeq[\Phi][\Delta][L'] \metalevel \metalevel \delta{\delta'}{\Gamma[\phi][\sigma]}$, we must
    show
    \[
      \D :: \ldctxeq[\Phi][L'] \metalevel \metalevel{\sigma(g)}{\sigma'(g)}
      \tand
      \ldsubeq[\Phi][\Delta][L'] \metalevel \metalevel{\wk_g^{k'}[\sigma] \circ \delta}{\wk_g^{k'}[\sigma'] \circ \delta'}\D
    \]
    $\D$ is obtained by looking up $g$ in $\sigma$, from which we get
    \[
      \ldSctxeq[\Phi][L'] \metalevel \proglevel{\sigma(g)}{\sigma'(g)}
    \]
    We have $\D$ by layering restriction. %

    For composition, we have
    \begin{gather*}
      \wk_g^{k'}[\sigma] \circ \delta = \wk_{\sigma(g)}^{k'} \circ \delta \\
      \wk_g^{k'}[\sigma'] \circ \delta' = \wk_{\sigma'(g)}^{k'} \circ \delta'
    \end{gather*}
    Effectively, this is the same as popping off $k'$ terms from $\delta$ and
    $\delta'$ simultaneously. %
    We have the goal by irrelevance.
  \item[Case $i = \proglevel$] Similar. 
  \item[Case $i \in \{\varlevel, \codelevel\}$]
    Similarly, we use the previous case.
  \end{itemize}
\end{proof}

\begin{lemma}
    \begin{mathpar}
    \inferrule
    {\typing[L]l\Level \\ \lssubeq i {\delta}{\delta'}{\Delta} \\ \lstypwf[\Psi][\Delta]{\typeof i}T l \\
      \lstyequiv i {t}{t'}{T[\delta]} l}
    {\lssubeq i {\delta, t/x}{\delta', t'/x}{\Delta, x : T \at l}}
  \end{mathpar}
\end{lemma}
\begin{proof}
  From $\lssubeq i {\delta}{\delta'}{\Delta}$ we have $\lsctxwf {\typeof i} \Gamma$,
  $\lsctxwf{\typeof i} \Delta$ and $\lsgctx \Psi$. %
  Then from a previous lemma, we further have
  $\lsctxwf{\typeof i}{\Delta, x : T \at l}$.

  Now we assume $\tyequiv[L']{\phi}{\phi'}L$,
  $\ldSgsubeq[\Phi][L'] \sigma{\sigma'}{\Psi[\phi]}$. %
  We should consider all possible $i$.
  \begin{itemize}
  \item[Case $i = \metalevel$] Then assuming
    $\ldsubeq[\Phi][\Delta'][L'] \metalevel \metalevel{\delta_1}{\delta_1'}{\Gamma[\phi][\sigma]}$, we must
    show
    \begin{itemize}
    \item $\D :: \ldctxeq[\Phi][L'] \metalevel \metalevel{(\Delta, x : T \at
        l)[\phi][\sigma]}{(\Delta, x : T \at l)[\phi'][\sigma']}$ and 
    \item
      $\ldsubeq[\Phi][\Delta'][L'] \metalevel \metalevel{(\delta, t/x)[\phi][\sigma] \circ
        \delta_1}{(\delta', t'/x)[\phi'][\sigma'] \circ \delta_1'}\D$.
    \end{itemize}
    We expand the composition:
    \begin{align*}
      (\delta, t/x)[\phi][\sigma] \circ \delta_1 &= (\delta[\phi][\sigma] \circ
                                                   \delta_1),
                                                   (t[\phi][\sigma][\delta_1])/x \\
      (\delta', t'/x)[\phi'][\sigma'] \circ \delta_1' &= (\delta'[\phi'][\sigma'] \circ
                                                        \delta_1'),
                                                        (t'[\phi'][\sigma'][\delta_1'])/x 
    \end{align*}
    
    We can conclude the goal by using $\lsctxwf{\typeof i}{\Delta, x : T \at l}$,
    $\lstyequiv i {t}{t'}{T[\delta]} l$ and irrelevance. %
    
  \item[Case $i = \proglevel$] Similar.
    
  \item[Case $i = \codelevel$] Similar to the previous pattern, we apply the local substitution
    lemma and use the previous case to discharge the obligations. 
  \item[Case $i = \varlevel$] Similar.
  \end{itemize}
\end{proof}

\begin{lemma}
  \begin{mathpar}
    \inferrule
    { }
    {\lsgctx\cdot}

    \inferrule
    {\lsgctx\Psi}
    {\lsgctx{\Psi, g : \Ctx}}

    \inferrule
    {\lsgctx\Psi \\ \lsctxwf \proglevel \Gamma \\\\ \typing[L]l\Level \\ i \in \{\codelevel, \proglevel\}}
    {\lsgctx{\Psi, U : \DTyp i l}}

    \inferrule
    {\lsgctx\Psi \\ \lstypwf \proglevel T l \\\\ \typing[L]l\Level \\ i \in \{\varlevel, \codelevel\}}
    {\lsgctx{\Psi, u : \DTrm i T l}}
  \end{mathpar}
\end{lemma}
\begin{proof}
  Immediate. %
  We take advantage of the fact that a universe weakening is a special universe
  substitution. %
\end{proof}

\begin{lemma}
  \begin{mathpar}
    \inferrule
    {\lsctxwf{\typeof i}\Gamma \\ x : T \at l \in \Gamma}
    {\lstyequiv i x x T l}
  \end{mathpar}
\end{lemma}
\begin{proof}
  From $\lsctxwf {\typeof i} \Gamma$, we also know $\lsgctx \Psi$. %

  To construct the semantic judgment for type $T$, we first assuming
  $\tyequiv[L']{\phi}{\phi'}L$, $\ldSgsubeq[\Phi][L'] \sigma{\sigma'}{\Psi[\phi]}$ and
  some $k \ge i$. %
  We have
  \[
    \ldSctxeq[\Phi][L'] k{\typeof i} {\Gamma[\phi][\sigma]}{\Gamma[\phi'][\sigma']}
  \]
  Our goal is to construct
  \[
    \ldtypeq[\Phi][\Delta][L'] k{\typeof
      i}{T[\phi][\sigma][\delta]}{T'[\phi'][\sigma'][\delta']}{l[\phi]}
  \]
  with further assuming
  $\ldsubeq[\Phi][\Delta][L'] k{\typeof i}\delta{\delta'}{\Gamma[\phi][\sigma]}$. %
  This is done by doing induction on $x : T \at l \in \Gamma$. %

  Then we consider the term. %
  Since it is the variable case, $i$ can take all four layers.
  \begin{itemize}
  \item[Case $i = \metalevel$] Then assuming
    $\ldsubeq[\Phi][\Delta][L'] \metalevel \metalevel \delta{\delta'}{\Gamma[\phi][\sigma]}$, we must
    show
    \[
      \D :: \ldtypeq[\Phi][\Delta][L'] \metalevel \metalevel{T[\phi][\sigma][\delta]}{T'[\phi'][\sigma'][\delta']}{l[\phi]}
      \tand
      \ldtyequiv[\Phi][\Delta][L'] \metalevel \metalevel{\delta(x)}{\delta'(x)}\D
    \]
    We proceed by doing induction on $x : T \at l \in \Gamma$. %
    We weaken the universe and global contexts to obtain the goal.
    
  \item[Case $i = \proglevel$] This case works similarly at different layers. %
    We omit it here.
  \item[Case $i = \codelevel$] In this case, we consider $k = \codelevel$ and apply the local
    substitution lemma.  %
    Based on the rule in \Cref{sec:dt:semglob} and the previous case, we have the
    goal.
  \item[Case $i = \varlevel$] This case makes use of the entire previous case and also in
    addition must prove the same for $k = \varlevel$. %
    But this is virtually identical to the previous case.
  \end{itemize}
\end{proof}

Combining the semantic rules for local substitutions, we derive that
\begin{corollary}[Local Weakening Substitutions]\labeledit{lem:dt:sem:lwk}
  $\lssubst[\Psi][\Gamma,\Delta] i{\wk_{\Gamma}^{|\Delta|}}\Gamma$
\end{corollary}

\begin{lemma}
  \begin{mathpar}
    \inferrule
    {\lsctxwf{\typeof i}\Gamma \\ u : \DTrm[\Delta]{i'}T l \in \Psi \\ i' \in \{\varlevel, \codelevel\} \\ i \in
      \{\varlevel, \codelevel, \proglevel, \metalevel\} \\ i' \le i \\ \Ac :: \lssubeq i \delta{\delta'} \Delta}
    {\lstyequiv i{u^\delta}{u^{\delta'}}{T[\delta]}{l}}
  \end{mathpar}
\end{lemma}
\begin{proof}
  From $\lsctxwf {\typeof i} \Gamma$, we also know $\lsgctx \Psi$. %

  To construct the semantic judgment for type $T$, we first assuming
  $\tyequiv[L']{\phi}{\phi'}L$, $\ldSgsubeq[\Phi][L'] \sigma{\sigma'}{\Psi[\phi]}$ and
  some $k \ge i$, and then
  $\ldsubeq[\Phi][\Delta'][L'] k{\typeof i}{\delta_1}{\delta_1'}{\Gamma[\phi][\sigma]}$. %
  We have
  \[
    \ldSctxeq[\Phi][L'] k{\typeof i} {\Gamma[\phi][\sigma]}{\Gamma[\phi'][\sigma']}
  \]
  Our goal is to construct
  \[
    \ldtypeq[\Phi][\Delta'][L'] k{\typeof
      i}{T[\phi][\sigma][\delta[\phi][\sigma] \circ \delta_1]}{T[\phi'][\sigma'][\delta'[\phi'][\sigma'] \circ \delta_1']}{l[\phi]}
  \]
  We obtain this by looking up $\ldSgctx[L']{\Psi[\phi]}$ using
  $u : \DTrm[\Delta]{i'}T l \in \Psi$, and use the invariant that $T$ is stable under
  global and local substitutions. 
  
  Then we consider the term. %
  \begin{itemize}
  \item[Case $i = \metalevel$] Then assuming
    $\ldsubeq[\Phi][\Delta'][L'] \metalevel \metalevel{\delta_1}{\delta_1'}{\Gamma[\phi][\sigma]}$, we must show
    \[
      \ldtyequivt[\Phi][\Delta'][L'] \metalevel \metalevel{\sigma(u)[\delta[\phi][\sigma]\circ
        \delta_1]}{\sigma'(u)[\delta'[\phi'][\sigma']\circ
        \delta_1']}{T[\phi][\sigma][\delta[\phi][\sigma] \circ \delta_1]}{l[\phi]}
    \]
    Looking up $\ldSgsubeq[\Phi][L'] \sigma{\sigma'}{\Psi[\phi]}$, we know that
    \[
      \Bc :: \ldStyequivge[\Phi][\Delta[\phi][\sigma]][L']{i'} \proglevel{\sigma(u)}{\sigma'(u)}{T[\phi][\sigma]} l
    \]
    and
    \[
      \ldctxwf[\Phi][L'] \proglevel \proglevel{\Delta[\phi][\sigma]}
    \]
    We also know the following from $\Ac$
    \[
      \ldsubeq[\Phi][\Delta'][L'] \metalevel \metalevel{\delta[\phi][\sigma] \circ
        \delta_1}{\delta'[\phi'][\sigma'] \circ \delta_1'}{\Delta[\phi][\sigma]}
    \]
    Therefore, we can apply layering restriction and obtain
    \[
      \ldsubeq[\Phi][\Delta'][L'] \metalevel \proglevel{\delta[\phi][\sigma] \circ
        \delta_1}{\delta'[\phi'][\sigma'] \circ \delta_1'}{\Delta[\phi][\sigma]}
    \]
    This is because we are sure $\Delta[\phi][\sigma]$ only contains types from MLTT.
    
    We want to lift $\sigma(u)$ and $\sigma'(u)$ to $\metalevel$, so we instantiate $\Bc$ with
    the related local substitutions above and obtain
    \[
      \ldtyequivt[\Phi][\Delta'][L'] \metalevel \proglevel{\sigma(u)[\delta[\phi][\sigma]\circ
        \delta_1]}{\sigma'(u)[\delta'[\phi'][\sigma']\circ
        \delta_1']}{T[\phi][\sigma][\delta[\phi][\sigma] \circ \delta_1]}{l[\phi]}
    \]
    The goal is achieved by another layering restriction. 
  \item[Case $i = \proglevel$] Similar to the previous case except that there is no need for
    layering restriction, as we are evaluating terms right inside of MLTT and
    therefore no lifting occurs. 
  \item[Case $i \in \{\varlevel, \codelevel\}$] In this case, we also follow similar footsteps as the
    results of looking up global substitutions must be stable under local
    substitutions. %
  \end{itemize}
\end{proof}

\begin{lemma}
  \begin{mathpar}
    \inferrule
    {\lsctxwf{\typeof i}\Gamma \\ U : \DTyp[\Delta]{i'} l \in \Psi \\ i' \in \{\codelevel, \proglevel\} \\ i' \le i \\ \lssubeq i \delta{\delta'} \Delta}
    {\lstypeq i{U^\delta}{U^{\delta'}}{l}}
  \end{mathpar}
\end{lemma}
\begin{proof}
  Similar to above, but simpler. %
  Use layering restriction as well when $i = \metalevel$.
\end{proof}

\begin{lemma}
  \begin{mathpar}
      \inferrule*
      {\lsgctx \Psi}
      {\lsgsubeq{\cdot}{\cdot}{\cdot}}

      \inferrule*
      {\lsgsubeq{\sigma}{\sigma'}{\Phi} \\ \lsctxeq \proglevel \Gamma\Delta}
      {\lsgsubeq{\sigma, \Gamma/g}{\sigma', \Delta/g}{\Phi, g: \Ctx}}
    \end{mathpar}
\end{lemma}
\begin{proof}
  Immediate.
\end{proof}

\begin{lemma}
  \begin{mathpar}
    \inferrule*
    {\lsgsubeq{\sigma}{\sigma'}{\Phi} \\ \lstypwf[\Phi] \proglevel T l \\ \typing[L]l\Level \\
      i \in \{\varlevel, \codelevel\} \\ \Ac :: \lstyequiv[\Psi][\Gamma[\sigma]] i t{t'}{T[\sigma]}l}
    {\lsgsubeq[\Psi]{\sigma, t/u}{\sigma', t'/u}{\Phi, u : \DTrm{i} T l}}
  \end{mathpar}
\end{lemma}
\begin{proof}
  Assuming $\tyequiv[L']{\phi}{\phi'}L$ and
  $\ldSgsubeq[\Psi'][L']{\sigma_1}{\sigma_1'}{\Psi[\phi]}$, we have to show
  \[
    \ldSgsubeq[\Psi'][L']{\sigma[\phi] \circ \sigma_1, t[\phi][\sigma_1]/u}{\sigma'[\phi] \circ
      \sigma_1', t[\phi][\sigma_1']/u}{(\Phi, u : \DTrm{i} T l)[\phi]}
  \]
  Our goal is to show, without loss of generality,
  \[
    \ldStyequivge[\Psi][\Gamma[\sigma[\phi] \circ \sigma_1]] i
    \proglevel{t[\phi][\sigma_1]}{t'[\phi][\sigma_1']}{T[\sigma[\phi] \circ \sigma_1]} l
  \]
  This is given by $\Ac$, modulo converting weakenings to substitutions. %
\end{proof}

\begin{lemma}
  \begin{mathpar}
    \inferrule*
    {\lsgsubeq{\sigma}{\sigma'}{\Phi} \\ \lsctxwf[\Phi] \proglevel \Gamma \\ \typing[L]l\Level
      \\ i \in \{\codelevel, \proglevel\} \\
      \lstypeq[\Psi][\Gamma[\sigma]] i{T}{T'}l}
    {\lsgsubeq{\sigma, T/U}{\sigma', T'/U}{\Phi, u : \DTyp i l}}    
  \end{mathpar}
\end{lemma}
\begin{proof}
  Similar to the previous case but simpler. 
\end{proof}

\begin{corollary}[Global Weakening Substitutions]\labeledit{lem:dt:sem:gwk}
  $\lsgsubst[\Psi,\Phi] {\wk_{\Psi}^{|\Phi|}}{\Psi}$
\end{corollary}

\begin{lemma}
  \begin{mathpar}
    \inferrule
    {\tyequiv[L]{l_1}{l_3}\Level \\ \tyequiv[L]{l_2}{l_4}\Level \\
      \lstypeq i S{S'}{l_1} \\ \lstypeq[\Psi][\Gamma, x : S \at{l_1}] iT{T'}{l_2}}
    {\lstypeq i{\PI{l_1}{l_2} x S T}{\PI{l_3}{l_4} x{S'}{T'}}{l_1 \sqcup l_2}}
  \end{mathpar}
\end{lemma}
\begin{proof}
  From $\lstypeq i S{S'}{l_1}$, we can conclude $\lsgctx \Psi$ and $\lsctxwf {\typeof
    i} \Gamma$. %
  Now assuming $\tyequiv[L']{\phi}{\phi'}L$ and
  $\ldSgsubeq[\Phi][L'] \sigma{\sigma'}{\Psi[\phi]}$, we have to consider all
  $i \ge \codelevel$. %
  \begin{itemize}
  \item[Case $i = \metalevel$] Then assuming
    $\ldsubeq[\Phi][\Delta][L'] \metalevel \metalevel \delta{\delta'}{\Gamma[\phi][\sigma]}$, we must
    show
    \[
      \ldtypeq[\Phi][\Delta][L'] \metalevel \metalevel{\PI{l_1}{l_2} x S T[\phi][\sigma][\delta]}{\PI{l_3}{l_4} x{S'}{T'}[\phi'][\sigma'][\delta']}{(l_1 \sqcup l_2)[\phi]}
    \]
    From escape lemmas, we are able to establish the reduction premises, typing
    premises and the generic equivalence. %
    We then only focus on the semantic premises. %
    First, from $\lstypeq i S{S'}{l_1}$, we obtain
    \[
      \ldtypeq[\Phi][\Delta][L'] \metalevel
      \metalevel{S[\phi][\sigma][\delta]}{S'[\phi'][\sigma'][\delta']}{l_1[\phi]}
    \]
    Then we further assume
    $\psi :: L'' \sep \Phi'; \Delta' \To_m L' \sep \Phi; \Delta$ and
    $\ldtyequivt[\Phi'][\Delta'][L''] \metalevel \metalevel{s}{s'}{S[\phi][\sigma][\delta]}{l_1[\phi]}$,
    we should prove
    \[
      \ldtypeq[\Phi'][\Delta'][L''] \metalevel
      \metalevel{T[\phi][\sigma][\delta, s/x]}{T'[\phi'][\sigma'][\delta',s'/x]}{l_1[\phi]}
    \]
    We are almost there, as long as we can provide
    \[
      \ldsubeq[\Phi'][\Delta'][L''] \metalevel \metalevel {\delta,s/x}{\delta',s'/x}{(\Gamma, x : S \at{l_1})[\phi][\sigma]}
    \]
    To prove these two local substitutions are related, we are interested in showing \break
    $\lsctxeq{\typeof i}{\Gamma, x : S \at{l_1}}{\Gamma, x : S' \at{l_3}}$, but
    this is immediate from a previous lemma.
  \item[Case $i = \proglevel$] This case follows similarly to the previous case. %
    It must range over $k \in \{\proglevel, \metalevel\}$ so a similar reasoning must be repeated twice.
    
  \item[Case $i = \codelevel$]
    This case is much simpler by using the local substitution lemma to remove the need
    to assume another local substitution. %
    Then we can simply apply identity local substitutions to all premises and use the
    previous case to conclude
    \[
      \ldStypwfge[\Phi][\Gamma[\phi][\sigma]][L'] \proglevel \proglevel{\PI{l_1}{l_2} x S T}{(l_1\sqcup l_2)[\phi]}
    \]
    If we introduce another local substitution, then we must reason about extending a
    local variable to an arbitrary local substitution, which is quite verbose and
    unnecessary. %
  \end{itemize}
\end{proof}

\begin{lemma}
  \begin{mathpar}
    \inferrule
    {\tyequiv[L]{l}{l'}\Level \\ \lstyequiv i t{t'}{\Ty l}{1 + l}}
    {\lstypeq i{\Elt l t}{\Elt{l'}{t'}}{l}}
  \end{mathpar}
\end{lemma}
\begin{proof}
  We use the fact that $\Ty{l}$ reduces only to itself and therefore it is only
  possible to expand $t \approx t'$ to the universe case. %  
\end{proof}

\begin{lemma}
  \begin{mathpar}
    \inferrule
    {\tyequiv[L]{l_1}{l_3}\Level \\ \tyequiv[L]{l_2}{l_4}\Level \\\\
      \lstyequiv i s{s'}{\Ty{l_1}}{1 + l_1} \\
      \lstyequiv[\Psi][\Gamma, x : \Elt{l_1} s \at{l_1}] i t{t'}{\Ty{l_2}}{1 + l_2}}
    {\lstyequiv i{\PI{l_1}{l_2} x s t}{\PI{l_3}{l_4} x{s'}{t'}}{\Ty{l_1 \sqcup l_2}}{\su{(l_1 \sqcup l_2)}}}
  \end{mathpar}
\end{lemma}
\begin{proof}
  Very similar to the previous proof. %
  We use the previous lemma and know that
  \[
    \lstypeq i{\Elt {l_1} s}{\Elt{l_3}{s'}}{l_1}
  \]
  We do the same for $t \approx t'$. %
  This gives us
  \[
    \lstypeq i{\PI {l_1}{l_2} x{\Elt{l_1} s}{\Elt{l_2} t}}{\PI {l_3}{l_4}
      x{\Elt{l_3}{s'}}{\Elt{l_4}{t'}}}{l_1 \sqcup l_2}
  \]
  When $i \in \{\proglevel, \metalevel\}$, we use 
  \begin{align*}
    \Elt{l_1 \sqcup l_2}{\PI{l_1}{l_2} x s t} &\redd {\PI {l_1}{l_2} x{\Elt{l_1}
                                                s}{\Elt{l_2} t}} \\
    \Elt{l_3 \sqcup l_4}{\PI{l_3}{l_4} x{s'}{t'}} &\redd {\PI {l_3}{l_4}
                                                    x{\Elt{l_3}{s'}}{\Elt{l_4}{t'}}}
  \end{align*}
  We have the goal using reduction expansion. 
\end{proof}

\begin{lemma}
  \begin{mathpar}
    \inferrule
    {\typing[L]l\Level  \\ \typing[L]{l'}\Level \\
      \lstyping i s{\Ty l}{1 + l} \\ \lstyping[\Psi][\Gamma, x : \Elt l s \at l] i t{\Ty{l'}}{1 + l'}}
    {\lstypeq i{\PI l{l'} x{\Elt l s}{\Elt{l'} t}}{\Elt{l \sqcup l'}{\PI l{l'} x s t}}{l \sqcup l'}}
  \end{mathpar}
\end{lemma}
\begin{proof}
  Here $i \in \{\proglevel, \metalevel\}$. %
  The proof is similar to the previous lemma except that we only use reduction
  expansion on one side. 
\end{proof}

\begin{lemma}
  \begin{mathpar}
    \inferrule
    {\tyequiv[L]{l_1}{l_3}\Level \\ \tyequiv[L]{l_2}{l_4}\Level \\
      \lstypeq i S{S'}{l_1} \\
      \lstyequiv[\Psi][\Gamma, x : S \at {l_1}] i t{t'}{T}{l_2}}
    {\lstyequiv i{\LAM {l_1}{l_2} x S t}{\LAM {l_3}{l_4} x{S'}{t'}}{\PI l{l'} x S T}{l_1 \sqcup l_2}}
  \end{mathpar}
\end{lemma}
\begin{proof}
  From $\lstypeq i S{S'}{l_1}$, we can conclude $\lsgctx \Psi$ and
  $\lsctxwf {\typeof i} \Gamma$. %
  The premise $\lstyequiv[\Psi][\Gamma, x : S \at {l_1}] i t{t'}{T}{l_2}$ also gives
  us $\lstypwf[\Psi][\Gamma, x : S \at{l_1}] iT{l_2}$. %
  Combining symmetry, transitivity and the previous lemma, we have
  \[
    \lstypwf i{\PI{l_1}{l_2} x S T}{l_1 \sqcup l_2}
  \]
  Now assuming $\tyequiv[L']{\phi}{\phi'}L$ and
  $\ldSgsubeq[\Phi][L'] \sigma{\sigma'}{\Psi[\phi]}$, we have to consider all
  $i \ge \codelevel$. %
  \begin{itemize}
  \item[Case $i = \metalevel$] Then assuming
    $\ldsubeq[\Phi][\Delta][L'] \metalevel \metalevel \delta{\delta'}{\Gamma[\phi][\sigma]}$, we must
    show
    \begin{itemize}
    \item
      $\D :: \ldtypeq[\Phi][\Delta][L'] \metalevel \metalevel{\PI{l_1}{l_2} x S
        T[\phi][\sigma][\delta]}{\PI{l_3}{l_4} x{S}{T}[\phi'][\sigma'][\delta']}{(l_1
        \sqcup l_2)[\phi]}$, and
    \item
      $\ldtyequiv[\Phi][\Delta][L'] \metalevel \metalevel{\LAM {l_1}{l_2} x S
        t[\phi][\sigma][\delta]}{\LAM {l_3}{l_4}
        x{S'}{t'}[\phi'][\sigma'][\delta']}\D$.
    \end{itemize}

    To obtain the goal we shall proceed in two steps. %
    First, we obtain
    \[
      \ldsubeq[\Phi][\Delta, x : S[\phi][\sigma][\delta] \at{l_1[\phi]}][L'] \metalevel \metalevel
      {\delta,x/x}{\delta',x/x}{(\Gamma, x : S \at{l_1})[\phi][\sigma]}
    \]. %
    Giving it to $\lstyequiv[\Psi][\Gamma, x : S \at {l_1}] i t{t'}{T}{l_2}$, we
    have
    \[
      \ldtyequivt[\Phi][\Delta, x : S[\phi][\sigma][\delta] \at{l_1[\phi]}][L'] \metalevel \metalevel
      {t[\phi][\sigma][\delta,x/x]}{t'[\phi'][\sigma'][\delta',x/x]}{T[\phi][\sigma][\delta,x/x]}{l_2[\phi]}
    \]
    A further escape gives us
    \[
      \lttrmgeq[\Phi][\Delta, x : S[\phi][\sigma][\delta] \at{l_1[\phi]}][L'] \metalevel
      {t[\phi][\sigma][\delta,x/x]}{t'[\phi'][\sigma'][\delta',x/x]}{T[\phi][\sigma][\delta,x/x]}{l_2[\phi]}
    \]
    From this, we conclude 
    \[
      \lttrmgeq[\Phi][\Delta][L'] \metalevel {(\LAM {l_1}{l_2} x S
        t)[\phi][\sigma][\delta]}{(\LAM {l_3}{l_4}
        x{S'}{t'})[\phi'][\sigma'][\delta']}{T[\phi][\sigma][\delta]}{l_2[\phi]}
    \]
    modulo the law of weak head closure.

    In the second step, we assume
    $\psi :: L'' \sep \Phi'; \Delta' \To_m L' \sep \Phi; \Delta$ and 
    $\ldtyequivt[\Phi'][\Delta'][L''] \metalevel \metalevel{s}{s'}{S[\phi][\sigma][\delta]}{l_1[\phi]}$,
    we should prove
    \[
      \ldtyequivt[\Phi'][\Delta'][L''] \metalevel \metalevel{t[\phi][\sigma][\delta,
        s/x]}{t'[\phi'][\sigma'][\delta',s'/x]}{T[\phi][\sigma][\delta, s/x]}{l_1[\phi]}
    \]
    then we have the goal modulo weak head expansion. %
    We are almost there, as long as we can provide
    \[
      \ldsubeq[\Phi'][\Delta'][L''] \metalevel \metalevel {\delta,s/x}{\delta',s'/x}{(\Gamma, x : S \at{l_1})[\phi][\sigma]}
    \]
    But in the previous lemma, we have seen it obvious.
    
  \item[Case $i = \proglevel$] Repeat the previous case at different layers twice. %
  \item[Case $i = \codelevel$] This case is very similar to that of $\Pi$. %
  \end{itemize}  
\end{proof}

\begin{lemma}
  \begin{mathpar}
    \inferrule
    {\tyequiv[L]{l_1}{l_3}\Level  \\ \tyequiv[L]{l_2}{l_4}\Level \\ \lstypeq i S{S'}{l_1}
      \\
      \lstypeq[\Psi][\Gamma, x : S \at{l_1}] i{T}{T'}{l_2} \\
      \Ac :: \lstyequiv i t{t'}{\PI {l_1}{l_2} x S T}{l_1 \sqcup l_2} \\ \lstyequiv i s{s'} S{l_1}}
    {\lstyequiv i{\APP t {l_1}{l_2} x S T s}{\APP{t'} {l_3}{l_4} x {S'}{T'}{s'}}{T[s/x]}{l_2}}
  \end{mathpar}
\end{lemma}
\begin{proof}
  From $\lstypeq i S{S'}{l_1}$, we can conclude $\lsgctx \Psi$ and
  $\lsctxwf {\typeof i} \Gamma$. %
  To show\break
  $\lstypwf{\typeof i}{T[s/x]}{l_2}$, we first show $\lssubst{\typeof i}{s/x}{\Gamma,x
    : S \at{l_1}}$ and get the goal using the local substitution lemma. %
  This is immediate by the semantic rule for local substitutions and
  \Cref{lem:dt:sem:lwk}. %
  
  Now assuming $\tyequiv[L']{\phi}{\phi'}L$ and
  $\ldSgsubeq[\Phi][L'] \sigma{\sigma'}{\Psi[\phi]}$, we have to consider all
  $i \ge \codelevel$. %
  \begin{itemize}
  \item[Case $i = \metalevel$] Then assuming
    $\ldsubeq[\Phi][\Delta][L'] \metalevel \metalevel \delta{\delta'}{\Gamma[\phi][\sigma]}$, we must
    show
    \begin{itemize}
    \item
      $\D :: \ldtypeq[\Phi][\Delta][L'] \metalevel \metalevel{T[\phi][\sigma][\delta,s[\phi][\sigma][\delta]/x]}{T'[\phi'][\sigma'][\delta',s'[\phi'][\sigma'][\delta']/x]}{l_2[\phi]}$, and
    \item
      $\ldtyequiv[\Phi][\Delta][L'] \metalevel \metalevel{{\APP t {l_1}{l_2} x S T
          s}[\phi][\sigma][\delta]}{{\APP{t'} {l_3}{l_4} x
          {S'}{T'}{s'}}[\phi'][\sigma'][\delta']}\D$.
    \end{itemize}
    $\D$ is easily concluded from $\lstypwf \metalevel {T[s/x]}{l_2}$. %

    We obtain the goal by instantiating $\Ac$, from which we get
    \[
      \ldtyequiv[\Phi][\Delta][L'] \metalevel \metalevel{t[\phi][\sigma][\delta]}{t'[\phi'][\sigma'][\delta']}{\PI {l_1}{l_2} x S T[\phi][\sigma][\delta]}{(l_1 \sqcup l_2)[\phi]}
    \]
    The semantics of ${\PI {l_1}{l_2} x S T[\phi][\sigma][\delta]}$ gives us the
    goal, up to irrelevance. 
  \item[Case $i = \proglevel$] Similar.
  \item[Case $i = \codelevel$] Follow the previous pattern, we use the local substitution
    lemma. 
  \end{itemize}
\end{proof}

\begin{lemma}
  \begin{mathpar}
    \inferrule
    {\typing[L]l\Level \\ \typing[L]{l'}\Level \\ 
      \lstyping[\Psi][\Gamma, x : S \at{l}] i t{T}{l'} \\ \lstyping i s{S}{l}}
    {\lstyequiv i{t[s/x]}{\APP {\LAM {l}{l'} x S t} {l}{l'} x{S} T s}{T[s/x]}{l'}}
  \end{mathpar}
\end{lemma}
\begin{proof}
  Here $i \in \{\proglevel, \metalevel\}$. %
  We obtain $\lsgctx \Psi$, $\lsctxwf {\typeof i} \Gamma$, and
  $\lstypwf{\typeof i}{T[s/x]}{l'}$ following similar lines as the previous lemma. %

  From the local substitution lemma, we also have $\lstyping i{t[s/x]}{T}{l'}$ so we
  are one reduction step away, which can be concluded by the reduction expansion lemma. %
\end{proof}

\begin{lemma}
  \begin{mathpar}
    \inferrule
    {\typing[L]l\Level \\ \typing[L]{l'}\Level \\ \lstypwf i S{l} \\
      \lstypwf[\Psi][\Gamma, x : S \at{l}] i{T}{l'} \\
      \lstyping i t{\PI {l}{l'}x S T}{l \sqcup l'}}
    {\lstyequiv i{\LAM {l}{l'} x{S} {\APP t {l}{l'}x S T x}}{t}{\PI {l}{l'} x S T}{l \sqcup l'}}
  \end{mathpar}
\end{lemma}
\begin{proof}
  Here $i \in \{\proglevel, \metalevel\}$. %
  We assume $\phi \approx \phi'$, $\sigma \approx \sigma'$ and $\delta \approx
  \delta'$, and finally $s \approx s'$. Then see
  \begin{align*}
    &~ \APP {\LAM {l}{l'} x{S} {\APP{t[\phi][\sigma][\delta]} {l}{l'}x S T x}} {l}{l'}x S
      T s \\
    \redd&~ {\APP {t[\phi][\sigma][\delta]}
           {l}{l'}x S T x}[s/x] \\
    =&~ {\APP {t[\phi][\sigma][\delta]} {l}{l'}x S T s} \\
    \approx &~ {\APP {t[\phi'][\sigma'][\delta']} {l}{l'}x S T{s'}}
  \end{align*}
  Therefore, we obtain the goal by weak head expansion of logical relations. 
\end{proof}

\begin{lemma}
  If
  \begin{itemize}
  \item $(k, i) \in \{(\proglevel, \proglevel), (\metalevel, \proglevel), (\metalevel, \metalevel)\}$,
  \item $\tyequiv[L]l{l'}\Level$,
  \item $\lstypeq[\Psi][\Gamma, x : \Nat \at 0] i M{M'} l$,
  \item $\lstyequiv i {s_1}{s_3} {M[\ze/x]}l$, 
  \item $\Bc :: \lstyequiv[\Psi][\Gamma, x : \Nat \at 0, y : M \at l] i{s_2}{s_4}{M[\su
      x/x]}l$,
  \item $\tyequiv[L']{\phi}{\phi'}L$,
  \item $\ldSgsubeq[\Phi][L'] \sigma{\sigma'}{\Psi[\phi]}$,
  \item $\ldsubeq[\Phi][\Delta][L'] k{\typeof i} \delta{\delta'}{\Gamma[\phi][\sigma]}$,
  \item $\ttrmreds[\Phi][\Delta][L'] k t w \Nat \ze$,
  \item $\ttrmreds[\Phi][\Delta][L'] k{t'}{w'}\Nat \ze$,
  \item $\lttrmgeq[\Phi][\Delta][L'] k w{w'}\Nat\ze$, 
  \item $\Ac :: \ldtynfeqnat[\Phi][\Delta][L'] k w{w'}$,
  \item
    $t_1 = {\ELIMN
      l{x.M[\phi][\sigma][\delta,x/x]}{(s_1[\phi][\sigma][\delta])}{x,y. s_2[\phi][\sigma][\delta,x/x,y/y]}{t}}$,
  \item $t_2 = {\ELIMN
      l{x.M'[\phi'][\sigma'][\delta',x/x]}{(s_3[\phi'][\sigma'][\delta'])}{x,y. s_4[\phi'][\sigma'][\delta',x/x,y/y]}{t'}}$
  \end{itemize}
  then
  \[
    \ldtyequivt[\Phi][\Delta][L'] k{\typeof i}{t_1}{t_2}{M[\phi][\sigma][\delta, t/x]}{l[\phi]}
  \]
\end{lemma}
\begin{proof}
  We do induction on $\Ac$.
  \begin{itemize}
  \item If $w = w' = \ze$, then we hit the base case. In this case, we use $\lstyequiv
    i {s_1}{s_3} {M[\ze/x]}l$ and reduction expansion to almost obtain the goal. %
    The only missing piece is to prove
    \[
      M[\phi][\sigma][\delta, t/x] \approx M[\phi][\sigma][\delta, \ze/x]
    \]
    This holds from symmetry and the fact that $t \approx \ze$, so that we can extend
    $\delta \approx \delta$.
  \item If $w = \su s$ and $w' = \su s'$, then we apply IH and also $\Bc$ to obtain
    the relation between recursive calls for $s$ and $s'$. %
    We perform a similar analysis to handle the types.
  \item If $w = \varlevel$ and $w = \varlevel'$ for some neutrals, then we relate them using the
    reflexivity of neutral.
  \end{itemize}
\end{proof}

\begin{lemma}
  \begin{mathpar}
    \inferrule
    {\tyequiv[L]l{l'}\Level \\ \lstypeq[\Psi][\Gamma, x : \Nat \at 0] i M{M'} l \\
      \lstyequiv i {s_1}{s_3} {M[\ze/x]}l \\
      \lstyequiv[\Psi][\Gamma, x : \Nat \at 0, y : M \at l] i{s_2}{s_4}{M[\su x/x]}l \\
      \Ac :: \lstyequiv i t{t'} \Nat \ze}
    {\lstyequiv i{\ELIMN l{x.M}{s_1}{x,y. s_2}t}{\ELIMN{l'}{x.M'}{s_3}{x,y. s_4}{t'}}{M[t/x]}{l}}
  \end{mathpar}
\end{lemma}
\begin{proof}
  Instantiate $\Ac$ and use the previous lemma for $i \in \{\proglevel, \metalevel\}$. %
  Notice that we know $t$ and $t'$ are related by $\Nat$ which must reduce to itself
  so we can supply all premises required by the previous lemma. %
  When $i = \codelevel$, we reuse the proof when $i = \proglevel$. 
\end{proof}

\begin{lemma}
  \begin{mathpar}
    \inferrule
    {\typing[L]l\Level \\ \lstypwf[\Psi][\Gamma, x : \Nat \at 0] i M l \\
      \lstyping i s {M[\ze/x]}l \\
      \lstyping[\Psi][\Gamma, x : \Nat \at 0, y : M \at l] i {s'}{M[\su x/x]}l}
    {\lstyequiv i{s}{\ELIMN l{x.M}s{x,y. s'}\ze}{M[\ze/x]}{l}}

    \inferrule
    {\typing[L]l\Level \\ \lstypwf[\Psi][\Gamma, x : \Nat \at 0] i M l \\
      \lstyping i s {M[\ze/x]}l \\
      \lstyping[\Psi][\Gamma, x : \Nat \at 0, y : M \at l] i {s'}{M[\su x/x]}l \\
      \lstyping i t \Nat \ze}
    {\lstyequiv i{s'[t/x,\ELIMN l{x.M}s{x,y. s'}t/y]}{\ELIMN l{x.M}s{x,y. s'}{(\su t)}}{M[\su t/x]}{l}}
  \end{mathpar}
\end{lemma}
\begin{proof}
  We use reduction expansion, semantic local substitution lemma and the previous lemma. %
\end{proof}

\begin{lemma}
  \begin{mathpar}
    \inferrule
    {\lstyequiv i t{t'}{T'}l \\ \lstypeq{\typeof i}T{T'}l}
    {\lstyequiv i t {t'} T l}
  \end{mathpar}
\end{lemma}
\begin{proof}
  Use irrelevance when $i \in \{\proglevel, \metalevel\}$. %
  When $i \in \{\varlevel, \codelevel\}$, we do a case analysis on the semantics of related terms. %
  Notice that all rules for terms in \Cref{sec:dt:semglob} contain a type relation. %
  We apply transitivity of related types and irrelevance. %
\end{proof}

\subsection{More Semantic Rules}

In the previous section, we have considered all possible (non-trivial) rules for all
layers. %
Among these rules, we have looked at the rules for global variables and see how
layering restriction enables code running in the semantics. %
In this section, we will finish the proof by considering rules that are available at
layer $\metalevel$. %
This will make our proofs in some sense simpler; there is only one layer to
consider. %
On the other hand, we will look into another important feature, recursors for code,
and how it is semantically justified. %

\begin{lemma}
  \begin{mathpar}
    \inferrule
    {\lsctxwf \metalevel\Gamma \\ \lsctxeq \proglevel \Delta{\Delta'} \\ \lstypeq[\Psi][\Delta] \proglevel T{T'} l \\ \tyequiv[L]l{l'}\Level}
    {\lstypeq \metalevel{\CTrm[\Delta]T l}{\CTrm[\Delta']{T'}{l'}}{0}}
  \end{mathpar}
\end{lemma}
\begin{proof}
  Assuming $\tyequiv[L']{\phi}{\phi'}L$,
  $\ldSgsubeq[\Phi][L'] \sigma{\sigma'}{\Psi[\phi]}$ and
  $\ldsubeq[\Phi][\Delta_1][L'] \metalevel \metalevel \delta{\delta'}{\Gamma[\phi][\sigma]}$, the
  focuses are
  \begin{itemize}
  \item $\Ac :: \ldSctxeqge[\Phi][L'] \proglevel \proglevel {\Delta[\phi][\sigma]}{\Delta'[\phi'][\sigma']}$,
  \item $\Bc :: \ldStypwfge[\Phi][\Delta[\phi][\sigma]][L'] \proglevel \proglevel{T[\phi][\sigma]}{l[\phi]}$ and
  \item $\Cc :: \ldStypwfge[\Psi][\Delta'[\phi][\sigma]][L'] \proglevel \proglevel{T'[\phi][\sigma]}{l[\phi]}$.
  \end{itemize}
  $\Ac$ is simple as we only need to convert weakenings to substitutions. %
  Then we can use $\lsctxeq \proglevel \Delta{\Delta'}$. %
  $\Bc$ and $\Cc$ are symmetric so we only focus on $\Cc$ which is slightly more
  complex. %
  We similarly want to convert weakenings to substitutions so that we can apply
  $\lstypeq[\Psi][\Delta] \proglevel T{T'} l$ to obtain the goal. %
  The local contexts are mismatched, though it is not a problem. %
  The local contexts are used when further assuming related local substitutions, and
  by $\lsctxeq \proglevel \Delta{\Delta'}$, we can use irrelevance to swap the local contexts. 
\end{proof}

\begin{lemma}
  \begin{mathpar}
    \inferrule
    {\lsctxwf \metalevel\Gamma \\ \lsctxeq \proglevel \Delta{\Delta'} \\ \tyequiv[L]l{l'}\Level}
    {\lstypeq \metalevel{\CTyp[\Delta]l}{\CTyp[\Delta']{l'}}{0}}
  \end{mathpar}
\end{lemma}
\begin{proof}
  Similar to above but simpler. 
\end{proof}

\begin{lemma}
  \begin{mathpar}
    \inferrule
    {\lsctxwf \metalevel \Gamma \\ \lstyping[\Psi][\Delta]\codelevel t T l}
    {\lstyequiv \metalevel{\boxit t}{\boxit{t}}{\CTrm[\Delta] T l}{0}}
  \end{mathpar}
\end{lemma}
\begin{proof}
  Assuming $\tyequiv[L']{\phi}{\phi'}L$,
  $\ldSgsubeq[\Phi][L'] \sigma{\sigma'}{\Psi[\phi]}$ and
  $\ldsubeq[\Phi][\Delta'][L'] \metalevel \metalevel \delta{\delta'}{\Gamma[\phi][\sigma]}$, then
  the goal requires
  \[
    \ldtyping[\Phi][\Delta][L']\codelevel \proglevel{t[\phi][\sigma]}{T[\phi][\sigma]}{l[\phi]}
  \]
  Notice that the effect of local substitutions is irrelevant anymore. %
  This goal can be instantiate by\break $\lstyping[\Psi][\Delta]\codelevel t T l$ by passing in
  the same universe and global substitutions, and the identity local substitution. %
\end{proof}

\begin{lemma}
  \begin{mathpar}
  \inferrule
  {\lsctxwf \metalevel \Gamma \\ \lstypwf[\Psi][\Delta]\codelevel T l}
  {\lstyequiv \metalevel{\boxit T}{\boxit{T}}{\CTyp[\Delta] l}{0}}
  \end{mathpar}
\end{lemma}
\begin{proof}
  Similar to above but simpler. 
\end{proof}

\begin{lemma}
  \begin{mathpar}
    \inferrule
    {\tyequiv[L]{l_1}{l_3}\Level \\ \tyequiv[L]{l_2}{l_4}\Level \\
      \lpequiv \proglevel \Delta{\Delta'} \\
      \lstypeq \proglevel T{T'}{l_2} \\
      \lstyequiv \metalevel{t}{t'}{\CTrm[\Delta]T{l_2}}{0} \\
      \lstypeq[\Psi][\Gamma,x_T : \CTrm[\Delta]T{l_2} \at{0}]\metalevel{M}{M'}{l_1} \\
      \lstyequiv[\Psi, u : \DTrm[\Delta]\codelevel T {l_2}]\metalevel{t_1}{t_2}{M[\boxit{u^\id}/x_t]}{l_1}}
    {\lstyequiv \metalevel{\LETBTRM{l_1}{l_2} \Delta{T}{x_t.M}{u}{t_1}t}{\LETBTRM{l_3}{l_4}{\Delta'}{T'}{x_T.M'}{u}{t_2}{t'}}{M[t/x_t]}{l_1}}
  \end{mathpar}
\end{lemma}
\begin{proof}
  Assuming $\tyequiv[L']{\phi}{\phi'}L$,
  $\ldSgsubeq[\Phi][L'] \sigma{\sigma'}{\Psi[\phi]}$ and
  $\ldsubeq[\Phi][\Delta'][L'] \metalevel \metalevel \delta{\delta'}{\Gamma[\phi][\sigma]}$, then
  the goal requires
  \[
    \ldtyequivt[\Phi][\Delta][L']\metalevel \metalevel{s_1}{s_2}{M[\phi][\sigma][\delta,t[\phi][\sigma][\delta]/x_t]}{l_1[\phi]}
  \]
  where $s_1 = {\LETBTRM{l_1}{l_2} \Delta{T}{x_t.M}{u}{t_1}t[\phi][\sigma][\delta]}$
  and\\
  $s_2 = {\LETBTRM{l_3}{l_4}{\Delta'}{T'}{x_T.M'}{u}{t_2}{t'}[\phi'][\sigma'][\delta']}$.
  By instantiating\\
  $\lstyequiv \metalevel{t}{t'}{\CTrm[\Delta]T{l_2}}{0}$, we have
  \[
    \ldtyequivt[\Phi][\Delta'][L']\metalevel \metalevel{t[\phi][\sigma][\delta]}{{t'}[\phi'][\sigma'][\delta']}{\CTrm[\Delta]T{l_2}[\phi][\sigma][\delta]}{0}
  \]
  Unfolding it, there are two possibilities for this relation.
  \begin{itemize}[label=Case]
  \item We know for some $t''$,
    \begin{itemize}
    \item $\ttrmreds[\Phi][\Delta'][L']
      k{t[\phi][\sigma][\delta]}{\boxit{t''}}{\CTrm[\Delta]T{l_2}[\phi][\sigma][\delta]}{0}$,
    \item $\ttrmreds[\Phi][\Delta'][L']
      k{t[\phi'][\sigma'][\delta']}{\boxit{t''}}{\CTrm[\Delta]T{l_2}[\phi][\sigma][\delta]}{0}$,
      and
    \item $\ldtyping[\Phi][\Delta[\phi][\sigma]][L']\codelevel \proglevel{t''}{T[\phi][\sigma]}{0}$
    \end{itemize}
    We observe that
    \[
      \ldSgsubeq[\Phi][L']{\sigma,t''/u}{\sigma', t''/u}{\Psi[\phi], u :
        \DTrm[\Delta[\phi][\sigma]] \codelevel{T[\phi][\sigma]}{0}}
    \]
    Giving it to $\lstyequiv[\Psi, u : \DTrm[\Delta]\codelevel T
    {l_2}]\metalevel{t_1}{t_2}{M[\boxit{u^\id}/x_t]}{l_1}$, we obtain
    \[
      \ldtyequivt[\Phi][\Delta'][L']\metalevel
      \metalevel{t_1[\phi][\sigma,t''/u][\delta]}{{t_2}[\phi'][\sigma',t''/u][\delta']}{M[\phi][\sigma][\delta,
        \boxit{t''}/x_t]}{l_1[\phi]}
    \]
    By reduction expansion, we have 
    \[
      \ldtyequivt[\Phi][\Delta'][L']\metalevel \metalevel{s_1}{s_2}{M[\phi][\sigma][\delta,
        \boxit{t''}/x_t]}{l_1[\phi]}
    \]
    which is almost the goal.  %
    To tame the goal, we observe that
    \[
      \ldtyequivt[\Phi][\Delta'][L']\metalevel \metalevel{t[\phi][\sigma][\delta]}{\boxit{t''}}{\CTrm[\Delta]T{l_2}[\phi][\sigma][\delta]}{0}
    \]
    which further allows us to conclude
    \[
      \ldsubeq[\Phi][\Delta'][L'] \metalevel
      \metalevel{\delta,t[\phi][\sigma][\delta]/x_t}{\delta',\boxit{t''}/x_t}{(\Gamma, x_T : \CTrm[\Delta]T{l_2} \at{0})[\phi][\sigma]}
    \]
    Apply $\lstypeq[\Psi][\Gamma,x_T : \CTrm[\Delta]T{l_2} \at{0}]\metalevel{M}{M'}{l_1}$
    gives us the goal using irrelevance and weak head expansion. %
  \item We know for some $\nu$ and $\nu'$, 
    \begin{itemize}
    \item $\ttrmreds[\Phi][\Delta'][L']
      k{t[\phi][\sigma][\delta]}{\nu}{\CTrm[\Delta]T{l_2}[\phi][\sigma][\delta]}{0}$,
    \item $\ttrmreds[\Phi][\Delta'][L']
      k{t[\phi'][\sigma'][\delta']}{\nu'}{\CTrm[\Delta]T{l_2}[\phi][\sigma][\delta]}{0}$,
      and
    \item $\lttrmgneeq[\Phi][\Delta'][L'] \metalevel \nu{\nu'}{\CTrm[\Delta] T{l_2}[\phi][\sigma][\delta]}{0}$
    \end{itemize}

    The idea here is to use the law of congruence for neutrals to establish a generic
    equivalence between neutrals, and then we use reflexivity for neutrals to relate
    two neutrals using the logical relations. %

    The process requires us to provide
    \[
      \ldsubeq[\Phi][\Delta',x_T : \CTrm[\Delta]T{l_2}[\phi][\sigma]
      \at{0}][L'] \metalevel \metalevel {\delta,x_T/x_T}{\delta',x_T/x_T}{(\Gamma, x_T :
        \CTrm[\Delta]T{l_2} \at{0})[\phi][\sigma]}
    \]
    This is immediate. %
    We also need
    \[
      \ldSgsubeq[\Phi,u : \DTrm[\Delta]\codelevel T {l_2}[\phi][\sigma]][L']{\sigma,u^\id/u}{\sigma',u^\id/u}{(\Psi, u : \DTrm[\Delta]\codelevel T {l_2})[\phi]}
    \]
    In this case, we should prove
    \[
      \ldStypingge[\Phi,u : \DTrm[\Delta]\codelevel T
      {l_2}[\phi][\sigma]][\Delta[\phi][\sigma]][L'] \proglevel
      \proglevel{u^\id}{T[\phi][\sigma]}{0}
    \]
    This turns out to have been checked by the lemma of reflexive global weakening in
    \Cref{sec:dt:semglob}. %
    Finally, $\ldsubeq[\Phi][\Delta'][L'] \metalevel \metalevel \delta{\delta'}{\Gamma[\phi][\sigma]}$
    is weakened before applying. 
  \end{itemize}
\end{proof}

\begin{lemma}
  \begin{mathpar}
    \inferrule
    {\tyequiv[L]{l_1}{l_3}\Level \\ \tyequiv[L]{l_2}{l_4}\Level \\ \lpequiv \proglevel \Delta{\Delta'} \\
      \lstyequiv \metalevel{t}{t'}{\CTyp[\Delta]{l_2}}{0} \\
      \lstypeq[\Psi][\Gamma,x_T : \CTyp[\Delta]{l_2} \at{0}]\metalevel{M}{M'}{l_1} \\
      \lstyequiv[\Psi, U : \DTyp[\Delta]\codelevel{l_2}]\metalevel{t_1}{t_2}{M[\boxit{U^\id}/x_T]}{l_1}}
    {\lstyequiv \metalevel{\LETBTYP{l_1}{l_2} \Delta{x_T.M}{U}{t_1}t}{\LETBTYP{l_3}{l_4}{\Delta'}{x_T.M'}{U}{t_2}{t'}}{M[t/x_T]}{l_1}}
  \end{mathpar}
\end{lemma}
\begin{proof}
  Similar to above but simpler. 
\end{proof}

Next, we consider the semantics for the recursive principles. %
The following lemma needs to be mutually proved.
\begin{lemma}
  If
  \begin{itemize}
  \item $S_A$,
  \item $\tyequiv[L']{\phi}{\phi'}L$,
  \item $\ldSgsubeq[\Phi][L'] \sigma{\sigma'}{\Psi[\phi]}$,
  \item $\ldsubeq[\Phi][\Delta_1][L'] \metalevel \metalevel \delta{\delta'}{\Gamma[\phi][\sigma]}$,
  \item $\tyequiv[L']l{l'}\Level$,
  \item $\ldSctxeqge[\Phi][L']\proglevel \proglevel\Delta{\Delta'}$,
  \item $\ttrmreds[\Phi][\Delta_1][L'] \metalevel t{\boxit{T_1}} {\CTyp[\Delta] l}{0}$,
  \item $\ttrmreds[\Phi][\Delta_1][L'] \metalevel{t'}{\boxit{T_1}}{\CTyp[\Delta] l}{0}$,
  \item $\Ac :: \ldtypwf[\Psi][\Delta][L'] \codelevel \proglevel {T_1}l$, 
  \item
    $t_1 = \ELIMTYPn{l_1[\phi]}{l_2[\phi]}{(\vect M[\phi][\sigma][\delta])}{(\vect
      b[\phi][\sigma][\delta])}l \Delta t$, and
  \item
    $t_2 = \ELIMTYPn{l_1[\phi]}{l_2[\phi]}{(\vect M'[\phi][\sigma][\delta])}{(\vect
      b'[\phi][\sigma][\delta])}{l'}{\Delta'}{t'}$,
  \end{itemize}
  then
  \[
    \ldtyequivt[\Phi][\Delta_1][L'] \metalevel \metalevel{t_1}{t_2}{M[l/\ell,\Delta/g,t/x_T][\phi][\sigma][\delta]}{l_1[\phi]}
  \]
\end{lemma}
\begin{lemma}
  If
  \begin{itemize}
  \item $S_A$,
  \item $\tyequiv[L']{\phi}{\phi'}L$,
  \item $\ldSgsubeq[\Phi][L'] \sigma{\sigma'}{\Psi[\phi]}$,
  \item $\ldsubeq[\Phi][\Delta_1][L'] \metalevel \metalevel \delta{\delta'}{\Gamma[\phi][\sigma]}$,
  \item $\tyequiv[L']l{l'}\Level$,
  \item $\ldSctxeqge[\Phi][L']\proglevel \proglevel\Delta{\Delta'}$,
  \item $\ldStypeqge[\Phi][\Delta][L'] \proglevel \proglevel T{T'} l$,
  \item $\ttrmreds[\Phi][\Delta_1][L'] \metalevel t{\boxit t} {\CTrm[\Delta] T l} 0$,
  \item $\ttrmreds[\Phi][\Delta_1][L'] \metalevel{t'}{\boxit{t_1}}{\CTrm[\Delta] T l} 0$,
  \item $\Bc :: \ldtyping[\Psi][\Delta][L'] \codelevel \proglevel {t_1}T l$,
  \item
    $t_1 = \ELIMTRMn{l_1[\phi]}{l_2[\phi]}{(\vect M[\phi][\sigma][\delta])}{(\vect
      b[\phi][\sigma][\delta])}l \Delta T t$, and
  \item
    $t_2 = \ELIMTRMn{l_1[\phi]}{l_2[\phi]}{(\vect M'[\phi][\sigma][\delta])}{(\vect
      b'[\phi][\sigma][\delta])}{l'}{\Delta'}{T'}{t'}$,
  \end{itemize}
  then
  \[
    \ldtyequivt[\Phi][\Delta_1][L'] \metalevel \metalevel{t_1}{t_2}{M'[l/\ell,\Delta/g,T/U_T,t/x_t][\phi][\sigma][\delta]}{l_2[\phi]}
  \]
\end{lemma}
where $S_A$ is the set containing all semantic judgments for motives and branches in
$L \sep \Psi; \Gamma$.
\begin{proof}
  The idea is to do a mutual induction on $\Ac$ and $\Bc$. %
  Recall that they are mutually defined structures as shown in
  \Cref{sec:dt:semglob}. %
  Let us pick two cases to discuss:
  \begin{itemize}[label=Case]
  \item
    \begin{mathpar}
      \inferrule
      {u : \DTrm[\Delta_2]{i'}{T_2}{l''} \in \Phi \\ i' \in \{\varlevel, \codelevel\} \\ 
        \ldsubst[\Phi][\Delta][L'] \codelevel \proglevel{\delta_2}{\Delta_2} \\
        \tyequiv[L]l{l''}\Level \\
        \ldStypeqge[\Phi][\Delta][L'] \proglevel \proglevel{T}{T_2[\delta_2]}l \\
        \ldStypingge[\Phi][\Delta][L'] \proglevel \proglevel{u^{\delta_2}}{T}l}
      {\ldtyping[\Phi][\Delta][L'] \codelevel \proglevel{u^{\delta_2}}{T}l}
    \end{mathpar}
    In this case, we must block the evaluation. %
    The idea follows closely to the neutral case for $\tletbox$ and we should apply
    the law of neutral recursion on code. %
    In this case, we see that $\ltctxgeq[\Phi][L'] \proglevel{\Delta}{\Delta'}$,
    $\tconvtyp[\Phi][\Delta][L'] \proglevel T{T'}{l}$. %
    Then we have to show that the motives and the branches are related by generic
    equivalence. %
    The idea is to extend all substitutions if necessary. %
    We have shown that all substitutions can be extended by identities in the
    semantics (and universe substitutions are the same in both syntax and
    semantics). %
    
  \item
    \begin{mathpar}
      \inferrule
      {\typing[L']{l_3}\Level  \\ \typing[L']{l_4}\Level \\ \ldtypwf[\Phi][\Delta][L'] \codelevel \proglevel{S_2}{l_3} \\
        \ldtyping[\Phi][\Delta, x :{S_2} \at{l_3}][L'] \codelevel \proglevel{t_2}{T_2}{l_4} \\
        \tyequiv[L']l{l_3 \sqcup l_4}\Level \\
        \ldStypeqge[\Phi][\Delta][L'] \proglevel \proglevel{T}{\PI{l_3}{l_4} x{S_2}{T_2}}l \\
        \ldStypingge[\Phi][\Delta][L'] \proglevel \proglevel{\LAM {l_3}{l_4} x{S_2}{t_2}}{T}l}
      {\ldtyping[\Phi][\Delta][L'] \codelevel \proglevel{\LAM {l_3}{l_4} x{S_2}{t_2}}{T}{l}}
    \end{mathpar}
    In this case, we should go down and recurse on $S_2$ and $t_2$. %
    We then use the semantic rule for $t_{\lambda}$ to substitute in the results of
    the recursive calls for $S_2$ and $t_2$. %
    We also obtain $\ldStypwfge[\Phi][\Delta, x :{S_2} \at{l_3}][L'] \proglevel \proglevel{T_2}l$. %
    Therefore we have everything we need to use the semantic rule for $t_\lambda$. %
    The only missing piece is that the conclusion requires relation between $T$ and
    $T'$. %
    Meanwhile, $t_\lambda$ only gives us ${\PI{l_3}{l_4} x{S_2}{T_2}}$. %
    The solution lies in
    \begin{gather*}
      \ldStypeqge[\Phi][\Delta][L'] \proglevel \proglevel T{T'} l \\
      \ldStypeqge[\Phi][\Delta][L'] \proglevel \proglevel{T}{\PI{l_3}{l_4} x{S_2}{T_2}}l
    \end{gather*}
    Thus these three types are related. %
    Since ${\PI{l_3}{l_4} x{S_2}{T_2}}$ is already in normal form, so we know both
    $T'$ and $T'$ must reduce to it. %
    Together with the $\beta$ rule when hitting the $\lambda$ case, we use weak head
    expansion to obtain the desired goal. %
  \end{itemize}
\end{proof}

These two lemmas in the semantics give the recursions on code of types and terms and
actual do the recursions. %
For the semantic rules for the recursors, we first use these lemmas to prove the
congruence rules. %
We are almost done wit the congruence rules except that we have to handle the neutral
cases. %
This is virtually identical to the global variable cases, where we use the law of
neutral recursion of code to relate neutral terms and use reflexivity for neutrals to
establish the logical relations. %
The rest are the $\beta$ rules. %
They are even simpler due to access to the congruence rules. %
Then we use reduction expansion to achieve the goals. %

What are left now are the meta-functions including universe-polymorphic functions. %
Fist the meta-functions for local contexts and types are very similar. %
\begin{lemma}
  \begin{mathpar}
    \inferrule
    {\lsctxwf \metalevel\Gamma \\ \lsctxeq \proglevel \Delta{\Delta'} \\ \lstypeq[\Psi, U : \DTyp[\Delta] \proglevel l] \metalevel T{T'}{l'} \\ \tyequiv[L]{l_1}{l_3}\Level \\ \tyequiv[L]{l_2}{l_4}\Level}
    {\lstypeq \metalevel{\TPI U[\Delta]{l_1}{l_2}T}{\TPI U[\Delta']{l_3}{l_4}{T'}}{l_2}}
  \end{mathpar}
\end{lemma}
\begin{proof}
  Assuming $\tyequiv[L']{\phi}{\phi'}L$,
  $\ldSgsubeq[\Phi][L'] \sigma{\sigma'}{\Psi[\phi]}$ and
  $\ldsubeq[\Phi][\Delta'][L'] \metalevel \metalevel \delta{\delta'}{\Gamma[\phi][\sigma]}$, then
  we should prove
  \[
    \ldtypeq[\Phi][\Delta'][L'] \metalevel \metalevel{\TPI
      U[\Delta]{l_1}{l_2}T[\phi][\sigma][\delta]}{\TPI
      U[\Delta']{l_3}{l_4}{T'}[\phi'][\sigma'][\delta']}{l_2[\phi]}
  \]
  Most premises are simple. The tricky part is to show $\E_1$ and $\E_2$, which are
  symmetric. %
  Knowing $\lsctxeq \proglevel \Delta{\Delta'}$, it is sufficient to only prove one of them. %
  Using $\lstypeq[\Psi, U : \DTyp[\Delta] \proglevel l] \metalevel T{T'}{l'}$, it is easy to see that
  both $\E_1$ and $E_2$ hold, by extending $\sigma$ and $\sigma'$. 
\end{proof}

\begin{lemma}
  \begin{mathpar}
    \inferrule
    {\lsctxwf \metalevel\Gamma \\ \lstypeq[\Psi, g : \Ctx] \metalevel T{T'} l \\ \tyequiv[L]l{l'}\Level}
    {\lstypeq \metalevel{\CPI g l T}{\CPI g{l'}{T'}}{l}}
  \end{mathpar}
\end{lemma}
\begin{proof}
  Similar to above but simpler. 
\end{proof}

\begin{lemma}
  \begin{mathpar}
    \inferrule
    {\lsctxwf \metalevel\Gamma \\ \lstyequiv[\Psi, U : \DTyp[\Delta] \proglevel{l_1}]\metalevel{t}{t'}{T}{l_2} \\
      \tyequiv[L]{l_1}{l_3}\Level \\ \tyequiv[L]{l_2}{l_4}\Level}
    {\lstyequiv \metalevel{\TLAM{l_1}{l_2}U t}{\TLAM{l_3}{l_4}U{t'}}{\TPI U[\Delta]{l_1}{l_2}
        T}{l_2}}
  \end{mathpar}
\end{lemma}
\begin{proof}
  In this proof, we are asked to prove two symmetric proof which is essentially to
  show that the results of applying ${\TLAM{l_1}{l_2}U t}$ and
  ${\TLAM{l_3}{l_4}U{t'}}$ are related. %
  This is immediate by extending related global substitutions and applying them to
  $\lstyequiv[\Psi, U : \DTyp[\Delta] \proglevel{l_1}]\metalevel{t}{t'}{T}{l_2}$. %
\end{proof}

\begin{lemma}
  \begin{mathpar}
    \inferrule
    {\lstyequiv \metalevel{t}{t'}{\TPI U[\Delta] l{l'}{T''}}{l'} \\ \lstypeq[\Psi][\Delta] \proglevel{T}{T'}l}
    {\lstyequiv \metalevel{\TAPP t{T}}{\TAPP{t'}{T'}}{T''[T/U]}{l'}}
  \end{mathpar}
\end{lemma}
\begin{proof}
  We obtain the goal quite easily by using the semantics of related terms of type
  \break $\TPI U[\Delta] l{l'}{T''}$.
\end{proof}

\begin{lemma}
  \begin{mathpar}
    \inferrule
    {\lsctxwf \metalevel \Gamma \\ \lstyping[\Psi, U : \DTyp[\Delta] \proglevel l]\metalevel{t}{T'}{l'} \\ \typing[L]{l}\Level \\ \typing[L]{l'}\Level \\ \lstypwf[\Psi][\Delta] \proglevel{T} l}
    {\lstyequiv \metalevel{t[T/U]}{\TAPP{(\TLAM{l}{l'}U t)}{T}}{T'[T/U]}{l'}}

    \inferrule
    {\lstyping \metalevel{t}{\TPI U[\Delta] l{l'}{T'}}{1 + l \sqcup l'}}
    {\lstyequiv \metalevel{\TLAM{l}{l'}U{(\TAPP t{U^\id})}}{t}{\TPI U[\Delta] l{l'}{T'}}{l'}}
  \end{mathpar}
\end{lemma}
\begin{proof}
  These two rules are also very simple to prove and follow similar lines to those of
  $\Pi$ types. %
  For the $\beta$ rule, we use reduction expansion. For the $\eta$ rule, we see it by
  noticing applying related global substitutions. %
\end{proof}

The corresponding introduction, elimination, $\beta$ and $\eta$ rules for
meta-functions for local contexts are proved similarly but just simpler. %

Now the only last piece is universe-polymorphic functions.

\begin{lemma}
  \begin{mathpar}
    \inferrule
    {\lstypwf[\Psi][\Gamma][L,\vect\ell] \metalevel T l \\
      \lstypeq[\Psi][\Gamma][L,\vect\ell] \metalevel T{T'} l \\ |\vect \ell| > 0 \\ \tyequiv[L,\vect\ell]l{l'}\Level}
    {\lstypeq \metalevel {\UPI \ell l T}{\UPI \ell{l'}{T'}}{\omega}}
  \end{mathpar}
\end{lemma}
\begin{proof}
  Assuming $\tyequiv[L']{\phi}{\phi'}L$,
  $\ldSgsubeq[\Phi][L'] \sigma{\sigma'}{\Psi[\phi]}$ and
  $\ldsubeq[\Phi][\Delta'][L'] \metalevel \metalevel \delta{\delta'}{\Gamma[\phi][\sigma]}$, then
  we should prove
  \[
    \ldtypeq[\Phi][\Delta'][L'] \metalevel \metalevel{\UPI \ell l T[\phi][\sigma][\delta]}{\UPI
      \ell{l'}{T'}[\phi'][\sigma'][\delta']}{\omega}
  \]
  Looking at the semantics of universe-polymorphic functions, we see that
  $\lstypwf[\Psi][\Gamma][L,\vect\ell] \metalevel T l$ has already provided the goal. %
  It is important to see that $\typing[L']{l[\phi, \vect l/\vect \ell]}\Level$ if all
  universe levels in $\vect l$ are well-formed so it is strictly smaller than
  $\omega$. %
  Therefore, we still have a well-founded semantics. 
\end{proof}

\begin{lemma}
  \begin{mathpar}
    \inferrule
    {\lstyequiv[\Psi][\Gamma][L, \vect \ell]\metalevel{t}{t'}{T}{l} \\ |\vect \ell| > 0 \\ \tyequiv[L,\vect\ell]l{l'}\Level}
    {\lstyequiv \metalevel {\ULAM l \ell t}{\ULAM {l'} \ell{t'}}{\UPI \ell l T}{\omega}}

    \inferrule
    {\lstypwf[\Psi][\Gamma][L,\vect\ell] \metalevel T l \\
      \lstyequiv \metalevel{t}{t'}{\UPI \ell l T}{\omega} \\ |\vect\ell| = |\vect l| = |\vect
      l'| > 0 \\
      \forall 0 \le n < |\vect l| ~.~ \tyequiv[L]{\vect l(n)}{\vect l'(n)}\Level}
    {\lstyequiv \metalevel{\UAPP t l}{t'~\$~\vect l'}{T[\vect l/\vect \ell]}{l[\vect l/\vect \ell]}}

    \inferrule
    {\lsctxwf \metalevel \Gamma \\ \lstyping[\Psi][\Gamma][L, \vect \ell]\metalevel{t}{T}{l} \\ \typing[L,\vect\ell]l\Level \\
      |\vect\ell| = |\vect l| > 0 \\ \forall l' \in \vect l ~.~ \typing[L]{l'}\Level}
    {\lstyequiv \metalevel{t[\vect l/\vect \ell]}{(\ULAM l \ell t)~\$~\vect l}{T[\vect l/\vect \ell]}{l[\vect l/\vect \ell]}}

    \inferrule
    {\lstyping \metalevel{t}{\UPI \ell l T}{\omega}}
    {\lstyequiv \metalevel{\ULAM l \ell{(t~\$~\vect\ell)}}{t}{\UPI \ell l T}{\omega}}
  \end{mathpar}
\end{lemma}
\begin{proof}
  All these rules are relatively simple. %
  The congruence rule for introduction can be proved by following the definition of
  related terms of type ${\UPI \ell l T}$.

  The congruence rule for elimination can be proved by using the definition of related
  terms of type ${\UPI \ell l T}$.

  The $\beta$ rule can be derived from using reduction expansion.

  The $\eta$ rule can be given after applying arbitrary equivalent universe
  substitutions. %
\end{proof}

At this point, we have proved all semantic rules and thus the fundamental theorems
hold.

\section{Consequences and Decidability of Convertibility}

In the previous section, we have established the fundamental theorems, and using the
escape lemma, we see that all syntactic equivalent types and terms are also
semantically related by the Kripke logical relations. %
In this section, we will instantiate the generic equivalence so that we can obtain
desired properties that are difficult to prove syntactically. %

\subsection{First Instantiation: Syntactic Equivalence}

First we instantiate the generic equivalence with syntactic equivalences of types and
terms. %
The laws are all easily instantiated by the corresponding equivalence rules. %
We also see that the derived equivalences between local contexts and local
substitutions are equivalent to the corresponding syntactic equivalences. %
From the fundamental theorems and the escape lemma, we are able to derive the
following lemmas.

\begin{lemma}[Injectivity of Type Constructors] $ $
  \begin{itemize}
  \item If $\lttypeq i{\PI{l}{l'} x S T}{\PI{l}{l'} x{S'}{T'}}{l \sqcup l'}$ and
    $i \in \{\proglevel, \metalevel\}$, then $\lttypeq i{S}{S'}l$ and
    $\lttypeq[\Psi][\Gamma, x : S \at l]i{T}{T'}{l'}$.
  \item If
    $\lttyequiv i{\PI{l}{l'} x s t}{\PI{l}{l'} x{s'}{t'}}{\Ty{l \sqcup l'}}{\su{(l
        \sqcup l')}}$ and $i \in \{\proglevel, \metalevel\}$, then \break
    $\lttyequiv i{s}{s'}{\Ty l}{1 + l}$ and
    $\lttyequiv[\Psi][\Gamma, x : \Elt l s \at l]i{t}{t'}{\Ty{l'}}{1 + l'}$.
  \item If $\lttypeq \metalevel{\CTyp[\Delta]l}{\CTyp[\Delta']l}{0}$, then
    $\lpequiv \proglevel \Delta{\Delta'}$.
  \item If $\lttypeq \metalevel{\CTrm[\Delta] T l}{\CTrm[\Delta'] T l}{0}$, then
    $\lpequiv \proglevel \Delta{\Delta'}$ and $\lttypeq[\Psi][\Delta] \proglevel{T}{T'}l$.
  \item If $\lttypeq \metalevel{\CPI g l T}{\CPI g l{T'}}l$, then
    $\lttypeq[\Psi, g : \Ctx] \metalevel{T}{T'}l$.
  \item If
    $\lttypeq \metalevel{\TPI{U}[\Delta]{l}{l'}{T}}{\TPI{U}[\Delta']{l}{l'}{T'}}{l'}$, then $\lpequiv \proglevel \Delta{\Delta'}$ and
    $\lttypeq[\Psi, U : \DTyp[\Delta] \proglevel l] \metalevel{T}{T'}{l'}$.
  \item If $\lttypeq \metalevel{\UPI \ell l T}{\UPI \ell l{T'}}\omega$, then
    $\lttypeq[\Psi][\Gamma][L, \vect \ell] \metalevel{T}{T'}{l}$.
  \end{itemize}
\end{lemma}
\begin{proof}
  By fundamental theorems, we pass in all identity substitutions and then we can
  extract this fact from the logical relations of types. 
\end{proof}

We follow \citet{hu_jang_pientka_2023} and give us the consistency proof of \delamlang.
\begin{lemma}[Consistency]
  There is no term $t$ that satisfies this typing judgment:
  \[
    \lttyping[\cdot][\cdot][\ell] \metalevel t{\PI{1 + \ell}{\ell}x{\Ty \ell}x}{1 + \ell}
  \]
\end{lemma}
That is, there is not a generic way to construct a term of an arbitrary type.
\begin{proof}
  The lemma effectively asks to reject the following derivation after applying the
  fundamental theorems:
  \[
    \lttyping[\cdot][x : \Ty \ell \at{1 + \ell}][\ell] \metalevel t x \ell
  \]
  Now by the logical relations of the neutral type $x$, we know that $t$ must also be
  neutral, so we now move on to rejecting 
  \[
    \lttyping[\cdot][x : \Ty \ell \at{1 + \ell}][\ell] \metalevel \varlevel x \ell
  \]
  Then we do induction on $\varlevel$. %
  It is impossible to do any operation on $\varlevel$ but to refer to $x$ ultimately, but then
  $x$ has type $\Ty \ell$ which cannot be equivalent to $x$, as they both have reached
  normal forms and it is not possible to relate any two distinguished normal forms in
  any case by the logical relations. %
\end{proof}

\subsection{Second Instantiation: Convertibility Checking}

In this section, we perform our second instantiation by specifying the generic
equivalence to be the convertibility checking judgments. %
This second instantiation is not very immediate for some laws for the generic
equivalence, so we must do some verification, using the results from the first
instantiation. %

\begin{itemize}
\item $\lttypgneeq i V{V'} l$ is $\tconvtypne i V{V'}l$.
\item $\lttypgeq i {T}{T'} l$ is $\tconvtyp i T{T'} l$.
\item $\lttrmgneeq i \nu{\nu'} T l$ is $\tconvtrmne i \nu{\nu'}{T'} l$ and
  $\lttypeq i T{T'} l$.
\item $\lttrmgeq i t{t'} T l$ is $\tconvtrm i t{t'}T l$.
\end{itemize}

Most laws are quite straightforward, except the conversion laws and the PER laws. %
Let us consider the conversion laws first. %
We must show that the convertibility checking algorithm is invariant under contexts
and types. %
Intuitively, this should be true, but it is not very easy to prove, until we obtain
the injectivity lemma after the first instantiation.

\begin{lemma}[Conversion] $ $
  \begin{itemize}
  \item if $\tconvtyp i T{T'} l$, $\gequiv\Phi\Psi$ and
    $\lpequiv[\Phi]{i}\Delta\Gamma$, then $\tconvtyp[\Phi][\Delta] i T{T'} l$.
  \item if $\tconvtypnf i W{W'} l$, $\gequiv\Phi\Psi$ and
    $\lpequiv[\Phi]{i}\Delta\Gamma$, then $\tconvtypnf[\Phi][\Delta] i W{W'} l$.
  \item if $\tconvtypne i V{V'}l$, $\gequiv\Phi\Psi$ and
    $\lpequiv[\Phi]{i}\Delta\Gamma$, then $\tconvtypne[\Phi][\Delta] i V{V'}l$.
  \item if $\tconvctx i \Gamma \Delta$ and $\gequiv\Phi\Psi$, then $\tconvctx[\Phi] i \Gamma \Delta$.
  \item if $\tconvtrm i t{t'}T l$, $\gequiv\Phi\Psi$, $\lpequiv[\Phi]{i}\Delta\Gamma$
    and $\lttypeq[\Phi][\Delta] i{T'} T l$, then
    $\tconvtrm[\Phi][\Delta] i t{t'}{T'} l$.
  \item if $\tconvtrmnf i w{w'}W l$, $\gequiv\Phi\Psi$, $\lpequiv[\Phi]{i}\Delta\Gamma$
    and $\lttypeq[\Phi][\Delta] i{W'}W l$, then $\lttyequiv[\Phi][\Delta] i w{w'}{W'} l$;
  \item if $\tconvtrmne i \nu{\nu'}T l$, $\gequiv\Phi\Psi$ and
    $\lpequiv[\Phi]{i}\Delta\Gamma$, then $\tconvtrmne[\Phi][\Delta] i \nu{\nu'}{T'} l$
    and $\lttypeq i{T'}{T}l$ for some $T'$. 
  \item if $\tconvtrmnee i \nu{\nu'}W l$, $\gequiv\Phi\Psi$ and
    $\lpequiv[\Phi]{i}\Delta\Gamma$, then $\lttyequiv[\Phi][\Delta] i \nu{\nu'}{W'} l$
    and $\lttypeq i{W'}{W}l$ for some $W'$.
  \item if $\tconvsub i{\delta}{\delta'}\Delta$, $\gequiv\Phi\Psi$,
    $\lpequiv[\Phi]{i}{\Gamma'}\Gamma$ and $\lpequiv[\Phi]{i}{\Delta'}\Delta$, then
    $\tconvsub[\Phi][\Gamma'] i{\delta}{\delta'}{\Delta'}$.
  \end{itemize}
\end{lemma}
\begin{proof}
  We do induction on all derivations. %
  Most cases are immediate. %
  When we get under binders, we need to extend the context equivalences, in which case
  we should use the soundness lemma to obtain the well-formedness of the type that
  needs to be extended to the contexts. %

  We consider a few cases.
  \begin{itemize}[label=Case]
  \item
    \begin{mathpar}
      \inferrule
      { T \reds W \\ \ttrmreds i t w T l \\ \ttrmreds i
        {t'}{w'}T l \\ \tconvtrmnf i w{w'}W l}
      {\tconvtrm i t{t'}T l}
    \end{mathpar}
    The most important thing to notice is that by $\lttypeq[\Phi][\Delta] i{T'} T l$
    and $T \reds W$, we know that $T' \reds W'$ for some $W'$. Therefore we can prove
    this case using IH. %
    Other premises are satisfied by syntactic context equivalence lemmas. %
    
  \item
    \begin{mathpar}
      \inferrule
      {\lttypwf i S l \\ \lttyping i w {\PI{l}{l'}x{S}{T}}{l \sqcup l'} \\ \lttyping i {w'}{\PI{l}{l'}x{S}{T}}{l \sqcup l'} \\ \tconvtrm[\Psi][\Gamma, x : S \at{l}] i {\APP{w}{l}{l'}x{S}{T}{x}}{\APP{w'}{l}{l'}x{S}{T}{x}}{T}{l'}}
      {\tconvtrmnf i{w}{w'}{\PI{l}{l'}x{S}{T}}{l \sqcup l'}}
    \end{mathpar}
    In this case, we know $\lttypeq[\Phi][\Delta] i{W'}{\PI{l}{l'}x{S}{T}}{l \sqcup
      l'}$. %
    By the fundamental theorems, we know that $W'$ can only reduce to some $\Pi$
    type. %
    Since $W'$ is already normal, it is only possible for $W'$ to be some $\Pi$
    type. %
    Say $W' = {\PI{l}{l'}x{S'}{T'}}$. %
    Then by injectivity, we know $S \approx S'$ and $T \approx T'$. %
    We obtain our goal by extending the local contexts. 
    
  \item
    \begin{mathpar}
      \inferrule
      {\tconvtrmnee i \nu{\nu'}{W} l}
      {\tconvtrmnf i \nu{\nu'}V l}
    \end{mathpar}
    In this case, the type equivalence is irrelevant. %
    When hitting neutral types, we simply ignore the type $W$ inferred by
    $\tconvtrmnee i \nu{\nu'}{W} l$.

  \item
    \begin{mathpar}
      \inferrule
      {\tconvtrmne i \nu{\nu'}{T}{l} \\ T \reds W}
      {\tconvtrmnee i \nu{\nu'}W l}
    \end{mathpar}
    In this case, we simply apply IH. Then we know
    \[
      T \reds W \tand T' \reds W'
    \]
    Our goal is to show $W \approx W'$. %
    But this is immediate from the determinacy lemma of multi-step reduction and the
    fundamental theorems. %

  \item
    \begin{mathpar}
      \inferrule
      {\lpjudge i \Gamma \\ x : T \at l \in \Gamma}
      {\tconvtrmne i x{x}T l}
    \end{mathpar}
    The goal is given by the equivalence between $\Gamma$ and $\Delta$. 
    
  \item
    \begin{mathpar}
      \inferrule
      {\tyequiv[L]{l_1}{l_3}\Level  \\ \tyequiv[L]{l_2}{l_4}\Level \\ \tconvtyp i S{S'}{l_1}
        \\
        \tconvtyp[\Psi][\Gamma, x : S \at{l_1}] i{T}{T'}{l_2} \\
        \tconvtrmnee i \nu{\nu'}{\PI {l_1}{l_2} x{S''}{T''}}{l_1 \sqcup l_2} \\ \tconvtrm i s{s'} S{l_1}}
      {\tconvtrmne i{\APP \nu {l_1}{l_2} x S T s}{\APP{\nu'}{l_3}{l_4} x {S'}{T'}{s'}}{T[s/x]}{l_2}}
    \end{mathpar}
    By IH, we know from
    $\tconvtrmnee i \nu{\nu'}{\PI {l_1}{l_2} x{S''}{T''}}{l_1 \sqcup l_2}$ that there
    must be $W'$, so that
    \[
      \lttypeq i{W'}{\PI {l_1}{l_2} x{S''}{T''}}{l_1 \sqcup l_2}
    \]
    By the fundamental theorems, we know $W'$ must be some $\Pi$ types. %
    The return type is fixed so we do not have to do anything. %
  \end{itemize}
\end{proof}

From the conversion lemma, we see that both $\tconvtrmne i \nu{\nu'}T l$ and
$\tconvtrmnee i \nu{\nu'}W l$ return some equivalent types because they are in fact
inference steps. %
Therefore, it makes sense when we instantiate $\lttrmgneeq i \nu{\nu'} T l$, we hide a
syntactic equivalence judgment inside. %

For the PER laws, we see that the difficulties primarily come from the convertibility
checking of neutrals, because again they are inference steps. %
Moreover, due to their left-biased design, in general it is not true that the inferred
types can be replaced by their equivalence. %
However, since we are hiding an equivalence judgment in $\lttrmgneeq i \nu{\nu'} T l$
during instantiation, following the same principle as the conversion lemma, we are
able to erase the effect of the left bias and establish the PER laws. %

The instantiation immediately gives us the convertibility lemma once we apply the
escape lemma:
\begin{theorem}[Convertibility] $ $
  \begin{itemize}
  \item If $\lttypwf i T l$, then $\tconvtyp i T T l$. 
  \item If $\lttypeq i T{T'} l$, then $\tconvtyp i T{T'} l$. 
  \item If $\lttyping i t T l$, then $\tconvtrm i t t T l$.
  \item If $\lttyequiv i t{t'} T l$, then $\tconvtrm i t{t'} T l$.
  \end{itemize}
\end{theorem}
In particular, we see that the convertibility checking algorithm is both sound and
complete w.r.t. the syntactic equivalence judgments. %
The decidability of type checking requires the following lemma by attempting to relate
two reflexively convertible types or terms. 

\begin{lemma}[Decidability] $ $
  \begin{itemize}
  \item if $\tconvtyp[\Phi][\Delta] i T{T} l$, $\tconvtyp i {T'}{T'} l$,
    $\gequiv\Phi\Psi$ and $\lpequiv[\Phi]{i}\Delta\Gamma$, then whether
    $\tconvtyp[\Phi][\Delta] i T{T'} l$ is decidable.
  \item if $\tconvtypnf[\Phi][\Delta] i W{W} l$, $\tconvtypnf i{W'}{W'} l$,
    $\gequiv\Phi\Psi$ and $\lpequiv[\Phi]{i}\Delta\Gamma$, then whether
    $\tconvtypnf[\Phi][\Delta] i W{W'} l$ is decidable.
  \item if $\tconvtypne[\Phi][\Delta] i V{V}l$, $\tconvtypne i {V'}{V'}l$,
    $\gequiv\Phi\Psi$ and $\lpequiv[\Phi]{i}\Delta\Gamma$, then whether
    $\tconvtypne[\Phi][\Delta] i V{V'}l$ is decidable. %
  \item if $\tconvctx[\Phi] i \Gamma \Gamma$, $\tconvctx i \Delta \Delta$ and
    $\gequiv\Phi\Psi$, then whether $\tconvctx[\Phi] i \Gamma \Delta$ is decidable. 
  \item if $\tconvtrm[\Phi][\Delta] i t{t}T l$, $\tconvtrm i{t'}{t'}{T} l$,
    $\gequiv\Phi\Psi$ and $\lpequiv[\Phi]{i}\Delta\Gamma$, then whether
    $\tconvtrm[\Phi][\Delta] i t{t'}{T} l$ is decidable.
  \item if $\tconvtrmnf[\Phi][\Delta] i w{w}W l$, $\tconvtrmnf i {w'}{w'}{W} l$,
    $\gequiv\Phi\Psi$ and $\lpequiv[\Phi]{i}\Delta\Gamma$, then whether
    $\lttyequiv[\Phi][\Delta] i w{w'}{W} l$ is decidable.
  \item if $\tconvtrmne[\Phi][\Delta] i \nu{\nu}T l$,
    $\tconvtrmne i {\nu'}{\nu'}{T'} l$, $\gequiv\Phi\Psi$ and
    $\lpequiv[\Phi]{i}\Delta\Gamma$, then whether
    $\tconvtrmne[\Phi][\Delta] i \nu{\nu'}{T''} l$ for some $T''$ is decidable.
  \item if $\tconvtrmnee[\Phi][\Delta] i \nu{\nu}W l$,
    $\tconvtrmnee i {\nu'}{\nu'}{W'} l$, $\gequiv\Phi\Psi$ and
    $\lpequiv[\Phi]{i}\Delta\Gamma$, then whether
    $\lttyequiv[\Phi][\Delta] i \nu{\nu'}{W''} l$ is decidable.
  \item if $\tconvsub[\Phi][\Gamma'] i{\delta}{\delta}{\Delta}$,
    $\tconvsub i{\delta'}{\delta'}{\Delta}$, $\gequiv\Phi\Psi$ and
    $\lpequiv[\Phi]{i}{\Gamma'}\Gamma$, then whether
    $\tconvsub[\Phi][\Gamma'] i{\delta}{\delta'}{\Delta}$ is decidable. %
  \end{itemize}
\end{lemma}
\begin{proof}
  We do a mutual induction on the first derivations and then invert the second ones. %
  We can reject most cases when they have different root derivations. %
  We consider a few cases.
  \begin{itemize}[label=Case]
  \item
    \begin{mathpar}
      \inferrule
      {\ttypreds[\Phi][\Delta] i T W l \\\\ \tconvtypnf i W{W} l}
      {\tconvtyp[\Phi][\Delta] i T{T} l}

      \inferrule
      {\ttypreds i{T'}{W'} l \\\\ \tconvtypnf i{W'}{W'} l}
      {\tconvtyp i T{T'} l}
    \end{mathpar}
    In this case, we apply IH to decide whether $W$ and $W'$ are convertible. %

  \item
    \begin{mathpar}
      \inferrule
      {\tconvtyp[\Phi][\Delta] i S{S} l \\ \tconvtyp[\Phi][\Delta, x : S \at l] i T{T}{l'}}
      {\tconvtypnf[\Phi][\Delta] i{\PI l{l'}x{S}{T}}{\PI l{l'}x{S}{T}}{l \sqcup l'}}

      \inferrule
      {\tconvtyp i{S'}{S'} l \\ \tconvtyp[\Psi][\Gamma, x : S' \at l] i{T'}{T'}{l'}}
      {\tconvtypnf i{\PI l{l'}x{S'}{T'}}{\PI l{l'}x{S'}{T'}}{l \sqcup l'}}
    \end{mathpar}
    We can decide whether $S$ and $S'$ are convertible by IH. %
    When we decide $T$ and $T'$, we must extend the equivalent local contexts. %
    
  \item
    \begin{mathpar}
      \inferrule
      {T \reds W \\ \ttrmreds[\Phi][\Delta] i t w T l \\\\  \tconvtrmnf[\Phi][\Delta] i w{w}W l}
      {\tconvtrm[\Phi][\Delta] i t{t}T l}

      \inferrule
      {T' \reds W' \\ \ttrmreds i{t'}{w'}{T'} l \\\\  \tconvtrmnf i {w'}{w'}{W'} l}
      {\tconvtrm i {t'}{t'}{T} l}
    \end{mathpar}
    By fundamental theorems, we know $W \approx W'$. %
    Then by IH, we can decide whether $w$ and $w'$ are convertible. %
    
  \item
    \begin{mathpar}
      \inferrule
      {\lttypwf[\Phi][\Delta] i S l \\ \lttyping[\Phi][\Delta] i w
        {\PI{l}{l'}x{S}{T}}{l \sqcup l'} \\
        \tconvtrm[\Phi][\Delta, x : S \at{l}] i {\APP{w}{l}{l'}x{S}{T}{x}}{\APP{w}{l}{l'}x{S}{T}{x}}{T}{l'}}
      {\tconvtrmnf[\Phi][\Delta] i{w}{w}{\PI{l}{l'}x{S}{T}}{l \sqcup l'}}

      \inferrule
      {\lttypwf i{S} l \\ \lttyping i{w'}{\PI{l}{l'}x{S}{T}}{l \sqcup l'} \\
        \tconvtrm[\Psi][\Gamma, x : S \at{l}] i {\APP{w'}{l}{l'}x{S}{T}{x}}{\APP{w'}{l}{l'}x{S}{T}{x}}{T}{l'}}
      {\tconvtrmnf i{w'}{w'}{\PI{l}{l'}x{S}{T}}{l \sqcup l'}}
    \end{mathpar}
    Again, it is quite immediate by IH. %
    
  \item
    \begin{mathpar}
      \inferrule
      {\tconvtrmnee[\Phi][\Delta] i \nu{\nu}{W} l}
      {\tconvtrmnf[\Phi][\Delta] i \nu{\nu}V l}

      \inferrule
      {\tconvtrmnee i {\nu'}{\nu'}{W'} l}
      {\tconvtrmnf i {\nu'}{\nu'}{V} l}
    \end{mathpar}
    By IH, we know $\tconvtrmnee[\Phi][\Delta] i \nu{\nu'}{W''} l$ for some $W''$. 

  \item
    \begin{mathpar}
      \inferrule
      {\tconvtrmne[\Phi][\Delta] i \nu{\nu}{T}{l} \\ T \reds W}
      {\tconvtrmnee[\Phi][\Delta] i \nu{\nu}W l}

      \inferrule
      {\tconvtrmne i {\nu'}{\nu'}{T'}{l} \\ T' \reds W'}
      {\tconvtrmnee i {\nu'}{\nu'}{W'} l}
    \end{mathpar}
    By IH, we have $\tconvtrmne i {\nu'}{\nu'}{T''}{l}$ for some $T''$. %
    $T''$ will reduce to some normal form by the fundamental theorems. %
  \end{itemize}
\end{proof}

\begin{theorem}[Decidability of Convertibility] $ $
  \begin{itemize}
  \item If $\lttypwf i T l$ and $\lttypwf i{T'} l$, then whether $\tconvtyp i T{T'} l$
    is decidable.
  \item If $\lttyping i t T l$ and $\lttyping i{t'} T l$, then whether
    $\tconvtrm i t{t'}T l$ is decidable.
  \end{itemize}
\end{theorem}
\begin{proof}
  First we use the fundamental theorems from the second instantiation to show that
  both types (or terms, resp.) are reflexively convertible. %
  Then we use the decidability lemma above. %
\end{proof}

At this point, we have justified the decidability of convertibility checking of
\delamlang and therefore conclude our investigations. %

% \bibliography{ref}


\end{document}